\documentclass{article}
\usepackage{graphicx} 
\usepackage{amsmath,amsthm,mathrsfs}
\numberwithin{equation}{section}
\usepackage[a4paper,top=3.2cm,width=16cm, bottom=3.2cm, left=2cm]{geometry}
\usepackage{tikz, tikz-cd, tikz-3dplot}
\tikzcdset{scale cd/.style={every label/.append style={scale=#1},
    cells={nodes={scale=#1}}}}
\usepackage{quiver}
\usepackage{hyperref}
\usepackage{float}
\usepackage{physics}
\usepackage{yhmath}
\usepackage{wasysym}
\usepackage{stmaryrd}
\usepackage{amsfonts}
\usepackage{subcaption}
\usepackage{stmaryrd}
\usepackage[bbgreekl]{mathbbol}
\usepackage{relsize}
\usepackage{hyperref}
\hypersetup{
    colorlinks=true,
    linkcolor=blue,
    filecolor=magenta,      
    urlcolor=cyan,
    citecolor=blue
}
\DeclareSymbolFontAlphabet{\mathbbm}{bbold}
\DeclareSymbolFontAlphabet{\mathbb}{AMSb}%

\newtheorem{thm}{Theorem}
\newtheorem{prp}[thm]{Proposition}
\newtheorem{cor}[thm]{Corollary}
\newtheorem{lma}[thm]{Lemma}

\newtheorem{dfn}[thm]{Definition}

\theoremstyle{remark}
\newtheorem{rmk}[thm]{Remark}
\newtheorem*{exmp*}{Example}

\newcommand{\Sk}{\operatorname{Sk}}
\newcommand{\SkAlg}{\operatorname{SkAlg}}
\newcommand{\bfU}{\boldsymbol{U}}
\newcommand{\CA}{\mathcal{A}}
\newcommand{\CB}{\mathcal{B}}
\newcommand{\CI}{\mathcal{I}}
\newcommand{\CL}{\mathcal{L}}
\newcommand{\CM}{\mathcal{M}}
\newcommand{\CN}{\mathcal{N}}
\newcommand{\CO}{\mathcal{O}}
\newcommand{\CP}{\mathcal{P}}
\newcommand{\CT}{\mathcal{T}}
\newcommand{\CW}{\mathcal{W}}
\newcommand{\IR}{\mathbb{R}}
\newcommand{\IC}{\mathbb{C}}
\newcommand{\IZ}{\mathbb{Z}}

\newcommand{\be}{\begin{equation}}
\newcommand{\ee}{\end{equation}}

\definecolor{amber}{rgb}{1.0, 0.49, 0.0}
\definecolor{Green}{rgb}{0.0, 0.5, 0.0}
\definecolor{purple}{rgb}{0.7,0,0.7}

\newcommand{\Qsym}[2]{{\mathfrak{S}}({#1},{#2})}

\usepackage[draft]{say}

\usepackage{subfiles} 

\usepackage{authblk} 
\usepackage{orcidlink}

\title{Quivers and BPS states in 3d and 4d}

\author[1]{Piotr Kucharski \orcidlink{0000-0002-9599-5658}\thanks{piotr.kucharski@mimuw.edu.pl}}

\author[2,3,4]{Pietro Longhi \orcidlink{0000-0002-5558-0386}\thanks{pietro.longhi@physics.uu.se}}

\author[5]{Dmitry Noshchenko \orcidlink{0000-0002-9639-5603}\thanks{dsnoshchenko@stp.dias.ie}}

\author[6,7]{\break Sunghyuk Park \orcidlink{0000-0002-6132-0871}\thanks{sunghyukpark@math.harvard.edu}}

\author[8]{Piotr Su{\l}kowski \orcidlink{0000-0002-6176-6240}\thanks{psulkows@fuw.edu.pl}}

\affil[1]{Institute of Mathematics, University of Warsaw, ul. Banacha 2, 02-097 Warsaw, Poland}
\affil[2]{Department of Physics and Astronomy, Uppsala University, Box 516, 751 20 Uppsala, Sweden}
\affil[3]{Department of Mathematics, Uppsala University, Box 480, 751 06 Uppsala, Sweden}
\affil[4]{Centre for Geometry and Physics, Uppsala University, Box 516, 751 20 Uppsala, Sweden}
\affil[5]{School of Theoretical Physics, Dublin Institute for Advanced Studies,\break 10 Burlington Road, Dublin 4, D04 C932, Ireland}
\affil[6]{Department of Mathematics, Harvard University, Cambridge, MA 02138, USA}
\affil[7]{Center of Mathematical Sciences and Applications, Harvard University, Cambridge, MA 02138, USA}
\affil[8]{Faculty of Physics, University of Warsaw, ul. Pasteura 5, 02-093 Warsaw, Poland}

\date{}

\begin{document}

\hfill
\begin{tabular}{r}
UUITP-22/25 \\
DIAS-STP-25-19
\end{tabular}
{\let\newpage\relax\maketitle}

\begin{abstract}
We propose a symmetrization relation between BPS quivers encoding 4d $\mathcal{N}=2$ theories and symmetric quivers associated to 3d $\mathcal{N}=2$ theories. 
We analyse in detail the symmetrization of BPS quivers for a series of $A_m$ Argyres-Douglas theories by engineering 3d-4d systems in geometric backgrounds involving appropriate 3-manifolds and Riemann surfaces. 
We discuss properties of these geometric backgrounds and derive the corresponding quiver partition functions from the perspective of skein modules, which forms the foundation of the symmetrization map for the minimal chamber.
We also prove that the structure of wall-crossing in 4d $A_m$ Argyres-Douglas theories is isomorphic to the structure of unlinking of symmetric quivers encoding their partner 3d theories, which allows for a proper definition of the symmetrization map outside the minimal chamber. 
Finally, we show that the Schur indices of 4d theories are captured by symmetric quivers that include symmetrization of 4d BPS quivers. 
\end{abstract}

\newpage

\tableofcontents

\newpage

\section{Introduction and summary}

The study of BPS states in field theories with extended supersymmetry has been an active area of research in the past decades. 
Apart from deep physical motivations, it makes contact with important developments in contemporary mathematics. 
One interesting phenomenon in this context is that of wall-crossing, which is the~statement that the stability of BPS states depends on certain moduli of the theory under consideration, so that these may form bound states and their spectrum might jump when values of the~moduli cross so-called walls of marginal stability. 
The pattern of these jumps is governed by wall-crossing formulae of various kinds \cite{Cecotti:1992rm, Joyce:2008pc,KS0811}, which state that a certain operator constructed from the data of stable BPS states is invariant upon crossing walls of marginal stability. 
In the context of 4d $\mathcal{N}=2$ theories, the relevant operator was introduced by Kontsevich and Soibelman \cite{KS0811, Gaiotto:2009hg, Gaiotto:2010be}, and its trace is known to encode other important observables, such as the Schur index of the theory and other specializations of the superconformal index \cite{Cecotti:2010fi, Cordova:2015nma, Cecotti:2015lab}. 

Properties of BPS states in theories in various dimensions can be encoded by various types of quivers. 
For 4d $\mathcal{N}=2$ theories, such quivers are referred to as BPS quivers and have been studied extensively since the seminal works \cite{Alim:2011ae,Alim:2011kw}. 
Quiver vertices correspond to a certain type of fundamental BPS states, and arrows encode interactions among them determined by the Dirac pairing. 
On the other hand, BPS states of 3d $\mathcal{N}=2$ theories have also been found to admit a quiver description, albeit captured by symmetric quivers (meaning that for each of their arrows there is an arrow in the opposite direction) \cite{KRSS1707short,KRSS1707long}. 
In this case, quiver vertices correspond to $U(1)$ factors of the gauge group, while arrows encode (mixed, effective) Chern-Simons couplings \cite{Ekholm:2018eee,Ekholm:2019lmb}. 

In this work, we establish a direct relationship between these two classes of theories and their corresponding BPS quivers. 
For a certain class of 4d $\mathcal{N}=2$ theories, we identify corresponding 3d $\mathcal{N}=2$ theories with a matching set of observables, and whose structure is encoded by symmetric quivers closely related to the BPS quiver of the corresponding 4d theory. 
In the simplest settings, the 3d $\mathcal{N}=2$ symmetric quiver arises as the symmetrization of the 4d $\mathcal{N}=2$ BPS quiver, i.e., by adding an arrow in the opposite direction to each arrow of the latter. 
More generally, the structure of a symmetric quiver $Q$ may be more involved and it depends also on a choice of stability conditions in the 4d theory. 
We thus define the \emph{symmetrization map} $\mathfrak{S}$ between 4d BPS quivers with stability data and 3d symmetric quivers. 
Moreover, we also introduce the CPT-doubled symmetrization map $\mathfrak{S}^{CPT}$ that involves both BPS and anti-BPS states of the 4d theory, and produces CPT-doubled symmetric quivers $Q^{CPT}$ that contain twice as many nodes. 
These maps act as
\begin{equation}
(Q_{4\text{d}},\text{stab.\,data}) \stackrel{\mathfrak{S}}{\longmapsto}Q\, , \qquad \qquad (Q_{4\text{d}},\text{stab.\,data}) \stackrel{\mathfrak{S}^{CPT}}{\longmapsto}Q^{CPT}\, , \label{Sigma-maps}
\end{equation}
and their construction is one of the main goals of this work. 
We show that appropriately engineered symmetric quivers, and thus also the 3d $\mathcal{N}=2$ theories they represent, encode information about BPS states in the partner 4d theories, their wall-crossing, as well as Schur indices.

The combinatorial definition of symmetrization map $\mathfrak{S}$  in  \eqref{Sigma-maps} that we introduce (see definition in Section \ref{sec-wall-crossing}) applies to BPS quivers $Q_{4d}$ of any 4d $\mathcal{N}=2$ Argyres-Douglas theory of type $A_m$ with $u$ belonging to any BPS chamber of the Coulomb branch. However the general idea behind our definition and its properties are expected to generalize, and we illustrate this in the case of the Argyres-Douglas $D_4$ theory.
In addition to the combinatorial definition we offer direct geometric and physical constructions of this map for $A_m$ Argyres-Douglas theories with $u$ in the minimal BPS chamber in Sections \ref{sec:3d4d} and \ref{sec:algebra}. 
The geometric construction hinges crucially on the realization of these QFTs as class $S$ theories of type $A_1$, since these admit a UV description based on ideally triangulated Riemann surfaces \cite{Gaiotto:2009hg}, and since their BPS spectrum defines a canonical 3-manifold constructed by attaching a tetrahedron for each stable state \cite{CCV}.
The physical derivation is tied to the geometric data, since the triangulation by tetrahedra defines the 3d-3d dual QFT for the 3-manifolds defined by BPS states of 4d $\mathcal{N}=2$ Argyres-Douglas theories of type $A_m$ in their minimal chambers. 
The 3d quiver $Q$ defines an IR description of the corresponding 3d $\mathcal{N}=2$ theory.
While in this work we mainly focus on $A_m$ Argyres-Douglas theories, we stress that the general ideas behind our definition of symmetrization map and behind its geometric and physical constructions have a broader scope. 
In particular we foresee a feasible generalization of this 4d-3d relation to any 4d QFT in class $S$ of type $A_1$, and even to certain theories of higher rank, see Section \ref{sec:future-work}. It would also be interesting to generalize geometric constructions discussed in this work to  analogous 4d-3d relation between Kronecker quivers (consiting of two nodes connected by an arbitrary number of arrows) and symmetric quivers, analysed in a parallel work \cite{Bryan:2025mwi} in the context of wild wall-crossing phenomena.

The relation between 4d $\mathcal{N}=2$ and 3d $\mathcal{N}=2$ QFTs arising in the context of 2d-4d and 3d-3d correspondences has a long history, starting with \cite{Terashima:2011qi, Dimofte:2011ju}.
Our work establishes a precise relationship between the above-mentioned classes of models and their BPS spectra. 
A related discussion of the relation between 4d BPS states and certain 3d QFTs was recently developed in \cite{Gaiotto:2024ioj}, whose work focuses on SCFTs that admit a supersymmetry enhancement in the infrared, see also \cite{Gang:2018huc, Gang:2021hrd,  Gang:2023rei, Gang:2024loa}. 
A connection to our work comes through the CPT-doubled symmetrization map $\mathfrak{S}^{CPT}$. Its output is a 3d quiver $Q^{CPT}$ whose motivic DT partition function computes the Schur index of the 4d QFT, which in turn is related to vacuum characters of vertex operator algebras. Related construction of these indices from 4d BPS states have been recently explored by \cite{Kim:2024dxu, Go:2025ixu} in a closely related setting. 
In recent years the relation between 4d SCFTs and 3d TQFTs (and therefore their boundary 2d VOAs) has been studied intensively by a number of works \cite{Dedushenko:2018bpp, Dedushenko:2023cvd, ArabiArdehali:2024ysy} which also consider the Argyres-Douglas theories that we consider in this paper.

\subsection*{Summary of main results}

\paragraph{Physics and topology of 3d-4d systems that lead to symmetrization of quivers.} 

Both 3d and 4d $\mathcal{N}=2$ theories that we study arise from a geometric construction, involving a~pair of M5-branes partially wrapped on a 3-manifold with boundary \cite{Terashima:2011qi, Dimofte:2011ju, Chung:2014qpa}. 
The 4d theories arise from compactification on the boundary Riemann surfaces, and their BPS quivers are encoded by ideal triangulations of the latter \cite{2013arXiv1302.7030B}. 
The 3d theories are encoded by the 3-manifolds, and admit effective Lagrangian descriptions encoded by decompositions of the latter into ideal tetrahedra \cite{Dimofte:2011ju}.
While these properties have been known for some time, we reformulate and present them in terms of skein modules, which provide a novel, interesting perspective. 
In particular, in Section \ref{Topology}, we show how skein modules -- for 3-manifolds built from a sequence of flips on an ideal triangulation of a~surface -- naturally encode Nahm sums and the pentagon relation (\ref{pentagon}); the latter corresponds to the 2–3 Pachner move and admits a reinterpretation in terms of an unlinking operation. 
Furthermore, we establish a~precise match between Lagrangian descriptions of 3d $\mathcal{N}=2$ theories defined, respectively, by tetrahedron decompositions \cite{Dimofte:2011ju} and by 3d symmetric quivers \cite{Ekholm:2018eee}. 
We identify the origin of this 4d-3d correspondence in the topological nature of 4d theta-terms in abelian low-energy descriptions of 4d $\mathcal{N}=2$ theories, which are directly related to Dirac pairings of 4d BPS particles via Witten's effect, and whose restriction to the boundary gives rise to Chern-Simons couplings encoded by the 3d symmetric quiver.

While we expect that relations between 4d and 3d theories and corresponding quivers hold in general, we analyze various aspects of these relations across different classes of theories. 
For $A_m$ Argyres-Douglas theories, we provide an explicit construction (generalizing those in \cite{CCV}) of a 3-manifold that relates the chamber with minimal number of BPS states to 3d theory, and we provide geometric and physical derivations of the corresponding symmetric quiver. 
For a much larger class of Argyres-Douglas theories, represented by BPS quivers whose nodes are connected by at most one arrow, we find a representation of their wall-crossing phenomena in terms of symmetric quivers, based on algebraic considerations. 
Continuing this analysis, we assign a~quantum torus algebra of rank $m$ to any symmetric quiver with $m$ nodes, in such a way that its motivic generating series can be reproduced as an expectation value of the appropriate operators. 
For another, yet larger class of 4d $\mathcal{N}=2$ theories, we find symmetric quivers whose generating series reproduce Schur indices upon appropriate specialization. 
On the other hand, related phenomena for 4d $\mathcal{N}=2$ theories whose BPS quivers include arbitrary Kronecker subquivers (consisting of two nodes connected by an arbitrary number of arrows, which gives rise to wild wall-crossing \cite{Galakhov:2013oja}) were independently analysed in a recent parallel work \cite{Bryan:2025mwi}.

\paragraph{Two constructions of the symmetrization map.}

The physical and topological considerations summarized above suggest the existence of a symmetrization map between 4d BPS quivers in the minimal chamber and symmetric quivers corresponding to 3d theories. 
One of the main novelties of our work is the generalization of this map beyond the minimal chamber. 
The key ingredient is the isomorphism between the wall crossing relations connecting different chambers of the 4d theory (excluding those with non-trivial superpotential) and unlinking relations connecting dual 3d~theories and corresponding symmetric quivers. 
In consequence, the symmetrization map can be understood in terms of two equivalent constructions (we prove it for $A_m$ Argyres-Douglas theories and their BPS quivers):
\begin{itemize}
    \item One starts from assigning $Q$ -- the symmetrization of $Q_{4\text{d}}$ -- to the minimal chamber and then directly uses the isomorphism between wall-crossing and unlinking (see Sections \ref{sec:pentagon to unlinking}--\ref{sec:wcf-connector theorem}).
    \item The other goes outside the minimal chamber using a path polytope, derived from the oriented exchange graph of $Q_{4\text{d}}$ (this approach is undertaken in Section \ref{sec:connectors and path polytopes}).
\end{itemize}

\paragraph{Quivers and DT invariants.}
Quiver representation theory of a quiver $Q$ assigns a vector space (of dimension $d_i$) to each node~$i$ and a~linear map to each arrow. 
Betti numbers (or their generalizations) of the moduli space of representations are then captured by the motivic Donaldson-Thomas (DT) invariants $\Omega_{\boldsymbol{d},s}$, which are indexed by the~dimension vector $\boldsymbol{d}=(d_1,\ldots,d_{m})\in \mathbb{N}^{m}$ and by the integer $s$ \cite{KS0811,KS1006,2011arXiv1103.2736E,MR1411,FR1512}. 
These invariants can be extracted from the motivic generating series, which for a symmetric quiver $Q$ with trivial potential takes the form
\begin{equation}\label{eq:motivic_generating_series}
    P_Q(\boldsymbol{x},q) =
    \sum_{\boldsymbol{d}}(-q)^{\boldsymbol{d} \cdot Q\cdot\boldsymbol{d}}\frac{\boldsymbol{x}^{\boldsymbol{d}}}{(q^{2};q^{2})_{\boldsymbol{d}}}
    = \sum_{d_1,\dots,d_{m}= 0}^\infty (-q)^{\sum_{i,j=1}^{m} d_i Q_{ij} d_j}\prod_{i=1}^{m}\frac{x_i^{d_i}}{(q^{2};q^{2})_{d_i}},
\end{equation}
where $Q_{ij}(= Q_{ji})$ are components of the adjacency matrix of $Q$ and denote the number of arrows between nodes $i$ and $j$. 
The DT invariants are determined by the~product decomposition of this series into quantum dilogarithms:
\begin{equation}\label{PQ-product}
    P_Q(\boldsymbol{x},q) = \prod_{\boldsymbol{d},s}(\boldsymbol{x}^{\boldsymbol{d}}q^s;q^2)_{\infty}^{\Omega_{\boldsymbol{d},s}} = \prod_{\boldsymbol{d}\in \mathbb{N}^{m}\setminus \boldsymbol{0}} \prod_{s\in\mathbb{Z}} \prod_{k\geq 0} \Big(1 - (x_1^{d_1}\cdots x_{m}^{d_{m}}) q^{2k+s} \Big)^{\Omega_{(d_1,\ldots,d_{m}),s}}.
\end{equation}
These invariants are proven to be integers \cite{2011arXiv1103.2736E}, which reflects their physical interpretation as counts of BPS states. 
Specifically, one role of this motivic generating series has been found in the context of the knots-quivers correspondence \cite{KRSS1707short,KRSS1707long}, where DT invariants determine LMOV invariants (thus implying their integrality), or equivalently count BPS states in a corresponding 3d $\mathcal{N}=2$ theory $T[Q]$~\cite{Ekholm:2018eee}. 
Sometimes we consider expressions of the form (\ref{eq:motivic_generating_series}) with some specialization of $x_i$ in terms of $q$ and in that case we call them Nahm sums. 

\paragraph{Isomorphism between wall-crossing and unlinking. }

An important property of the motivic generating series of $Q$ is its invariance under the operation of \emph{unlinking} $U(ij)$, which produces a quiver $U(ij)Q$ with one pair of arrows between nodes $i$ and $j$ removed and with one extra node whose pattern of arrows is given in (\ref{eq:arbitrary quiver}--\ref{eq:unlinking_definition}) \cite{Ekholm:2019lmb}. 
The extra node in $U(ij)Q$ comes with its own generating parameter, whose specialization to $q^{-1}x_{i}x_{j}$ yields the equality
\begin{equation} \label{unlinking}
P_{Q}(\boldsymbol{x},q)=P_{U(ij)Q}(\boldsymbol{x},q^{-1}x_{i}x_{j},q)\,.
\end{equation}

For a specific case of a quiver with two nodes connected by a pair of arrows, the unlinking operation produces a quiver with three nodes (we denote the node obtained from unlinking $i$ and $j$ as $(ij)$):
\be\label{eq:A_2 unlinking}
\begin{tikzcd}[column sep={0.5cm},row sep={-0.2cm}]
&&&&& \bullet \\
	\bullet && \bullet & {\xrightarrow{U(12)}} & \bullet & 12 & \bullet \\
	1 && 2 && 1 && 2
	\arrow[from=1-6, to=1-6, loop, in=55, out=125, distance=10mm]
	\arrow[curve={height=-6pt}, from=2-1, to=2-3]
	\arrow[curve={height=-6pt}, from=2-3, to=2-1]
\end{tikzcd}
\ee
In this case, (\ref{unlinking}) can be regarded as a reinterpretation of the pentagon relation for quantum dilogarithms $\Psi(x):= (qx;q^2)_\infty^{-1} $ (for details, see Section \ref{sec:pentagon to unlinking}):
\begin{equation}  \label{pentagon}
    \Psi(-X_{\alpha_1})\Psi(-X_{\alpha_2}) = \Psi(-X_{\alpha_2})\Psi(-X_{\alpha_1+\alpha_2})\Psi(-X_{\alpha_1}),
\end{equation}
where the operator arguments form a quantum torus algebra with commutation relations
\begin{equation} \label{Xrelation}
X_{\alpha_1} X_{\alpha_2} = q^{2\langle \alpha_1,\alpha_2\rangle} X_{\alpha_2} X_{\alpha_1}
\end{equation}
with $\langle \alpha_1,\alpha_2\rangle=1$.
On the other hand, such quantum torus algebras are associated to 4d $\mathcal{N}=2$ theories, with $\alpha_i$ representing charges of BPS states and $\langle \alpha_i,\alpha_j\rangle$ their Dirac pairing. 
From this perspective, the pentagon identity (\ref{pentagon}) is a prototype example of wall-crossing identity, in this case for the $A_2$ Argyres-Douglas theory, whose BPS quiver is
\be\label{eq:A_2 dynkin quiver}
\begin{tikzcd}[column sep={0.5cm},row sep={-0.2cm}]
\bullet && \bullet \\
{\alpha_1} && {\alpha_2}
\arrow[from=1-1, to=1-3]
\end{tikzcd}\,.
\ee
Both sides of \eqref{pentagon} correspond to the Kontsevich-Soibelman operator, constructed from a~number of quantum dilogarithms in one-to-one correspondence with stable BPS states in two chambers of moduli space, multiplied in the order of the phases of their BPS central charges. 
Specifically, \eqref{pentagon} can be read as reverting the phase order from $\arg \alpha_1 < \arg \alpha_2$ to $\arg \alpha_2 < \arg \alpha_1$ and creating a~new charge $\alpha_{1}+\alpha_2$.
In general (for other 4d $\mathcal{N}=2$ theories) Kontsevich-Soibelman operators take the form 
\begin{equation}  \label{KS-intro}
\prod^\curvearrowleft_{\alpha}
\prod_{s\in \mathbb{Z}} \Psi(-q^sX_{\alpha})^{\Omega^{4\text{d}}_{\alpha,s}},
\end{equation}
with $\alpha$ labelling charges of 4d BPS states with degeneracies $\Omega^{4\text{d}}_{\alpha,s}$\footnote{Note that for $A_m$ Argyres-Douglas theories, all BPS degeneracies $\Omega^{4\text{d}}_{\alpha,s}$ are equal to 1.} in a given chamber, and $s$ their spin. 
We relate pentagon identities such as (\ref{pentagon}) to the unlinking operation (\ref{unlinking}) by finding a \emph{3d-4d homomorphism} between the quantum torus algebra of $X_{\bullet}$ and a quantum torus algebra of twice the rank that is realized by operators $\hat{x}_{i}$ and $\hat{y}_{j}$ that satisfy the relations
\begin{equation}
\hat{y}_{i}\hat{x}_{j}=q^{2\delta_{ij}}\hat{x}_{j}\hat{y}_{i}.\label{eq:Big_q-torus_algebra-intro}
\end{equation}
The 3d-4d homomorphism expresses operators $X_{\bullet}$ as certain monomials in $\hat{x}_{i}$ and $\hat{y}_{j}$, so that evaluating both sides of (\ref{pentagon}) on appropriate ground states yields the $q$-series identity (\ref{unlinking}). 
It then also follows that the resulting symmetric quiver $Q$ is a symmetrization of the 4d BPS quiver encoding the appropriate Kontsevich-Soibelman operator \eqref{KS-intro}. 
In the specific case of $A_2$ quiver with two nodes connected by one arrow, corresponding to the $A_2$ Argyres-Douglas theory, symmetrization gives a quiver with two nodes and a pair of arrows with opposite orientations, thus relating the diagrams \eqref{eq:A_2 unlinking} and \eqref{eq:A_2 dynkin quiver}. 

\paragraph{Path polytopes.}

On the other hand, the construction above can be rephrased in pure combinatorial terms using the notion of \emph{oriented exchange graph} \cite{keller2011cluster,keller2013quiver, garver2017maximal,garver2019lattice,padrol2023associahedra} which encapsulates the structure of the Konstevich-Soibelman operator (see Section \ref{sec:connectors and path polytopes}). 
For $A_2$ quiver, such graph is equivalent to two-dimensional associahedron $K_4$ where the two paths correspond to the left- and right-hand side of \eqref{pentagon}, respectively: 
\[\begin{tikzcd}[column sep=0.1]
& {a((bc)d)} && {a(b(cd))} \\
	{(a(bc))d} & {} && {} & {(ab)(cd)} \\
	&& {((ab)c)d}
	\arrow[from=1-2, to=1-4]
	\arrow[from=2-1, to=1-2]
	\arrow[from=2-5, to=1-4]
	\arrow[from=3-3, to=2-1]
	\arrow[from=3-3, to=2-5]
\end{tikzcd}\]
Therefore, unlinking operator $U(12)$ can be seen as a homotopy between the two paths in the associahedron. 
In other words, the structure of symmetrization map for $A_2$ quiver can be described by the \emph{path polytope} of $K_4$. 
In this case, it is simply one unlinking and the path polytope is just $* \xrightarrow{U(12)} *$. 
We also find that the path polytope of the exchange graph of $A_3$ quiver (i.e., a 3d associahedron $K_5$) is a~hexagon, which can be directly interpreted as the unlinking hexagon found in \cite{KLNS}. 
This property allows us to generalize this example to any $A_m$ quiver and even go beyond that. 
For higher $m$ and other types of root systems, we also find the corresponding higher-dimensional polytopes -- see Figure \ref{fig:A4_polytope_example} for an example. 
On the other hand, the same structures can be derived solely from the properties of unlinking using the notion of a \emph{connector} (see Section \ref{sec:connectors}). 
\begin{figure}[h!]
    \centering
    \includegraphics[width=0.6\linewidth]{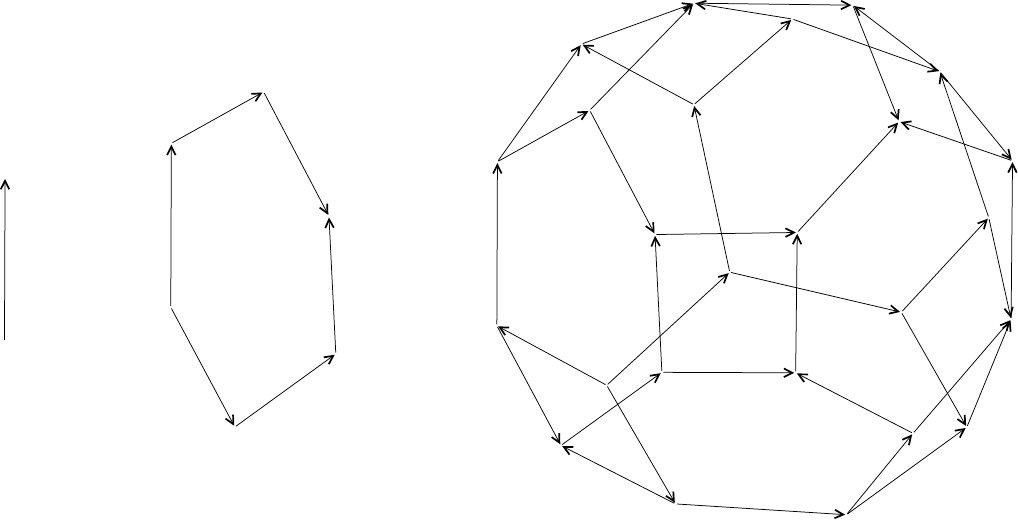}
    \caption{The path polytopes corresponding to $A_2$, $A_3$ and $A_4$ quivers with linear orientation. 
    On one hand, they can be viewed as a set of unlinking operators (where every directed edge is some unlinking $U(ij)$ and consecutive arrows define composition), and define the symmetrization maps for the respective 4d BPS quivers. 
    On the other hand, they are in a way dual to associahedra $K_4$, $K_5$ and $K_6$, respectively.}
    \label{fig:A4_polytope_example}
\end{figure}

\paragraph{Properties of the symmetrization map.}

This means that the symmetrization map is based on a homomorphism between quantum torus algebras that preserves the structure of wall-crossing, mapping it to symmetric quivers related by unlinking. 
Thanks to this property, we can define the symmetrization map $\mathfrak{S}$ between 4d BPS quivers with stability data and 3d symmetric quivers (\ref{Sigma-maps}) -- see Section \ref{sec:sym-map-def}. 
The full structure of the symmetrization map that includes all possible orientations of the 4d BPS quiver can be very involved, but there are several universal properties. 
First, since for the minimal chamber the quiver $Q=\mathfrak{S}(Q_{4\text{d}},\text{min})$ is just a symmetrization of $Q_{4\text{d}}$, it remains the same for any orientation of arrows in $Q_{4\text{d}}$. 
Second, the arrows connecting targets of different stability data always form two building blocks -- the \emph{hexagon} and the \emph{square} -- which reflect identities satisfied by unlinking (see Section \ref{sec:connectors}). 
Third, each $A_m$ quiver is a subquiver in a suitably oriented $A_{m+1}$. 
As a result, the symmetrization map respects this inclusion, and the corresponding polytopes are contained in each other, as we also see in Figure \ref{fig:A4_polytope_example}.

\paragraph{Schur index and symmetric quivers.}
Furthermore, the appearance of Kontsevich-Soibelman operators, generalizing those on both sides of the identity (\ref{pentagon}), enables us to express Schur indices in terms of quiver generating series (\ref{eq:motivic_generating_series}). 
Indeed, as shown in \cite{Cordova:2015nma}, Schur indices can be expressed as traces of a version of Kontsevich-Soibelman operators with twice as many quantum dilogarithms, representing both BPS and anti-BPS states in a given theory
\begin{equation}
\mathcal{I}=(q^2;q^2)_{\infty}^{2r} \, \textrm{Tr}[\mathcal{O}(q)], \qquad \mathcal{O}(q) =  \prod^\curvearrowleft_{\alpha:\ \textrm{BPS \& anti-BPS states}} \Psi(-X_{\alpha}), \label{Schur-intro}
\end{equation}
where $(q^2;q^2)_{\infty}^2$ is the contribution from a $U(1)$ vector multiplets, and $r$ is the rank of the Coulomb branch. 
Evaluating such traces also leads to expressions that can be written in the form (\ref{eq:motivic_generating_series}), from which we can read off the structure of underlying symmetric quivers $Q^{CPT}$. 
More precisely, these are CPT-doubled symmetric quivers, which arise from a CPT-doubled symmetrization map $\mathfrak{S}^{CPT}$  (\ref{Sigma-maps}) -- these quivers include a symmetrization of the BPS quiver of 4d theory, as well as extra nodes and arrows that arise as a consequence of including anti-BPS states into considerations.

\subsection*{Organization of the paper}
Section \ref{Topology} lays the topological foundation of the symmetrization map, reviewing how skein modules for sequences of signed flips give rise to symmetric quivers. 
Section \ref{sec:3d4d} then introduces the relevant physical 3d-4d systems, explains how their low-energy descriptions lead to a natural relation between 3d and 4d quivers, and constructs the symmetrization map in the minimal chamber for $A_m$ Argyres-Douglas theories. 
In Section~\ref{sec:algebra}, we discuss the relation between line operators in 3d and 4d, formulate it algebraically in terms of 3d-4d homomorphism, and work it out for the minimal chamber of $A_m$ quivers. 
Section \ref{sec-wall-crossing} contains our main proposal: the general construction of the symmetrization map, together with a~description of its chamber dependence in terms of unlinking. 
Section \ref{sec-schur} records a further connection, showing how Schur indices can be encoded in CPT-doubled quivers obtained from the 4d quiver via the symmetrization framework. 
Finally, Section \ref{sec:future-work}, discusses avenues for future work and Appendix \ref{sec:notations} summarizes the notations and conventions used in the paper.


\section{Topology of 2d and 3d manifolds and associated quivers}\label{Topology}

This section provides a review the combinatorial and topological structures that are relevant for the remainder of the paper. Many are well-known (see, e.g., \cite{CCV, GMN_SpecNet}), but we present them from the perspective of skein modules, in the spirit of \cite{Ekholm:2018eee, Ekholm:2019lmb, EkholmShende}.
This allows us to explain how sequences of signed flips of ideal triangulations give rise to symmetric quivers, which encode skein-valued counts of holomorphic disks arising from the geometry of the associated 3d bordisms. This construction forms the topological foundation of the symmetrization map~$\mathfrak{S}$. 

\subsection{Sequence of signed flips}
In this section, we discuss signed flips of ideal triangulations and corresponding elementary bordisms.
\newline

Let $C$ be a surface with an ideal triangulation $\tau$. 
Then, there is an associated branched double cover $\Sigma_{\tau} \rightarrow C$, obtained by putting one branch point at the barycenter of each ideal triangle; see Figure \ref{fig:TriangleDoubleCover}. 
That is, each ideal triangle has a branched double cover given by an ideal hexagon, and $\Sigma_\tau$ is obtained by gluing those ideal hexagons according to the gluing pattern of $\tau$. 
\begin{figure}[H]
\centering
$\vcenter{\hbox{
\begin{tikzpicture}
\begin{scope}[scale = 1.5]
    \coordinate (o) at (0, 0);
    \coordinate (a) at (0, 1);
    \coordinate (a1) at (0, -1/2);
    \coordinate (b) at ({sqrt(3)/2}, -1/2);
    \coordinate (b1) at ({-sqrt(3)/4}, 1/4);
    \coordinate (c) at (-{sqrt(3)/2}, -1/2);
    \coordinate (c1) at ({sqrt(3)/4}, 1/4);

    \tikzset{decoration={snake, amplitude=.4mm, segment length=2mm, post length=0mm, pre length=0mm}}
    \draw[orange, decorate] (o) -- (a1);
    \draw[orange, decorate] (o) -- (b1);
    \draw[orange, decorate] (o) -- (c1);
    
    \draw[very thick] (a) -- (b);
    \draw[very thick] (a) -- (c);
    \draw[very thick] (b) -- (c);
    \filldraw[red] (o) circle (0.1em);
    \filldraw (a) circle (0.03em);
    \filldraw (b) circle (0.03em);
    \filldraw (c) circle (0.03em);
\end{scope}
\end{tikzpicture}
}}$
\caption{Branch cuts for an ideal triangle}
\label{fig:TriangleDoubleCover}
\end{figure}
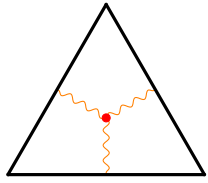
Suppose that $\tau'$ is an ideal triangulation of $C$ obtained from $\tau$ by a flip. 
Then, we can consider two bordisms
\[
T_{\pm} : \Sigma_{\tau} \rightarrow \Sigma_{\tau'}
\]
from $\Sigma_{\tau}$ to $\Sigma_{\tau'}$, giving branched double covers of $C \times I$; 
see Figure \ref{fig:signed_flip}, which specifies how the branch locus evolves along the $I$-direction in $C \times I$. 
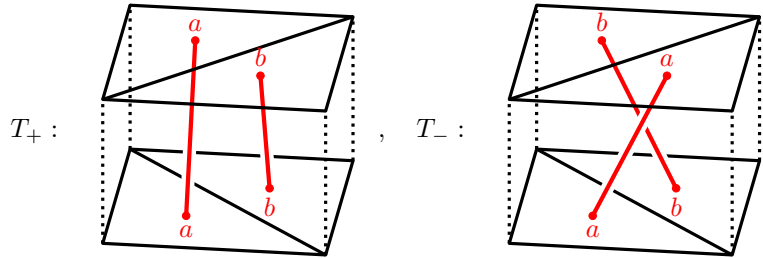
\begin{figure}[H]
\centering
\begin{math}
T_+ : 
\vcenter{\hbox{
\tdplotsetmaincoords{65}{52}
\begin{tikzpicture}[tdplot_main_coords]
\begin{scope}[scale = 0.7, tdplot_main_coords]
    \newcommand*{\defcoords}{
        \coordinate (o) at (0, 0, 0);
        \coordinate (a) at (3, 0, 0);
        \coordinate (b) at ({3*cos(90)}, {3*sin(90)}, 0);
        \coordinate (c) at ({3*cos(2*90)}, {3*sin(2*90)}, 0);
        \coordinate (d) at ({3*cos(3*90)}, {3*sin(3*90)}, 0);
        \coordinate (abc) at ($1/3*(a)+1/3*(b)+1/3*(c)$);
        \coordinate (abd) at ($1/3*(a)+1/3*(b)+1/3*(d)$);
        \coordinate (bcd) at ($1/3*(b)+1/3*(c)+1/3*(d)$);
        \coordinate (acd) at ($1/3*(a)+1/3*(c)+1/3*(d)$);
    }
    \defcoords
    \draw[very thick] (a) -- (b);
    \draw[very thick] (b) -- (c);
    \draw[very thick] (c) -- (d);
    \draw[very thick] (d) -- (a);
    \draw[very thick] (a) -- (c);
    \draw[very thick, dotted] (a) -- (3, 0, 3);
    \draw[very thick, dotted] (b) -- ({3*cos(90)}, {3*sin(90)}, 3);
    \draw[very thick, dotted] (c) -- ({3*cos(2*90)}, {3*sin(2*90)}, 3);
    \draw[very thick, dotted] (d) -- ({3*cos(3*90)}, {3*sin(3*90)}, 3);
    \filldraw (a) circle (0.05em);
    \filldraw (b) circle (0.05em);
    \filldraw (c) circle (0.05em);
    \filldraw (d) circle (0.05em);
    \draw[white, line width=5] (abc) -- ($(abd)+(0,0,3)$);
    \draw[ultra thick, red] (abc) -- ($(abd)+(0,0,3)$);
    \filldraw[red] (abc) circle (0.2em);
    \draw[white, line width=5] (acd) -- ($(bcd)+(0,0,3)$);
    \draw[ultra thick, red] (acd) -- ($(bcd)+(0,0,3)$);
    \filldraw[red] (acd) circle (0.2em);
    \node[red, anchor=north] at (acd){$a$};
    \node[red, anchor=north] at (abc){$b$};
    \begin{scope}[shift={(0, 0, 3)}]
        \defcoords
        \draw[very thick] (a) -- (b);
        \draw[very thick] (b) -- (c);
        \draw[very thick] (c) -- (d);
        \draw[very thick] (d) -- (a);
        \draw[very thick] (b) -- (d);
        \filldraw (a) circle (0.05em);
        \filldraw (b) circle (0.05em);
        \filldraw (c) circle (0.05em);
        \filldraw (d) circle (0.05em);
        \filldraw[red] (abd) circle (0.2em);
        \filldraw[red] (bcd) circle (0.2em);
        \node[red, anchor=south] at (bcd){$a$};
        \node[red, anchor=south] at (abd){$b$};
    \end{scope}
\end{scope}
\end{tikzpicture}
}}
\;\;,\quad
T_- : 
\vcenter{\hbox{
\tdplotsetmaincoords{65}{52}
\begin{tikzpicture}[tdplot_main_coords]
\begin{scope}[scale = 0.7, tdplot_main_coords]
    \newcommand*{\defcoords}{
        \coordinate (o) at (0, 0, 0);
        \coordinate (a) at (3, 0, 0);
        \coordinate (b) at ({3*cos(90)}, {3*sin(90)}, 0);
        \coordinate (c) at ({3*cos(2*90)}, {3*sin(2*90)}, 0);
        \coordinate (d) at ({3*cos(3*90)}, {3*sin(3*90)}, 0);
        \coordinate (abc) at ($1/3*(a)+1/3*(b)+1/3*(c)$);
        \coordinate (abd) at ($1/3*(a)+1/3*(b)+1/3*(d)$);
        \coordinate (bcd) at ($1/3*(b)+1/3*(c)+1/3*(d)$);
        \coordinate (acd) at ($1/3*(a)+1/3*(c)+1/3*(d)$);
    }
    \defcoords
    \draw[very thick] (a) -- (b);
    \draw[very thick] (b) -- (c);
    \draw[very thick] (c) -- (d);
    \draw[very thick] (d) -- (a);
    \draw[very thick] (a) -- (c);
    \draw[very thick, dotted] (a) -- (3, 0, 3);
    \draw[very thick, dotted] (b) -- ({3*cos(90)}, {3*sin(90)}, 3);
    \draw[very thick, dotted] (c) -- ({3*cos(2*90)}, {3*sin(2*90)}, 3);
    \draw[very thick, dotted] (d) -- ({3*cos(3*90)}, {3*sin(3*90)}, 3);
    \filldraw (a) circle (0.05em);
    \filldraw (b) circle (0.05em);
    \filldraw (c) circle (0.05em);
    \filldraw (d) circle (0.05em);
    \draw[white, line width=5] (abc) -- ($(bcd)+(0,0,3)$);
    \draw[ultra thick, red] (abc) -- ($(bcd)+(0,0,3)$);
    \filldraw[red] (abc) circle (0.2em);
    \node[red, anchor=north] at (abc){$b$};
    \draw[white, line width=5] (acd) -- ($(abd)+(0,0,3)$);
    \draw[ultra thick, red] (acd) -- ($(abd)+(0,0,3)$);
    \filldraw[red] (acd) circle (0.2em);
    \node[red, anchor=north] at (acd){$a$};
    \begin{scope}[shift={(0, 0, 3)}]
        \defcoords
        \draw[very thick] (a) -- (b);
        \draw[very thick] (b) -- (c);
        \draw[very thick] (c) -- (d);
        \draw[very thick] (d) -- (a);
        \draw[very thick] (b) -- (d);
        \filldraw (a) circle (0.05em);
        \filldraw (b) circle (0.05em);
        \filldraw (c) circle (0.05em);
        \filldraw (d) circle (0.05em);
        \filldraw[red] (abd) circle (0.2em);
        \filldraw[red] (bcd) circle (0.2em);
        \node[red, anchor=south] at (bcd){$b$};
        \node[red, anchor=south] at (abd){$a$};
    \end{scope}
\end{scope}
\end{tikzpicture}
}}
\end{math}
\caption{Elementary bordisms corresponding to signed flips.}
\label{fig:signed_flip}
\end{figure}
These elementary bordisms can be thought of as branched double covers of the taut 
ideal tetrahedron (i.e., an ideal tetrahedron with dihedral angles $0, 0, \pi$) thought of as a bordism from $(C,\tau)$ to $(C,\tau')$. 
We call the pair of a flip and a choice of either $T_+$ or $T_-$ a \emph{signed flip}. 
Note that the branched double cover corresponding to a triangulation of a square is an annulus, so both $T_+$ and $T_-$, thought of as bordisms from the branched double cover for 
\tikz[scale=0.3]{\draw (0, 0) -- (1, 0) -- (1, 1) -- (0, 1) -- cycle; \draw (1, 0) -- (0, 1);}
to that of 
\tikz[scale=0.3]{\draw (0, 0) -- (1, 0) -- (1, 1) -- (0, 1) -- cycle; \draw (0, 0) -- (1, 1);}, 
are topologically the annulus times $I$ (i.e., the solid torus -- see Figure \ref{fig:solid_torus}) which double covers the square times $I$ (i.e., the 3-ball) with branch locus given by the red tangle shown in Figure \ref{fig:signed_flip}.\footnote{To make it clear whether we are drawing figures in the base 3-manifold or in the branched double cover, we draw the branch locus on the base 3-manifold in red, and draw the branch locus on the double cover in orange.} 
\begin{figure}[H]
\centering
\[
\vcenter{\hbox{
\begin{tikzpicture}
\begin{scope}[scale=0.5]
\draw[very thick, dotted] (0, 0) ellipse (1 and 0.5);
\draw[very thick] (3, 0) arc (0:-180:3 and 1.5);
\draw[very thick, dotted] (3, 0) arc (0:180:3 and 1.5);
\draw[very thick] (3, 0) -- (3, 4);
\draw[very thick] (-3, 0) -- (-3, 4);
\draw[very thick, dotted] (1, 0) -- (1, 4);
\draw[very thick, dotted] (-1, 0) -- (-1, 4);
\filldraw[orange] (2, 0) circle (0.1);
\filldraw[orange] (2, 4) circle (0.1);
\draw[orange, very thick] (2, 0) -- (2, 4);
\filldraw[orange] (-2, 0) circle (0.1);
\filldraw[orange] (-2, 4) circle (0.1);
\draw[orange, very thick] (-2, 0) -- (-2, 4);
\draw[very thick] (0, 4) ellipse (1 and 0.5);
\draw[very thick] (0, 4) ellipse (3 and 1.5);
\end{scope}
\end{tikzpicture}
}}
\quad\overset{2:1}{\rightarrow}\;\;
\vcenter{\hbox{
\tdplotsetmaincoords{65}{52}
\begin{tikzpicture}[tdplot_main_coords]
\begin{scope}[scale = 0.7, tdplot_main_coords]
    \newcommand*{\defcoords}{
        \coordinate (o) at (0, 0, 0);
        \coordinate (a) at (3, 0, 0);
        \coordinate (b) at ({3*cos(90)}, {3*sin(90)}, 0);
        \coordinate (c) at ({3*cos(2*90)}, {3*sin(2*90)}, 0);
        \coordinate (d) at ({3*cos(3*90)}, {3*sin(3*90)}, 0);
        \coordinate (abc) at ($1/3*(a)+1/3*(b)+1/3*(c)$);
        \coordinate (abd) at ($1/3*(a)+1/3*(b)+1/3*(d)$);
        \coordinate (bcd) at ($1/3*(b)+1/3*(c)+1/3*(d)$);
        \coordinate (acd) at ($1/3*(a)+1/3*(c)+1/3*(d)$);
    }
    \defcoords
    \draw[very thick] (a) -- (b);
    \draw[very thick, dotted] (b) -- (c);
    \draw[very thick, dotted] (c) -- (d);
    \draw[very thick] (d) -- (a);
    \draw[very thick] (a) -- (3, 0, 3);
    \draw[very thick] (b) -- ({3*cos(90)}, {3*sin(90)}, 3);
    \draw[very thick, dotted] (c) -- ({3*cos(2*90)}, {3*sin(2*90)}, 3);
    \draw[very thick] (d) -- ({3*cos(3*90)}, {3*sin(3*90)}, 3);
    \filldraw (a) circle (0.05em);
    \filldraw (b) circle (0.05em);
    \filldraw (c) circle (0.05em);
    \filldraw (d) circle (0.05em);
    \draw[white, line width=5] (0.5,0.5,0) -- ($(0.5,0.5,0)+(0,0,3)$);
    \draw[ultra thick, red] (0.5,0.5,0) -- ($(0.5,0.5,0)+(0,0,3)$);
    \filldraw[red] (0.5,0.5,0) circle (0.2em);
    \draw[white, line width=5] (-0.5,-0.5,0) -- ($(-0.5,-0.5,0)+(0,0,3)$);
    \draw[ultra thick, red] (-0.5,-0.5,0) -- ($(-0.5,-0.5,0)+(0,0,3)$);
    \filldraw[red] (-0.5,-0.5,0) circle (0.2em);
    \begin{scope}[shift={(0, 0, 3)}]
        \defcoords
        \draw[very thick] (a) -- (b);
        \draw[very thick] (b) -- (c);
        \draw[very thick] (c) -- (d);
        \draw[very thick] (d) -- (a);
        \filldraw (a) circle (0.05em);
        \filldraw (b) circle (0.05em);
        \filldraw (c) circle (0.05em);
        \filldraw (d) circle (0.05em);
        \filldraw[red] (0.5,0.5,0) circle (0.2em);
        \filldraw[red] (-0.5,-0.5,0) circle (0.2em);
    \end{scope}
\end{scope}
\end{tikzpicture}
}}
\]
\caption{Solid torus double-covering the 3-ball.}
\label{fig:solid_torus}
\end{figure}
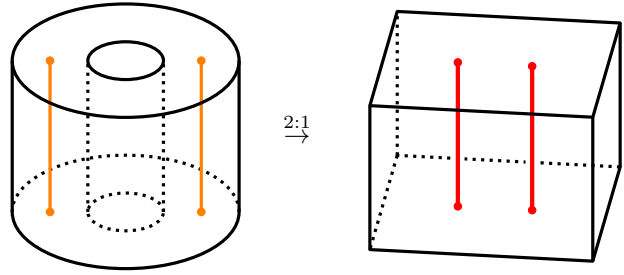

Now, given a sequence of signed flips
\[
\tau = \tau_0 \overset{\pm}{\rightarrow} \tau_1 \overset{\pm}{\rightarrow} \cdots \overset{\pm}{\rightarrow} \tau_m = \tau',
\]
we can glue the elementary bordisms together to get a bordism
\[
Y : \Sigma_{\tau} \rightarrow \Sigma_{\tau'},
\]
which is a branched double cover over $C \times I$. 

\subsection{Skein module}
In this section we recall the definition of the skein module and focus on the case of solid torus; we also discuss the skein algebra of the boundary torus.
\newline

\begin{dfn}\label{defn:gl1_skein}
Let $Y$ be a 3-manifold decorated with a branch locus $B$, an embedded 1-manifold in $Y$. 
Fix a ring $R$ containing $\mathbb{Z}[q,q^{-1}]$ (e.g., $\mathbb{Z}[q,q^{-1}]$ or $\mathbb{C}(q)$). 
The \emph{$\mathfrak{gl}_1$-skein module} of $Y$ is defined as
\[
\Sk^{\mathfrak{gl}_1}_q(Y) := \frac{R\langle \text{isotopy classes of framed oriented links in }Y\setminus B\rangle}{\langle \mathfrak{gl}_1\text{-skein relations}\rangle}
\]
where the $\mathfrak{gl}_1$-skein relations are given by
\begin{gather}
q^{-1}\;\vcenter{\hbox{
\begin{tikzpicture}[scale=0.7]
\draw[dotted] (0,0) circle (1);
\draw[ultra thick, ->] ({sqrt(2)/2},{-sqrt(2)/2}) -- ({-sqrt(2)/2},{sqrt(2)/2});
\draw[white, line width=2.5mm] ({-sqrt(2)/2},{-sqrt(2)/2}) -- ({sqrt(2)/2},{sqrt(2)/2});
\draw[ultra thick, ->] ({-sqrt(2)/2},{-sqrt(2)/2}) -- ({sqrt(2)/2},{sqrt(2)/2});
\end{tikzpicture}
}}
\;\;=\;\;
\vcenter{\hbox{
\begin{tikzpicture}[scale=0.7]
\draw[dotted] (0,0) circle (1);
\draw[ultra thick, <-] ({sqrt(2)/2},{sqrt(2)/2}) arc (135:225:1);
\draw[ultra thick, ->] ({-sqrt(2)/2},{-sqrt(2)/2}) arc (-45:45:1);
\end{tikzpicture}
}}
\;\;=\;\;
q\;
\vcenter{\hbox{
\begin{tikzpicture}[scale=0.7]
\draw[dotted] (0,0) circle (1);
\draw[ultra thick, ->] ({-sqrt(2)/2},{-sqrt(2)/2}) -- ({sqrt(2)/2},{sqrt(2)/2});
\draw[white, line width=2.5mm] ({sqrt(2)/2},{-sqrt(2)/2}) -- ({-sqrt(2)/2},{sqrt(2)/2});
\draw[ultra thick, ->] ({sqrt(2)/2},{-sqrt(2)/2}) -- ({-sqrt(2)/2},{sqrt(2)/2});
\end{tikzpicture}
}}
\;, \label{eq:skeinrel1}
\\
\vcenter{\hbox{
\begin{tikzpicture}[scale=0.7]
\draw[dotted] (0,0) circle (1);
\draw[ultra thick, ->] (0.5,0) arc (0:370:0.5);
\end{tikzpicture}
}}
\;\;=\;\;
\vcenter{\hbox{
\begin{tikzpicture}[scale=0.7]
\draw[dotted] (0,0) circle (1);
\end{tikzpicture}
}}
\;, \label{eq:skeinrel2}
\\
\vcenter{\hbox{
\begin{tikzpicture}[scale=0.7]
\draw[dotted] (0,0) circle (1);
\draw[ultra thick, orange] (-1, 0) -- (1, 0);
\draw[white, line width=2.5mm] (0, -1) -- (0, 1);
\draw[ultra thick, ->] (0, -1) -- (0, 1);
\end{tikzpicture}
}}
\;\;=\;\;
(-1)\cdot
\vcenter{\hbox{
\begin{tikzpicture}[scale=0.7]
\draw[dotted] (0,0) circle (1);
\draw[ultra thick, ->] (0, -1) -- (0, 1);
\draw[white, line width=2.5mm] (-1, 0) -- (1, 0);
\draw[ultra thick, orange] (-1, 0) -- (1, 0);
\end{tikzpicture}
}}
\;, \label{eq:skeinrel3}
\end{gather}
where the orange line in the last skein relation is part of the branch locus $B$. 

For a surface $\Sigma$ (possibly with branch points), the skein module $\Sk^{\mathfrak{gl}_1}_q(\Sigma \times I)$ has a natural algebra structure given by
\[
[L_1] \cdot [L_2] := [L_1 \cdot L_2]
\]
for any framed oriented links $L_1, L_2 \subset \Sigma \times I$, where $L_1 \cdot L_2$ denotes the link obtained by stacking $L_1$ above $L_2$ along the $I$-direction. 
We will call the skein module $\Sk^{\mathfrak{gl}_1}_q(\Sigma \times I)$, equipped with this algebra structure, the \emph{$\mathfrak{gl}_1$-skein algebra} of $\Sigma$ and denote it by $\SkAlg^{\mathfrak{gl}_1}_q(\Sigma)$. 
\end{dfn}
It is clear from the definition that if $Y \hookrightarrow Y'$ is an embedding of 3-manifolds, then there is an induced $R$-linear map
\[
\Sk^{\mathfrak{gl}_1}_q(Y) \rightarrow \Sk^{\mathfrak{gl}_1}_q(Y'). 
\]

\paragraph{Solid torus.}
The structure of the skein module for the solid torus is well known: 
\[
\Sk_q^{\mathfrak{gl}_1}\qty(
\vcenter{\hbox{
\begin{tikzpicture}
\begin{scope}[scale=0.3]
\draw[very thick, dotted] (0, 0) ellipse (1 and 0.5);
\draw[very thick] (3, 0) arc (0:-180:3 and 1.5);
\draw[very thick, dotted] (3, 0) arc (0:180:3 and 1.5);
\draw[very thick] (3, 0) -- (3, 4);
\draw[very thick] (-3, 0) -- (-3, 4);
\draw[very thick, dotted] (1, 0) -- (1, 4);
\draw[very thick, dotted] (-1, 0) -- (-1, 4);
\filldraw[orange] (2, 0) circle (0.1);
\filldraw[orange] (2, 4) circle (0.1);
\draw[orange, very thick] (2, 0) -- (2, 4);
\filldraw[orange] (-2, 0) circle (0.1);
\filldraw[orange] (-2, 4) circle (0.1);
\draw[orange, very thick] (-2, 0) -- (-2, 4);
\draw[very thick] (0, 4) ellipse (1 and 0.5);
\draw[very thick] (0, 4) ellipse (3 and 1.5);
\end{scope}
\end{tikzpicture}
}}
)
\;\cong\;
R[x,x^{-1}],
\]
where $x^d \in R[x,x^{-1}]$ represents the following skein:
\[
x^d \leftrightarrow
\vcenter{\hbox{
\begin{tikzpicture}
\begin{scope}[scale=0.5]
\draw[very thick, dotted] (0, 0) ellipse (1 and 0.5);
\draw[very thick] (3, 0) arc (0:-180:3 and 1.5);
\draw[very thick, dotted] (3, 0) arc (0:180:3 and 1.5);
\draw[very thick] (3, 0) -- (3, 4);
\draw[very thick] (-3, 0) -- (-3, 4);
\draw[very thick, dotted] (1, 0) -- (1, 4);
\draw[very thick, dotted] (-1, 0) -- (-1, 4);
\filldraw[orange] (2, 0) circle (0.1);
\filldraw[orange] (2, 4) circle (0.1);
\draw[orange, very thick] (2, 0) -- (2, 4);
\filldraw[orange] (-2, 0) circle (0.1);
\filldraw[orange] (-2, 4) circle (0.1);
\draw[orange, very thick] (-2, 0) -- (-2, 4);
\draw[very thick] (0, 4) ellipse (1 and 0.5);
\draw[very thick] (0, 4) ellipse (3 and 1.5);
\draw[ultra thick, <-, blue] (0, 0.8) arc (-90:270:1.5 and 0.75);
\node[anchor=south, blue] at (0, 0.8){$d$};
\end{scope}
\end{tikzpicture}
}}
\;\;\in\;\;
\Sk_q^{\mathfrak{gl}_1}\qty(
\vcenter{\hbox{
\begin{tikzpicture}
\begin{scope}[scale=0.3]
\draw[very thick, dotted] (0, 0) ellipse (1 and 0.5);
\draw[very thick] (3, 0) arc (0:-180:3 and 1.5);
\draw[very thick, dotted] (3, 0) arc (0:180:3 and 1.5);
\draw[very thick] (3, 0) -- (3, 4);
\draw[very thick] (-3, 0) -- (-3, 4);
\draw[very thick, dotted] (1, 0) -- (1, 4);
\draw[very thick, dotted] (-1, 0) -- (-1, 4);
\filldraw[orange] (2, 0) circle (0.1);
\filldraw[orange] (2, 4) circle (0.1);
\draw[orange, very thick] (2, 0) -- (2, 4);
\filldraw[orange] (-2, 0) circle (0.1);
\filldraw[orange] (-2, 4) circle (0.1);
\draw[orange, very thick] (-2, 0) -- (-2, 4);
\draw[very thick] (0, 4) ellipse (1 and 0.5);
\draw[very thick] (0, 4) ellipse (3 and 1.5);
\end{scope}
\end{tikzpicture}
}}
),
\]
where $d$ denotes $d$ parallel strands (and $|d|$ parallel strands in the opposite orientation if $d < 0$). 
Vertical stacking gives an algebra structure on the skein module of the annulus times $I$, and it is easy to see that the above isomorphism is an isomorphism of $R$-algebras:
\[
\SkAlg^{\mathfrak{gl}_1}_q\qty(
\vcenter{\hbox{
\begin{tikzpicture}
\begin{scope}[scale=0.3]
\draw[very thick, fill=lightgray!30] (0, 0) circle (3);
\draw[very thick, fill=white] (0, 0) circle (1);
\filldraw[orange] (2, 0) circle (0.1);
\filldraw[orange] (-2, 0) circle (0.1);
\end{scope}
\end{tikzpicture}
}}
)
\;\cong
R[x, x^{-1}]. 
\]

\paragraph{Boundary torus.}
The skein module of the solid torus is naturally a module over the skein algebra of the boundary torus. 
Let $\hat{x}$ and $\hat{y}$ be the following elements of the $\mathfrak{gl}_1$-skein algebra of the torus (with four branch points):
\begin{equation}\label{eq:xhat-yhat-A1}
\hat{x} := 
\vcenter{\hbox{
\begin{tikzpicture}
\begin{scope}[scale=0.5]
\draw[very thick, dotted] (0, 0) ellipse (1 and 0.5);
\draw[very thick] (3, 0) arc (0:-180:3 and 1.5);
\draw[very thick, dotted] (3, 0) arc (0:180:3 and 1.5);
\draw[very thick] (3, 0) -- (3, 4);
\draw[very thick] (-3, 0) -- (-3, 4);
\draw[very thick, dotted] (1, 0) -- (1, 4);
\draw[very thick, dotted] (-1, 0) -- (-1, 4);
\filldraw[orange] (2, 0) circle (0.1);
\filldraw[orange] (2, 4) circle (0.1);
\filldraw[orange] (-2, 0) circle (0.1);
\filldraw[orange] (-2, 4) circle (0.1);
\draw[very thick] (0, 4) ellipse (1 and 0.5);
\draw[very thick] (0, 4) ellipse (3 and 1.5);
\draw[ultra thick, <-, blue] (0, 3.25) arc (-90:270:1.5 and 0.75);
\end{scope}
\end{tikzpicture}
}}
\;,\quad
\hat{y} := 
\vcenter{\hbox{
\begin{tikzpicture}
\begin{scope}[scale=0.5]
\draw[very thick, dotted] (0, 0) ellipse (1 and 0.5);
\draw[very thick] (3, 0) arc (0:-180:3 and 1.5);
\draw[very thick, dotted] (3, 0) arc (0:180:3 and 1.5);
\draw[very thick] (3, 0) -- (3, 4);
\draw[very thick] (-3, 0) -- (-3, 4);
\draw[very thick, dotted] (1, 0) -- (1, 4);
\draw[very thick, dotted] (-1, 0) -- (-1, 4);
\filldraw[orange] (2, 0) circle (0.1);
\filldraw[orange] (2, 4) circle (0.1);
\filldraw[orange] (-2, 0) circle (0.1);
\filldraw[orange] (-2, 4) circle (0.1);
\draw[very thick] (0, 4) ellipse (1 and 0.5);
\draw[very thick] (0, 4) ellipse (3 and 1.5);
\draw[ultra thick, blue] ($(0, 4) - ({1/sqrt(2)}, {0.5/sqrt(2)})$) -- ($(0, 4) - ({3/sqrt(2)}, {1.5/sqrt(2)})$) -- ($(0, 0) - ({3/sqrt(2)}, {1.5/sqrt(2)})$) -- ($(0, 0) - ({1/sqrt(2)}, {0.5/sqrt(2)})$) -- cycle;
\draw[ultra thick, ->, blue] ($(0, 0) - ({3/sqrt(2)}, {1.5/sqrt(2)})$) -- ($(0, 2) - ({3/sqrt(2)}, {1.5/sqrt(2)})$);
\end{scope}
\end{tikzpicture}
}}
\;\;\in\;\;
\SkAlg^{\mathfrak{gl}_1}_q \qty(
\vcenter{\hbox{
\begin{tikzpicture}
\begin{scope}[scale=0.3]
\draw[very thick, dotted] (0, 0) ellipse (1 and 0.5);
\draw[very thick] (3, 0) arc (0:-180:3 and 1.5);
\draw[very thick, dotted] (3, 0) arc (0:180:3 and 1.5);
\draw[very thick] (3, 0) -- (3, 4);
\draw[very thick] (-3, 0) -- (-3, 4);
\draw[very thick, dotted] (1, 0) -- (1, 4);
\draw[very thick, dotted] (-1, 0) -- (-1, 4);
\filldraw[orange] (2, 0) circle (0.1);
\filldraw[orange] (2, 4) circle (0.1);
\filldraw[orange] (-2, 0) circle (0.1);
\filldraw[orange] (-2, 4) circle (0.1);
\draw[very thick] (0, 4) ellipse (1 and 0.5);
\draw[very thick] (0, 4) ellipse (3 and 1.5);
\end{scope}
\end{tikzpicture}
}}
)
\;.
\end{equation}
Then the $\mathfrak{gl}_1$-skein algebra of the torus (with four branch points) is a quantum torus generated by these two elements: 
\begin{equation}
\SkAlg^{\mathfrak{gl}_1}_q \qty(
\vcenter{\hbox{
\begin{tikzpicture}
\begin{scope}[scale=0.3]
\draw[very thick, dotted] (0, 0) ellipse (1 and 0.5);
\draw[very thick] (3, 0) arc (0:-180:3 and 1.5);
\draw[very thick, dotted] (3, 0) arc (0:180:3 and 1.5);
\draw[very thick] (3, 0) -- (3, 4);
\draw[very thick] (-3, 0) -- (-3, 4);
\draw[very thick, dotted] (1, 0) -- (1, 4);
\draw[very thick, dotted] (-1, 0) -- (-1, 4);
\filldraw[orange] (2, 0) circle (0.1);
\filldraw[orange] (2, 4) circle (0.1);
\filldraw[orange] (-2, 0) circle (0.1);
\filldraw[orange] (-2, 4) circle (0.1);
\draw[very thick] (0, 4) ellipse (1 and 0.5);
\draw[very thick] (0, 4) ellipse (3 and 1.5);
\end{scope}
\end{tikzpicture}
}}
)
\cong 
\frac{R\langle \hat{x}^{\pm 1}, \hat{y}^{\pm 1}\rangle}{(\hat{y}\hat{x} = q^2\hat{x}\hat{y})}.
\end{equation}
The skein algebra of the torus acts on the skein module of the solid torus in the following way: 
\begin{equation}\label{eq:quantum_torus_algebra}
    \hat{x}\cdot x^d = x^{d+1},\quad \hat{y} \cdot x^d = q^{2d}x^d.
\end{equation}

\subsection{Quantum dilogarithm and holomorphic disks}
In this section, we recall the definition of the quantum dilogarithm -- which is a crucial object in this work -- and discuss some of its properties and interpretations.
\newline

$\Psi = \Psi(x)$ is a distinguished element of (a completion of) the skein module of the solid torus. 
It is uniquely characterized by the following $q$-difference equation: 
\begin{equation}
    (1-q \hat{x}-\hat{y})\Psi = 0,
\end{equation}
with the initial condition that $\Psi = 1 + \cdots$, where $\cdots \in R[[x]]x$, and called the \emph{quantum dilogarithm}. 
Explicitly, 
\begin{equation}\label{eqn:qdilog}
\Psi = \frac{1}{(qx;q^2)_\infty} = \sum_{d\geq 0} \frac{q^d x^d}{(q^2;q^2)_d} \in R[[x]],
\end{equation}
where
\[
(z;q^2)_d := \prod_{k=0}^{d-1}(1-q^{2k}z)\quad\text{and}\quad (z;q^2)_{\infty} := \prod_{k\geq 0}(1-q^{2k}z). 
\]
Thinking of the skein module of the solid torus as the skein algebra of the annulus, $\Psi$ is invertible, with inverse
\begin{equation}
\Psi^{-1} = (q x;q^2)_{\infty} = \sum_{d\geq 0} \frac{(-1)^d q^{d^2}x^d}{(q^2;q^2)_d} \in R[[x]],
\end{equation}
satisfying the $q$-difference equation: 
\begin{equation}
    (1-q^{-1}\hat{x} - \hat{y}^{-1}) \Psi^{-1} = 0.
\end{equation}

Now, back to the 3-dimensional bordisms $Y : \Sigma_{\tau} \rightarrow \Sigma_{\tau'}$ constructed out of a sequence of signed flips, we consider the distinguished element $Z$ of (a completion of) the $\mathfrak{gl}_1$-skein module $\widehat{\Sk}_{q}^{\mathfrak{gl}_1}(Y)$ of $Y$, obtained by inserting the quantum dilogarithm $\Psi$ for each positive elementary bordism $T_+$ and the inverse quantum dilogarithm $\Psi^{-1}$ for each negative elementary bordism $T_-$. 
This element, $Z$, is secretly a (skein-valued) count of holomorphic curves ending on a certain Lagrangian diffeomorphic to $Y$ inside $T^*(C\times I)$, and $\Psi$ (resp. $\Psi^{-1}$) is the count of a holomorphic disk (resp. a holomorphic anti-disk) and its multiple covers. 
Pictorially, we will draw these skeins as in Figure \ref{fig:holomorphic_disks}. 
\begin{figure}[H]
\centering
\begin{math}
T_+ : 
\vcenter{\hbox{
\tdplotsetmaincoords{65}{52}
\begin{tikzpicture}[tdplot_main_coords]
\begin{scope}[scale = 0.7, tdplot_main_coords]
    \newcommand*{\defcoords}{
        \coordinate (o) at (0, 0, 0);
        \coordinate (a) at (3, 0, 0);
        \coordinate (b) at ({3*cos(90)}, {3*sin(90)}, 0);
        \coordinate (c) at ({3*cos(2*90)}, {3*sin(2*90)}, 0);
        \coordinate (d) at ({3*cos(3*90)}, {3*sin(3*90)}, 0);
        \coordinate (abc) at ($1/3*(a)+1/3*(b)+1/3*(c)$);
        \coordinate (abd) at ($1/3*(a)+1/3*(b)+1/3*(d)$);
        \coordinate (bcd) at ($1/3*(b)+1/3*(c)+1/3*(d)$);
        \coordinate (acd) at ($1/3*(a)+1/3*(c)+1/3*(d)$);
    }
    \defcoords
    \draw[very thick] (a) -- (b);
    \draw[very thick] (b) -- (c);
    \draw[very thick] (c) -- (d);
    \draw[very thick] (d) -- (a);
    \draw[very thick] (a) -- (c);
    \draw[very thick, dotted] (a) -- (3, 0, 3);
    \draw[very thick, dotted] (b) -- ({3*cos(90)}, {3*sin(90)}, 3);
    \draw[very thick, dotted] (c) -- ({3*cos(2*90)}, {3*sin(2*90)}, 3);
    \draw[very thick, dotted] (d) -- ({3*cos(3*90)}, {3*sin(3*90)}, 3);
    \filldraw (a) circle (0.05em);
    \filldraw (b) circle (0.05em);
    \filldraw (c) circle (0.05em);
    \filldraw (d) circle (0.05em);
    \draw[white, line width=5] (abc) -- ($(abd)+(0,0,3)$);
    \draw[ultra thick, red] (abc) -- ($(abd)+(0,0,3)$);
    \filldraw[red] (abc) circle (0.2em);
    \draw[white, line width=5] (acd) -- ($(bcd)+(0,0,3)$);
    \draw[ultra thick, red] (acd) -- ($(bcd)+(0,0,3)$);
    \filldraw[red] (acd) circle (0.2em);
    \node[red, anchor=north] at (acd){$a$};
    \node[red, anchor=north] at (abc){$b$};
    \draw[blue, line width=3] ($1/2*(abc) + 1/2*(abd) + (0, 0, 3/2)$) -- ($1/2*(acd) + 1/2*(bcd) + (0, 0, 3/2)$);
    \node[blue, anchor=north] at ($1/4*(abc) + 1/4*(abd) + 1/4*(acd) + 1/4*(bcd) + (0, 0, 3/2)$){$\Psi$};
    \begin{scope}[shift={(0, 0, 3)}]
        \defcoords
        \draw[very thick] (a) -- (b);
        \draw[very thick] (b) -- (c);
        \draw[very thick] (c) -- (d);
        \draw[very thick] (d) -- (a);
        \draw[very thick] (b) -- (d);
        \filldraw (a) circle (0.05em);
        \filldraw (b) circle (0.05em);
        \filldraw (c) circle (0.05em);
        \filldraw (d) circle (0.05em);
        \filldraw[red] (abd) circle (0.2em);
        \filldraw[red] (bcd) circle (0.2em);
        \node[red, anchor=south] at (bcd){$a$};
        \node[red, anchor=south] at (abd){$b$};
    \end{scope}
\end{scope}
\end{tikzpicture}
}}
\;\;,\quad
T_- : 
\vcenter{\hbox{
\tdplotsetmaincoords{65}{52}
\begin{tikzpicture}[tdplot_main_coords]
\begin{scope}[scale = 0.7, tdplot_main_coords]
    \newcommand*{\defcoords}{
        \coordinate (o) at (0, 0, 0);
        \coordinate (a) at (3, 0, 0);
        \coordinate (b) at ({3*cos(90)}, {3*sin(90)}, 0);
        \coordinate (c) at ({3*cos(2*90)}, {3*sin(2*90)}, 0);
        \coordinate (d) at ({3*cos(3*90)}, {3*sin(3*90)}, 0);
        \coordinate (abc) at ($1/3*(a)+1/3*(b)+1/3*(c)$);
        \coordinate (abd) at ($1/3*(a)+1/3*(b)+1/3*(d)$);
        \coordinate (bcd) at ($1/3*(b)+1/3*(c)+1/3*(d)$);
        \coordinate (acd) at ($1/3*(a)+1/3*(c)+1/3*(d)$);
    }
    \defcoords
    \draw[very thick] (a) -- (b);
    \draw[very thick] (b) -- (c);
    \draw[very thick] (c) -- (d);
    \draw[very thick] (d) -- (a);
    \draw[very thick] (a) -- (c);
    \draw[very thick, dotted] (a) -- (3, 0, 3);
    \draw[very thick, dotted] (b) -- ({3*cos(90)}, {3*sin(90)}, 3);
    \draw[very thick, dotted] (c) -- ({3*cos(2*90)}, {3*sin(2*90)}, 3);
    \draw[very thick, dotted] (d) -- ({3*cos(3*90)}, {3*sin(3*90)}, 3);
    \filldraw (a) circle (0.05em);
    \filldraw (b) circle (0.05em);
    \filldraw (c) circle (0.05em);
    \filldraw (d) circle (0.05em);
    \draw[white, line width=5] (abc) -- ($(bcd)+(0,0,3)$);
    \draw[ultra thick, red] (abc) -- ($(bcd)+(0,0,3)$);
    \filldraw[red] (abc) circle (0.2em);
    \node[red, anchor=north] at (abc){$b$};
    \draw[white, line width=5] (acd) -- ($(abd)+(0,0,3)$);
    \draw[cyan, line width=3] ($1/2*(abc) + 1/2*(bcd) + (0, 0, 3/2)$) -- ($1/2*(acd) + 1/2*(abd) + (0, 0, 3/2)$);
    \node[cyan, anchor=east] at ($1/4*(abc) + 1/4*(bcd) + 1/4*(acd) + 1/4*(abd) + (0, 0, 3/2)$){$\Psi^{-1}$};
    \draw[ultra thick, red] (acd) -- ($(abd)+(0,0,3)$);
    \filldraw[red] (acd) circle (0.2em);
    \node[red, anchor=north] at (acd){$a$};
    \begin{scope}[shift={(0, 0, 3)}]
        \defcoords
        \draw[very thick] (a) -- (b);
        \draw[very thick] (b) -- (c);
        \draw[very thick] (c) -- (d);
        \draw[very thick] (d) -- (a);
        \draw[very thick] (b) -- (d);
        \filldraw (a) circle (0.05em);
        \filldraw (b) circle (0.05em);
        \filldraw (c) circle (0.05em);
        \filldraw (d) circle (0.05em);
        \filldraw[red] (abd) circle (0.2em);
        \filldraw[red] (bcd) circle (0.2em);
        \node[red, anchor=south] at (bcd){$b$};
        \node[red, anchor=south] at (abd){$a$};
    \end{scope}
\end{scope}
\end{tikzpicture}
}}
\end{math}
\caption{Elementary bordisms and holomorphic (anti) disks.}
\label{fig:holomorphic_disks}
\end{figure}
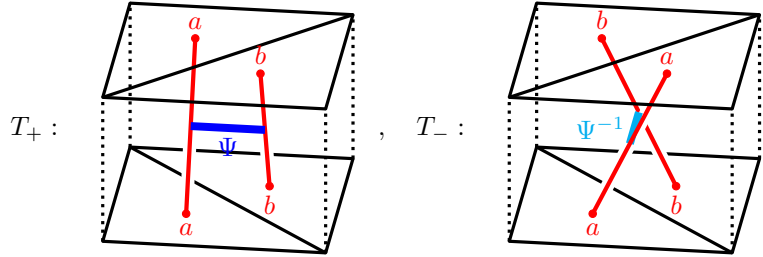

\paragraph{Quantum cluster transformation.}

There is a rich story behind this construction: 
Given a surface $C$ and a branched double cover $\Sigma_{\tau}$ induced by an ideal triangulation, there is an algebra homomorphism called the \emph{quantum UV--IR map} \cite{NeitzkeYan}
\[
F_{\tau} : \SkAlg^{\mathfrak{gl}_2}_q(C) \rightarrow \SkAlg^{\mathfrak{gl}_1}_q(\Sigma_\tau)
\]
from the $\mathfrak{gl}_2$-skein algebra of $C$ to the $\mathfrak{gl}_1$-skein algebra of $\Sigma_{\tau}$; 
and if $\tau'$ is an ideal triangulation obtained from $\tau$ by a flip, then
\begin{equation}
    F_{\tau} \Psi = \Psi F_{\tau'},
\end{equation}
i.e., the two maps $F_\tau$ and $F_{\tau'}$ are related by conjugation by a quantum dilogarithm, which is also known as the \emph{quantum cluster transformation}. 
See \cite{Ekholm:2025anq} for its generalization to HOMFLYPT skeins and the interpretation in terms of holomorphic curves. 

\paragraph{Disk-antidisk cancellation and pentagon relation.}
It is important to point out two properties of the distinguished element $Z \in \widehat{\Sk}_{q}^{\mathfrak{gl}_1}(Y)$ that we constructed above. 
First, a flip composed with a flip back with opposite signs gives the trivial bordism $Y = \Sigma_{\tau} \times I$, and in this case $Z = \Psi \Psi^{-1} = [\emptyset] \in \Sk_{q}^{\mathfrak{gl}_1}(Y)$; see Figure \ref{fig:disk_antidisk}. 
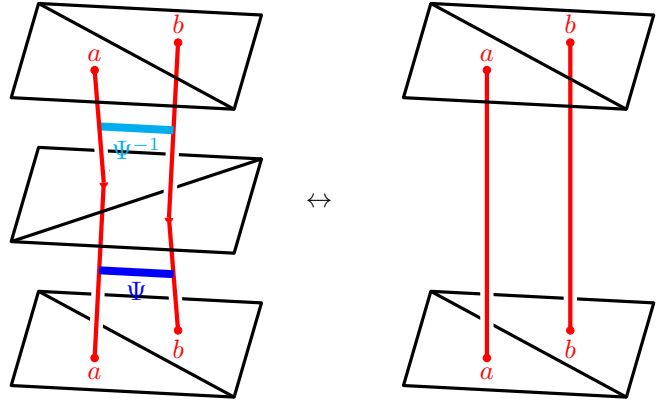
\begin{figure}[H]
\centering
\begin{math}
\vcenter{\hbox{
\tdplotsetmaincoords{65}{52}
\begin{tikzpicture}[tdplot_main_coords]
\begin{scope}[scale = 0.7, tdplot_main_coords]
    \newcommand*{\defcoords}{
        \coordinate (o) at (0, 0, 0);
        \coordinate (a) at (3, 0, 0);
        \coordinate (b) at ({3*cos(90)}, {3*sin(90)}, 0);
        \coordinate (c) at ({3*cos(2*90)}, {3*sin(2*90)}, 0);
        \coordinate (d) at ({3*cos(3*90)}, {3*sin(3*90)}, 0);
        \coordinate (abc) at ($1/3*(a)+1/3*(b)+1/3*(c)$);
        \coordinate (abd) at ($1/3*(a)+1/3*(b)+1/3*(d)$);
        \coordinate (bcd) at ($1/3*(b)+1/3*(c)+1/3*(d)$);
        \coordinate (acd) at ($1/3*(a)+1/3*(c)+1/3*(d)$);
    }
    \defcoords
    \draw[very thick] (a) -- (b);
    \draw[very thick] (b) -- (c);
    \draw[very thick] (c) -- (d);
    \draw[very thick] (d) -- (a);
    \draw[very thick] (a) -- (c);
    \filldraw (a) circle (0.05em);
    \filldraw (b) circle (0.05em);
    \filldraw (c) circle (0.05em);
    \filldraw (d) circle (0.05em);
    \draw[white, line width=5] (abc) -- ($(abd)+(0,0,3)$);
    \draw[ultra thick, red] (abc) -- ($(abd)+(0,0,3)$);
    \filldraw[red] (abc) circle (0.2em);
    \draw[white, line width=5] (acd) -- ($(bcd)+(0,0,3)$);
    \draw[ultra thick, red] (acd) -- ($(bcd)+(0,0,3)$);
    \filldraw[red] (acd) circle (0.2em);
    \node[red, anchor=north] at (acd){$a$};
    \node[red, anchor=north] at (abc){$b$};
    \draw[blue, line width=3] ($1/2*(abc) + 1/2*(abd) + (0, 0, 3/2)$) -- ($1/2*(acd) + 1/2*(bcd) + (0, 0, 3/2)$);
    \node[blue, anchor=north] at ($1/4*(abc) + 1/4*(abd) + 1/4*(acd) + 1/4*(bcd) + (0, 0, 3/2)$){$\Psi$};
    \begin{scope}[shift={(0, 0, 3)}]
        \defcoords
        \draw[very thick] (a) -- (b);
        \draw[very thick] (b) -- (c);
        \draw[very thick] (c) -- (d);
        \draw[very thick] (d) -- (a);
        \draw[very thick] (b) -- (d);
        \filldraw (a) circle (0.05em);
        \filldraw (b) circle (0.05em);
        \filldraw (c) circle (0.05em);
        \filldraw (d) circle (0.05em);
        \filldraw[red] (abd) circle (0.2em);
        \filldraw[red] (bcd) circle (0.2em);
        \node[red, anchor=south] at (bcd){$a$};
        \node[red, anchor=south] at (abd){$b$};
        \draw[white, line width=5] (abd) -- ($(abc)+(0,0,3)$);
        \draw[ultra thick, red] (abd) -- ($(abc)+(0,0,3)$);
        \draw[white, line width=5] (bcd) -- ($(acd)+(0,0,3)$);
        \draw[cyan, line width=3] ($1/2*(abc) + 1/2*(abd) + (0, 0, 3/2)$) -- ($1/2*(acd) + 1/2*(bcd) + (0, 0, 3/2)$);
        \node[cyan, anchor=north] at ($1/4*(abc) + 1/4*(abd) + 1/4*(acd) + 1/4*(bcd) + (0, 0, 3/2)$){$\Psi^{-1}$};
        \draw[ultra thick, red] (bcd) -- ($(acd)+(0,0,3)$);
        \begin{scope}[shift={(0, 0, 3)}]
            \defcoords
            \draw[very thick] (a) -- (b);
            \draw[very thick] (b) -- (c);
            \draw[very thick] (c) -- (d);
            \draw[very thick] (d) -- (a);
            \draw[very thick] (a) -- (c);
            \filldraw (a) circle (0.05em);
            \filldraw (b) circle (0.05em);
            \filldraw (c) circle (0.05em);
            \filldraw (d) circle (0.05em);
            \filldraw[red] (acd) circle (0.2em);
            \filldraw[red] (abc) circle (0.2em);
            \node[red, anchor=south] at (acd){$a$};
            \node[red, anchor=south] at (abc){$b$};
        \end{scope}
    \end{scope}
\end{scope}
\end{tikzpicture}
}}
\quad\leftrightarrow \quad 
\vcenter{\hbox{
\tdplotsetmaincoords{65}{52}
\begin{tikzpicture}[tdplot_main_coords]
\begin{scope}[scale = 0.7, tdplot_main_coords]
    \newcommand*{\defcoords}{
        \coordinate (o) at (0, 0, 0);
        \coordinate (a) at (3, 0, 0);
        \coordinate (b) at ({3*cos(90)}, {3*sin(90)}, 0);
        \coordinate (c) at ({3*cos(2*90)}, {3*sin(2*90)}, 0);
        \coordinate (d) at ({3*cos(3*90)}, {3*sin(3*90)}, 0);
        \coordinate (abc) at ($1/3*(a)+1/3*(b)+1/3*(c)$);
        \coordinate (abd) at ($1/3*(a)+1/3*(b)+1/3*(d)$);
        \coordinate (bcd) at ($1/3*(b)+1/3*(c)+1/3*(d)$);
        \coordinate (acd) at ($1/3*(a)+1/3*(c)+1/3*(d)$);
    }
    \defcoords
    \draw[very thick] (a) -- (b);
    \draw[very thick] (b) -- (c);
    \draw[very thick] (c) -- (d);
    \draw[very thick] (d) -- (a);
    \draw[very thick] (a) -- (c);
    \filldraw (a) circle (0.05em);
    \filldraw (b) circle (0.05em);
    \filldraw (c) circle (0.05em);
    \filldraw (d) circle (0.05em);
    \draw[white, line width=5] (abc) -- ($(abc)+(0,0,6)$);
    \draw[ultra thick, red] (abc) -- ($(abc)+(0,0,6)$);
    \filldraw[red] (abc) circle (0.2em);
    \draw[white, line width=5] (acd) -- ($(acd)+(0,0,6)$);
    \draw[ultra thick, red] (acd) -- ($(acd)+(0,0,6)$);
    \filldraw[red] (acd) circle (0.2em);
    \node[red, anchor=north] at (acd){$a$};
    \node[red, anchor=north] at (abc){$b$};
    \begin{scope}[shift={(0, 0, 6)}]
        \defcoords
        \draw[very thick] (a) -- (b);
        \draw[very thick] (b) -- (c);
        \draw[very thick] (c) -- (d);
        \draw[very thick] (d) -- (a);
        \draw[very thick] (a) -- (c);
        \filldraw (a) circle (0.05em);
        \filldraw (b) circle (0.05em);
        \filldraw (c) circle (0.05em);
        \filldraw (d) circle (0.05em);
        \filldraw[red] (acd) circle (0.2em);
        \filldraw[red] (abc) circle (0.2em);
        \node[red, anchor=south] at (acd){$a$};
        \node[red, anchor=south] at (abc){$b$};
    \end{scope}
\end{scope}
\end{tikzpicture}
}}
\end{math}
\caption{Cancelling a pair of a disk and an anti-disk}
\label{fig:disk_antidisk}
\end{figure}
Second, but more importantly, the quantum dilogarithm satisfies the \emph{pentagon relation}:
\begin{equation}\label{eq:pentagon_relation}
    \Psi_{ab}\Psi_{bc} = \Psi_{bc}\Psi_{ac}\Psi_{ab},
\end{equation}
which in this context becomes the invariance of $Z$ under the 2-3 Pachner move changing the taut ideal triangulation of $C \times I$, as illustrated in Figure \ref{fig:Pachner_unlinking}. 
In the context of symmetric quivers and their geometric interpretation in terms of holomorphic disks, this is closely related to the operation called \emph{unlinking} \cite{Ekholm:2019lmb}.\footnote{In the language of \emph{skein categories}, one can summarize this by saying that we have an embedding 
\begin{align*}
\mathbf{Tri}(C) &\rightarrow \mathbf{Bimod} \\
\tau &\mapsto \mathrm{SkCat}(\Sigma_{\tau})\\
(T_{\pm} : \tau \rightarrow \tau') &\mapsto (\Psi^{\pm 1} : \mathrm{SkCat}(\Sigma_{\tau}) \otimes \mathrm{SkCat}(\Sigma_{\tau'})^{op} \rightarrow \mathbf{Vect})
\end{align*}
of the 2-category $\mathbf{Tri}(C)$ of ideal triangulations of $C$, signed flips between them, and 2-3 Pachner moves into the 2-category $\mathbf{Bimod}$ of categories (in our case, $\mathfrak{gl}_1$-skein categories of branched double covers), bimodules between them, and bimodule homomorphisms.}  
\begin{figure}[htbp] 
\centering
\begin{math} 
\vcenter{\hbox{
\tdplotsetmaincoords{65}{52}
\begin{tikzpicture}[tdplot_main_coords]
\begin{scope}[scale = 0.7, tdplot_main_coords]
    \newcommand*{\defcoords}{
        \coordinate (o) at (0, 0, 0);
        \coordinate (a) at (3, 0, 0);
        \coordinate (b) at ({3*cos(72)}, {3*sin(72)}, 0);
        \coordinate (c) at ({3*cos(2*72)}, {3*sin(2*72)}, 0);
        \coordinate (d) at ({3*cos(3*72)}, {3*sin(3*72)}, 0);
        \coordinate (e) at ({3*cos(4*72)}, {3*sin(4*72)}, 0);
        \coordinate (abc) at ($1/3*(a)+1/3*(b)+1/3*(c)$);
        \coordinate (abe) at ($1/3*(a)+1/3*(b)+1/3*(e)$);
        \coordinate (ace) at ($1/3*(a)+1/3*(c)+1/3*(e)$);
        \coordinate (bcd) at ($1/3*(b)+1/3*(c)+1/3*(d)$);
        \coordinate (bce) at ($1/3*(b)+1/3*(c)+1/3*(e)$);
        \coordinate (bde) at ($1/3*(b)+1/3*(d)+1/3*(e)$);
        \coordinate (cde) at ($1/3*(c)+1/3*(d)+1/3*(e)$);
    }
    \defcoords
    \draw[very thick] (a) -- (b);
    \draw[very thick] (b) -- (c);
    \draw[very thick] (c) -- (d);
    \draw[very thick] (d) -- (e);
    \draw[very thick] (e) -- (a);
    \draw[very thick] (a) -- (c);
    \draw[very thick] (c) -- (e);
    \filldraw (a) circle (0.05em);
    \filldraw (b) circle (0.05em);
    \filldraw (c) circle (0.05em);
    \filldraw (d) circle (0.05em);
    \filldraw (e) circle (0.05em);
    \draw[white, line width=5] (cde) -- ($(cde)+(0,0,3)$);
    \draw[ultra thick, red] (cde) -- ($(cde)+(0,0,3)$);
    \draw[white, line width=5] (abc) -- ($(abe)+(0,0,3)$);
    \draw[ultra thick, red] (abc) -- ($(abe)+(0,0,3)$);
    \draw[white, line width=5] (ace) -- ($(bce)+(0,0,3)$);
    \draw[ultra thick, red] (ace) -- ($(bce)+(0,0,3)$);
    \filldraw[red] (abc) circle (0.2em);
    \filldraw[red] (ace) circle (0.2em);
    \filldraw[red] (cde) circle (0.2em);
    \node[red, anchor=north] at (cde){$a$};
    \node[red, anchor=north] at (ace){$b$};
    \node[red, anchor=north] at (abc){$c$};
    \draw[blue, line width=3] ($1/2*(abc) + 1/2*(abe) + (0, 0, 3/2)$) -- ($1/2*(ace) + 1/2*(bce) + (0, 0, 3/2)$);
    \node[blue, anchor=north] at ($1/4*(abc) + 1/4*(abe) + 1/4*(ace) + 1/4*(bce) + (0, 0, 3/2)$){$\Psi_{bc}$};
    \begin{scope}[shift={(0, 0, 3)}]
        \defcoords
        \draw[very thick] (a) -- (b);
        \draw[very thick] (b) -- (c);
        \draw[very thick] (c) -- (d);
        \draw[very thick] (d) -- (e);
        \draw[very thick] (e) -- (a);
        \draw[very thick] (b) -- (e);
        \draw[very thick] (c) -- (e);
        \filldraw (a) circle (0.05em);
        \filldraw (b) circle (0.05em);
        \filldraw (c) circle (0.05em);
        \filldraw (d) circle (0.05em);
        \filldraw (e) circle (0.05em);
        \draw[white, line width=5] (abe) -- ($(abe)+(0,0,3)$);
        \draw[ultra thick, red] (abe) -- ($(abe)+(0,0,3)$);
        \draw[white, line width=5] (cde) -- ($(bcd)+(0,0,3)$);
        \draw[ultra thick, red] (cde) -- ($(bcd)+(0,0,3)$);
        \draw[white, line width=5] (bce) -- ($(bde)+(0,0,3)$);
        \draw[ultra thick, red] (bce) -- ($(bde)+(0,0,3)$);
        \filldraw[red] (abe) circle (0.2em);
        \filldraw[red] (bce) circle (0.2em);
        \filldraw[red] (cde) circle (0.2em);
        \draw[blue, line width=3] ($1/2*(cde) + 1/2*(bcd) + (0, 0, 3/2)$) -- ($1/2*(bce) + 1/2*(bde) + (0, 0, 3/2)$);
        \node[blue, anchor=north] at ($1/4*(cde) + 1/4*(bcd) + 1/4*(bce) + 1/4*(bde) + (0, 0, 3/2)$){$\Psi_{ab}$};
    \end{scope}
    \begin{scope}[shift={(0, 0, 6)}]
        \defcoords
        \draw[very thick] (a) -- (b);
        \draw[very thick] (b) -- (c);
        \draw[very thick] (c) -- (d);
        \draw[very thick] (d) -- (e);
        \draw[very thick] (e) -- (a);
        \draw[very thick] (b) -- (e);
        \draw[very thick] (b) -- (d);
        \filldraw (a) circle (0.05em);
        \filldraw (b) circle (0.05em);
        \filldraw (c) circle (0.05em);
        \filldraw (d) circle (0.05em);
        \filldraw (e) circle (0.05em);
        \filldraw[red] (abe) circle (0.2em);
        \filldraw[red] (bcd) circle (0.2em);
        \filldraw[red] (bde) circle (0.2em);
        \node[red, anchor=south] at (bcd){$a$};
        \node[red, anchor=south] at (bde){$b$};
        \node[red, anchor=south] at (abe){$c$};
    \end{scope}
\end{scope}
\end{tikzpicture}
}}
\quad
\;\;\leftrightarrow\;\;
\vcenter{\hbox{
\tdplotsetmaincoords{65}{52}
\begin{tikzpicture}[tdplot_main_coords]
\begin{scope}[scale = 0.8, tdplot_main_coords]
    \newcommand*{\defcoords}{
        \coordinate (o) at (0, 0, 0);
        \coordinate (a) at (3, 0, 0);
        \coordinate (b) at ({3*cos(72)}, {3*sin(72)}, 0);
        \coordinate (c) at ({3*cos(2*72)}, {3*sin(2*72)}, 0);
        \coordinate (d) at ({3*cos(3*72)}, {3*sin(3*72)}, 0);
        \coordinate (e) at ({3*cos(4*72)}, {3*sin(4*72)}, 0);
        \coordinate (abc) at ($1/3*(a)+1/3*(b)+1/3*(c)$);
        \coordinate (abd) at ($1/3*(a)+1/3*(b)+1/3*(d)$);
        \coordinate (abe) at ($1/3*(a)+1/3*(b)+1/3*(e)$);
        \coordinate (acd) at ($1/3*(a)+1/3*(c)+1/3*(d)$);
        \coordinate (ace) at ($1/3*(a)+1/3*(c)+1/3*(e)$);
        \coordinate (ade) at ($1/3*(a)+1/3*(d)+1/3*(e)$);
        \coordinate (bcd) at ($1/3*(b)+1/3*(c)+1/3*(d)$);
        \coordinate (bce) at ($1/3*(b)+1/3*(c)+1/3*(e)$);
        \coordinate (bde) at ($1/3*(b)+1/3*(d)+1/3*(e)$);
        \coordinate (cde) at ($1/3*(c)+1/3*(d)+1/3*(e)$);
    }
    \defcoords
    \draw[very thick] (a) -- (b);
    \draw[very thick] (b) -- (c);
    \draw[very thick] (c) -- (d);
    \draw[very thick] (d) -- (e);
    \draw[very thick] (e) -- (a);
    \draw[very thick] (a) -- (c);
    \draw[very thick] (c) -- (e);
    \filldraw (a) circle (0.05em);
    \filldraw (b) circle (0.05em);
    \filldraw (c) circle (0.05em);
    \filldraw (d) circle (0.05em);
    \filldraw (e) circle (0.05em);
    \draw[white, line width=5] (abc) -- ($(abc)+(0,0,3)$);
    \draw[ultra thick, red] (abc) -- ($(abc)+(0,0,3)$);
    \draw[white, line width=5] (cde) -- ($(acd)+(0,0,3)$);
    \draw[ultra thick, red] (cde) -- ($(acd)+(0,0,3)$);
    \draw[white, line width=5] (ace) -- ($(ade)+(0,0,3)$);
    \draw[ultra thick, red] (ace) -- ($(ade)+(0,0,3)$);
    \filldraw[red] (abc) circle (0.2em);
    \filldraw[red] (ace) circle (0.2em);
    \filldraw[red] (cde) circle (0.2em);
    \node[red, anchor=north] at (cde){$a$};
    \node[red, anchor=north] at (ace){$b$};
    \node[red, anchor=north] at (abc){$c$};
    \draw[blue, line width=3] ($1/2*(cde) + 1/2*(acd) + (0, 0, 3/2)$) -- ($1/2*(ace) + 1/2*(ade) + (0, 0, 3/2)$);
    \node[blue, anchor=north] at ($1/4*(cde) + 1/4*(acd) + 1/4*(ace) + 1/4*(ade) + (0, 0, 3/2)$){$\Psi_{ab}$};
    \begin{scope}[shift={(0, 0, 3)}]
        \defcoords
        \draw[very thick] (a) -- (b);
        \draw[very thick] (b) -- (c);
        \draw[very thick] (c) -- (d);
        \draw[very thick] (d) -- (e);
        \draw[very thick] (e) -- (a);
        \draw[very thick] (a) -- (c);
        \draw[very thick] (a) -- (d);
        \filldraw (a) circle (0.05em);
        \filldraw (b) circle (0.05em);
        \filldraw (c) circle (0.05em);
        \filldraw (d) circle (0.05em);
        \filldraw (e) circle (0.05em);
        \draw[white, line width=5] (abc) -- ($(abd)+(0,0,3)$);
        \draw[ultra thick, red] (abc) -- ($(abd)+(0,0,3)$);
        \draw[white, line width=5] (acd) -- ($(bcd)+(0,0,3)$);
        \draw[ultra thick, red] (acd) -- ($(bcd)+(0,0,3)$);
        \draw[white, line width=5] (ade) -- ($(ade)+(0,0,3)$);
        \draw[ultra thick, red] (ade) -- ($(ade)+(0,0,3)$);
        \filldraw[red] (abc) circle (0.2em);
        \filldraw[red] (acd) circle (0.2em);
        \filldraw[red] (ade) circle (0.2em);
        \draw[blue, line width=3] ($1/2*(abc) + 1/2*(abd) + (0, 0, 3/2)$) -- ($1/2*(acd) + 1/2*(bcd) + (0, 0, 3/2)$);
        \node[blue, anchor=north] at ($1/4*(abc) + 1/4*(abd) + 1/4*(acd) + 1/4*(bcd) + (0, 0, 3/2)$){$\Psi_{ac}$};
    \end{scope}
    \begin{scope}[shift={(0, 0, 6)}]
        \defcoords
        \draw[very thick] (a) -- (b);
        \draw[very thick] (b) -- (c);
        \draw[very thick] (c) -- (d);
        \draw[very thick] (d) -- (e);
        \draw[very thick] (e) -- (a);
        \draw[very thick] (b) -- (d);
        \draw[very thick] (a) -- (d);
        \filldraw (a) circle (0.05em);
        \filldraw (b) circle (0.05em);
        \filldraw (c) circle (0.05em);
        \filldraw (d) circle (0.05em);
        \filldraw (e) circle (0.05em);
        \draw[white, line width=5] (bcd) -- ($(bcd)+(0,0,3)$);
        \draw[ultra thick, red] (bcd) -- ($(bcd)+(0,0,3)$);
        \draw[white, line width=5] (abd) -- ($(abe)+(0,0,3)$);
        \draw[ultra thick, red] (abd) -- ($(abe)+(0,0,3)$);
        \draw[white, line width=5] (ade) -- ($(bde)+(0,0,3)$);
        \draw[ultra thick, red] (ade) -- ($(bde)+(0,0,3)$);
        \filldraw[red] (abd) circle (0.2em);
        \filldraw[red] (bcd) circle (0.2em);
        \filldraw[red] (ade) circle (0.2em);
        \draw[blue, line width=3] ($1/2*(abd) + 1/2*(abe) + (0, 0, 3/2)$) -- ($1/2*(ade) + 1/2*(bde) + (0, 0, 3/2)$);
        \node[blue, anchor=north] at ($1/4*(abd) + 1/4*(abe) + 1/4*(ade) + 1/4*(bde) + (0, 0, 3/2)$){$\Psi_{bc}$};
    \end{scope}
    \begin{scope}[shift={(0, 0, 9)}]
        \defcoords
        \draw[very thick] (a) -- (b);
        \draw[very thick] (b) -- (c);
        \draw[very thick] (c) -- (d);
        \draw[very thick] (d) -- (e);
        \draw[very thick] (e) -- (a);
        \draw[very thick] (b) -- (e);
        \draw[very thick] (b) -- (d);
        \filldraw (a) circle (0.05em);
        \filldraw (b) circle (0.05em);
        \filldraw (c) circle (0.05em);
        \filldraw (d) circle (0.05em);
        \filldraw (e) circle (0.05em);
        \filldraw[red] (abe) circle (0.2em);
        \filldraw[red] (bcd) circle (0.2em);
        \filldraw[red] (bde) circle (0.2em);
        \node[red, anchor=south] at (bcd){$a$};
        \node[red, anchor=south] at (bde){$b$};
        \node[red, anchor=south] at (abe){$c$};
    \end{scope}
\end{scope}
\end{tikzpicture}
}}
\end{math}
\caption{Invariance under 2-3 Pachner move from pentagon (unlinking) relation} 
\label{fig:Pachner_unlinking} 
\end{figure}
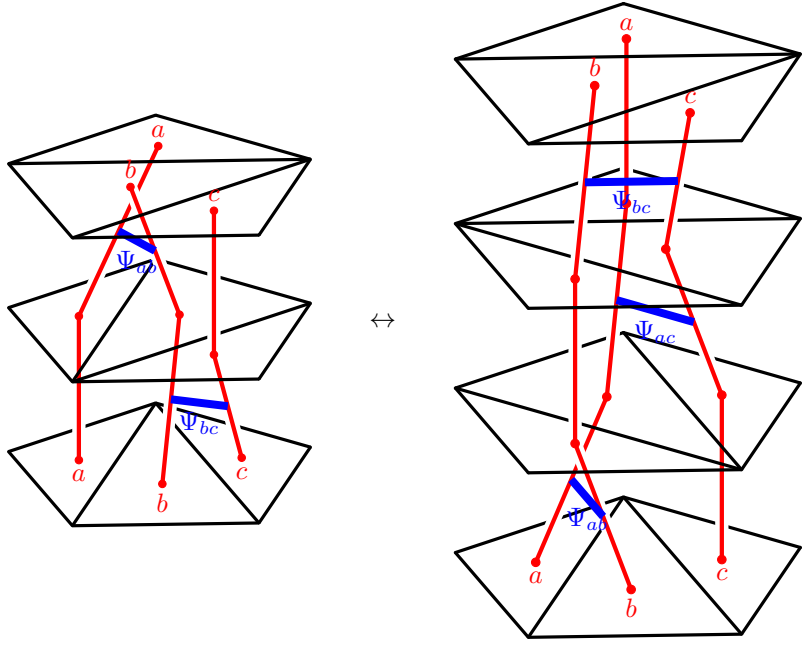

\subsection{Extracting Nahm sums from skeins}\label{sec:Z-from-skein}
This section shows how to obtain Nahm sums from the topological construction we introduced earlier. Importantly for later considerations, Nahm sums are specializations of the motivic generating series of symmetric quivers.
\newline

In order to get a number out of the skein $Z \in \widehat{\Sk}_q^{\mathfrak{gl}_1}(Y)$, we will choose an embedding of $Y$ into a bigger 3-manifold $\overline{Y}$ where all the boundaries of the holomorphic disks appearing in $Z$ become homologically trivial. 
Then, the embedding $Y \overset{\iota}{\hookrightarrow} \overline{Y}$ will induce a linear map
\begin{align*}
\widehat{\Sk}_{q}^{\mathfrak{gl}_1}(Y) &\overset{\iota_*}{\rightarrow} \widehat{\Sk}_{q}^{\mathfrak{gl}_1}(\overline{Y})\\
Z &\mapsto \iota_*(Z),
\end{align*}
where the image $\iota_*(Z)$ is now a number (in some completion $\widehat{\mathbb{C}(q)}$ of our base field $\mathbb{C}$) times the empty skein $[\emptyset] \in \widehat{\Sk}_{q}^{\mathfrak{gl}_1}(\overline{Y})$. 
This is because homologically trivial skeins can be made into linear combinations of unlinks (and, in turn, empty links) using the $\mathfrak{gl}_1$-skein relations. 
In this way, we get a ``partition function'' $\iota_*(Z) \in \widehat{\mathbb{C}(q)}$. 
We will be most interested in the case when $Y$ is constructed entirely out of positive flips $T_+$, in which case $\iota_*(Z)$ can be easily computed to be
\begin{equation}\label{eq:Q-q-series}
\iota_*(Z) =  \sum_{d_1,\dots,d_{m} = 0}^{\infty} q^{\sum_{i,j=1}^{m}d_i Q_{ij} d_j} \prod_{i=1}^{m}\frac{((-1)^{l_i}q)^{d_i}}{(q^2; q^2)_{d_i}}\,,
\end{equation}
where $m$ is the number of holomorphic disks (i.e., the number of flips), $Q$ is the $m \times m$ linking matrix computed in $\overline{Y}$ using the embedding $Y \overset{\iota}{\hookrightarrow} \overline{Y}$, 
and $l_i \in \mathbb{Z}/2$ is the mod 2 linking number between the boundary of the $i$-th holomorphic disk and the branch locus. 
Note that this $q$-series is a Nahm sum, which is a specialization of the \emph{motivic generating series of the symmetric quiver} \eqref{eq:motivic_generating_series} associated to the linking matrix $Q$. 

Let us describe this procedure more explicitly. 
One general method to create such $\overline{Y}$ is by capping off the in- and out-boundaries of $Y$ appropriately to create an integer homology sphere $\overline{Y}$; then the desired condition that $Z$ is homologically trivial will be automatically satisfied. 
The two boundaries of $Y$, $\Sigma_{\tau}$ and $\Sigma_{\tau'}$ are surfaces of the same topological type, as the bordism $Y : \Sigma_{\tau} \rightarrow \Sigma_{\tau'}$ simply braids the branch points around. 
Hence, $Y$ is topologically a mapping cylinder of the homeomorphism $\Sigma_{\tau} \overset{\sim}{\rightarrow} \Sigma_{\tau'}$ induced by this braiding of branch points. 
For simplicity, let us assume that $C$ is a punctured surface without any boundary, so that $\Sigma_{\tau}$ and $\Sigma_{\tau'}$ are topologically both surfaces of genus $g$ with $n$ punctures, for some $g$ and $n$; 
if $C$ had any boundary intervals, we can always attach some more ideal triangles to close up the boundary. 
In this setup, we choose a cup bordism
\[
Y_{in} : \emptyset \rightarrow \Sigma_{\tau}
\]
which, if we reverse time, can be thought of as a bordism that first fills in all the punctures to get a~surface of genus $g$ and no punctures, and then attaches a~genus $g$ handlebody; 
that is, it contracts $g$ number of ``$\mu$-curves''. 
Likewise, we can choose a set of ``$\nu$-curves'' on $\Sigma_{\tau'}$ so that $\{\mu_i, \nu_i\}_{i = 1, \cdots, g}$ form a~symplectic basis, 
and use them to get a cap bordism
\[
Y_{out} : \Sigma_{\tau'} \rightarrow \emptyset. 
\]
In this way, the composite bordism 
\[
\emptyset \overset{Y_{in}}{\rightarrow}\Sigma_{\tau} \overset{Y}{\rightarrow} \Sigma_{\tau'}\overset{Y_{out}}{\rightarrow} \emptyset
\]
becomes an integral homology sphere $\overline{Y}$ after filling in all the punctures. 
In this way, we get the desired embedding
\[
Y \overset{\iota}{\hookrightarrow} \overline{Y}. 
\]

To get an induced map of $\mathfrak{gl}_1$-skein modules, however, we have to make sure that the $\overline{Y}$ carries the branch locus that is compatible with that of $Y$. 
Since we are assuming that $C$ has no boundary, it has an even number of branch points, 
so we can choose the $\mu$- and $\nu$-curves in such a way that contraction of each $\mu$- and $\nu$-curve corresponds to colliding two branch points, annihilating the pair. 
Then, the branch locus extends naturally to both the cap and cup bordisms, and we get the desired embedding $Y \overset{\iota}{\hookrightarrow} \overline{Y}$ that is compatible with the branch locus. 
In this way, we get the desired map
\[
\widehat{\Sk}_{q}^{\mathfrak{gl}_1}(Y) \overset{\iota_*}{\rightarrow} \widehat{\Sk}_{q}^{\mathfrak{gl}_1}(\overline{Y}),
\]
and the associated Nahm sum $\iota_*(Z) \in \widehat{\mathbb{C}(q)}$. 

For some special classes of $Y$, there are simpler ways to construct $\overline{Y}$ making $Z$ homologically trivial, without constructing an integer homology sphere. 
One such example is given in Figure \ref{fig:capping_example}, where $Y$ is a~branched double cover of a pentagon times $I$ given by a sequence of two flips (so $Y$ is topologically the torus with one boundary component times $I$), and $\overline{Y}$ is a bordism from a disk to a disk (each with one branch point) obtained by connecting some pairs of branch points on the bottom and the top of $Y$. 
It is easy to see that both $\Psi_{ab}$ and $\Psi_{bc}$ are homologically trivial in $\overline{Y}$: we can pull $\Psi_{ab}$ down and pull $\Psi_{bc}$ up. 
\begin{figure}[H]
\centering
\begin{math}
\vcenter{\hbox{
\tdplotsetmaincoords{65}{52}
\begin{tikzpicture}[tdplot_main_coords]
\begin{scope}[scale = 0.7, tdplot_main_coords]
    \newcommand*{\defcoords}{
        \coordinate (o) at (0, 0, 0);
        \coordinate (a) at (3, 0, 0);
        \coordinate (b) at ({3*cos(72)}, {3*sin(72)}, 0);
        \coordinate (c) at ({3*cos(2*72)}, {3*sin(2*72)}, 0);
        \coordinate (d) at ({3*cos(3*72)}, {3*sin(3*72)}, 0);
        \coordinate (e) at ({3*cos(4*72)}, {3*sin(4*72)}, 0);
        \coordinate (abc) at ($1/3*(a)+1/3*(b)+1/3*(c)$);
        \coordinate (abe) at ($1/3*(a)+1/3*(b)+1/3*(e)$);
        \coordinate (ace) at ($1/3*(a)+1/3*(c)+1/3*(e)$);
        \coordinate (bcd) at ($1/3*(b)+1/3*(c)+1/3*(d)$);
        \coordinate (bce) at ($1/3*(b)+1/3*(c)+1/3*(e)$);
        \coordinate (bde) at ($1/3*(b)+1/3*(d)+1/3*(e)$);
        \coordinate (cde) at ($1/3*(c)+1/3*(d)+1/3*(e)$);
    }
    \defcoords
    \draw[ultra thick, red] (cde) .. controls ($(cde) + (0, 0, -1)$) and ($(ace) + (0, 0, -1)$) .. (ace);
    \draw[very thick] (a) -- (b);
    \draw[very thick] (b) -- (c);
    \draw[very thick] (c) -- (d);
    \draw[very thick] (d) -- (e);
    \draw[very thick] (e) -- (a);
    \draw[very thick] (a) -- (c);
    \draw[very thick] (c) -- (e);
    \filldraw (a) circle (0.05em);
    \filldraw (b) circle (0.05em);
    \filldraw (c) circle (0.05em);
    \filldraw (d) circle (0.05em);
    \filldraw (e) circle (0.05em);
    \draw[white, line width=5] (cde) -- ($(cde)+(0,0,3)$);
    \draw[ultra thick, red] (cde) -- ($(cde)+(0,0,3)$);
    \draw[white, line width=5] (abc) -- ($(abe)+(0,0,3)$);
    \draw[ultra thick, red] (abc) -- ($(abe)+(0,0,3)$);
    \draw[white, line width=5] (ace) -- ($(bce)+(0,0,3)$);
    \draw[ultra thick, red] (ace) -- ($(bce)+(0,0,3)$);
    \filldraw[red] (abc) circle (0.2em);
    \filldraw[red] (ace) circle (0.2em);
    \filldraw[red] (cde) circle (0.2em);
    \node[red, anchor=east] at (cde){$a$};
    \node[red, anchor=west] at (ace){$b$};
    \node[red, anchor=north] at (abc){$c$};
    \draw[blue, line width=3] ($1/2*(abc) + 1/2*(abe) + (0, 0, 3/2)$) -- ($1/2*(ace) + 1/2*(bce) + (0, 0, 3/2)$);
    \node[blue, anchor=north] at ($1/4*(abc) + 1/4*(abe) + 1/4*(ace) + 1/4*(bce) + (0, 0, 3/2)$){$\Psi_{bc}$};
    \begin{scope}[shift={(0, 0, 3)}]
        \defcoords
        \draw[very thick] (a) -- (b);
        \draw[very thick] (b) -- (c);
        \draw[very thick] (c) -- (d);
        \draw[very thick] (d) -- (e);
        \draw[very thick] (e) -- (a);
        \draw[very thick] (b) -- (e);
        \draw[very thick] (c) -- (e);
        \filldraw (a) circle (0.05em);
        \filldraw (b) circle (0.05em);
        \filldraw (c) circle (0.05em);
        \filldraw (d) circle (0.05em);
        \filldraw (e) circle (0.05em);
        \draw[white, line width=5] (abe) -- ($(abe)+(0,0,3)$);
        \draw[ultra thick, red] (abe) -- ($(abe)+(0,0,3)$);
        \draw[white, line width=5] (cde) -- ($(bcd)+(0,0,3)$);
        \draw[ultra thick, red] (cde) -- ($(bcd)+(0,0,3)$);
        \draw[white, line width=5] (bce) -- ($(bde)+(0,0,3)$);
        \draw[ultra thick, red] (bce) -- ($(bde)+(0,0,3)$);
        \filldraw[red] (abe) circle (0.2em);
        \filldraw[red] (bce) circle (0.2em);
        \filldraw[red] (cde) circle (0.2em);
        \draw[blue, line width=3] ($1/2*(cde) + 1/2*(bcd) + (0, 0, 3/2)$) -- ($1/2*(bce) + 1/2*(bde) + (0, 0, 3/2)$);
        \node[blue, anchor=north] at ($1/4*(cde) + 1/4*(bcd) + 1/4*(bce) + 1/4*(bde) + (0, 0, 3/2)$){$\Psi_{ab}$};
    \end{scope}
    \begin{scope}[shift={(0, 0, 6)}]
        \defcoords
        \draw[very thick] (a) -- (b);
        \draw[very thick] (b) -- (c);
        \draw[very thick] (c) -- (d);
        \draw[very thick] (d) -- (e);
        \draw[very thick] (e) -- (a);
        \draw[very thick] (b) -- (e);
        \draw[very thick] (b) -- (d);
        \draw[white, line width=5] (bde) .. controls ($(bde) + (0, 0, +1)$) and ($(abe) + (0, 0, +1)$) .. (abe);
        \draw[ultra thick, red] (bde) .. controls ($(bde) + (0, 0, +1)$) and ($(abe) + (0, 0, +1)$) .. (abe);
        \filldraw (a) circle (0.05em);
        \filldraw (b) circle (0.05em);
        \filldraw (c) circle (0.05em);
        \filldraw (d) circle (0.05em);
        \filldraw (e) circle (0.05em);
        \filldraw[red] (abe) circle (0.2em);
        \filldraw[red] (bcd) circle (0.2em);
        \filldraw[red] (bde) circle (0.2em);
        \node[red, anchor=south] at (bcd){$a$};
        \node[red, anchor=east] at (bde){$b$};
        \node[red, anchor=west] at (abe){$c$};
    \end{scope}
\end{scope}
\end{tikzpicture}
}}
\end{math}
\caption{The two holomorphic disks $\Psi_{ab}$ and $\Psi_{bc}$ in this figure have a well-defined linking number. }
\label{fig:capping_example}
\end{figure}
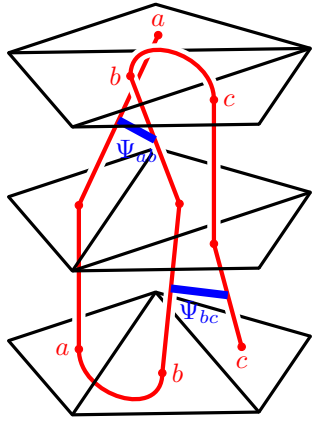
This example generalizes to a class of examples called the minimal chamber of $A_{m}$ Argyres-Douglas theory (Section \ref{sec:A2n-theory}), where a choice of $\overline{Y}$ can be made by pairing the branch locus in a similar zig-zag manner as in Figure \ref{fig:capping_pattern_A2n}. 
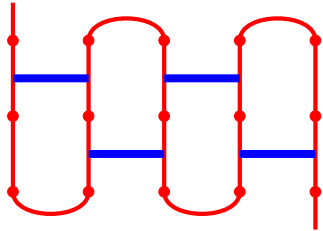
\begin{figure}[H]
\centering
\begin{math}
\vcenter{\hbox{
\begin{tikzpicture}
\filldraw[red] (0, 0) circle (0.2em);
\filldraw[red] (0, 1) circle (0.2em);
\filldraw[red] (0, 2) circle (0.2em);
\draw[ultra thick, red] (0, 2.5) -- (0, 0) to[out=-90, in=-90] (1, 0) -- (1, 2) to[out=90, in=90] (2, 2) -- (2, 0) to[out=-90, in=-90] (3, 0) -- (3, 2) to[out=90, in=90] (4, 2) -- (4, -0.5);
\filldraw[red] (1, 0) circle (0.2em);
\filldraw[red] (1, 1) circle (0.2em);
\filldraw[red] (1, 2) circle (0.2em);
\filldraw[red] (2, 0) circle (0.2em);
\filldraw[red] (2, 1) circle (0.2em);
\filldraw[red] (2, 2) circle (0.2em);
\filldraw[red] (3, 0) circle (0.2em);
\filldraw[red] (3, 1) circle (0.2em);
\filldraw[red] (3, 2) circle (0.2em);
\filldraw[red] (4, 0) circle (0.2em);
\filldraw[red] (4, 1) circle (0.2em);
\filldraw[red] (4, 2) circle (0.2em);
\draw[blue, line width=3] (0, 1.5) -- (1, 1.5);
\draw[blue, line width=3] (1, 0.5) -- (2, 0.5);
\draw[blue, line width=3] (2, 1.5) -- (3, 1.5);
\draw[blue, line width=3] (3, 0.5) -- (4, 0.5);
\end{tikzpicture}
}}
\end{math}
\caption{Branch locus pattern for a choice of $\overline{Y}$ in the minimal chamber of $A_{m}$ Argyres-Douglas theory}
\label{fig:capping_pattern_A2n}
\end{figure}

\subsection{Left and right modules of the quantum torus algebra}
This section provides a topological perspective on the left and right modules of the quantum torus algebra; broader discussion is presented in Section \ref{sec:embedding}.
\newline

It is worth noting that the skein module of a 3-manifold obtained by gluing two handlebodies can be computed by taking the relative tensor product:\footnote{For a general skein module, we need to take the invariant part of the relative tensor product of \emph{internal skein modules} \cite[Thm. 2]{GJS}, but for $\mathfrak{gl}_1$-skein modules, the internal skein module coincides with the usual skein module (i.e., it is already invariant), and we get the simple formula.} 
\begin{equation}
    \Sk_q^{\mathfrak{gl}_1}(\overline{Y}) 
\cong 
\Sk_{q}^{\mathfrak{gl}_1}(Y_{out}) \underset{\SkAlg_{q}^{\mathfrak{gl}_1}(\Sigma)}{\otimes} \Sk_{q}^{\mathfrak{gl}_1}(Y_{in})
. 
\end{equation}
Therefore, it is useful to understand $\Sk_{q}^{\mathfrak{gl}_1}(Y_{out})$ and $\Sk_{q}^{\mathfrak{gl}_1}(Y_{in})$ as right and left modules over the quantum torus algebra $\SkAlg_{q}^{\mathfrak{gl}_1}(\Sigma)$. 

Suppose that the handlebody $Y_{in}$ is obtained by contracting $\mu$-curves $\mu_1, \cdots, \mu_g$, each surrounding exactly two branch points which get annihilated in pairs in $Y_{in}$.
That would mean that, under the action of corresponding mutually commutative operators $L_{\mu_1}, \cdots, L_{\mu_g} \in \SkAlg_{q}^{\mathfrak{gl}_1}(\Sigma)$, 
the empty skein $[\emptyset] \in \Sk_{q}^{\mathfrak{gl}_1}(Y_{in})$ remains invariant. 
Hence, we can view $\Sk_{q}^{\mathfrak{gl}_1}(Y_{in})$ as a vector space with the ``ground state''
\begin{equation}
    \ket{0, \cdots, 0} := [\emptyset] \in \Sk_{q}^{\mathfrak{gl}_1}(Y_{in})
\end{equation}
and a basis
\begin{equation}
    \ket{d_1, \cdots, d_g} := L_{\nu_1}^{d_1} \cdots L_{\nu_g}^{d_g}[\emptyset]\in \Sk_{q}^{\mathfrak{gl}_1}(Y_{in}),\quad d_1, \cdots, d_g \in \mathbb{Z}.
\end{equation}
Hence, when $\overline{Y}$ is an integer homology sphere so that $\Sk_q^{\mathfrak{gl}_1}(\overline{Y}) \cong \widehat{\mathbb{C}(q)}$, 
we can express our partition function $\iota_*(Z)$ in the following form:
\begin{equation}\label{eq:matix_element_from_topology}
    \bra{0, \cdots, 0} \Psi_m \cdots \Psi_1 \ket{0, \cdots, 0} \in \widehat{\mathbb{C}(q)},
\end{equation}
where $\Psi_i$ denotes the element of $\widehat{\Sk}^{\mathfrak{gl}_1}_q(Y)$ corresponding to the $i$-th flip, 
and $\bra{0, \cdots, 0} := [\emptyset] \in \Sk_{q}^{\mathfrak{gl}_1}(Y_{out})$.


\section{Physics of 3d and 4d BPS states and associated quivers}\label{sec:3d4d}

In this section, we describe in detail a class of 3d-4d systems engineered by a pair of M5-branes wrapping a 3-manifold with boundary.
We discuss the relevant geometries and derive the low-energy descriptions for both the 3d and 4d theories involved, explaining how each is related to the respective quivers (BPS quivers for 4d $\mathcal{N}=2$ theories, and symmetric quivers for 3d $\mathcal{N}=2$ theories). 
Our analysis leads to a~natural relation between 3d and 4d quivers, whose underlying physical origin lies in the Witten effect for Abelian (effective) gauge theories on a manifold with boundary.
One of the main results of this section is the first instance of the \emph{symmetrization map} $\mathfrak{S}$ that relates 4d and 3d quivers, which will be further extended in later sections.

\subsection{3d-4d systems from M-theory}
In this section we describe the construction of 3d-4d system that is the environment in which the symmetrization map works.
\newline

Triangulated 3-manifolds describe BPS sectors of M-theory in the background of $T^*M\times S^1\times \mathbb{R}^4$ for some smooth 3-manifold $M$. 
The worldvolume dynamics of a stack of $N=2$ M5 branes 
wrapping $M\times S^1\times \IR^2$ is described by a twisted compactification of the 6d $(2,0)$ theory of type $\mathfrak{g}=A_1$, which gives rise to a 3d $\CN=2$ QFT denoted $T[M]$, on $S^1\times \IR^2$ \cite{Dimofte:2011ju, Terashima:2011qi}. 
We consider a deformation of this setup defined by a splitting of $M$
\begin{equation}
\label{eq:M-split}
	M  = M_+ \cup M_0 \cup M_-\,,
\end{equation}
where $M_0\simeq C\times I$ for some punctured Riemann surface $C$, and $M_\pm$ are bordisms to/from the empty set and $C$. 
More precisely, we will consider a certain degeneration of the mapping cylinder $M_0 \simeq C\times I$, defined by shrinking $I\to \{\text{pt}\}$ along the boundary $\partial C$. 
As a result, the boundary components of 3-manifolds that we consider are
\be
	\partial M_- = C\,,
	\quad
	\partial M_+ = \overline{C}\,,
	\quad
	\partial M_0 = C\cup \overline{C}\,,
\ee

Given any component $M_\theta\in\{M_0, M_+, M_-\}$ of the decomposition, we consider a pair of M5-branes wrapping a manifold with corners 
\be
	\left(M_\theta \times S^1\times \IR^2 \right) 
	\cup
	\left(\partial M_\theta \times S^1\times \IR^2 \times \IR_{> 0} \right) \,.
\ee
M5-branes compactified on this background are described at low energies by a 3d-4d system involving a 4d $\CN=2$ QFT $T[C]$ of class $S$ for each component of $\partial M_\theta$, coupled to a 3d $\CN=2$ QFT $T[M_\theta]$.
We obtain, in this way, the following systems: 
\begin{itemize}
\item A copy of $T[C]$ on the half-space $S^1\times \IR^2\times \IR_{> 0}$, coupled to either of the boundary QFTs $T[M_\pm]$
\item A domain wall QFT $T[M_0]$ interpolating between two copies of $T[C]$ on half spaces glued to either side of the wall
\end{itemize}
We also consider gluing these pieces along the semi-infinite direction of $T[C]$, giving the 3d-4d system defined by twisted compactification of the 6d $A_1$ theory on
\be
	\left[M_- \cup \left(C \times I_-\right)  \cup M_0 \cup \left(C \times I_+\right) \cup M_+\right] \times \IR^2 \times S^1\,.
\ee
The worldvolume description is given by the 3d-4d system summarized in Figure \ref{fig:M-theory-setup}, schematically denoted
\be
	T[M_-] \star_{T[C]}  T[M_0] \star_{T[C]}   T[M_+]
\ee
where $\star_{T[C]}$ denotes a coupling between two 3d $\CN=2$ QFTs mediated by $T[C]$ on $S^1\times \IR^2\times I_\pm$. The two 4d QFTs interact through the domain wall theory $T[M_0]$ along $S^1\times \IR^2$, and are bounded at the opposite ends of respective intervals $I_\pm$ by $T[M_\pm]$ on $S^1\times \IR^2$. 

Note that, since $\partial M_- = C\times \{0\}$ and $\partial M_0 = C\times \{1\}$ are located at \emph{opposite} ends of the interval $I_- = [0,1]$ in a direction transverse to that of the filling of $C$, the different pieces of the 3-manifold $M$ do not match up, but instead feature corners.
The 3d theory $T[M]$ corresponds to a limit of this 3d-4d system obtained by shrinking both $I_\pm$ so that the 4d theories $T[C]$ disappear, and the three manifolds glue up as in \eqref{eq:M-split}: 
\be
	T[M] \simeq \lim_{I_\pm \to \mathrm{pt}} T[M_-] \star_{T[C]}  T[M_0] \star_{T[C]}   T[M_+]\,.
\ee

\begin{figure}[h!]
\begin{center}
\includegraphics[width=0.6\textwidth]{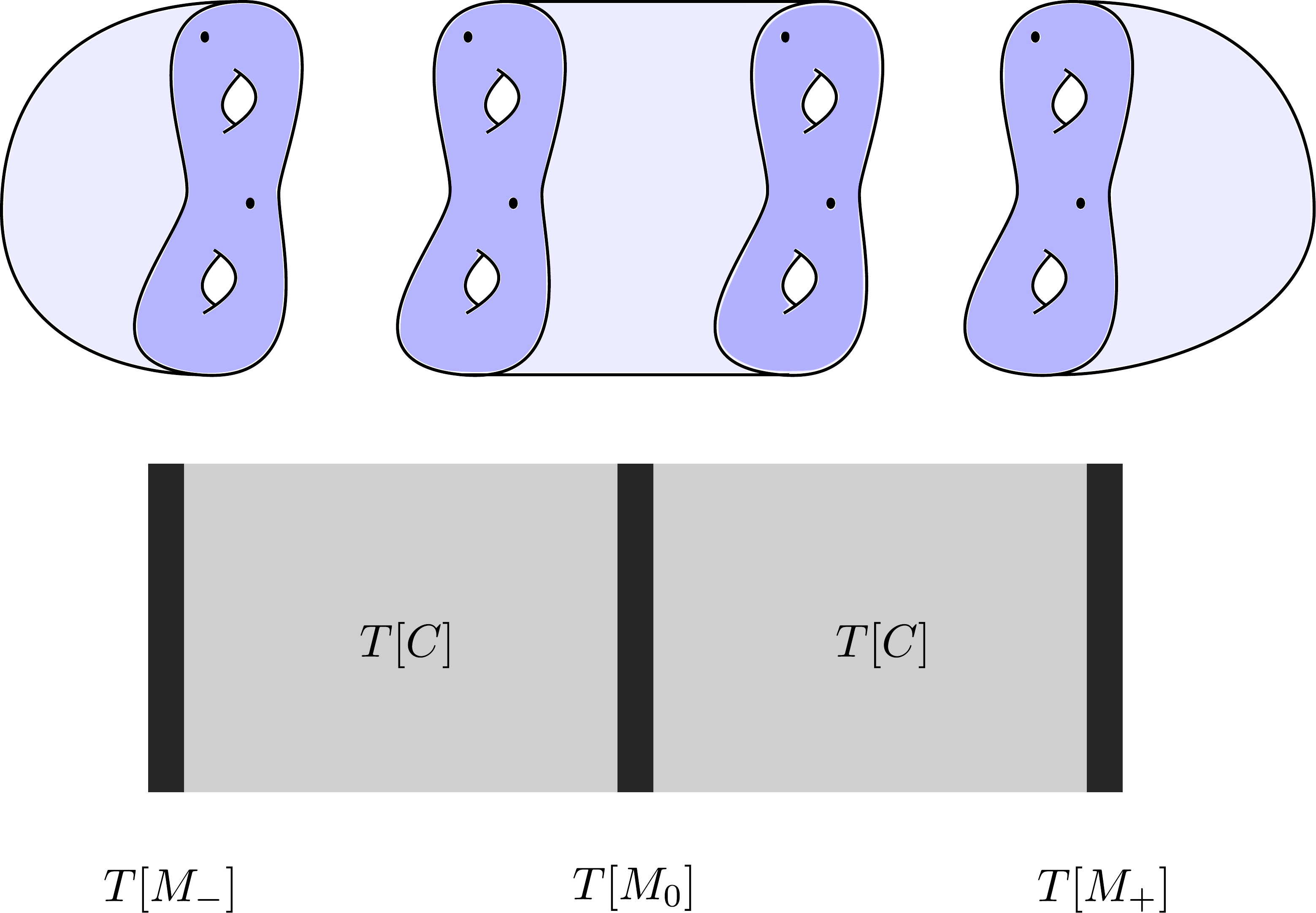}
\caption{The internal manifold wrapped by M5 branes (top) and the dual QFTs (bottom).}
\label{fig:M-theory-setup}
\end{center}
\end{figure}

\subsection{Lagrangian descriptions}

In this section we analyse the low-energy behaviour of the 3d-4d system that we introduced above, which will provide the basis for two quiver descriptions connected by the symmetrization map $\mathfrak{S}$.

\paragraph{Triangulation and polarization data}

Following \cite{Dimofte:2011ju}, the 3d $\CN=2$ QFT associated to a 3-manifold $M$ with boundary is uniquely defined by the topology and geometry of $M$.
However, a specific Lagrangian description $T[M_0; \tau_{M_0}, \Pi]$ is labeled by additional data: a choice of ideal triangulation $\tau_M$ by tetrahedra, and a choice of polarization $\Pi$ for the $SL(2,\IC)$ character variety of its boundary 
$S=\partial M$. 

Since $\tau_M$ induces a boundary triangulation $\tau_S$, there is a natural local parametrization of the character variety $\CP_{S} :=\CM_{\mathrm{flat}}(S,SL(2,\IC))$ in terms of complexified shear coordinates
 \cite{NEUMANN1985307, fock1997dual, thurston1998minimal, fock2006moduli}.
Each edge $E$ of $\tau_S$ corresponds to a coordinate $\Xi_E \in \IC / 2\pi i \IZ$, and the Weil-Petersson Poisson structure on $\CP_S$ is encoded by incidence relations among edges: 
\be\label{eq:surface-FG-Poisson-structure}
	\{\Xi_E,\Xi_{E'}\} = \sum_{\text{shared faces}} \pm 1 \qquad  \in \{0, \pm1, \pm 2\}
\ee
where the sign is positive if $E$ is counterclockwise from $E'$ in a shared face, and negative otherwise.\footnote{Our conventions agree with those of \cite{Gaiotto:2009hg, Dimofte:2011ju}.} 
A~polarization $\Pi$ of $\CP_S$ consists of a choice of Darboux basis, namely a set of coordinates with canonical Poisson brackets
\be\label{eq:Darboux}
	\{Y_i, X_j\} = \delta_{ij} ,\qquad 
	\{X_i, X_j\} = \{Y_i, Y_j\} = 0\,, \qquad i,j\in\{1,\dots, d\}\,.
\ee


\paragraph{The tetrahedron theory}
The building block of Lagrangian descriptions of 3d $\CN=2$ QFTs associated with ideally triangulated 3-manifolds is the tetrahedron theory
\be\label{eq:tetrahedron-theory}
	\begin{split}
	\CT_\Delta := 
	\end{split}
	\ \ 
	\left\{\begin{split}
	& \text{3d $\CN=2$ chiral multiplet coupled to a background $U(1)$ symmetry}
	\\[-4pt]
	&\text{with Chern-Simons level $-\frac{1}{2}$}
	\end{split}
	\right\}
\ee
The negative half-integer level for the Chern-Simons coupling is offset by one-loop corrections at the quantum level \cite{Aharony:1997bx}, resulting in an effective zero Chern-Simons level. 
This theory has a duality group $Sp(2,\IZ)\ltimes (i\pi \IZ)^2$, where the $T$-generator acts by shifts of the background Chern-Simons coupling, while the $S$-generator acts by gauging the background $U(1)$ symmetry. 
An important feature of $\CT_\Delta$ is that it enjoys a self-triality, in the sense that
\be
	\widetilde{ST} \circ \CT_\Delta \simeq \CT_\Delta
\ee
retains an identical Lagrangian description, but with a different identification of fundamental degrees of freedom. 
Here, $\widetilde{ST}$ is defined as the affine symplectic transformation corresponding to $ST$ on $(X, Y)^t$ followed by a shift of the position coordinate $X\to X-i\pi$. 
The geometric origin of triality is the existence of three equivalent choices of polarization. 
Let $\Xi, \Xi', \Xi''$ denote shear coordinates of boundary edges, as in Figure \ref{fig:tetrahedron} (opposite edges have identical coordinates).
These are linearly dependent through the following relation:
\be
	\Xi + \Xi' +\Xi'' = i\pi\,.
\ee
The three natural polarizations of the tetrahedron are then
\be\label{eq:delta-pol}
\begin{array}{c|cc}
	& X & Y \\
	\hline
	\Pi & \Xi & \Xi'\\
	\Pi' & \Xi' & \Xi''\\
	\Pi'' & \Xi'' & \Xi\\
\end{array}
\qquad \text{with} \qquad \{\Xi,\Xi''\} = \{\Xi',\Xi\} = \{\Xi'',\Xi'\} = 1\,.
\ee
On these, $\widetilde{ST}$ acts by a cyclic permutation: 
\be
	\widetilde{ST} : \ \Pi\to\Pi''\to\Pi'\to\Pi\,.
\ee
For example,
\be
	\widetilde{ST}
	\left(\begin{array}{c}
		\Xi\\
		\Xi'
	\end{array}\right)
	=
	\left(\begin{array}{cc}
		0 & -1\\
		1 & 0
	\end{array}\right)
	\left(\begin{array}{cc}
		1 & 0 \\
		1 & 1
	\end{array}\right)
	\left(\begin{array}{c}
		\Xi\\
		\Xi'
	\end{array}\right)
	+
	\left(\begin{array}{c}
		-\pi i\\
		0
	\end{array}\right)
	=
	\left(\begin{array}{c}
		\Xi'' \\
		\Xi
	\end{array}\right).
\ee

\begin{figure}[h!]
\begin{center}
\includegraphics[width=0.25\textwidth]{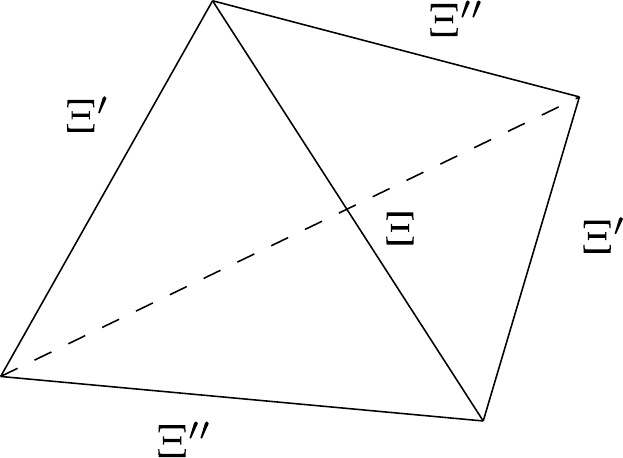}
\caption{Shear coordinates on edges of a tetrahedron.}
\label{fig:tetrahedron}
\end{center}
\end{figure}

\paragraph{Gluing tetrahedra}\label{sec:gluing}

To build $T[M;\tau_M,\Pi]$, we begin by fixing a choice of polarization for each of the individual tetrahedra in $\tau_M$, denoted $\{\Pi_1,\dots, \Pi_N\}$ from one of the three types in \eqref{eq:delta-pol}. 
Thanks to the triality property, the choice among $\Pi, \Pi',\Pi''$ for each tetrahedron is inessential at this stage. 
The product theory $\bigotimes_i \CT_{\Delta_i}$ is defined as a collection of $N$ free hypermultiplets with a global $U(1)^N$ symmetry. 
This has a duality group $Sp(2N,\IZ)\ltimes (i\pi \IZ)^{2N}$ which acts on Darboux coordinates $(X_1, \dots , X_N; Y_1,\dots Y_N)$ generated by three types of transformations
\begin{itemize}
\item
shifts of Chern-Simons couplings enacted by ``$T$-type'' transformations
\be\label{eq:TB-def}
	T_B:= \left(\begin{array}{cc}
	I_N & 0 \\
	B & I_N
	\end{array}\right) \qquad B = B^t
\ee
where $I_N$ is the $N\times N$ identity matrix
\item
gauging of symmetries
enacted by ``$S$-type'' transformations
\be\label{eq:SJ-def}
	S_{J}:=\left(\begin{array}{cc}
	I_N - {J} & -{J} \\
	{J} & I_N - {J}
	\end{array}\right) \qquad {J} = \text{diag}(j_1,\dots, j_N)
\ee
with $j_a\in\{0,1\}$ equal to 1 if $a$ labels $U(1)_a \subset U(1)^N$ to be gauged, and zero otherwise
\item
general linear transformations acting by field redefinitions
\be
	\left(\begin{array}{cc}
	H  & 0 \\
	0 & H 
	\end{array}\right) \qquad H \in GL(N,\IC)\,.
\ee
\end{itemize}
We then define $g\in Sp(2N,\IZ)\ltimes (i\pi \IZ)^{2N}$ as the unique duality transformation that relates the polarization of the product theory $\bigotimes_i \CT_{\Delta_i}$ to the polarization $\Pi$ of the desired Lagrangian description
\be\label{eq:change-of-polarization}
	\Pi = g\circ \{\Pi_i\}\,.
\ee
A decomposition of $g$ into generators of the three types described above defines a sequence of operations on $\bigotimes_i\CT_{\Delta_i}$ consisting of gaugings, shifts of Chern-Simons couplings, and field redefinitions, which eventually produces $T[M; \tau_{M}, \Pi]$.
We will illustrate this shortly with the example of main relevance to our work.
Last, but not least, if the triangulation contains any internal edges, these contribute terms to the superpotential of the theory
\be
	\CW = \sum_{I} \CO_I
\ee
where $\CO_I = \prod e^{X_{E_i}}$ involves a product of edges glued to $I$ from all incident tetrahedra. 
The role of $\CW$ is to enforce a breaking of $U(1)^N$ to an appropriate subgroup, completing the definition of the Lagrangian description associated to $\Pi$.

\subsection{Domain walls for \texorpdfstring{$A_m$}{Am} Argyres-Douglas theories and symmetric quivers}\label{sec:A2n-theory}

We now restrict attention to 4d $\mathcal{N}=2$ Argyres-Douglas theories of type $A_m$, also known as $(A_1, A_m)$ theories \cite{Gaiotto:2009hg, Cecotti:2010fi}.\footnote{The index $m$ labels the 4d theory under discussion, but in fact also coincides with the number of vertices of the corresponding 4d BPS quiver, as will be recalled below.}
We construct BPS domain walls associated with BPS states in the minimal chamber of their Coulomb branch, and show that the worldvolume theory on the domain wall admits a~low-energy Lagrangian description captured by a symmetric quiver $Q$. This relation between the 4d BPS quiver $Q_{4\text{d}}$ and the symmetric quiver $Q$ provides the first example of the symmetrization map $\mathfrak{S}$. 

\paragraph{The minimal chamber and 4d BPS quiver}

We now consider the specific case of $M_0$ with $C\simeq \IC$ given by a triangulated $(m+3)$-gon, as shown in Figure \ref{fig:A2n-M0} for $m=6$. 
The corresponding 4d $\CN=2$ theory $T[C]$ is (a deformation of) the $A_{m}$ Argyres-Douglas theory, and this specific type of triangulation arises in a chamber of its Coulomb branch with the minimal BPS spectrum \cite{Gaiotto:2009hg}.

In a generic Coulomb vacuum $u$, the theory $T[C]$ flows to a $U(1)^n$ effective gauge theory with $n = \lfloor \frac{m}{2} \rfloor$. 
Edges of the triangulation $\tau_C$ are labeled by IR electromagnetic charges $\alpha_1, \dots, \alpha_{m}$, corresponding to generators of the electromagnetic and flavour charge lattice. 
More precisely, if $m=2n$ the flavour symmetry is trivial, but when $m=2n+1$ there is a $U(1)$ flavour symmetry with a corresponding rank-one flavour sublattice generated by $\alpha_f = \sum_{i=1}^{2n+1} \alpha_i$.
The Dirac pairing on charges is given by
\be\label{eq:4d-Dirac-pairing}
    \langle\alpha_{2i},\alpha_{2i+1}\rangle=\langle\alpha_{2i},\alpha_{2i-1}\rangle = 1
\ee
and zero otherwise.
BPS states at a point $u$ in the Coulomb branch are encoded by a BPS quiver $Q_{4\text{d}}\equiv Q_{\tau_C}$ shown in Figure \ref{fig:A2n-quiver}.
It is well known \cite{2013arXiv1302.7030B, Alim:2011kw} that $Q_{\tau_C}$ is dual to the \emph{initial} triangulation of~$C$, i.e., the one shown at the bottom of Figure \ref{fig:A2n-M0}.

\begin{figure}[h!]
\begin{center}
\includegraphics[width=0.5\textwidth]{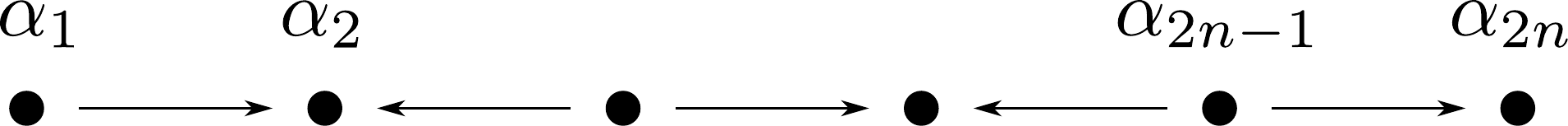}\\[15pt]
\includegraphics[width=0.63\textwidth]{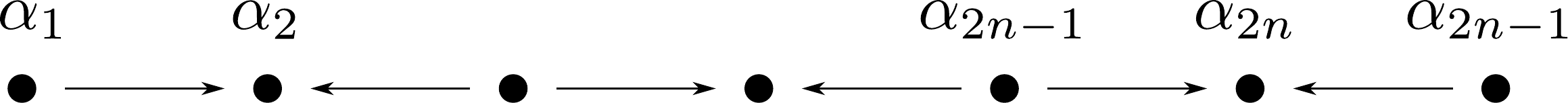}
\caption{4d BPS quivers $Q_{4\text{d}}\equiv Q_{\tau_C}$ corresponding to minimal chambers of $A_{2n}$ and $A_{2n+1}$ theories.}
\label{fig:A2n-quiver}
\end{center}
\end{figure} 

We will choose $u$ within the minimal chamber, defined as the region with BPS central charges ordered as follows: 
\be\label{eq:minimal-chamber-def}
	\arg Z^{4\text{d}}_{\alpha_{\text{even}}} < \arg Z^{4\text{d}}_{\alpha_{\text{odd}}}\,.
\ee
The relative ordering among odd (respectively, even) central charges does not matter, since they are mutually local, and since in this chamber the only stable BPS states are hypermultiplets of charge $\pm \alpha_i$ for each $i\in \{1,\dots, m\}$. 
The mapping cylinder $M_0$ is defined by the BPS spectrum of the minimal chamber as in Figure \ref{fig:A2n-M0}: by first including tetrahedra associated with BPS states of charge $\alpha_{\text{even}}$, and later tetrahedra associated with BPS states of charge $\alpha_{\text{odd}}$. 

\begin{figure}[h!]
\begin{center}
\includegraphics[width=0.7\textwidth]{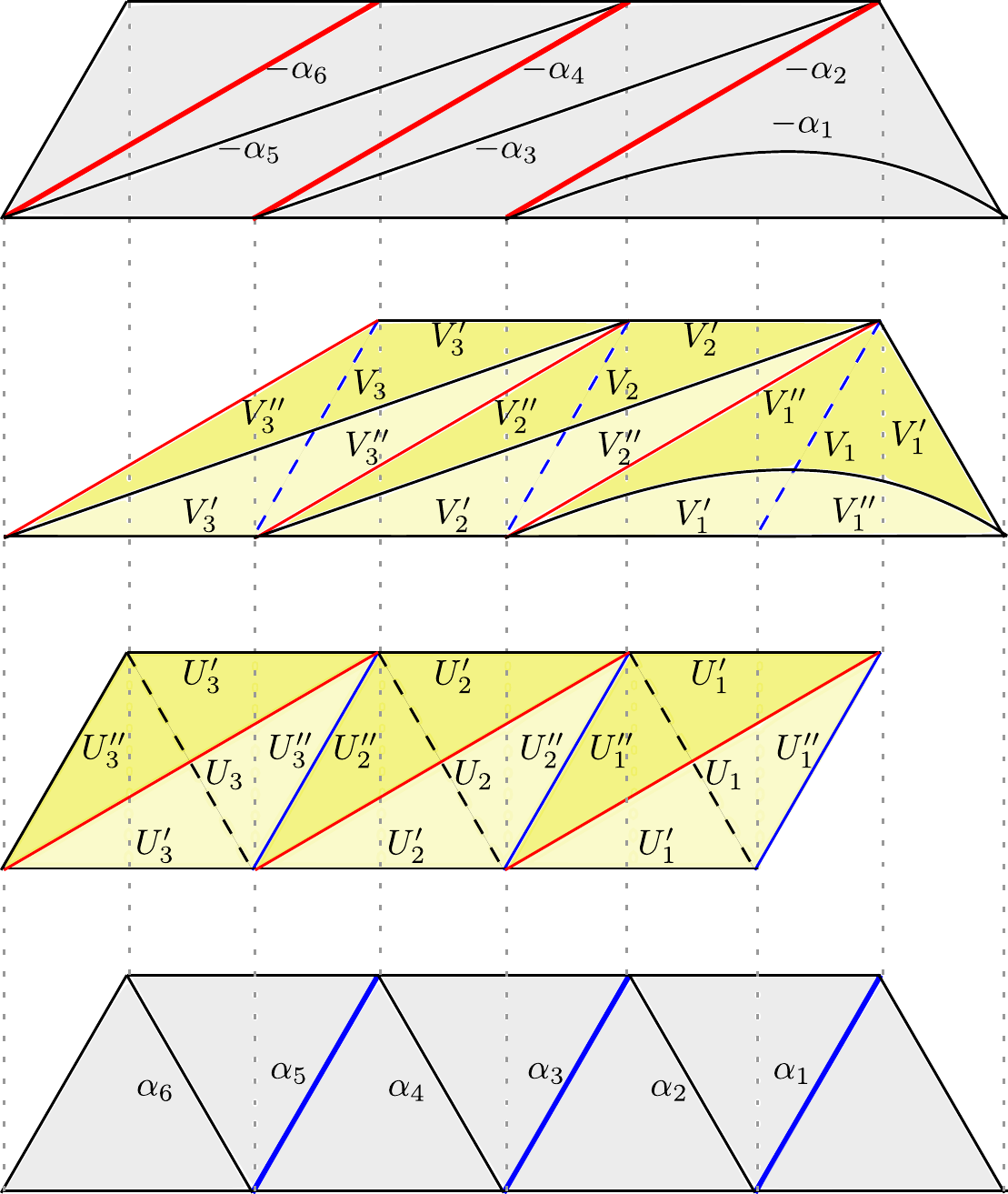}
\caption{The ideal triangulation of $M_0$. Position coordinates for the polarization $\Pi$ corresponding to edges of the boundary triangulation are highlighted in blue ($X_{2i-1}$) and in red $(X_{2i})$.}
\label{fig:A2n-M0}
\end{center}
\end{figure}

\paragraph{Lagrangian description and 3d symmetric quiver}

Let us denote by $U_i, U_i', U_i''$ the edge coordinates on the tetrahedron associated with a BPS state of charge $\alpha_{2i}$, and by  $V_i, V_i', V_i''$ those of the tetrahedron associated with a BPS state of charge $\alpha_{2i-1}$. 
We start with the theory $\bigotimes_{i=1}^{m} \CT_{\Delta_i}$ with polarizations of type $\Pi''$ from \eqref{eq:delta-pol}, given by positions
and momenta: 
\be
\{\Pi_i\}:\qquad\left\{
\begin{aligned}
	&(V''_1, U''_1 ,\dots, V''_n, U''_n; V_1, U_1 ,\dots, V_n, U_n )  \qquad\qquad (m=2n) \\
	&(V''_1, U''_1 ,\dots, V''_{n+1}, U''_{n+1}; V_1, U_1 ,\dots, V_n, U_n )  \quad (m=2n+1) \\
\end{aligned}\right. 
.
\ee
We then define a new polarization $\Pi$ whose positions and momenta
\be
	\Pi: \qquad (X_1, X_2, \dots, X_{m-1}, X_{m}; Y_1, Y_2, \dots, Y_{m-1}, Y_{m}) 
\ee
are related to the former by
\be\label{eq:tetrahedra-to-quiver}
\begin{split}
	X_{2i-1} &= -V_i - U_i'' - U_{i-1}''\\
	X_{2i} &= -U_i - V_i'' - V_{i+1}''
\end{split}
\qquad
\begin{split}
	Y_{2i-1} &= V_i''\\
	Y_{2i} &= U_i''
\end{split}
\ee
with the understanding that indices run from $1$ to $m$, and the respective coordinates are set to zero outside of this range. 
For later convenience, we also record the inverse relations: 
\be\label{eq:quiver-to-tetrahedra}
\begin{split}
	V_i &= -X_{2i-1} - Y_{2i} - Y_{2i-2}\\
	U_i &= -X_{2i} - Y_{2i-1} - Y_{2i+1}
\end{split}
\qquad
\begin{split}
	V_i'' & = Y_{2i-1}\\
	U_i'' & = Y_{2i}
\end{split}
\ee
as well as the relation to boundary shear coordinates for the bottom ($C_-$) and top $(C_+)$ surfaces: 
\be\label{eq:boundary-shear-XY}
\begin{split}
\CP_{C_-}:& \left\{\begin{split}
	X_{\alpha_{2i-1}} & = -V_i - U_i'' - U_{i-1}'' = X_{2i-1}\\
	X_{\alpha_{2i}} & = - U_i = X_{2i} + Y_{2i-1} + Y_{2i+1}\\
\end{split}\right.
\\
\CP_{C_+}:& \left\{\begin{split}
	X_{-\alpha_{2i-1}} & = -V_i = X_{2i-1} + Y_{2i} + Y_{2i-2}\\
	X_{-\alpha_{2i}} & = -U_i - V_{i}''- V_{i+1}'' = X_{2i}\\
\end{split}\right.
.
\end{split}
\ee
Here, a small difference between the cases $m = 2n$ and $m = 2n+1$ is worth mentioning. 
On the one hand, $X_i, Y_i$ with $i = 1,\dots, m$ always come in pairs as Darboux coordinates associated with tetrahedra (after a change of polarization). 
On the other hand, ${\mathrm{dim}}_{\IC} \CP_{C_\pm} = m$. 
If $m = 2n$ the map between bulk and boundary coordinates \eqref{eq:boundary-shear-XY} is invertible, but if $m=2n+1$, it is not. 
The reason is the existence of a flavour sublattice, since the corresponding coordinate is invariant under flips of the triangulation. 
Let $X_{\pm \alpha_{f}}:= X_{\pm(\alpha_1+\dots+\alpha_m)}$ denote the flavour coordinate on $\CP(C_\pm)$, then\footnote{In this description we are omitting the value of the coordinate along the interval, but it should be clear from context that the left hand side is evaluated at one end (the bottom) and the right-hand side at the other end (top).}
\be
	X_{\alpha_{f}} = X_{-\alpha_{f}}\,.
\ee
For illustration, consider the case $m = 1$ (corresponding to the $A_1$ theory). 
Both $C_\pm$ are triangulated 4-gons, related by a flip. 
In fact $C_\pm$ are half-boundaries of a single tetrahedron with coordinates $X_1 = V_1, Y_1 = V_1''$ on the character variety of its boundary. 
However, in the identification with shear coordinates of \emph{internal edges} of $C_\pm$, we only have $X_{\alpha_1} = X_1 = X_{-\alpha_1}$, while $Y_1$ does not appear at all. 
The reason is that $Y_1 = V_1''$ is associated with a \emph{boundary} edge, i.e., an edge that separates $C_+$ from $C_-$. 
The double-covering of $M_0$ in this case is a solid torus, and the cycles corresponding to $X_1$, $Y_1$ are shown in~\eqref{eq:xhat-yhat-A1}.

As a check, the Poisson structure \eqref{eq:Darboux} implies
\be\label{eq:Poisson-structure-alpha-I}
\begin{split}
	\{X_{\alpha_{2i}}, X_{\alpha_{2j-1}}\} &= \{Y_{2i-1} + Y_{2i+1}, X_{2j-1}\} = \delta_{i,j} + \delta_{i,j-1} \,,\\
	\{X_{\alpha_{2i}}, X_{\alpha_{2j}}\} &=  \{X_{\alpha_{2i-1}}, X_{\alpha_{2j-1}}\} = 0\,,\\
\end{split}
\ee
in agreement with \eqref{eq:4d-Dirac-pairing} for $C = C_-$. Similarly, we recover the Poisson structure on shear coordinates on $C_+$ (which has the opposite sign due to orientation reversal): 
\be\label{eq:Poisson-structure-alpha-II}
\begin{split}
	\{X_{-\alpha_{2i}}, X_{-\alpha_{2j-1}}\} &= \{X_{2i}, Y_{2j} + Y_{2j-2}\} = -\delta_{i,j} - \delta_{i,j-1} \,,\\
	\{X_{-\alpha_{2i}}, X_{-\alpha_{2j}}\} &=  \{X_{-\alpha_{2i-1}}, X_{-\alpha_{2j-1}}\} = 0\,,\\
\end{split}
\ee
and we have that shear coordinates on $C_-$ commute with those on $C_+$: 
\be\label{eq:Poisson-structure-alpha-III}
\begin{split}
	\{X_{\alpha_{2i}}, X_{-\alpha_{2j-1}}\} =\{X_{-\alpha_{2i}}, X_{\alpha_{2j-1}}\} = \{X_{\alpha_{2i}}, X_{-\alpha_{2j}}\} &=  \{X_{\alpha_{2i-1}}, X_{-\alpha_{2j-1}}\} = 0 \,.
\end{split}
\ee

For example, consider the case $m = 6$ shown in Figure \ref{fig:A2n-M0}. 
The change of polarization is given in this case by:
\be
\left(\begin{array}{c}
	X_1 \\
	X_2 \\
	X_3 \\
	X_4 \\
	X_5 \\
	X_6 \\
	\hline
	Y_1 \\
	Y_2 \\
	Y_3 \\
	Y_4 \\
	Y_5 \\
	Y_6 \\
\end{array}\right)
=
\left(\begin{array}{cccccc|cccccc}
	0&-1&&&&&-1&&&&& \\
	-1&0&-1&&&&&-1&&&& \\
	&-1&0&-1&&&&&-1&&& \\
	&&-1&0&-1&&&&&-1&& \\
	&&&-1&0&-1&&&&&-1& \\
	&&&&-1&0&&&&&&-1 \\
	\hline
	1&&&&&&&&&&& \\
	&1&&&&&&&&&& \\
	&&1&&&&&&&&& \\
	&&&1&&&&&&&& \\
	&&&&1&&&&&&& \\
	&&&&&1&&&&&& \\
\end{array}\right)
\left(\begin{array}{c}
	V_1'' \\
	U_1'' \\
	V_2'' \\
	U_2'' \\
	V_3'' \\
	U_3'' \\
	\hline
	V_1 \\
	U_1 \\
	V_2 \\
	U_2 \\
	V_3 \\
	U_3 \\
\end{array}\right).
\ee
Note that the positions $X_1,\dots, X_{m}$ correspond to a maximal collection of edges of the boundary triangulation $\tau_S$ that do not share any faces, as highlighted in Figure \ref{fig:A2n-M0}.
In fact, as already noted in \eqref{eq:boundary-shear-XY}, odd coordinates correspond to alternate edges of the \emph{initial} triangulation: 
\be\label{eq:surface-to-quiver-odd}
	X_{2i-1} \equiv X_{\alpha_{2i-1}} \,,
\ee
while even coordinates correspond to alternate edges of the \emph{final} triangulation: 
\be\label{eq:surface-to-quiver-even}
	X_{2i} \equiv X_{-\alpha_{2i}} \,.
\ee

It is straightforward to check that these are indeed Darboux coordinates, i.e., that \eqref{eq:Darboux} is satisfied, using \eqref{eq:delta-pol} for each tetrahedron.
The change of polarization therefore corresponds to an element of $Sp(2m,\IZ)$ given by
\be\label{eq:g-S-TB}
	g = S\cdot T_B 
\ee
where $S = S_J$, as given in \eqref{eq:SJ-def} with $J = I_{m}$ (all symmetries gauged), and $T_B$ is given by \eqref{eq:TB-def}, with
\be\label{eq:3d-quiver-adjacency-matrix}
	Q_{ij} = \delta_{i,j+1} + \delta_{i,j-1}. 
\ee

The Chern-Simons-Matter gauge theory description that emerges is as follows. 
Starting with $\bigotimes_i \CT_{\Delta_i}$, which is a theory of $k$ free chiral multiplets, we first act with $T_B$ by introducing (effective) mixed Chern-Simons couplings for the $U(1)^{m}$ flavour symmetry rotating the free chirals: 
\be\label{eq:CS-quiver}
	\kappa^{\text{effective}}_{ij} = Q_{ij}\,.
\ee
Following with the action by $S$, we maximally gauge the global symmetry, thereby passing to a $U(1)^{m}$ gauge theory with $m$ chiral multiplets, each charged under only one of the $U(1)$ gauge factors. 

The theory we just described belongs to the class of theories $T[Q]$ encoded by symmetric quivers $Q$, introduced in \cite{Ekholm:2018eee}. 
In this case, the quiver adjacency matrix is given by $Q_{ij}$, implying that $Q$ is a~symmetrization of the $A_{m}$ Dynkin quiver (see Figure \ref{fig:A2n-sym-quiver}).
\be
    \boxed{
    T[M_0; \tau_{M_0}, \Pi] \equiv T[Q]
    }
\ee

\begin{figure}[h!]
\begin{center}
\includegraphics[width=0.5\textwidth]{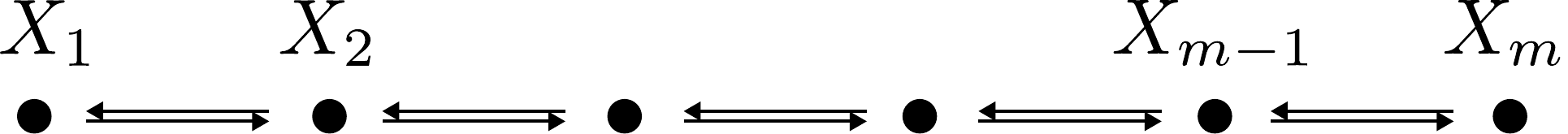}
\caption{The 3d symmetric quiver $Q$ associated with $T[Q]$.}
\label{fig:A2n-sym-quiver}
\end{center}
\end{figure}

\subsection{Symmetrization map in the case of minimal chamber}

There is a tantalizing similarity between the BPS quiver $Q_{4\text{d}}\equiv Q_{\tau_C}$, shown in Figure \ref{fig:A2n-quiver}, and the symmetric quiver $Q$ encoding mixed Chern-Simons couplings of the 3d theory $T[M_0]$, shown in Figure~\ref{fig:A2n-sym-quiver}. 
The quiver $Q$ obtained in this way is related to the 4d BPS quiver $Q_{4\text{d}}$ and to the choice of stability data, which is denoted by ``min'' (minimal chamber) and understood to encode both a choice of Coulomb moduli for the 4d theory as well as a choice of half-plane in the complex plane of central charges. 
The quiver $Q$ is obtained simply by doubling all arrows of $Q_{4\text{d}}$ -- this is our first example of the symmetrization map (\ref{Sigma-maps}): 
\be
	\Qsym{Q_{4\text{d}}}{\text{min}} = Q\,.
\ee 
More general definition of this map will be discussed in later sections. 
The simple relation between $Q$ and $Q_{4\text{d}}$ encountered here is a special feature of this choice of stability condition. 
This relation can be understood both from geometric and physical viewpoints, as follows. 

\paragraph{Quivers from triangulations}

The 4d quiver is dual to the triangulation of $C$ at either end of the interval. 
Indeed, $M_0$ is defined by gluing a tetrahedron for each BPS particle of $T[C]$, and each induces a~flip of the triangulation of $C$, which corresponds to a mutation of $Q_{4\text{d}}$. 
It follows that the entire sequence of mutations induced by $M_0$ takes $Q_{4\text{d}}$ back to itself \cite{Alim:2011kw}. 
Vertices of $Q_{4\text{d}}$ are internal edges of $\tau_C$ and arrows count shared triangles with orientation as in \eqref{eq:surface-FG-Poisson-structure}. 

The 3d quiver $Q$, on the other hand, depends on the data of $(\tau_M,\Pi)$. 
Each vertex corresponds to an external edge of $\tau_M$ labelling a position coordinate $X_i$ defined by the polarization $\Pi$. 
Links among vertices are in one-to-one correspondence with internal triangles of $\tau_M$ shared by these edges.\footnote{Since $X_i$ are in mutual involution, it follows immediately that they do not share any external faces.} 
For example, consider the link between vertices labeled by $X_1 = -V_1 - U_1''$ and $X_2=-U_1 -V_1''-V_2''$ in Figure \ref{fig:A2n-sym-quiver}. 
The link between them corresponds to the face $(U_1,U_1',U_1'')\simeq (V_1, V_1', V_1'')$ shared by the external edges labeled by $\alpha_1$ and $-\alpha_2$ in Figure \ref{fig:A2n-M0}. 

The geometric data that defines both 3d and 4d quivers is summarized by the following table:
\be\label{eq:geometric-quivers}
	\begin{array}{c|c|c}
	& Q_{4\text{d}} & Q \\
	\hline
	\text{vertices} & \text{internal edges of $\tau_C$} & \text{$X_i$-edges of $(\tau_{M_0},\Pi)$} \\
	\text{arrows} & \text{common triangles in $\tau_C$} & \text{common triangles in $\tau_{M_0}$}\\
	\end{array}.
\ee

\paragraph{Bulk and boundary Witten effects: from Dirac pairings to Chern-Simons couplings}

The striking similarity between 4d and 3d quivers can also be understood from a physical perspective. 
First, recall that vertices of $Q_{4\text{d}}$ are basic BPS states of $T[C]$, i.e., dyons whose boundstates generate the entire 4d BPS spectrum through interactions governed by the Dirac pairing \eqref{eq:4d-Dirac-pairing}, which is captured by the arrows between two vertices \cite{Denef:2002ru}. 
Conversely, vertices of the symmetric quiver $Q$ correspond to fundamental BPS vortices of $T[Q]$, whose own interactions are governed by mixed Chern-Simons couplings $\kappa_{ij}^{\text{eff}}$, which are encoded by the linking matrix of $Q$ through~\eqref{eq:CS-quiver}: 
\be\label{eq:physical-quiver-relations}
	\begin{array}{c|c|c}
	& Q_{4\text{d}} & Q \\
	\hline
	\text{vertices} & \text{basic BPS dyons of $T[C]$} & \text{fundamental vortices of $T[Q]$}\\
	\text{arrows} & \text{Dirac pairing}  & \text{Chern-Simons pairing}\\
	\end{array}.
\ee
A common feature is that in both cases the number of arrows between two vertices of the quiver is directly related to the \emph{orbital  spin} of their boundstate (respectively in 3 and in 2 space dimensions; see \cite{Coleman:1982cx} and \cite[Appendix A]{Gupta:2024wos}).
To understand why the two should be related, let us recall a few facts about the respective QFTs. 

On the Coulomb branch, $T[C]$ flows to an effective $U(1)^{n}$ gauge theory, and BPS dyons arise as massive particles charged under the IR abelian gauge group.
Consider a duality frame in which a stable BPS dyon of charge $\alpha$ is purely electric (has unit electric charge and is magnetically neutral).
In this duality frame, the dyon can be described by a local field coupled to the low-energy photons $F_{\mu\nu}$ as in standard QED. 
On the other hand, another dyon with charge $\alpha'$, such that $\langle\alpha,\alpha'\rangle\neq 0$, does not admit a~description in terms of local degrees of freedom in the same duality frame. 
However, there exists a dual frame in which $\alpha'$ is an elementary electric charge, and in which the new dyon is described by local fields coupling to a dual photon $F'_{\mu\nu}$. 

We thus have two duality frames $\Pi_\alpha, \Pi_{\alpha'}$ in which, respectively, $\alpha$ and $\alpha'$ correspond to purely electric charges. 
However, when viewing $\alpha$ in frame $\Pi_{\alpha'}$, it acquires $m$ units of magnetic charge, where $m = \langle\alpha,\alpha'\rangle$ is measured by the duality-invariant Dirac pairing. 
Conversely, the magnetic charge of $\alpha'$ in $\Pi_{\alpha}$ will be $-m$. 
The duality transformation between the two frames must therefore involve a field redefinition mixing electric and magnetic field strengths: 
\be
	F_{\mu\nu}' \sim \langle\alpha,\alpha'\rangle\, \tilde F_{\mu\nu} +\dots \,.
\ee
At the level of the Lagrangian, such a duality transformation can be implemented by a Legendre transform combined with a shift of the theta angle by $\langle\alpha,\alpha'\rangle$ units: 
\be
	\delta \mathcal{L} \sim \langle\alpha,\alpha'\rangle F\wedge F\,.
\ee
Such a shift adds a total derivative to the Lagrangian, leaving the local equations of motion invariant. 
But if the abelian gauge theory is considered on a manifold with boundary, as in our case, a shift of the theta angle induces a change in the boundary action. 
In particular, a shift of the theta angle corresponds to a shift of the appropriate Chern-Simons level for the boundary theory: 
\be\label{eq:top-theta}
	\frac{\langle\alpha,\alpha'\rangle}{8\pi} \int_{S^1\times \IR^2\times \IR_{\geq 0}} F\wedge F 
	= 
	\frac{\langle\alpha,\alpha'\rangle}{4\pi} \int_{S^1\times \IR^2\times \{0\}} A\wedge dA \,.
\ee
which leads to the identification 
\be
	\langle\alpha,\alpha'\rangle =\kappa\,.
\ee
This explains the relation between $Q_{4\text{d}}$ and $Q$ with structures given by \eqref{eq:physical-quiver-relations}. 

The relation between Dirac pairing and Chern-Simons coupling can also be understood as a relation between two well-known effects. 
On the one hand, in 4d abelian gauge theory, the Witten effect \cite{Witten:1979ey} is the observation that a shift of theta angles induces a shift of electric charges by magnetic ones. 
On the other hand, it is also well known that the inclusion of a Chern-Simons term in 3d abelian gauge theory has precisely the same effect: it induces a mixing of electric and magnetic charges. 
These observations are related by the topological nature \eqref{eq:top-theta} of the 4d theta term, i.e., the fact that it is a total derivative of the Chern-Simons term.

\subsection{Boundary conditions and duality frames for \texorpdfstring{$A_{m}$}{Am} Argyres-Douglas theories}\label{sec:4d-duality-frames}

We next explain how the observations above are realized in the case of main interest to us, i.e., the domain wall theory $T[M_0]$ defined by the minimal BPS spectrum of the $A_{m}$ Argyres-Douglas theory.
The description is very similar for both even and odd values of $m$. 
For illustration, we give details on the case $m = 2n$. 
\newline

We start with the Lagrangian description of the low-energy 4d $U(1)^n$ gauge theory associated with the initial surface $C$, corresponding to the triangulated $(2n+3)$-gon in Figure \ref{fig:A2n-M0}. 
In the minimal chamber of the Coulomb branch, the BPS spectrum includes hypermultiplets with charge $\alpha_i$ with $i=1,\dots, 2n$. 
Charges with odd (respectively, even) labels are mutually local: 
\be
	\langle \alpha_{2i-1},\alpha_{2j-1}\rangle = 0 = \langle \alpha_{2i},\alpha_{2j}\rangle \qquad i,j = 1\dots, n.
\ee
A Lagrangian description of the IR 4d theory involves a choice of duality frame, which we fix by demanding that a maximal collection of mutually local dyons are purely electric. 
Concretely, we may choose BPS dyons with charge $\alpha_{2i}$ for $i=1,\dots, n$ 
(with corresponding boundary shear coordinates $X_{\alpha_{2i}} = U_{i}$ with $i=1,\dots, n$),
and these would be described by local degrees of freedom of electrically charged 4d hypermultiplet fields. 

The 4d theory $T[C]$ is formulated on $S^1\times \IR^2\times \IR_{\geq 0}$, which has a boundary component at finite distance with topology $S^1\times \IR^2$. 
The boundary conditions for the 4d theory are described, in the dual geometry, by a triangulated 3-manifold with a boundary component the triangulated Riemann surface $C$. 
The $n$ tetrahedra glued to $C$ correspond to a maximal collection of mutually local charges. 
For example, for the collection of even charges $\alpha_{2i}$, the relevant tetrahedra are those in the bottom row (in yellow) in Figure \ref{fig:A2n-M0}. 
These tetrahedra define a boundary condition for the corresponding 4d theory $T[C]$ as follows. 

The boundary condition for the bottom (respectively top) copy of $T[C]$ involves a breaking of 4d $\CN=2$ supersymmetry down to 3d $\CN=2$, which is parameterized by a phase $\vartheta_e$ (respectively, $\vartheta_o$). 
The angle defines a splitting of 4d $\CN=2$ hypermultiplets into two 4d $\CN=1$ chiral multiplets $(\Phi,\Phi')$ with opposite flavour and gauge charges, and 
boundary conditions can be of Dirichlet type for $\Phi$ and of Neumann type for $\Phi'$, or vice versa; see \cite{Dimofte:2011ju}. 
The boundary degrees of freedom of the chiral field with Neumann conditions are described by a 3d $\CN=2$ chiral multiplet in $T[M_0]$. 

In our setup, the choice of phases $\vartheta_e, \vartheta_o$ is defined by the minimal chamber of the 4d Coulomb branch, if we require that they correspond to phases of BPS central charges: 
\be\label{eq:Z-condition}
	\arg Z_{\alpha_{2i}} = \vartheta_{e} < \arg Z_{\alpha_{2i-1}} = \vartheta_{o}\,.
\ee
This means, in particular, that tetrahedra glued to the lower boundary of $M_0$ all correspond to the same phase $\alpha_e$, while all upper tetrahedra have the phase $\alpha_o$. 
The abelian IR theory corresponding to the bottom copy of $T[C]$ features one hypermultiplet with charge $\alpha_{2i}$ for each $i=1,\dots, n$. 
The boundary condition defined by $\vartheta_e$ (for the bottom boundary) defines a splitting of these hypermultiplets into $\CN=1$ chirals. 
Of these, it is the field with Neumann boundary conditions which gives rise to the 3d chiral multiplet of the corresponding tetrahedron theory \eqref{eq:tetrahedron-theory} with shear coordinates $U_i$; see Figure \ref{fig:A2n-M0}: 
\be
	\text{Boundary theory for $T[C_{\text{bottom}}]$: } \ \ \CB_{\text{bottom}}:=\bigotimes_{i \text{ even}} \CT_{\Delta_i}.
\ee

Similarly, at the top boundary of $M_0$, namely the upper copy of the triangulated $(2n+3)$-gon in Figure \ref{fig:A2n-M0}, 
we have the same low-energy abelian $U(1)^n$ theory, but in a different duality frame where electrically charged hypermultiplets have charge $\alpha_{2i-1}$ for each $i=1,\dots, n$. 
Once again, each of these includes a $\CN=1$ chiral with Neumann boundary conditions, whose boundary degrees of freedom give rise to the 3d chiral multiplet of the tetrahedron theory \eqref{eq:tetrahedron-theory} with shear coordinates $V_i$; see Figure \ref{fig:A2n-M0}: 
\be
	\text{Boundary theory for $T[C_{\text{top}}]$: } \ \ \CB_{\text{top}}:=\bigotimes_{i \text{ odd}} \CT_{\Delta_i}. 
\ee

The full domain wall theory $T[M_0]$ is obtained by coupling $\CB_{\text{bottom}}$ to $\CB_{\text{top}}$ according to the rules outlined  in Section \ref{sec:gluing}. 
From the viewpoint of the 3d theory, this involves activating 3d Chern-Simons couplings, as we showed in Section \ref{sec:A2n-theory}. 
On the geometric side, the coupling between boundary theories is engineered by gluing along \emph{internal triangles}, as shown in Figure \ref{fig:A2n-M0}, where we glue the two rows of yellow tetrahedra. 
Moreover, each internal triangle connects edges associated with vertices of the 3d quiver $Q$, respectively on the bottom and top copies of $C$ (highlighted in blue and in red in Figure \ref{fig:A2n-M0}). 
This is consistent with the identification of arrows in $Q$ with internal faces shared by edges dual to vertices of~$Q$ -- see \eqref{eq:geometric-quivers}. 

The change of polarization defined in \eqref{eq:g-S-TB} encodes precisely the change of duality frame that relates the bottom and top copies of $T[C]$. The action of $T_B$ implements a shift (in fact, a mixing) of theta angles by $Q_{ij}$ units, while the $S$ operation implements the Legendre transform that switches between electric and magnetic frames.
These operations combined take us from the bottom 4d theory, which corresponds to the electric frame for $\alpha_{2i}$, to the top 4d theory, which corresponds to the electric frame for $\alpha_{2i-1}$. 

To summarize, we define the duality frame of the bottom (respectively, top) copy of $T[C]$ by requiring that the boundary condition determined by the bottom (respectively, top) row of tetrahedra corresponds to Dirichlet/Neumann conditions for elementary electric hypermultiplets. 
The duality transformation relating these two frames is encoded by gluing relations between the two sets of tetrahedra that define the respective boundary conditions. 
The transformation switching between the two frames involves an appropriate shift of the theta angles, followed by electric-magnetic duality, which are encoded by the Dirac pairing between $\alpha_{2i}$ and $\alpha_{2i-1}$. 
The topological nature of the 4d theta term implies that it can be equivalently recast as a 3d Chern-Simons term on the domain wall theory $T[M]$, thus explaining why the adjacency matrices of 4d and 3d quivers are directly related: recalling that matrix elements of the (antisymmetric) adjacency matrix of the 4d quiver are $(Q_{4\text{d}})_{ij} = \langle \alpha_{i},\alpha_{j}\rangle$, we have that
\be\label{eq:Dirac-pairing-symmetrization}
	|\langle \alpha_{i},\alpha_{j}\rangle| \equiv |(Q_{4\text{d}})_{ij}| = \Qsym{Q_{4\text{d}}}{\text{min}}_{ij} \equiv Q_{ij}\,.
\ee

\subsection{A geometric realization}

We now provide a geometric construction of the domain walls described above. 
For illustration, we focus on the case $m = 2n$ even, since the case of $m = 2n+1$ follows with minor adjustments.

\paragraph{Spectral networks of Argyres-Douglas theories}

To build a geometric model for $M_0$, we consider a family of quadratic differentials from the minimal chamber of the Argyres-Douglas theory $A_{2n}$, defined as follows: 
\be
	\phi_2(z) = \prod_{i=1}^{2n+1}(z-z_i)\, dz^2
\ee
where $z_i$ are chosen on the real axis. 
For convenience, we fix
\be
	z_i = i-n-1\,.
\ee
The quadratic differential then has zeroes at $2n+1$ integer points between $z=-n$ and $z=+n$ (inclusive). 
We define the spectral network $\CW(\vartheta)$ as the collection of the horizontal foliation: 
\be
	\mathrm{Im} \sqrt{e^{2i\vartheta}\phi_2(e^{-\frac{2i\vartheta}{2n+3}}z)} = 0 \,. 
\ee
Our definition differs slightly from the original one \cite{Gaiotto:2009hg} due to a rotation of the $z$-plane by $e^{-\frac{2i\vartheta}{2n+3}}$, which is included to keep the asymptotic behaviour of Stokes lines fixed at infinity. This is necessary for our construction of $M_0$ since we wish to glue the boundaries of the bottom and top copies of $C$. 
In our picture, it is the branch points that move instead, as $\vartheta$ changes. 

The evolution of the network for $\vartheta\in [\pi/4,5\pi/4]$ is shown in Figure \ref{fig:A2n-networks} for the case $n=3$. 
We observe the presence of saddles at phases 
\be
	\vartheta_e = \frac{\pi}{2}\,,
	\qquad
	\vartheta_o = \pi\,,
\ee
corresponding to periods of central charges $Z_{\alpha_{2i}}$ and $Z_{\alpha_{2i+1}}$, respectively. 
This gives a geometric realization of the stability condition \eqref{eq:Z-condition}. 

\begin{figure}[h!]
\begin{center}
\includegraphics[width=0.3\textwidth]{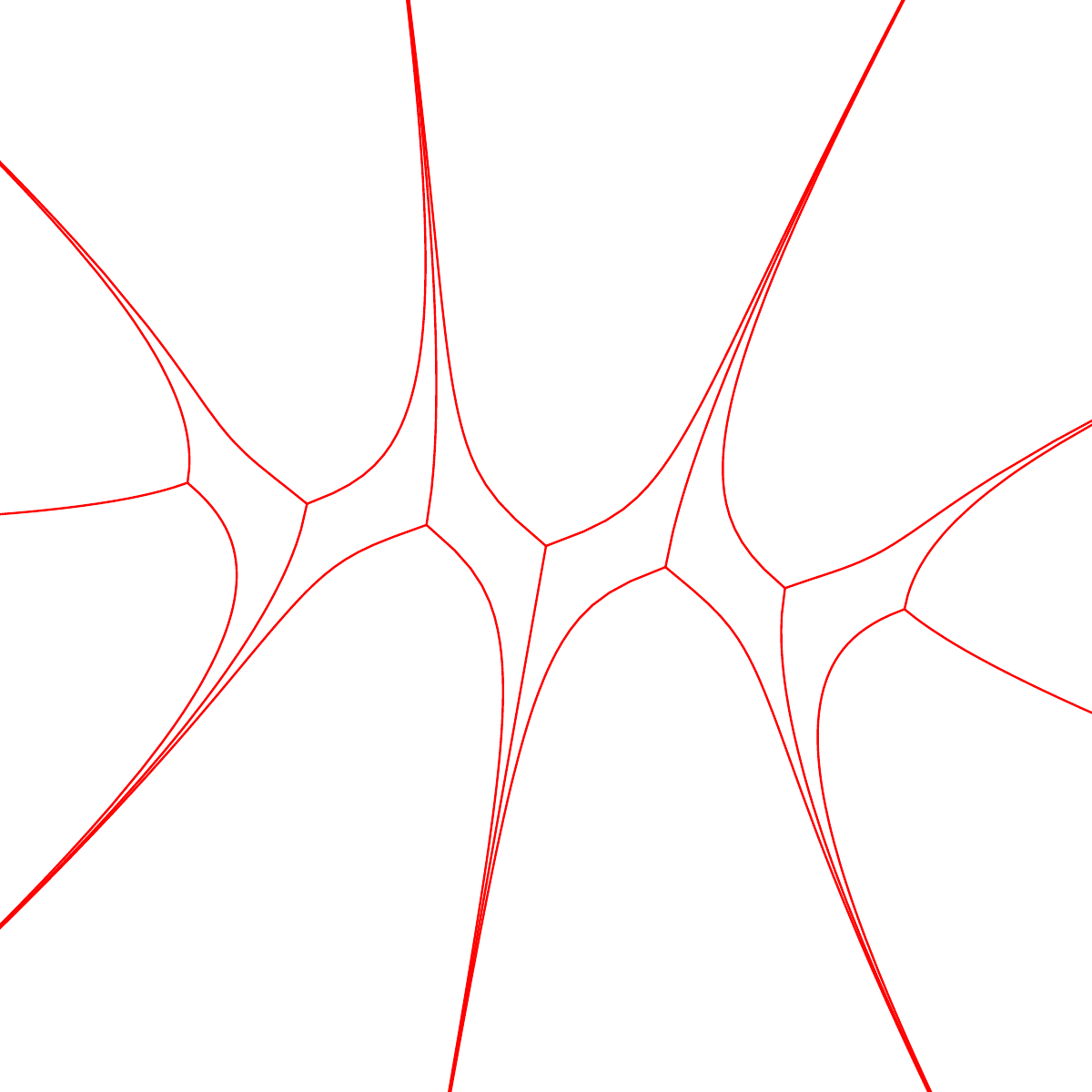}
\includegraphics[width=0.3\textwidth]{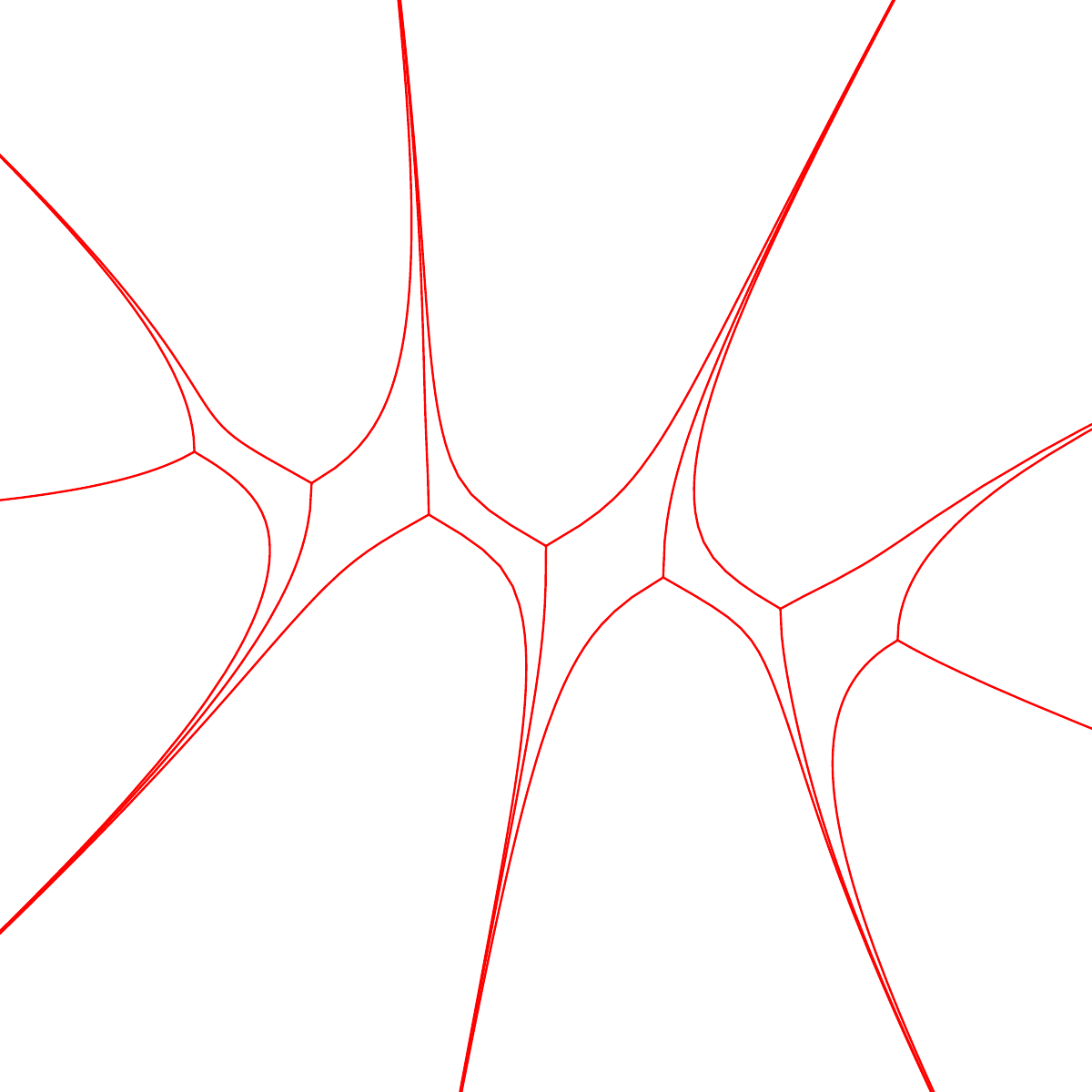}
\includegraphics[width=0.3\textwidth]{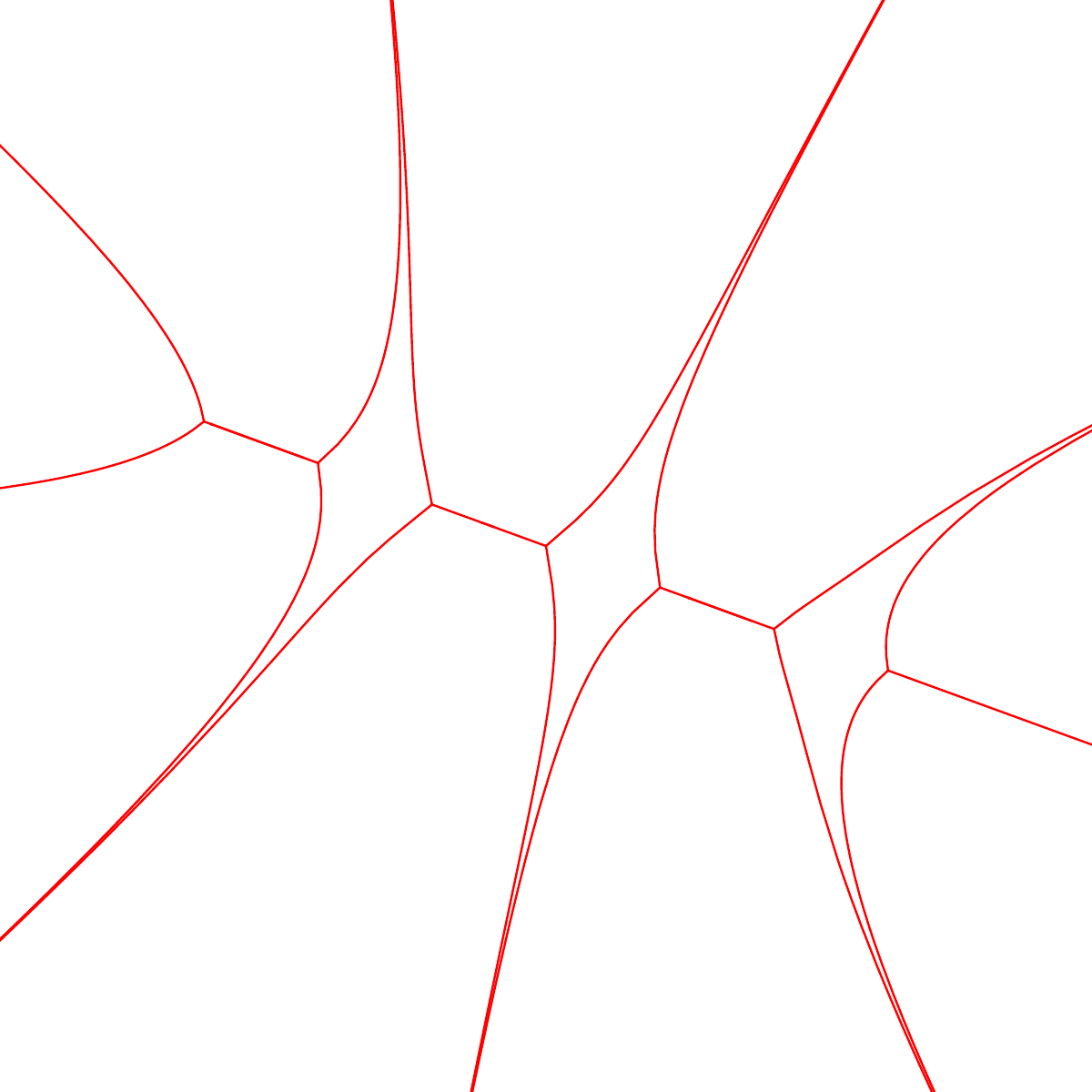}
\includegraphics[width=0.3\textwidth]{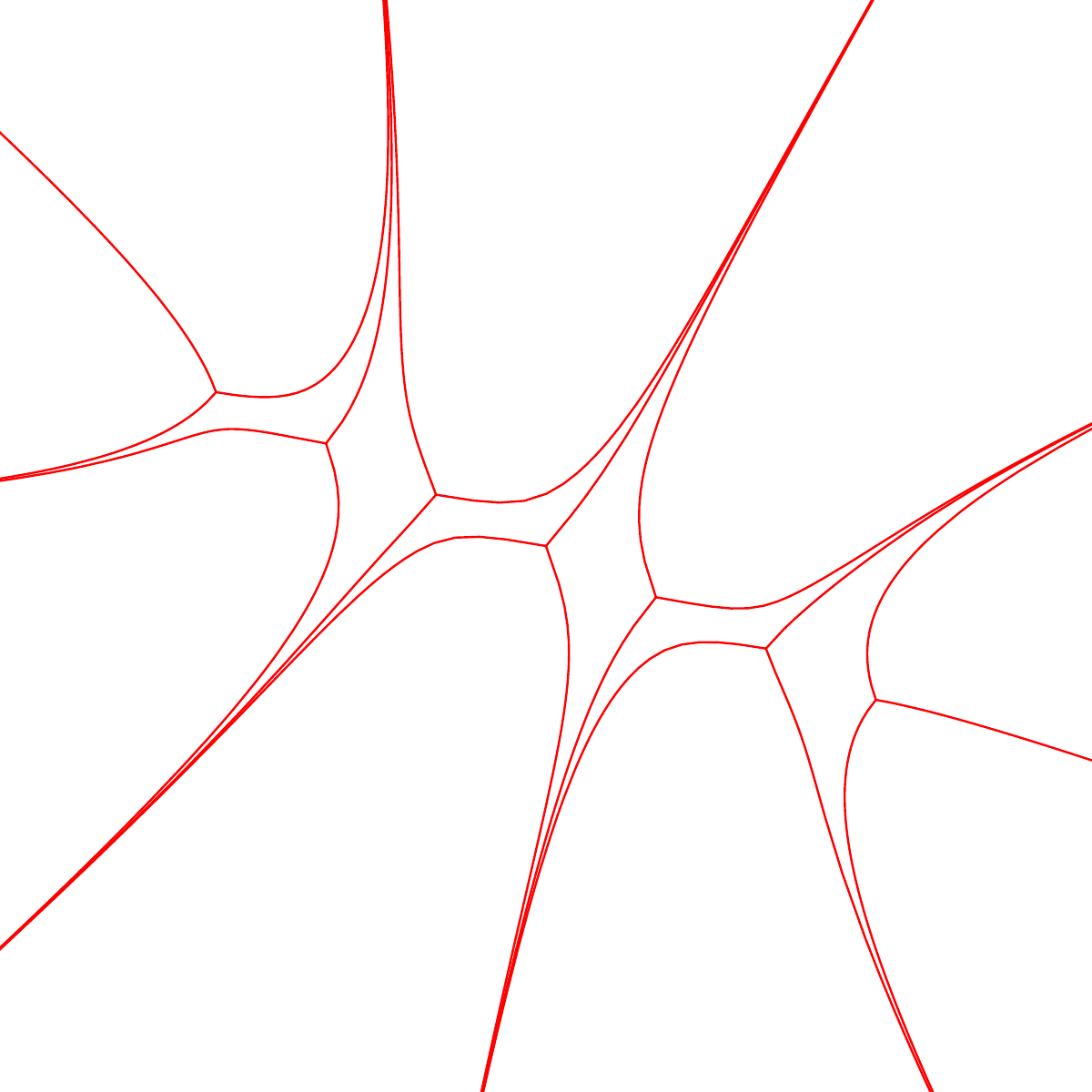}
\includegraphics[width=0.3\textwidth]{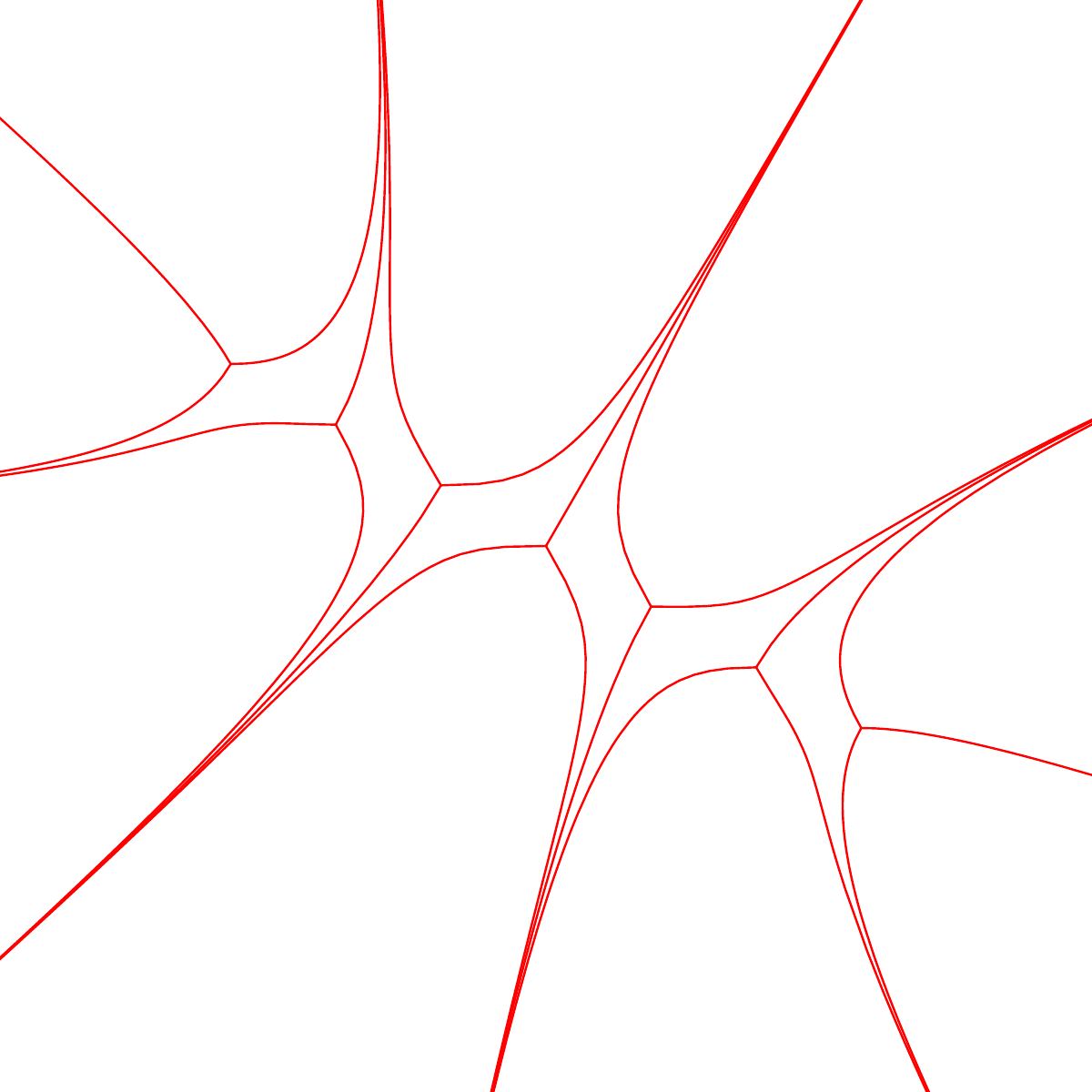}
\includegraphics[width=0.3\textwidth]{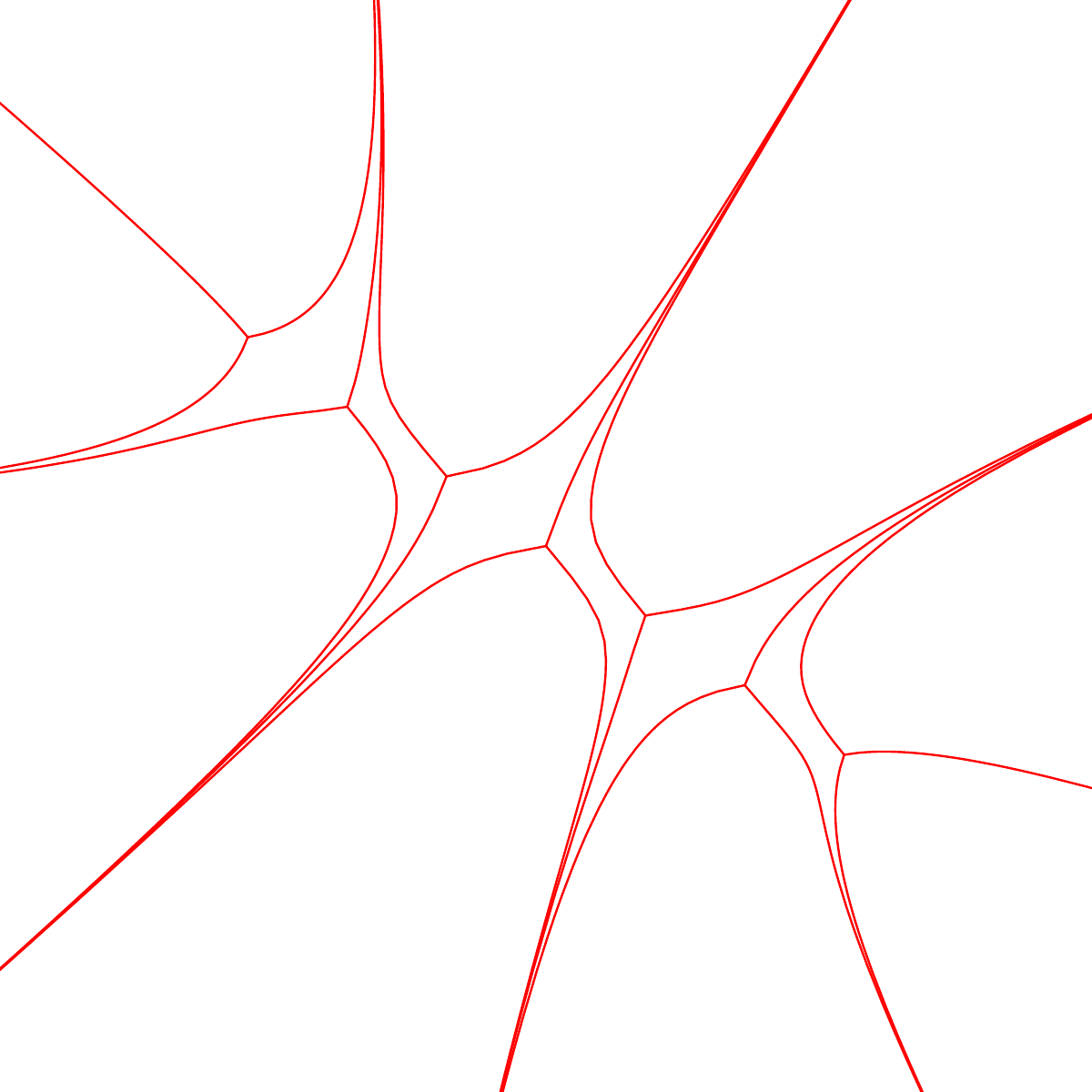}
\includegraphics[width=0.3\textwidth]{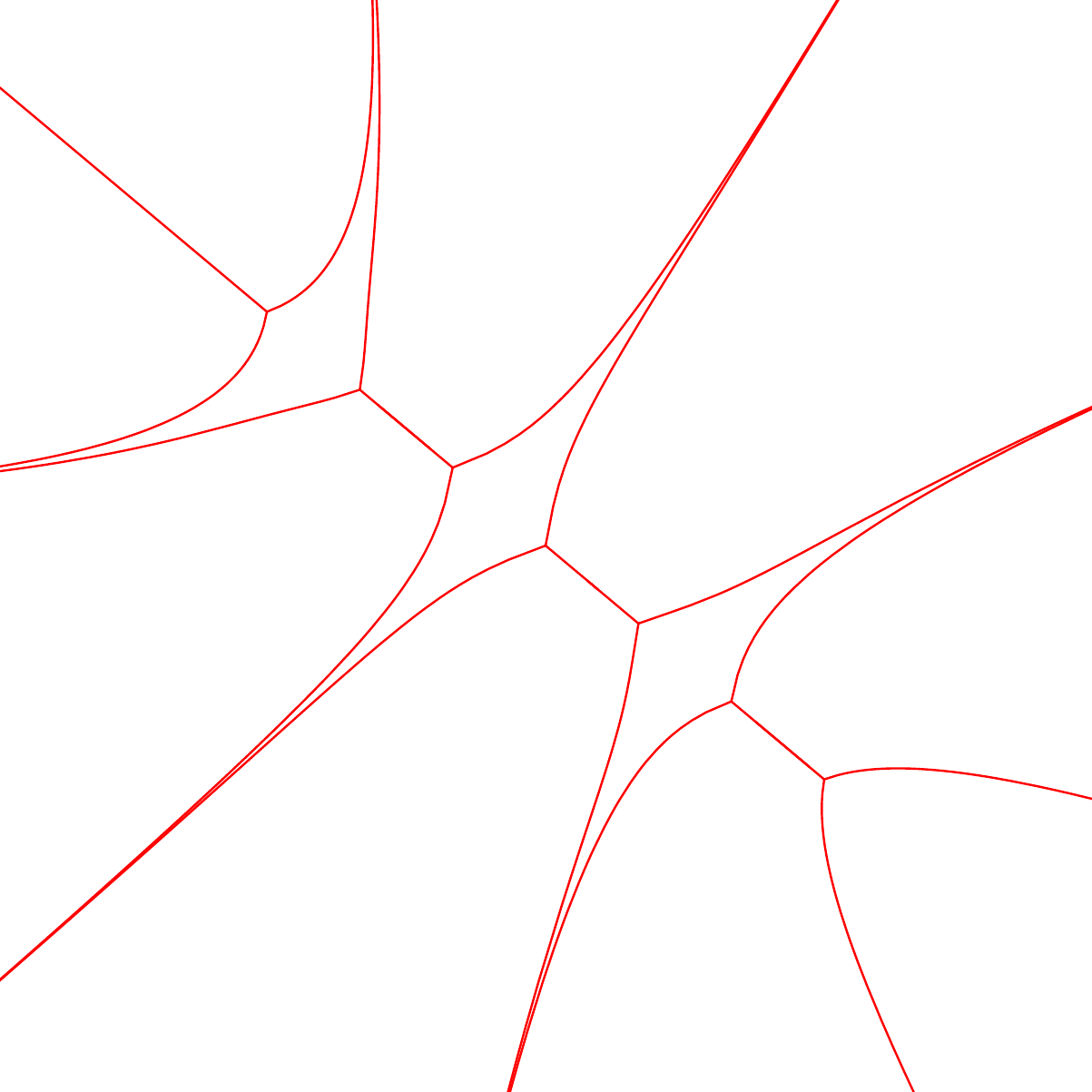}
\includegraphics[width=0.3\textwidth]{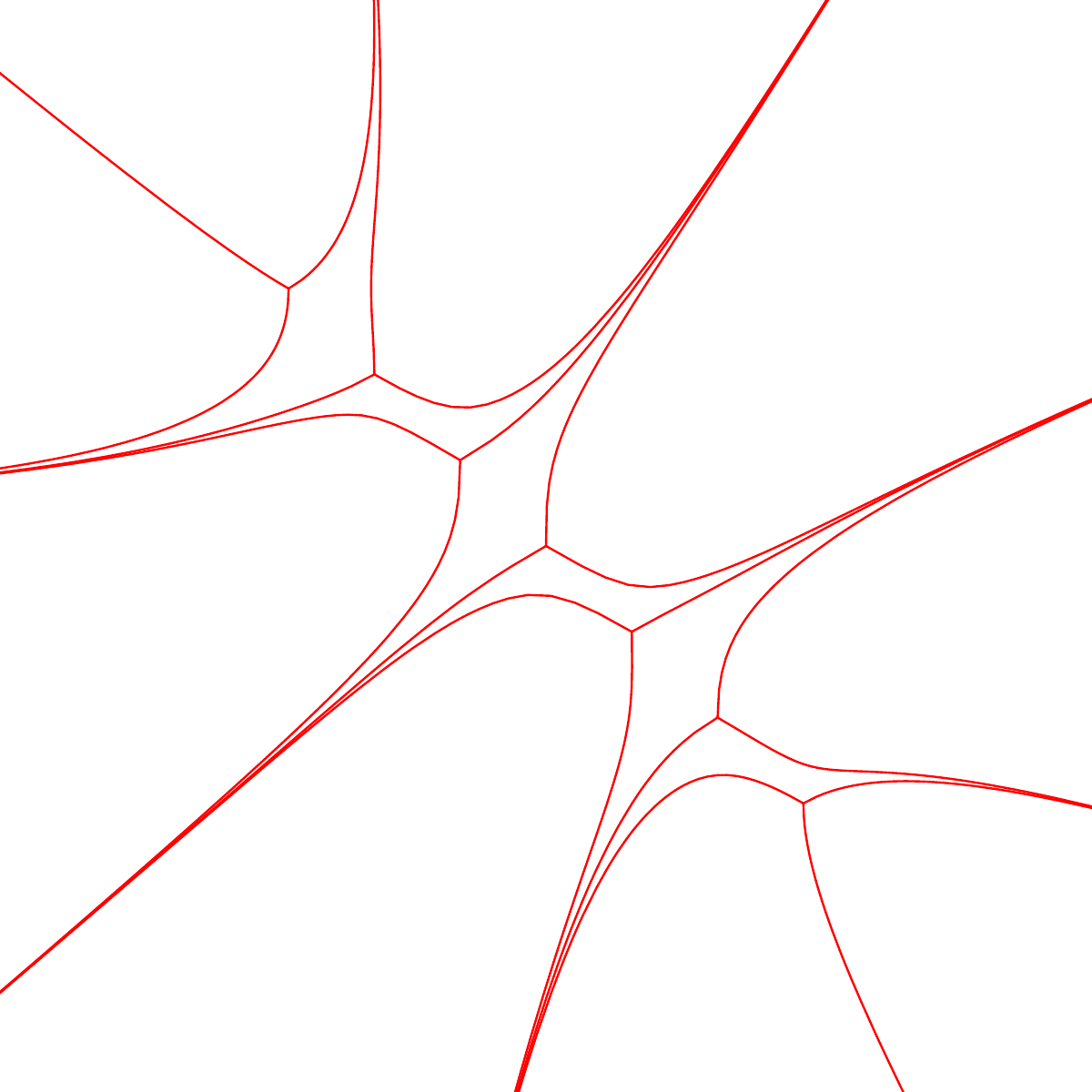}
\includegraphics[width=0.3\textwidth]{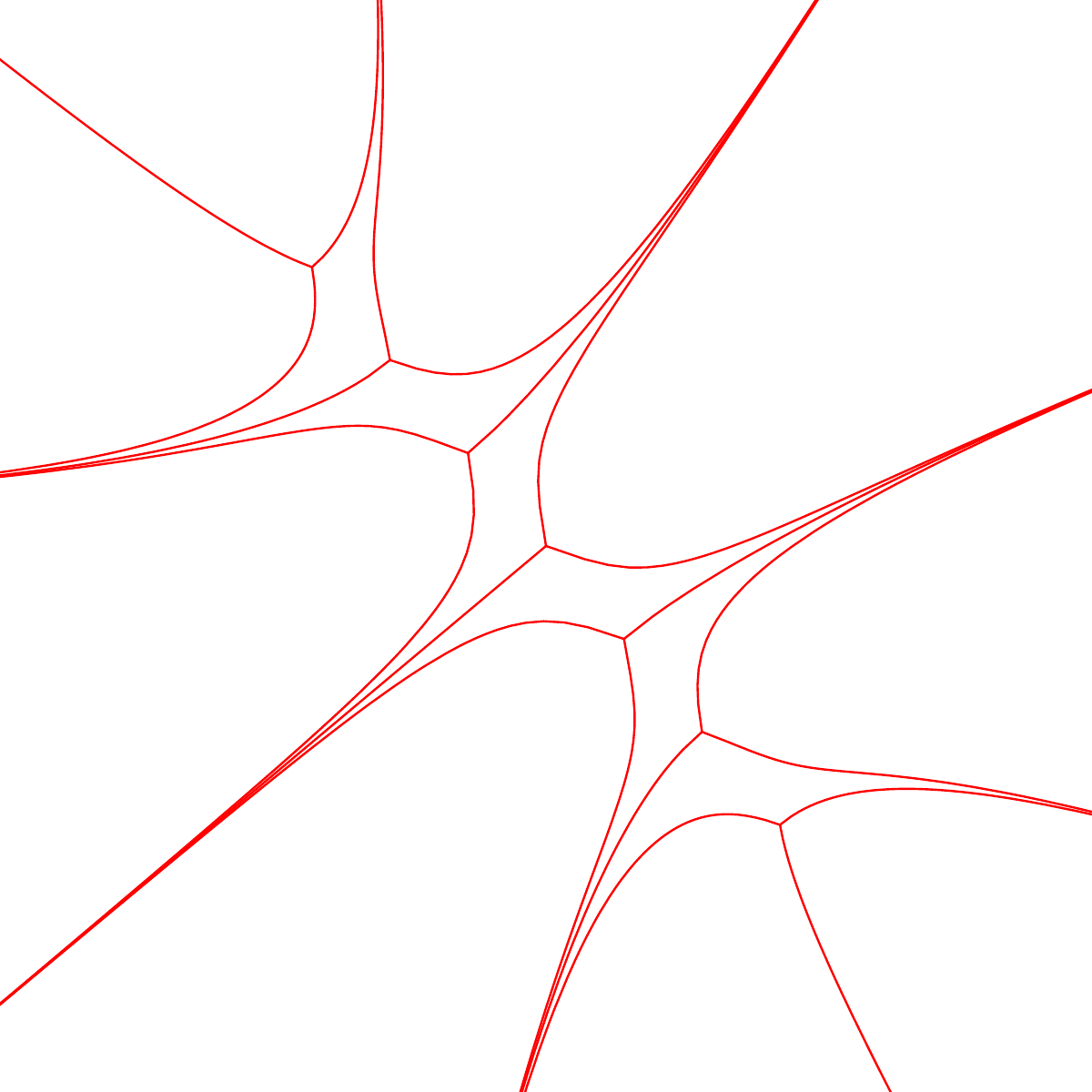}
\caption{Spectral networks for a family of quadratic differentials, corresponding to a family of IR 4d $\CN=2$ QFTs. 
There are two critical phases at which flips of the network occur (top-right and bottom-left frames). 
The family of 4d QFTs corresponds to the domain-wall theory described by the 3d-3d dual QFT to the 3-manifold in Figure \ref{fig:A2n-M0}, whose two rows of tetrahedra correspond to the two collections of flips featured here.}
\label{fig:A2n-networks}
\end{center}
\end{figure}

By a construction of \cite{Gaiotto:2009hg}, each spectral network induces an ideal WKB triangulation of $C$, which jumps by flips when saddles of the network appear.
In our case, the WKB triangulations correspond precisely to the initial, middle, and final triangulations of $C$ shown in Figure \ref{fig:A2n-M0}.

\paragraph{M2-M5 boundstates and linking}

A quadratic differential defines a ramified double covering of $\Sigma_\vartheta \to C$ inside $T^*C$. 
Considering a family over $\vartheta$ gives a double covering $L_0\to M_0$, defined in $T^*M_0$. 

As explained in Section \ref{sec:A2n-theory}, flips of the triangulation correspond to tetrahedra in a 3d triangulation of $M_0$. 
The double covering of each tetrahedron is a solid torus $S^1\times \IR^2$, and $L_0$ consists of $2n$ solid tori glued together along their boundaries, at the (double covers of) faces of the tetrahedra. 

Saddles of the spectral network correspond to the presence of holomorphic disks $D_i$ in $T^*C$, at phases $\vartheta_e$ and $\vartheta_o$, whose boundaries have homology classes $\alpha_{i} \in H_1(\Sigma,\IZ)$.
In the M-theory setup $L_0$ supports a~single M5-brane arising from recombination of the stacked pair of M5-branes on $M_0$, and each holomorphic disk supports an M2-brane with boundary on $L_0$.

The double covering $L_0$ is a manifold with boundary given by two copies of a double covering of $C$: 
\be
	\partial L_0 = \Sigma_+ \cup \Sigma_-\,.
\ee
It will be interesting, with later applications in mind, to consider a compactification of $L_0$ defined by gluing manifolds $L_\pm$ at the two ends, as described in Section \ref{sec:Z-from-skein}. 
To define $L_\pm$, let $\{\mu_i, \nu_i\}_{i=1}^{n}$ denote a~Darboux basis for $H_1(\Sigma,\IZ)$. 
Then, $L_\pm$ are handlebody fillings of $\Sigma_\pm$ obtained by filling in, respectively, the $\mu_i$ cycles for $\Sigma_+$ and the $\nu_i$ cycles for $\Sigma_-$. 
The resulting manifold has the topology of a three-sphere
\be
	L_+ \cup_{\Sigma_+} L_0 \cup_{\Sigma_-} L_- \approx S^3\,. 
\ee
Using this fact, we can consider the Gauss linking number of boundaries of holomorphic disks in $L_0$ in homology classes $\alpha_i$, and we observe that these coincide precisely with the linking matrix of the symmetric quiver $Q$: 
\be
	\mathrm{lk}(\partial D_i,\partial D_j) = Q_{ij}\,.
\ee

\section{Algebras of line operators and symmetric quivers}
\label{sec:algebra}

We discuss a relation between line operators of the low-energy effective 4d and 3d theories discussed in the previous section. 
We show that their geometric realization naturally gives rise to a bimodule for the quantum torus algebra associated to the symmetric quiver of the 3d $\mathcal{N}=2$ theory, that coincides with the $\mathfrak{gl}_1$ skein module of the (double-covering $L_0$ of the) mapping cylinder defined by the Kontsevich-Soibelman invariant. 
We also present how algebraic formulation of this relation, which we call 3d-4d homomorphism, can be generalized to any symmetric quiver of the 3d $\mathcal{N}=2$ theory. 
Finally, we discuss the 3d-4d homomorphism and the symmetrization map $\mathfrak{S}$ for the minimal chamber of $A_2$ and general $A_m$ quivers. 

\subsection{Line operators in 3d and 4d: boundary and domain-wall ideals from motivic DT series of symmetric quivers}\label{sec:line_operators}

In this section, we discuss the physical origins of the line operators in 3d and 4d theories, which forms the foundation of the algebraic perspective on the symmetrization map $\mathfrak{S}$.
\newline

The twisted compactification of the 6d $(2,0)$ theory of type $A_1$ on $C\times S^1\times \IR^3$ gives a 3d $\CN=4$ sigma model with target space $\CP_C$ with complex dimension $2n$, a hyperk\"ahler manifold isomorphic to the moduli space of $SL(2,\IC)$ flat connections on $C$ \cite{Seiberg:1996nz, Gaiotto:2009hg}.\footnote{In our case, the compactification involves the presence of boundaries for the 4d theory, further breaking supersymmetry to four supercharges.} 
Recall that $n = \lfloor \frac{m}{2} \rfloor$ here: when $m=2n+1$ is odd there is a $U(1)$ flavour symmetry that restricts the target from a space of complex dimension $2n+1$ to a complex-codimension-one subspace.

The emergence of this space can be understood directly from the six-dimensional viewpoint by noting that flat connections on $C$ arise after reducing the 2-form gauge potential along $S^1$.
Conversely, from the dual viewpoint of the 4d QFT, reducing the 2-form along cycles on $C$ gives rise to gauge fields of $T[C]$. 
This leads to the identification between holonomies of flat connections on $C$ and expectation values of line operators in $T[C]$ wrapping $S^1$ \cite{Drukker:2009tz, Drukker:2009id, Alday:2009fs}. 

Local functions on $\CP_C$ are provided by the shear coordinates $U_i, V_i$ and their counterparts defined earlier. 
The identification with line operators arises by observing that shear coordinates are labeled by electromagnetic charges $\alpha$ of the IR effective $U(1)^r$ gauge theory description. 
It was argued in \cite{Gaiotto:2009hg, Gaiotto:2010be} that shear coordinates correspond to vacuum expectation values of IR line operators with corresponding electromagnetic charges\footnote{The definition of BPS line operators involves a choice of preserved supercharges parameterized by a phase, which we suppress here.}
\be\label{eq:exp-x-alpha}
	x_\alpha = e^{X_{\alpha}} = \langle L_{\alpha}^{(IR)}\rangle\,.
\ee 
Introducing a half omega-background for the 4d theory on $S^1\times \IR^2\times \IR_{\geq 0}$, defined by twisting $\IR^2$ by a rotation along $S^1$ \cite{Moore:1997dj, Moore:1998et, Nekrasov:2002qd}, results in a quantization of the ring of line operators given by a quantum torus algebra
\be\label{eq:4d-q-torus}
	\CA_C^{IR}: \quad \hat x_\alpha \hat x_{\alpha'} = q^{\langle\alpha,\alpha'\rangle} \hat x_{\alpha+\alpha'}\,.
\ee

When $T[C]$ is formulated on a manifold with boundary, the boundary conditions for bulk fields induce relations among the line operators, encoded by a certain ideal $\CI$.
As a result, the algebra of line operators for the boundary 3d QFT takes the structure of a module over the algebra of 4d line operators \cite{Dimofte:2011ju}. 

In our case, we have two copies of the 4d theory $T[C]$, taken in distinct duality frames, as discussed in Section \ref{sec:4d-duality-frames}, and glued together by a domain wall theory $T[M_0]$. 
We can view shear coordinates on $\partial M$ as coordinates on $\CP_C\times \overline{\CP}_C$,\footnote{
$\overline{\mathcal{P}}_C$ has the opposite symplectic structure compared to $\mathcal{P}_C$, since the two components of the boundary come with opposite orientations.} whose quantization gives \emph{two} copies of the algebra of line operators of $T[C]$
\be\label{eq:quantization-CxC}
	\CP_{C}\times \overline{\CP}_C\quad\rightsquigarrow\quad \CA_C^{IR}\otimes \CA_C^{IR,\mathrm{op}}\,.
\ee
The manifold $M_0$ defines a bimodule for this algebra as follows. 

\paragraph{Before quantization.}
At $q=1$, each of the boundary shear coordinates corresponds to a certain linear combination of tetrahedron edge coordinates. 
Recall that the latter obey a 3-term relation, such as
\be
1 - e^{\Xi} - e^{\Xi'} = 0
\ee
or, equivalently, an analogue of this equation obtained by a $\IZ_3$ permutation $\Xi\to\Xi'\to \Xi''\to\Xi$ rotating the three choices of polarization in \eqref{eq:delta-pol}. 
With $2n$ tetrahedra without internal edges in $M_0$, this leaves exactly $2n$ independent coordinates. 
Through the relation between bulk and edge coordinates, this defines a Lagrangian submanifold of the boundary character variety: 
\be\label{eq:M0-Lag}
	\CL_{M_0} : = \{1 - e^{-U_i} - e^{U_i''} = 0\}_{i=1}^{n} \ \cap\ \{1 - e^{-V_j} - e^{V_j''} = 0\}_{j=1}^{n} \subset \CP_{C}\times \overline{\CP}_{C}.
\ee
Geometrically, this construction is known as a Lagrangian correspondence, and in our case it describes flat connections that extend from one boundary component to the other throughout the bulk of $M_0$. 
The change of polarization \eqref{eq:tetrahedra-to-quiver} allows us to rewrite \eqref{eq:M0-Lag} in terms of variables 
\be\label{eq:xi-yi-def}
	x_i := e^{X_i}\qquad
	y_i:= e^{Y_i}
\ee
adapted to the quiver description $T[Q]\simeq T[M_0]$, which reads
\be
	\CL_{M_0} : = \{1 - y_i - x_i y_{i-1} y_{i+1} = 0\}_{i=1}^{2n}\,.
\ee
The polynomials that define the Lagrangian correspondence coincide exactly with the quiver $A$-polynomials introduced in \cite{Ekholm:2018eee, Panfil:2018faz}
\be\label{eq:classical-Ai}
	A_i = 1- y_i - x_i(-y_i)^{Q_{ii}} \prod_{j\neq i} y_j^{Q_{ij}}\,,\qquad i=1,\dots, 2n
\ee
with the adjacency matrix given by \eqref{eq:3d-quiver-adjacency-matrix} in our case. 
From the viewpoint of $T[Q]$, the $A_i$ determine the twisted chiral ring of the 3d theory on $S^1$ (viewed as a Kaluza-Klein $(2,2)$ QFT). 

\paragraph{Deformation quantization.}
Next, we discuss \eqref{eq:quantization-CxC} with $q=e^\hbar\neq 1$.
Recall from \eqref{eq:exp-x-alpha} and \eqref{eq:xi-yi-def} that $x_\alpha$ and $x_i, y_i$ are related to $X_\alpha$ and $X_{i}, Y_{i}$, and that the latter are related to each other via the bulk-boundary relations \eqref{eq:boundary-shear-XY} at $C_\pm$, preserving the Poisson (symplectic when $m=2n$) structure. 
Deformation quantization promotes these two sets of coordinates to quantum torus algebra operators. 

On the one hand, we have the standard quantization of Darboux coordinates: 
\be\label{eq:canonical-quantization}
	\{Y_i,  X_j\} = \delta_{ij} \quad\mapsto\quad [\hat Y_i,\hat X_j] = \hbar\delta_{ij}
\ee
which leads to a quantum torus algebra with generators 
\be
	\hat x_i := e^{\hat X_i}\,,\qquad
	\hat y_i := e^{\hat Y_i}
\ee
obeying relations 
\be\label{eq:xi-hat-yi-hat-def}
	\hat y_i \hat x_j = q^{2 \delta_{ij}} \hat x_j\hat y_i\,,
	\qquad 
	[\hat x_i, \hat x_j] = [\hat y_i, \hat y_j] =0\,.
\ee
This is the quantum torus algebra associated with a 3d $\mathcal{N}=2$ symmetric quiver introduced in \cite{Ekholm:2019lmb}.
Note that it coincides with $2n$ mutually commuting copies of the quantum torus algebra \eqref{eq:quantum_torus_algebra}. 

On the other hand, the Poisson structure (\ref{eq:Poisson-structure-alpha-I}--\ref{eq:Poisson-structure-alpha-III}) leads to the quantum algebra of 4d IR line operators \eqref{eq:4d-q-torus}. 
The common origin and compatibility of these algebras play a central role in the 3d-4d correspondence between theories and related quivers.

In addition to a quantization of the algebra of functions over $\mathcal{P}_C\times \overline{\mathcal{P}}_C$, there is also a specific quantization of the Lagrangian submanifold defined by \eqref{eq:classical-Ai}, determined by the quiver motivic generating series (as will be explained shortly), given by \cite{Ekholm:2019lmb}
\be\label{eq:A-hat-quiver}
	\hat A_i = 1- \hat y_i - \hat x_i(-q \hat y_i)^{Q_{ii}} \prod_{j\neq i} \hat y_j^{Q_{ij}} \,,\qquad i=1,\dots, 2n\,. 
\ee
Taken together, these generate an ideal of $\CA^{IR}_{C}\otimes\CA^{IR,\mathrm{op}}_{C}$ that defines the $q$-difference $D$-module 
\be
	\CI^Q \ :=\  \CA^{IR}_{C}\otimes\CA^{IR,\mathrm{op}}_{C} / \{\hat A_i = 0\}\,.
\ee
The generator of this ideal is the motivic generating series of the symmetric 3d quiver $Q$ \cite{2011arXiv1103.2736E}
\be\label{eq:Efimov-PQ}
	P_{Q}(x_1,\dots,x_{2n},q)=\sum_{d_1,\dots,d_{2n} = 0}^{\infty} (-q)^{\sum_{i,j=1}^{2n} d_i Q_{ij}  d_j} \prod_{i=1}^{2n}\frac{x_i^{d_i}}{(q^2; q^2)_{d_i}}\,.
\ee
This has a physical interpretation as the $S^1\times \IR^2$ vortex partition function of the domain-wall theory $T[Q]\simeq T[M_0]$ \cite{Ekholm:2018eee}:
\be
	P_{Q}(x_1,\dots,x_{2n},q) = Z_{\text{vortex}}[T[Q]]\,.
\ee

\subsection{\texorpdfstring{$\mathfrak{gl}_1$}{gl1} skein module of \texorpdfstring{$L_0$}{L0}, quantum torus algebras, and \texorpdfstring{$P_Q$}{PQ} as a wavefunction}\label{sec:embedding}

 In this section, we relate line operators in 3d and 4d theories, which forms the algebraic counterpart of the symmetrization map $\mathfrak{S}$.
\newline 

Viewing $L_0 = \Sigma \times I$ as a bordism from $\Sigma_+$ to $\Sigma_-$, a TQFT (functor) defines an operator represented mathematically by an element of the skein algebra $\SkAlg_q^{\mathfrak{gl}_1}(\Sigma)$ acting on a Hilbert space associated to~$\Sigma$.
At the same time, recall from Section \ref{Topology} 
that $\SkAlg_q^{\mathfrak{gl}_1}(\Sigma)$ is 
a $\SkAlg_q^{\mathfrak{gl}_1}(\Sigma_+)$-$\SkAlg_q^{\mathfrak{gl}_1}(\Sigma_-)$ bimodule in its own right. 
This perspective leads to a natural realization of the quantum torus algebra \eqref{eq:xi-hat-yi-hat-def} of a~symmetric quiver, as we describe next.

As before, we consider $A_m$ Argyres-Douglas theory with $m=2n$ or $m=2n+1$. Since there are minimal differences between the two cases we focus on $m=2n$ here. 
The Riemann surface $\Sigma$ corresponding to the Seiberg-Witten curve of $A_{2n}$ Argyres-Douglas theory defines a physical charge lattice of rank $b_1(\Sigma)=2n$. 
It follows that 
\be
	b_1(L_0) = b_1(\Sigma\times I) = 2n\,,\qquad
	b_1(\partial L_0)  = b_1(\Sigma_{+}\sqcup\Sigma_{-}) = 4n\,.
\ee
The $\mathfrak{gl_1}$ skein module of $L_0 \cong \Sigma \times I$ (see Definition \ref{defn:gl1_skein}) is, as an $R = \IZ[q,q^{-1}]$-module, isomorphic to $R[H_1(\Sigma,\IZ)]$:
\be
	\Sk_q^{\mathfrak{gl}_1}(L_0) \simeq R[x_{1}^{\pm 1},\dots x_{2n}^{\pm1}]\,,
\ee
where $x_1, \cdots, x_{2n}$ are a choice of basis of $H_1(\Sigma, \mathbb{Z})$. 
This carries a natural bimodule structure for the algebra
\be
    \SkAlg_q^{\mathfrak{gl}_1}(\partial L_0)\simeq \SkAlg_q^{\mathfrak{gl}_1}(\Sigma) \otimes \SkAlg_q^{\mathfrak{gl}_1}(\Sigma)^{\mathrm{op}}.
\ee
To describe this, we introduce a basis of operators defined by inserting oriented curves on $\Sigma_+$ and $\Sigma_-$. 
We shall denote, as in Figure \ref{fig:A2n-M0}, 
\be
\begin{split}
	\{\alpha_i\}_{i=1}^{2n} 
	&\qquad \text{basis of cycles on $\Sigma_-$ (the bottom boundary)},
	\\
	\{-\alpha_i\}_{i=1}^{2n} 
	&\qquad \text{basis of cycles on $\Sigma_+$ (the top boundary)}.
\end{split}
\ee
Each of these $4n$ cycles is dual to an internal edge of the triangulations on $C_\pm$, and we denote by $X_{\alpha}$ the corresponding shear coordinates on $\CP_{C}\times \overline{\CP}_{C}$. 
Inverting \eqref{eq:boundary-shear-XY}, we can express the Darboux basis 
\eqref{eq:tetrahedra-to-quiver}
in terms of boundary shear coordinates $X_{\alpha_i}$ and $X_{-\alpha_i}$: 
\be\label{eq:Y-i-def}
\begin{split}
	X_{2i-1}  = X_{\alpha_{2i-1}} \,,
	\qquad&
	Y_{2i-1}  = \sum_{j=1}^{n-i+1} (-1)^{j+1} (X_{\alpha_{2j}} - X_{-\alpha_{2j}}) \,,\\ 
	X_{2i} = X_{-\alpha_{2i}} \,,
	\qquad&
	Y_{2i} = \sum_{j=1}^{i} (-1)^{j+i+1} (X_{\alpha_{2j-1}} - X_{-\alpha_{2j-1}})\,. \\ 
\end{split}
\ee
Relations (\ref{eq:canonical-quantization}--\ref{eq:xi-hat-yi-hat-def}) give an explicit isomorphism
\be\label{eq:q-torus-vs-skein}
	\CA_C^{IR}\otimes \CA_C^{IR,\mathrm{op}} \simeq \SkAlg_q^{\mathfrak{gl}_1}(\Sigma) \otimes \SkAlg_q^{\mathfrak{gl}_1}(\Sigma)^{\mathrm{op}},
\ee
thanks to the identification of $X_i, Y_i$ as insertions along cycles of $\Sigma_\pm$ given by (\ref{eq:surface-to-quiver-odd}--\ref{eq:surface-to-quiver-even}) and \eqref{eq:Y-i-def}; see Figure \ref{eq:L0-bimodule} for an example. 

Next, we give a concrete description of the (bi-)module structure of $\Sk_{q}^{\mathfrak{gl}_1}(L_0)$. 
We define the state associated with the trivial bordism: 
\be
	|\underbrace{0,\dots, 0}_{2n}\rangle\,. 
\ee
The $\hat x_i$ act by inserting lines from below (for $i$ odd) or from above (for $i$ even); see Figure \ref{fig:L0-sk-mod}. 
We generate a basis for $\Sk_q^{\mathfrak{gl}_1}(L_0)$ by their action, defined as
\be\label{eq:x_i_action_on_vector}
	\hat{x}_{i}|d_{1},\dots,d_{i},\dots,d_{2n}\rangle=  |d_{1},\dots,d_{i}+1,\dots,d_{2n}\rangle\,.
\ee
The $\hat y_i$ operators correspond to simultaneous insertions \emph{both} at the top and at the bottom \eqref{eq:Y-i-def}. 
Using the skein algebra \eqref{eq:xi-hat-yi-hat-def}, they act on our basis diagonally, as follows:\footnote{Note that we are choosing a certain polarization for the presentation of the skein module. 
A dual choice could have been, for instance:
\begin{equation}\nonumber
\begin{aligned}\hat{x}_{i}|\tilde{d}_{1},\dots,\tilde{d}_{i},\dots,\tilde{d}_{2n}\rangle= & q^{-2d_{i}}|\tilde{d}_{1},\dots,\tilde{d}_{i},\dots,\tilde{d}_{2n}\rangle,\\
\hat{y}_{i}|\tilde{d}_{1},\dots,\tilde{d}_{i},\dots,\tilde{d}_{2n}\rangle= & |\tilde{d}_{1},\dots,\tilde{d}_{i}+1,\dots,\tilde{d}_{2n}\rangle.
\end{aligned}
\end{equation}
}
\be\label{eq:y_i_action_on_vector}
	\hat{y}_{i}|d_{1},\dots,d_{i},\dots,d_{2n}\rangle= q^{2d_{i}}|d_{1},\dots,d_{i},\dots,d_{2n}\rangle.
\ee
\begin{figure}[h!]
\begin{center}
	\begin{subfigure}[b]{0.45\textwidth}
     	\includegraphics[width=\textwidth]{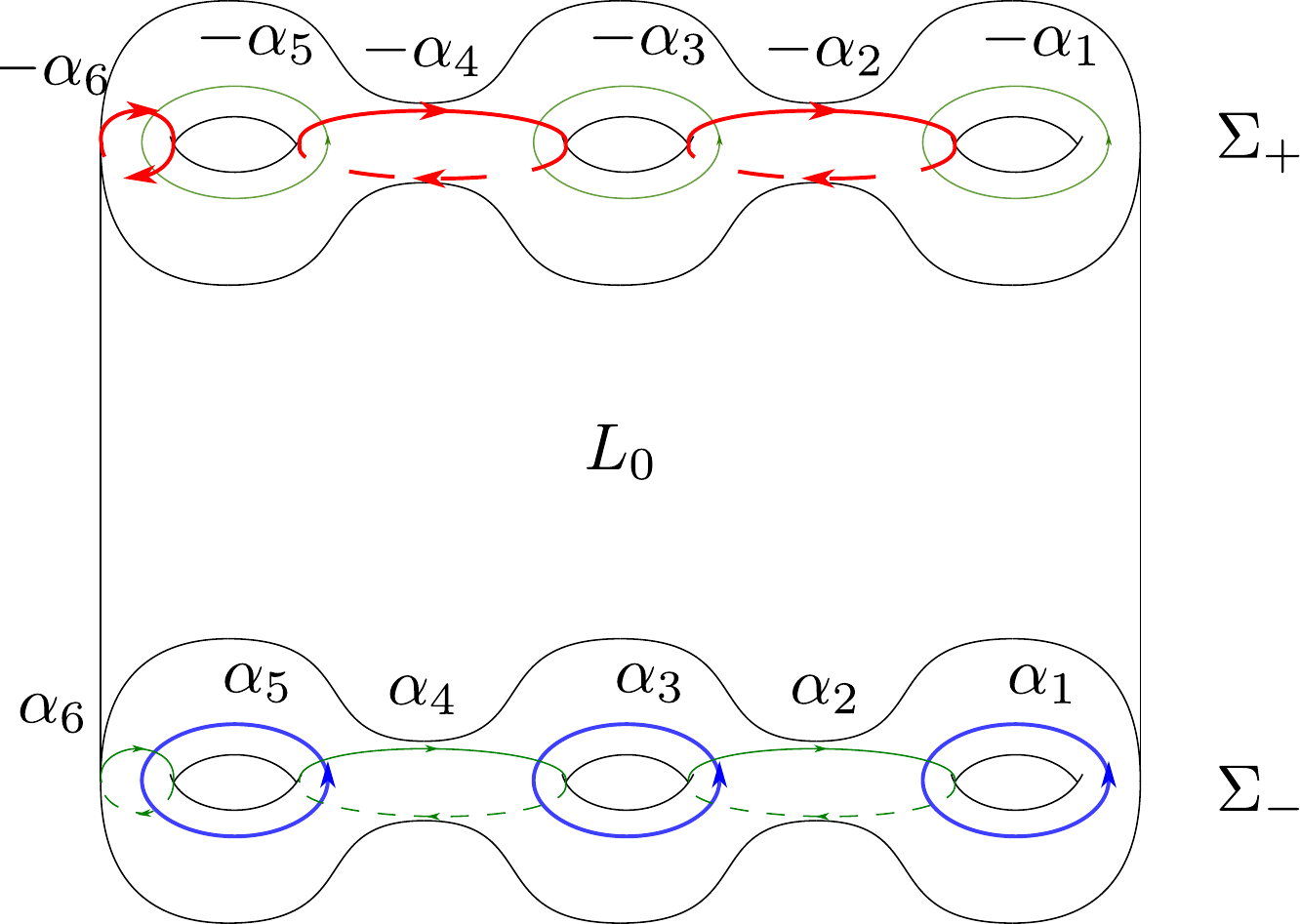}
      	\caption{Boundary curves on $L_0$. Highlighted in blue and red are the generators corresponding to $X_{1}\dots X_{2n}$.}
	\label{fig:L0-sk-mod}
    	\end{subfigure}
	\hfill
	\begin{subfigure}[b]{0.45\textwidth}
     	\includegraphics[width=\textwidth]{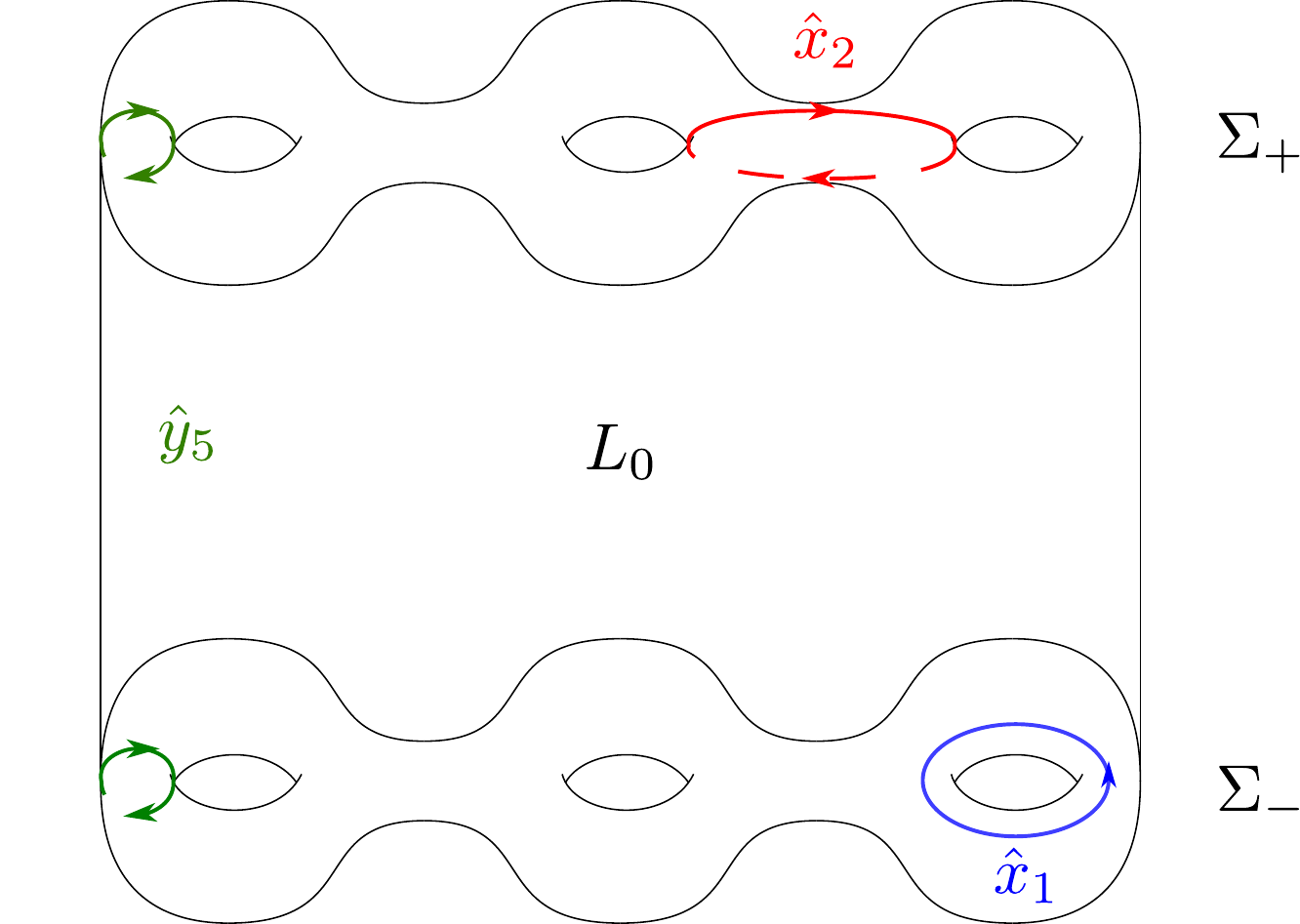}
     	\caption{Boundary curves corresponding to generators of the quiver quantum torus algebra $\hat x_i, \hat y_i$.}
	\label{fig:L0-sk-mod-quiver-variables}
    	\end{subfigure}
\end{center}
	\caption{}
	\label{eq:L0-bimodule}
\end{figure}
\paragraph{Relation to the motivic generating series of a symmetric quiver}

The collection of curves corresponding to boundaries of holomorphic disks in $L_0$ in the minimal chamber (Figure \ref{fig:A2n-M0}) defines the following element of the skein module: 
\be \label{eq:vector_for_symmetrization_of_A_2n}
\begin{split}
	|Q\rangle
	& := 
	\Psi\left(q^{-1}\hat{x}_{1}(-q\hat{y}_{1})^{Q_{11}}\hat{y}_{2}^{Q_{12}}\dots\hat{y}_{2n}^{Q_{1,2n}}\right)
	\times\dots\\
	&\dots\times 
	\Psi\left(q^{-1}\hat{x}_{i}(-q\hat{y}_{i})^{Q_{ii}}\hat{y}_{i+1}^{Q_{i,i+1}}\hat{y}_{i+2}^{Q_{i,i+2}}\dots\hat{y}_{2n}^{Q_{i,2n}}\right)\times\dots\\
	& \dots\times 
	\Psi\left(q^{-1}\hat{x}_{2n-1}(-q\hat{y}_{2n-1})^{Q_{2n-1,2n-1}}\hat{y}_{2n}^{Q_{2n-1,2n}}\right)
	\Psi\left(q^{-1}\hat{x}_{2n}(-q\hat{y}_{2n})^{Q_{2n,2n}}\right)
	\, |0,\dots,0\rangle\,.
\end{split}
\ee
We consider a representation of the $\mathfrak{gl}_1$ skein module over the field of formal series in $2n$ variables $x_1,\dots, x_{2n}$, defined in our basis by
\be\label{eq:wavefunction-rep}
	\langle x_1,\dots, x_{2n} | d_1, \dots, d_{2n}\rangle = x_1^{d_1}\dots x_{2n}^{d_{2n}}\,.
\ee
The use of normal ordering on the operators $\hat y_i$, as discussed in \cite{Ekholm:2019lmb}, shows that 
the skein element $|Q\rangle$ is mapped to the motivic Donaldson-Thomas partition function of the quiver $Q$ given by \eqref{eq:Efimov-PQ}
\be \label{eq:quiver_series_from_algebra}
	\langle x_1,\dots, x_{2n}|Q\rangle = P_Q(x_1,\dots, x_{2n},q) \,.
\ee

On the other hand, we may also consider capping off $L_0$ by gluing in handlebodies that fill in the~cycles
dual to $X_i$ -- namely, $\alpha_{2i-1}$ from the bottom surface and $\alpha_{2i}$ from the top one -- as discussed in Section~\ref{sec:Z-from-skein}. 
This defines an element in the dual vector space:
\be
	\langle \tilde 0,\dots,\tilde 0 |\,.
\ee
After gluing, we are left with an element of $\Sk^{\mathfrak{gl}_1}_q(S^3)$, which is the Nahm $q$-series \eqref{eq:Q-q-series}, and which coincides with \eqref{eq:Efimov-PQ} after setting $x_i=q$:\footnote{Note that for $A_{2n}$ quivers, $l_i$ and $Q_{ii}$ are even, so there are no minus signs in  \eqref{eq:Q-q-series} and \eqref{eq:Efimov-PQ}.} 
\be \label{eq:q_series_from_algebra}
	\langle \tilde 0,\dots,\tilde 0 | Q\rangle = \iota_*(Z)  = P_Q(q,\dots,q,q) \,.
\ee
 
%

\paragraph{Homomorphism between algebras of ranks \texorpdfstring{$2n$}{2n} and \texorpdfstring{$4n$}{4n}}

Relations \eqref{eq:boundary-shear-XY} define embeddings
\be
	\iota_\pm : \CP_{C_\pm} \hookrightarrow \CP_{\partial M_0} \simeq \CP_C\times \overline{\CP}_{C}
\ee
by expressing boundary shear coordinates $X_{\alpha_{i}}$ on $\CP_{C_-}$ (respectively, $X_{-\alpha_i}$ on $\CP_{C_+}$) as functions of $X_i, Y_i$.
The quantization defined by \eqref{eq:canonical-quantization} promotes these to embeddings of the surface skein algebras of $\Sigma_\pm$ (each of rank $2n$, with generators $\hat x_{\mp \alpha_i}$) in the quantum torus algebra of the quiver (of rank $4n$, with generators $\hat x_i, \hat y_i$): 
\be
	\hat\iota_\pm : \SkAlg^{\mathfrak{gl}_1}_{q}(\Sigma_\pm) \hookrightarrow \mathrm{QAlg} := 
	R\langle \hat{x}_1^{\pm 1},\dots, \hat x_{2n}^{\pm 1}, \hat{y}_1^{\pm 1},\dots, \hat{y}_{2n}^{\pm 1}\rangle / 
	(\hat{y_i}\hat{x_j} = q^{2\delta_{ij}}\hat{x_j}\hat{y_i}). 
\ee
Concretely, we define the quantum uplift of the embeddings as follows: 
\be\label{eq:embedding_iota}
\begin{split}
	\hat\iota_+(\hat x_{\alpha_{2i}}) &= -q^{-1} \hat x_{2i} \hat y_{2i-1} \hat y_{2i+1}, \\
	\hat\iota_+(\hat x_{\alpha_{2i-1}}) &= -q^{-1} \hat x_{2i-1}. \\
\end{split}
\ee
These can be summarized as follows for all $i=1,\dots, 2n$:
\be
	\hat\iota_+(\hat x_{\alpha_{i}}) = -q^{-1} \hat x_{2i} \prod_{j=1}^{2n} \hat y_{j}^{[\langle\alpha_i,\alpha_j\rangle]_+},\\
\ee
where
\be
    [z]_+ := \left\{
    \begin{array}{l}
        z \qquad \text{if }z>0\,, \\
        0 \qquad \text{if }z\leq 0 \,.\\
    \end{array}
    \right.
\ee
It follows that
\be
	\hat\iota_+(\hat x_{\alpha_i})
	\hat\iota_+(\hat x_{\alpha_j}) = q^{2\langle\alpha_i,\alpha_j\rangle} 
	\hat\iota_+(\hat x_{\alpha_j})
	\hat\iota_+(\hat x_{\alpha_i})
\ee
in agreement with the algebra defined by the quantum torus relations \eqref{eq:4d-q-torus}. 
Note that the intersection matrix \eqref{eq:4d-Dirac-pairing} 
\be\label{eq:antisymmetrization_embedding}
	(Q_{4\text{d}})_{ij} \equiv \langle\alpha_i,\alpha_j\rangle 
	= \Qsym{Q_{4\text{d}}}{\text{min}}_{ij} \equiv (-1)^{i}Q_{ij} = \left(\begin{array}{ccccc}
	0 & -1 & & & \\
	1 & 0 & 1 & & \\
	& -1 & 0 & -1 & \\
	& & & \ddots & \\
	\end{array}\right)
\ee
coincides with an antisymmetrization of the quiver adjacency matrix \eqref{eq:3d-quiver-adjacency-matrix}, as predicted on physical grounds in \eqref{eq:Dirac-pairing-symmetrization}.

\subsection{Generalization to arbitrary symmetric quivers}\label{sec:generalization}

In Section \ref{Topology}, we started from an ideally triangulated Riemann surface, discussed skein algebras together with their modules, and reached (\ref{eq:Q-q-series}--\ref{eq:matix_element_from_topology}), the $q$-series that arises as an evaluation of the motivic generating series of a symmetric quiver $Q$ with $x_i=(-1)^{l_i}q$. 
In Section \ref{sec:3d4d} we introduced an M-theory system that provides an interplay between 3d and 4d theories and focused on the case of $A_{m}$ Argyres-Douglas theories. 
We then followed a path from the 4d antisymmetric quiver to the 3d symmetric quiver, which was expressed in algebraic language in Sections \ref{sec:line_operators}--\ref{sec:embedding}. 
This allowed us to obtain the motivic generating series \eqref{eq:quiver_series_from_algebra} as well as its evaluation  \eqref{eq:q_series_from_algebra} from first principles of M-theoretic engineering of $A_{m}$ Argyres-Douglas theories. 
In this section, we will follow this path in reverse, starting from an arbitrary 3d symmetric quiver and trying to find the most general statements -- and to understand the limitations -- of the algebraic language introduced above. 
At this level of generality, the topological and geometric interpretations are still unknown, and we consider them to be exciting directions for future research.

\paragraph{Quantum torus algebras of rank $2m$ and their modules}
In Section \ref{sec:embedding}, we showed that the geometric M-theory construction of domain walls for $A_m$ Argyres-Douglas theories provides a quantum torus algebra \eqref{eq:q-torus-vs-skein} with an associated module (\ref{eq:x_i_action_on_vector}-\ref{eq:y_i_action_on_vector}), which contains an element $|Q\rangle$, such that $P_{Q}(\boldsymbol{x},q)$ can be reproduced as an appropriate representation of $|Q\rangle$, as in \eqref{eq:quiver_series_from_algebra}. 
Our goal here is to show that this algebraic construction generalizes in a straightforward way to arbitrary symmetric quivers, regardless of whether they arise from geometric data or not. 
The construction proceeds in exactly the same way, just in reverse, but we sketch the main steps for clarity. 

Consider the quantum torus algebra of rank $2m$ generated by operators $\{\hat{x}_{i}^{\pm 1},\hat{y}_{i}^{\pm 1}\}_{i=1,2,\dots,m}$ that satisfy the commutation relations \eqref{eq:xi-hat-yi-hat-def}. 
In analogy to (\ref{eq:x_i_action_on_vector}-\ref{eq:y_i_action_on_vector}), we also consider the module spanned by vectors $|d_{1},d_{2},\dots,d_{m}\rangle$, $d_{i}\in\mathbb{Z}$, for which
\begin{equation}\label{eq:right_module}
\begin{aligned}\hat{x}_{i}|d_{1},\dots,d_{i},\dots,d_{m}\rangle= & |d_{1},\dots,d_{i}+1,\dots,d_{m}\rangle,\\
\hat{y}_{i}|d_{1},\dots,d_{i},\dots,d_{m}\rangle= & q^{2d_{i}}|d_{1},\dots,d_{i},\dots,d_{m}\rangle.
\end{aligned}
\end{equation}

For an arbitrary symmetric quiver $Q$, we define an element of the module given by the adjacency matrix of the quiver: 
\begin{equation}
\begin{aligned}
|Q\rangle 
= & \Psi\left(q^{-1}\hat{x}_{1}(-q\hat{y}_{1})^{Q_{11}}\hat{y}_{2}^{Q_{12}}\dots\hat{y}_{m}^{Q_{1,m}}\right)\times\dots\\
\ldots\times & \Psi\left(q^{-1}\hat{x}_{i}(-q\hat{y}_{i})^{Q_{ii}}\hat{y}_{i+1}^{Q_{i,i+1}}\hat{y}_{i+2}^{Q_{i,i+2}}\dots\hat{y}_{m}^{Q_{i,m}}\right)\times\dots\\
\ldots\times & \Psi\left(q^{-1}\hat{x}_{m-1}(-q\hat{y}_{m-1})^{Q_{m-1,m-1}}\hat{y}_{m}^{Q_{m-1,m}}\right)\Psi\left(q^{-1}\hat{x}_{m}(-q\hat{y}_{m})^{Q_{mm}}\right)|0,\dots,0\rangle.
\end{aligned}
\label{eq:general_quiver}
\end{equation}
Introducing a wavefunction representation of the module, based on formal power series in variables $(x_1,\dots , x_{m})$ as in \eqref{eq:wavefunction-rep}, we obtain the motivic generating series as the wavefunction of $|Q\rangle$
\be
	\langle x_1,\dots, x_{m}|Q\rangle = P_Q(x_1,\dots, x_{m},q) \,.
\ee
This statement is formally identical to \eqref{eq:quiver_series_from_algebra}, but the construction of the algebra and of the vector $|Q\rangle$ now applies to arbitrary symmetric quivers.

\paragraph{3d-4d homomorphism}
In analogy to Section \ref{sec:embedding}, we can think of operators in \eqref{eq:general_quiver} as an algebra of rank $m$. 
More precisely, we define the \emph{3d-4d homomorphism} to be an embedding of quantum torus algebra of rank $m$, spanned by $X_{\alpha_{i}}$ that satisfy commutation relation \begin{equation}\label{eq:commutation_relation_quantum_torus_algebra}
X_{\alpha_{i}}X_{\alpha_{j}}=q^{2\langle\alpha_{i},\alpha_{j}\rangle}X_{\alpha_{j}}X_{\alpha_{i}}
\end{equation}
governed by pairing $\langle\alpha_{i},\alpha_{j}\rangle$, into the quantum torus algebra of rank $2m$, spanned by $\{\hat{x}_{i}^{\pm 1},\hat{y}_{i}^{\pm 1}\}_{i=1,2,\dots,m}$ that satisfy \eqref{eq:xi-hat-yi-hat-def}: 
\begin{equation}\label{eq:embedding_epsilon}
X_{\alpha_{i}}\overset{\epsilon}{\longmapsto}-q^{-1}\hat{x}_{i}(-q\hat{y}_{i})^{Q_{ii}}\hat{y}_{i+1}^{Q_{i,i+1}}\hat{y}_{i+2}^{Q_{i,i+2}}\dots\hat{y}_{m}^{Q_{i,m}}.
\end{equation}
The requirement of consistency between commutation rules \eqref{eq:xi-hat-yi-hat-def} and \eqref{eq:commutation_relation_quantum_torus_algebra} implies that the matrix $\langle\alpha_{i},\alpha_{j}\rangle$ is given by
\begin{equation}
\left[\begin{array}{ccccc}
0 & Q_{12} & Q_{13} & \dots & Q_{1m}\\
-Q_{12} & 0 & Q_{23} & \dots & Q_{2m}\\
-Q_{13} & -Q_{23} & \ddots & \ddots & \vdots\\
\vdots & \vdots & \ddots & 0 & Q_{m-1,m}\\
-Q_{1m} & -Q_{2m} & \dots & -Q_{m-1,m} & 0
\end{array}\right].\label{eq:antisymmetric_quiver}
\end{equation}
If we interpret $\langle\alpha_{i},\alpha_{j}\rangle$ as the intersection matrix of an antisymmetric quiver,\footnote{In this matrix $-Q_{ij}$ refers to the same arrows from $i$ to $j$ as $Q_{ij}$, but read as ``minus arrows from $j$ to $i$''.} then we can see that the direction of arrows matches the numbering of nodes: for $i < j$, each positive unit in $Q_{ij}$ corresponds to an arrow $i \longrightarrow j$. 
One can always define alternative embeddings, related to the numbering of nodes in $Q$ in a different way, but for clarity of presentation we will always use \eqref{eq:embedding_epsilon} and change the numbering of nodes if necessary. 

We can see that the 3d-4d homomorphism leads to the symmetrization of $\langle\alpha_{i},\alpha_{j}\rangle$ into $Q_{ij}$, with the numbers of loops being the parameters. 
On the other hand, the antisymmetric quiver can be understood as a description of how to build the symmetric quiver step by step, using the quantum dilogarithms. 
We can also combine $\epsilon$ with the knot-quiver correspondence \cite{KRSS1707short,KRSS1707long} and assign a quantum torus algebra of rank $m$ to a given knot.

\paragraph{$A_{2n}$ quivers} 
For $A_{2n}$ quivers, the construction discussed above with $m=2n$ and the one from Section \ref{sec:embedding} are equivalent. 
The only subtlety comes from the fact that the natural numbering of nodes in Figure \ref{fig:A2n-quiver} does not follow the direction of arrows. 
If we number nodes $\alpha_1,\alpha_3,\dots, \alpha_{2n-1}$ (sources) by $1,2,\dots,n$, and the nodes $\alpha_2,\alpha_4,\dots, \alpha_{2n}$ (sinks) by $n+1,n+2,\dots,2n$, we see that embeddings $\epsilon$ and $\hat{\iota}_+$ are the same, and that formulas (\ref{eq:embedding_epsilon}--\ref{eq:antisymmetric_quiver}) match perfectly with (\ref{eq:embedding_iota}--\ref{eq:antisymmetrization_embedding}).

\subsection{\texorpdfstring{$A_2$}{A2} quiver}\label{sec:A2_example}

In this section, we illustrate the construction presented in the previous one with the simplest nontrivial example of the symmetrization of the $A_2$ quiver.
\newline

Let us consider symmetric quiver $Q$ with two nodes and one pair of arrows (shown below), which is a~symmetrization of the 4d BPS quiver $Q_{4\text{d}}$ corresponding for the $A_{2}$ theory (see Figure \ref{fig:A2n-quiver}):
\begin{equation}\label{eq:sym_A2_quiver}
\Qsym{Q_{4\text{d}}}{\text{min}} = Q=\left[\begin{array}{cc}
0 & 1 \\
1 & 0 
\end{array}\right]\,
\qquad\qquad
\begin{tikzcd}
	\bullet & \bullet
	\arrow[curve={height=-6pt}, from=1-1, to=1-2]
	\arrow[curve={height=-6pt}, from=1-2, to=1-1]
\end{tikzcd}
\end{equation}
where the stability data for the minimal chamber are defined in Section \ref{sec:A2n-theory}. 
In this case, the motivic generating series is given by
\begin{equation}\label{eq:symA2_motivic_gen_series}
    P_Q(x_1,x_2,q)=\sum_{d_1,d_2\geq 0}(-q)^{2 d_1 d_2}\frac{x_1^{d_1}}{(q^{2};q^{2})_{d_1}}\frac{x_2^{d_2}}{(q^{2};q^{2})_{d_2}},
\end{equation}
and we would like to reproduce it as an expectation value of appropriate operators. 
In order to do so, we introduce a quantum torus algebra of rank $2m=4$ generated by $\{\hat{x}^{\pm1}_1,\hat{x}^{\pm1}_2,\hat{y}^{\pm1}_1,\hat{y}^{\pm1}_2\}$ that satisfy the following commutation relations:
\begin{equation}\label{eq:commutation_A2_rank_4}
    \hat{y}_1 \hat{x}_1 =q^2 \hat{x}_1 \hat{y}_1,\qquad \hat{y}_2 \hat{x}_2 =q^2 \hat{x}_2 \hat{y}_2 \,.
\end{equation}
The remaining commutation relations are trivial. 
We also consider the module spanned by vectors $|d_1,d_2\rangle$, $d_i\in \mathbb{Z}$, for which we have
\begin{equation}\label{eq:right_module_A2}
\begin{aligned}
\hat{x}_{1}|d_{1},d_{2}\rangle &= |d_{1}+1,d_{2}\rangle\,, & \hat{x}_{2}|d_{1},d_{2}\rangle &= |d_{1},d_{2}+1\rangle\,, \\
\hat{y}_{1}|d_{1},d_{2}\rangle &=  q^{2d_{1}}|d_{1},d_{2}\rangle\,, & \hat{y}_{2}|d_{1},d_{2}\rangle &=  q^{2d_{2}}|d_{1},d_{2}\rangle\,,
\end{aligned}
\end{equation}
We define a distinguished element of this module: 
\be
	|Q\rangle = 
	\Psi\left(q^{-1}\hat{x}_{1}\hat{y}_{2}\right) \Psi\left(q^{-1}\hat{x}_{2}\right) |0,0\rangle 
	=
	\sum_{d_1,d_2\geq 0}(-q)^{2 d_1 d_2}\frac{1}{(q^{2};q^{2})_{d_1}}\frac{1}{(q^{2};q^{2})_{d_2}} \, |d_1, d_2\rangle\,,
\ee
where we used  \eqref{eqn:qdilog}.
Thanks to the explicit expansion in the basis $|d_1, d_2\rangle$, it is now obvious that the wavefunction representation defined by
\be
	\langle x_1, x_2| d_1, d_2\rangle = x_1^{d_1}x_2^{d_2}
\ee
gives the motivic generating series \eqref{eq:symA2_motivic_gen_series}:
\begin{equation}\label{eq:evaluation_of_Psis_A2_case}
    \langle x_1, x_2| Q\rangle = P_Q(x_1,x_2,q)\,.
\end{equation}

Inspired by operators in \eqref{eq:evaluation_of_Psis_A2_case}, we can define an embedding -- the 3d-4d homomorphism -- of the quantum torus algebra of rank $m=2$ spanned by $\{X_{\alpha_1},X_{\alpha_2}\}$, satisfying the commutation relation \eqref{eq:commutation_relation_quantum_torus_algebra}  governed by the pairing $\langle\alpha_i,\alpha_j \rangle$, into the quantum torus algebra of rank $4$ defined above:
\begin{equation}\label{eq:embedding_A2}
X_{\alpha_{1}}\overset{\epsilon}{\longmapsto}-q^{-1}\hat{x}_{1}\hat{y}_{2}\,,\qquad X_{\alpha_{2}}\overset{\epsilon}{\longmapsto}-q^{-1}\hat{x}_{2}\,.
\end{equation}
By combining \eqref{eq:embedding_A2} with \eqref{eq:commutation_A2_rank_4} and \eqref{eq:commutation_relation_quantum_torus_algebra}, we see that the pairing is given by
\begin{equation}
    \begin{aligned}
    \langle\alpha_1,\alpha_1 \rangle &= 0\,, & \langle\alpha_1,\alpha_2 \rangle &= 1\,, \\
    \langle\alpha_2,\alpha_1 \rangle &= -1\,, & \langle\alpha_2,\alpha_2 \rangle &= 0\,.
\end{aligned}
\end{equation}
We can interpret  $\langle\alpha_{i},\alpha_{j}\rangle$  as the intersection matrix of $A_2$ quiver: 
\[
Q_{4\text{d}} \ =\ \  \stackrel{\alpha_1}{\bullet}\ \xrightarrow{\qquad}\ \stackrel{\alpha_2}{\bullet}\,.
\]
Comparing this with \eqref{eq:sym_A2_quiver}, we confirm that $Q$ indeed coincides with its symmetrization.

\subsection{\texorpdfstring{$A_m$}{Am} quiver}

In this section, we generalize the previous one to any $A_m$ quiver and its symmetrization. 
In order to avoid repetition, we choose the linear orientation of arrows (in contrast to alternating orientation in Figure~\ref{fig:A2n-quiver}; note that the construction from Section \ref{sec:generalization} works for any orientation) and reverse the direction of the reasoning from Sections \ref{sec:generalization} and \ref{sec:A2_example}: we will begin with the antisymmetric quiver and proceed toward its symmetrization (which is in line with the direction in Sections \ref{sec:line_operators} and \ref{sec:embedding}). 
\newline

Following the plan stated above, our starting point is the $A_m$ quiver in linear orientation: 
\[
\stackrel{\alpha_1}{\bullet}\ \xrightarrow{\qquad}\ \stackrel{\alpha_2}{\bullet}\ \xrightarrow{\qquad}\ \cdots \ \xrightarrow{\qquad}\ \stackrel{\alpha_m}{\bullet}\;.
\]
In the next step, we consider a quantum torus algebra of rank $m$ spanned by $\{X_{\alpha_1},X_{\alpha_2},\dots,X_{\alpha_m}\}$, satisfying the commutation relation \eqref{eq:commutation_relation_quantum_torus_algebra}, governed by the pairing $\langle\alpha_i,\alpha_j \rangle$ given by the intersection matrix of~$A_m$: 
\begin{equation}
\left[\begin{array}{ccccc}
0 & 1 & 0 & \dots &0\\
-1 & 0 & 1 & \dots & 0\\
0 & -1 & \ddots & \ddots & \vdots\\
\vdots & \vdots & \ddots & 0 & 1\\
0 & 0 & \dots & -1 & 0
\end{array}\right]\,.\label{eq:antisymmetric_quiver_Am}
\end{equation}
Based on \eqref{eq:embedding_epsilon} with $Q_{ij}=|\langle\alpha_i,\alpha_j \rangle|$,\footnote{This implicitly imposes the choice of the number of loops: $Q_{ii}=0$, but one could consider other choices as well.} we define an embedding (3d-4d homomorphism) into the quantum torus algebra of rank $2m$  generated by operators $\{\hat{x}_{i}^{\pm 1},\hat{y}_{i}^{\pm 1}\}_{i=1,2,\dots,m}$ that satisfy \eqref{eq:xi-hat-yi-hat-def}:
\begin{equation}
X_{\alpha_{i}}\overset{\epsilon}{\longmapsto}
\begin{cases}
    -q^{-1}\hat{x}_{i}\hat{y}_{i+1}, &i=1,2,\dots,m-1, \\
    -q^{-1}\hat{x}_{i}, &i=m.
\end{cases}
\end{equation}
If we consider the module given by \eqref{eq:right_module} and the wavefunction representation \eqref{eq:wavefunction-rep}, we immediately obtain the motivic generating series of a quiver from Figure \ref{fig:A2n-sym-quiver}, which is the symmetrization of the initial~$A_m$ quiver: 
\begin{equation}
\begin{split}
    \langle x_1,\dots , x_m| \Psi\left(q^{-1}\hat{x}_{1}\hat{y}_{2}\right)&\Psi\left(q^{-1}\hat{x}_{2}\hat{y}_{3}\right)\dots \Psi\left(q^{-1}\hat{x}_{m-1}\hat{y}_{m}\right)\Psi\left(q^{-1}\hat{x}_{m}\right) |0,0,\dots,0\rangle = \\
   = \sum_{d_1,d_2,\dots, d_m\geq 0}&(-q)^{2 d_1 d_2 + 2 d_2 d_3 + \dots + 2d_{m-1} d_m}\frac{x_1^{d_1}}{(q^{2};q^{2})_{d_1}}\frac{x_2^{d_2}}{(q^{2};q^{2})_{d_2}}\dots \frac{x_m^{d_m}}{(q^{2};q^{2})_{d_m}}\,.
\end{split}
\end{equation}

\section{The symmetrization map} \label{sec-wall-crossing}

In this section, we introduce our proposal for the symmetrization map (\ref{Sigma-maps}) which takes a 4d BPS quiver $Q_{4\text{d}}$ together with certain stability data $u$, and returns a 3d symmetric quiver $Q$:
\be
	\Qsym{Q_{4\text{d}}}{u} = Q\,.
\ee
In Section \ref{sec:3d4d}, we constructed this map in the special case where $Q_{4\text{d}}$ is a inward-outward linear quiver of an Argyres-Douglas theory of type $A_m$, and $u$ belongs to the minimal chamber. 
In that case, the symmetrization map is very simple: every vertex of $Q_{4\text{d}}$ maps to a vertex of $Q$ and every arrow is just doubled by a partner with opposite orientation. 

However, this simplicity hinges on the very specific choice of triangulations for $C$ and $M_0$, which is directly tied to the choice of a minimal chamber of the Coulomb branch of $T[C]$. 
Another limitation of the correspondence established so far is the implicit choice of a \emph{half} of the spectrum: while the BPS spectrum comes in CPT-conjugate pairs, the construction of $M_0$ involves assigning a tetrahedron only to half of the stable BPS states. 
In this section, we provide a general definition of the symmetrization map that extends far beyond the cases discussed so far. 
A key step towards this goal will be to establish a~precise relation between wall-crossing in 4d $\mathcal{N}=2$ theories and unlinking of 3d quivers, already hinted at in \cite{Ekholm:2019lmb}. 

\subsection{4d and 3d quivers across walls of the 1st and 2nd kind}\label{sec:mutations and wall-crossing to other chambers}


The description of 4d $\mathcal{N}=2$ BPS states via quivers involves two kinds of walls \cite{KS0811}. 
Walls of the ``first kind'' correspond to walls of marginal stability, i.e., real-codimension-one loci in the moduli space of a~theory (e.g., its Coulomb branch) across which the BPS spectrum jumps. 
Walls of the ``second kind'', instead, correspond to loci across which the imaginary part of the central charge of some BPS state changes sign: as this happens, the definition of particles vs. anti-particles changes, and the quiver description gets modified by a mutation \cite{Alim:2011ae}. 
While only walls of the first kind seem to play a~relevant role in connection to unlinking operations on 3d symmetric quivers, for completeness we discuss both, focusing on Argyres-Douglas theories of type $A_m$. 

\paragraph{Walls of the first kind.}
Moving to a chamber with a different number of BPS states does not change the 4d quiver (in general) but does change the triangulation $\tau_{M_0}$, thereby inducing a change in the symmetric quiver $Q$. 

In $A_{m}$ Argyres-Douglas theories, all wall-crossings are generated by the pentagon relation, which acts on $\tau_{M_0}$ by a 2-3 Pachner move. 
As shown in \cite{Dimofte:2011ju}, the Pachner move is directly related to the SQED/XYZ duality of 3d $\CN=2$ QFT. 
In \cite{Ekholm:2018eee, Ekholm:2019lmb}, it was further observed that this duality is also at play in the unlinking process of symmetric quivers. 
We therefore conclude that any change in the triangulation induced by 4d wall-crossing changes the 3d QFT description $T[M_0]$ in a way that is captured by the unlinking duality for the class of theories $T[Q]$. 

We can start in the minimal chamber and apply a number of pentagon relations to get to the maximal chamber.
At the level of the quantum torus algebra, the structure of such wall-crossing is well-understood and is, in fact, captured by the poset structure of the associated root system \cite{reineke2011cohomology,kontsevich2014wall}.
Namely, the charge lattice for the ${A}_m$ quiver is isomorphic to the root lattice of the $A_m$ root system. 
On the other hand, such root systems possess a natural structure of partially ordered sets. 
We denote $\Phi^0(A_m)$, $\Phi^+(A_m)$ as the sets of simple and positive roots in the root system of the Dynkin quiver of type $A_m$. 
Simple roots $\alpha_i$, which correspond to charges $\alpha_i$, label the vertices of the Dynkin diagram. 
Elements of $\Phi^+(A_m)\setminus\Phi^0(A_m)$ are of the form $\alpha_{i}+\alpha_{i+1}+\dots+\alpha_{i+l},\;l\in \mathbb{Z}_+$ and correspond to quiver segments (i.e., a sequence of nodes connected by arrows) 
\begin{equation}\label{eq:A_m segments}
\begin{tikzcd}[sep=small]
	{\alpha_i} & {\alpha_{i+1}} & \cdots & {\alpha_{i+l}}
	\arrow[no head, from=1-1, to=1-2]
	\arrow[no head, from=1-2, to=1-3]
	\arrow[no head, from=1-3, to=1-4]
\end{tikzcd}
\end{equation}
with appropriate orientation of arrows. 
In order to make our language simpler and more intuitive, instead of calling them positive non-simple roots, we will refer to them as composite roots. 

\begin{dfn}\label{dfn:partial ordering of A_m roots}
For any set of positive roots associated with the quiver ${A}_m$, we define a~partial ordering $\preceq$ based on the arrows of ${A}_m$ as follows. 
For positive roots $\beta_i=\alpha_{i}+\alpha_{i+1}+\dots 
 +\alpha_{i+l_i}$ and $\beta_j=\alpha_{j}+\alpha_{j+1}+\dots 
 +\alpha_{j+l_j}$, we write $\beta_i\preceq \beta_j$ if and only if, for some $k_i\in\{0,1,\dots,l_i\}$ and $k_j\in\{0,1,\dots,l_j\}$ there is an arrow $\alpha_{i+k_i}\to \alpha_{j+k_j}$  in the quiver ${A}_m$. 
\end{dfn}

\begin{dfn}\label{dfn:initial operator}
    We call 
    $$\boldsymbol{\Psi}_{{A}_m}=\Psi(-X_{\alpha_{i_1}})\dots \Psi(-X_{\alpha_{i_m}})$$
    an initial operator that follows the orientation of arrows in ${A}_m$ if $\alpha_{i_1},\dots \alpha_{i_m}\in \Phi^0(A_m)$ and $(\alpha_{i_k}\preceq \alpha_{i_l} \Rightarrow k<l)$, i.e., if there is an arrow from $\alpha_{i_k}$ to $\alpha_{i_l}$, then $\Psi(-X_{\alpha_{i_k}})$ appears before $\Psi(-X_{\alpha_{i_l}})$ in $\boldsymbol{\Psi}_{{A}_m}$. 
    
    We say that a sequence of pentagon relations follows the orientation of arrows in ${A}_m$ if the pentagon relation
    \[
    \boldsymbol{\Psi} = \dots \Psi(-X_{\beta_i}) \Psi(-X_{\beta_j}) \dots = \dots \Psi(-X_{\beta_j}) \Psi(-X_{\beta_i+\beta_j}) \Psi(-X_{\beta_i})  \dots
    \]
    requires $\beta_i\preceq \beta_{j}, \beta_i\nsucceq \beta_{j}$. 

    We call a sequence of pentagon relations that follow the orientation of arrows in ${A}_m$ \emph{maximal} if no further pentagon relations are allowed by the ordering $\preceq$, i.e., the final form of $\boldsymbol{\Psi}$ does not contain $\dots\Psi(-X_{\beta_i}) \Psi(-X_{\beta_j})\dots$ for which $\beta_i\preceq \beta_{j}$ and $\beta_i\nsucceq \beta_{j}$. 
\end{dfn}

Having said the above, the wall-crossing formula can be written as the following identity in $A_m$~quantum torus algebra (which follows directly from the pentagon relation): 
\begin{equation}\label{eq:Reineke_WCF}
    \prod^\curvearrowleft_{\alpha\in\Phi^0} \Psi(-X_{\alpha})\, = \prod^\curvearrowright_{\beta\in \Phi^+} \Psi(-X_{\beta}),
\end{equation}
where the products are taken over simple and positive roots $\alpha_i$ and $\beta_j$ of $A_m$ root system, respectively. 
The partial ordering of $\alpha$'s (``$\curvearrowleft$'', i.e., ordered from left to right by decreasing index) and $\beta$'s (``$\curvearrowright$'', i.e., ordered from left to right by increasing index) is fixed by Definition \ref{dfn:partial ordering of A_m roots}. 
\begin{rmk}
    Note that in $A_m$ Argyres-Douglas theories, (\ref{eq:Reineke_WCF}) relates the minimal and maximal BPS chambers. 
    All intermediate chambers, which are described by an orientation of $A_m$ quiver, are obtained by applying some number of pentagon relations to the right-hand side of (\ref{eq:Reineke_WCF}), where the ordering of $\alpha$'s on the right-hand side is only partially reversed. 
    We will discuss this in more detail when we present specific examples. 
\end{rmk}
\begin{rmk}
More general definition of the partial ordering from Definition \ref{dfn:partial ordering of A_m roots}, which is used, for example in \cite{reineke2011cohomology}, can be described in terms of quiver representation spaces associated to $Q_{4\text{d}}$: 
\begin{equation}
\begin{aligned}
   \alpha_j\rightarrow \alpha_i\ \qquad \quad \quad \iff & \quad
   \alpha_i \preceq \alpha_j \\
    \mathrm{Hom}(V(\beta_j),V(\beta_i)) = 0 \quad \iff & \quad \beta_i \preceq \beta_j\,.
\end{aligned}    
\end{equation}
Here, $\mathrm{Hom}(V(\beta_j),V(\beta_i))$ is the vector space of all morphisms from the indecomposable representation $V(\beta_j)$ to $V(\beta_i)$ of the corresponding Dynkin quiver. 
For example, for the ${A}_2$ quiver $\stackrel{\alpha_1}{\bullet}\ \xrightarrow{\qquad}\ \stackrel{\alpha_2}{\bullet}$, we have: 
\[
V(\alpha_1)=\mathbb{C} \longrightarrow 0,\quad 
V(\alpha_2)=0 \longrightarrow \mathbb{C},\quad 
V(\alpha_1+\alpha_2)=\mathbb{C} \longrightarrow \mathbb{C}\;.
\]
We can represent morphisms in $\mathrm{Hom}(V(\alpha_1),V(\alpha_1+\alpha_2))$ and $\mathrm{Hom}(V(\alpha_1+\alpha_2),V(\alpha_2))$ with the following commutative diagrams: 
\begin{equation}
    \begin{tikzcd}
	{\mathbb{C}} & 0 & {\mathbb{C}} & {\mathbb{C}} \\
	{\mathbb{C}} & {\mathbb{C}} & 0 & {\mathbb{C}}
	\arrow[from=1-1, to=1-2]
	\arrow[from=1-1, to=2-1]
	\arrow[from=1-2, to=2-2]
	\arrow[from=1-3, to=1-4]
	\arrow[from=1-3, to=2-3]
	\arrow[from=1-4, to=2-4]
	\arrow[from=2-1, to=2-2]
	\arrow[from=2-3, to=2-4]
\end{tikzcd}.
\end{equation}
Commutativity implies that the vertical arrows are zero, and therefore $\alpha_2 \preceq \alpha_1+\alpha_2 \preceq \alpha_1$. 
In turn, this dictates how the factors on the right-hand side of (\ref{eq:Reineke_WCF}) are ordered. 
As a result, the wall-crossing formula relating the minimal and maximal chambers of ${A}_2$ quiver is given by the pentagon relation:
\begin{equation}
    \Psi(-X_{\alpha_1})\Psi(-X_{\alpha_2}) = \Psi(-X_{\alpha_2})\Psi(-X_{\alpha_1+\alpha_2})\Psi(-X_{\alpha_1}). 
\end{equation}
\end{rmk}

The main question we pose in this section is: what happens on the 3d side of (\ref{eq:Reineke_WCF})?\footnote{The exact same question is studied for wild Kronecker quivers in \cite{Bryan:2025mwi}.} 
To answer it, we will map quantum torus algebra operators $\Psi(-X_*)$ to their analogues for symmetric quivers using the 3d-4d homomorphism from \eqref{eq:embedding_epsilon}.

\paragraph{Walls of the second kind.}
The distinction between 4d $\mathcal{N}=2$ BPS particles and anti-particles can be formulated as a choice of a half-plane in the complex plane of BPS central charges. 
Suppose that we remain in the strong coupling region but slightly perturb the configuration of central charges away from the condition \eqref{eq:Z-condition} in a generic way. 
After resolving in this way, we may consider tilting the choice of half-plane until the edge crosses the BPS ray corresponding to the charge with highest phase -- say, the ray corresponding to $Z_{\alpha_1}$. 
The effect of this tilting is well known: it induces a mutation in the quiver~$Q_{4\text{d}}$
\be
	Q_{4\text{d}} \ \ \mapsto\ \  \mu_{\alpha_1} \circ Q_{4\text{d}}\,,
\ee
while on the triangulation $\tau_C$, it induces a flip of the edge corresponding to $\alpha_1$. 

To see how the 3d quiver changes, we need to recall that the 3-manifold $M_0$ is built from the spectrum of BPS particles.
Tilting the choice of half-plane changes the BPS spectrum: $\alpha_1$ exits the half-plane, but its CPT conjugate $-\alpha_1$, whose central charge is $Z_{-\alpha_1} = -Z_{\alpha_1}$, enters as the state with the lowest phase. 
This means that $M_0$ changes by removing a tetrahedron corresponding to $\alpha_1$ at the top, and adding a~tetrahedron corresponding to $-\alpha_1$ at the bottom. 
We define in this way the action of a mutation $\mu_{\alpha_1}$ on~$M_0$: 
\be
	M_0 \ \ \mapsto\ \  \mu_{\alpha_1} \circ M_0\,.
\ee

We also wish to ask what happens at the level of $Q$, the 3d quiver. 
The answer can be deduced by repeating the analysis that we performed in the construction of the theory $T[M_0]$, now applied to $T[\mu_{\alpha_1} \circ M_0]$.  
The resulting 3d theory is described by a new symmetric quiver $Q'$, which coincides with the symmetrization of $\mu_{\alpha_1} \circ M_0$.
In particular, the 3d quiver can again be read off from the general dictionary in \eqref{eq:geometric-quivers}, with vertices in one-to-one correspondence with tetrahedra and arrows given by internal triangles shared by $X_i$-edges of the polarization $\Pi'$. 

We emphasize here that we only consider mutations $\mu_{\alpha_i}$ which do not create a 3-cycle in the 4d quiver (in other words, they do not introduce a superpotential). 
Clearly, combining wall-crossing with mutations of the half-plane can lead to interesting generalizations. 
In particular, mutations of 4d BPS quivers in chambers other than the strong coupling chamber can produce quivers with loops that include superpotentials. 
It is conceivable that the resulting 3d quiver $Q$ would then also differ in structure, compared to the $A_{m}$ quiver that we obtained in the strong coupling chamber. 
We leave a systematic exploration of these questions to future work.

\subsection{From pentagon relation to unlinking via the 3d-4d homomorphism}\label{sec:pentagon to unlinking}

In this section we take our first step towards a generalization of the symmetrization map $\mathfrak{S}$ outside the minimal chamber by analysing the pentagon relation.
We show how, for antisymmetric quivers in which different nodes are connected by at most one arrow (i.e., $\langle\alpha_{i},\alpha_{j}\rangle\leq1$),
the 3d-4d homomorphism defined in \eqref{eq:embedding_epsilon} maps the pentagon relation to the unlinking operation of \cite{Ekholm:2019lmb} on the same pair of nodes in the symmetrized quiver.\footnote{This includes quivers beyond the class of $A_m$ Argyres-Douglas theories -- in particular, one can consider other 4d quivers of finite mutation type coming from a triangulated surface. 
We also recall that for the $A_{2n}$ quivers $\epsilon=\hat{\iota}_+$ defined in \eqref{eq:embedding_iota}. 
}
\newline

Let us consider a quantum torus algebra of rank $m$ with commutation relation \eqref{eq:commutation_relation_quantum_torus_algebra}. If $\langle\alpha_{i},\alpha_{j}\rangle=0$, then
\begin{equation}
\Psi(-X_{\alpha_{i}})\Psi(-X_{\alpha_{j}})=\Psi(-X_{\alpha_{j}})\Psi(-X_{\alpha_{i}}),\label{eq:commutation}
\end{equation}
whereas for $\langle\alpha_{i},\alpha_{j}\rangle=1$, we have the pentagon relation:\footnote{Following the notation from (\ref{eq:general_quiver}), we assume that
$i<j$, which means that $\epsilon\left(\langle\alpha_{i},\alpha_{j}\rangle\right)=Q_{ij}$} 
\begin{equation}
\Psi(-X_{\alpha_{i}})\Psi(-X_{\alpha_{j}})=\Psi(-X_{\alpha_{j}})\Psi(-X_{\alpha_{i}+\alpha_{j}})\Psi(-X_{\alpha_{i}}).\label{eq:2-3_Pachner}
\end{equation}
The image of (\ref{eq:commutation}) under $\epsilon$ is trivial, but for (\ref{eq:2-3_Pachner}), we have:
\begin{align}\label{eq:LHS}
\epsilon\Big(\Psi(-X_{\alpha_{i}}) & \Psi(-X_{\alpha_{j}})\Big)=
 \Psi\left(\epsilon(-X_{\alpha_{i}})\right)\Psi\left(\epsilon(-X_{\alpha_{j}})\right) \\
= & \Psi\left(q^{-1}\hat{x}_{i}(-q\hat{y}_{i})^{Q_{ii}}\hat{y}_{i+1}^{Q_{i,i+1}}\hat{y}_{i+2}^{Q_{i,i+2}}\dots\hat{y}_{m}^{Q_{i,m}}\right)\Psi\left(q^{-1}\hat{x}_{j}(-q\hat{y}_{j})^{Q_{jj}}\hat{y}_{j+1}^{Q_{j,j+1}}\hat{y}_{j+2}^{Q_{j,j+2}}\dots\hat{y}_{m}^{Q_{j,m}}\right)\nonumber
\end{align}
and
\begin{align}
\epsilon\Big(\Psi & (-X_{\alpha_{j}})\Psi(-X_{\alpha_{i}+\alpha_{j}})\Psi(-X_{\alpha_{i}})\Big)
= \Psi\left(\epsilon(-X_{\alpha_{j}})\right)\Psi\left(\epsilon\left(-q^{\langle\alpha_{i},\alpha_{j}\rangle}X_{\alpha_{i}}X_{\alpha_{j}}\right)\right)\Psi\left(\epsilon(-X_{\alpha_{i}})\right) \nonumber\\
= & \Psi\left(\epsilon(-X_{\alpha_{j}})\right)\Psi\left(-q^{Q_{ij}}\epsilon\left(-X_{\alpha_{i}}\right)\epsilon\left(-X_{\alpha_{j}}\right)\right)\Psi\left(\epsilon(-X_{\alpha_{i}})\right) \label{eq:RHS}\\
= & \Psi\left(q^{-1}\hat{x}_{j}(-q\hat{y}_{j})^{Q_{jj}}\hat{y}_{j+1}^{Q_{j,j+1}}\hat{y}_{j+2}^{Q_{j,j+2}}\dots\hat{y}_{m}^{Q_{j,m}}\right) \nonumber\\
\times & \Psi\left(-q^{Q_{ij}-2}\hat{x}_{i}\hat{x}_{j}(-q\hat{y}_{i})^{Q_{ii}}\hat{y}_{j}^{Q_{i,j}}(-q\hat{y}_{j})^{Q_{jj}}\hat{y}_{i+1}^{Q_{i,i+1}}\dots\hat{y}_{j-1}^{Q_{i,j-1}}\hat{y}_{j+1}^{Q_{i,j+1}+Q_{j,j+1}}\hat{y}_{j+2}^{Q_{i,j+2}+Q_{j,j+2}}\dots\hat{y}_{m}^{Q_{i,m}+Q_{j,m}}\right) \nonumber\\
\times & \Psi\left(q^{-1}\hat{x}_{i}(-q\hat{y}_{i})^{Q_{ii}}\hat{y}_{i+1}^{Q_{i,i+1}}\hat{y}_{i+2}^{Q_{i,i+2}}\dots\hat{y}_{m}^{Q_{i,m}}\right). \nonumber
\end{align}
If we consider the following new variables:
\begin{equation}
\hat{x}'_{n}=q^{-1}\hat{x}_{i}\hat{x}_{j},\quad\hat{y}_{i}=\hat{y}'_{i}\hat{y}'_{n},\quad\hat{y}_{j}=\hat{y}'_{j}\hat{y}'_{n},\qquad\hat{x}_{k}=\hat{x}'_{k},\quad\hat{y}_{k}=\hat{y}'_{k}\quad\forall k\neq i,j,n
\end{equation}
and use the fact that $Q_{ij}=\langle\alpha_{i},\alpha_{j}\rangle=1$, then
\begin{equation*}
\Psi\left(\epsilon\left(-X_{\alpha_{i}+\alpha_{j}}\right)\right)=\Psi\left(q^{-1}\hat{x}'_{n}(-q\hat{y}'_{n})^{Q_{ii}+Q_{jj}+1}\hat{y}{'}_{i}^{Q_{ii}}\dots\hat{y}{'}_{j-1}^{Q_{i,j-1}}\hat{y}{'}_{j}^{1+Q_{jj}}\hat{y}{'}_{j+1}^{Q_{i,j+1}+Q_{j,j+1}}\dots\hat{y}{'}_{m}^{Q_{i,m}+Q_{j,m}}\right),
\end{equation*}
which looks like an operator generating a new node $n$. 
As a consequence, the equality between (\ref{eq:LHS}) and (\ref{eq:RHS}), considered in the presence of other unchanged operators inside \eqref{eq:general_quiver}, 
corresponds to the invariance of motivic generating series \eqref{eq:motivic_generating_series} under an operation called unlinking:
\begin{equation}
P_{Q}(\boldsymbol{x},q)=P_{U(ij)Q}(\boldsymbol{x},q^{-1}x_{i}x_{j},q)\,.
\end{equation}
Unlinking $U(ij)$ removes one pair of arrows between nodes $i$ and $j$ and increases the size of the quiver by one node. 
More precisely, for an arbitrary quiver\footnote{Since matrices in (\ref{eq:arbitrary quiver}-\ref{eq:unlinking_definition}) are symmetric and pretty big, for compactness and clarity of the presentation we write only the upper-triangular part.}
\begin{equation}
    \begin{split}\label{eq:arbitrary quiver}
    Q & =\left[\begin{array}{ccccccc}
 Q_{11}  &  \cdots  &  Q_{1i}  &  \cdots  &  Q_{1j}  &  \cdots  &  Q_{1m}\\
   &  \ddots\  &  \vdots &   &  \vdots  &   &  \vdots\\
 &  &  Q_{ii}  &  \cdots &  Q_{ij}  &  \cdots &  Q_{im}\\
 &  &  &  \ddots\  &  \vdots  &   &  \vdots\\
 &  &  &  &  Q_{jj}  &  \cdots  &  Q_{jm}\\
 &  &  &  &  &  \ddots &  \vdots\\
 &  &  &  &  &  & Q_{mm}
\end{array}\right]\,,
\end{split}
\end{equation}
we have:
\begin{equation}\label{eq:unlinking_definition}
    \begin{split}U(ij)Q & =\left[\begin{array}{cccccccc}
 Q_{11}  &  \cdots &  Q_{1i}  &  \cdots &  Q_{1j}  &  \cdots &  Q_{1m}  &  Q_{1i}+Q_{1j}\\
 &  \ddots\  &  \vdots &   &  \vdots &   &  \vdots &  \vdots\\
 &  &  Q_{ii}  &  \cdots &  Q_{ij}-1  &  \cdots &  Q_{im}  &  Q_{ii}+Q_{ij}-1\\
 &  &  &  \ddots\  &  \vdots &   &  \vdots &  \vdots\\
 &  &  &  &  Q_{jj}  &  \cdots &  Q_{jm}  &  Q_{ij}+Q_{jj}-1\\
 &  &  &  &  &  \ddots\  &  \vdots &  \vdots\\
 &  &  &  &  &  &  Q_{mm}  &  Q_{im}+Q_{jm}\\
 &  &  &  &  &  &    &  Q_{ii}+Q_{jj}+2Q_{ij}-1 
\end{array}\right]\,,
\end{split}
\end{equation}
and the generating parameter of the new node equals $q^{-1}x_{i}x_{j}$. Unlinking (as well as its partner operation, linking) was introduced in \cite{Ekholm:2019lmb}, together with an~interesting interpretation in terms of multi-cover skein relations and 3d $\mathcal{N}=2$ theories discussed earlier in \cite{Ekholm:2018eee, EkholmShende}. 

\begin{rmk}
We briefly comment on the generalization to symmetric quivers featuring pairs of nodes connected by more than one (pair of) arrows. 
The relation between the pentagon relation and unlinking can be generalized
to any symmetric quiver, but in that version it is much less natural
and direct. 
First of all, we have to alter the embedding $\epsilon\Big(\Psi(-X_{\alpha_{j}})\Big)$ in a way that depends on $\epsilon\Big(\Psi(-X_{\alpha_{i}})\Big)$: 
\begin{equation}
\begin{aligned}\epsilon\Big( & \Psi(-X_{\alpha_{i}})  \Psi(-X_{\alpha_{j}})\Big)=\\
= & \Psi\left(q^{-1}\hat{x}_{i}(-q\hat{y}_{i})^{Q_{ii}}\hat{y}_{i+1}^{Q_{i,i+1}}\hat{y}_{i+2}^{Q_{i,i+2}}\dots\hat{y}_{m}^{Q_{i,m}}\right)\Psi\left(q^{-1}\hat{x}_{j}(-q\hat{y}_{j})^{Q_{jj}}\hat{y}_{i}^{Q_{ij}-1}\hat{y}_{j+1}^{Q_{j,j+1}}\hat{y}_{j+2}^{Q_{j,j+2}}\dots\hat{y}_{m}^{Q_{j,m}}\right).
\end{aligned}
\end{equation}
Note that since there is no $\hat{x}_{i}$ appearing after $\epsilon\Big(\Psi(-X_{\alpha_{i}})\Big)$ in (\ref{eq:general_quiver}), the addition of $\hat{y}_{i}^{Q_{ij}-1}$ does not change the resulting generating series. 
Another alteration dependent on $Q_{ij}$ is required in the definition of new variables: 
\begin{equation}
\hat{x}'_{n}=q^{-Q_{ij}}\hat{x}_{i}\hat{x}_{j}\,,
\end{equation}
introducing a mismatch with the expression for generating parameters:
$x'_{n}=q^{-1}x_{i}x_{j}$. 
If we perform those modifications, we can match the pentagon relation with unlinking of nodes $i,j$ connected by arbitrary number of arrows $Q_{ij}$. 
However, if we want to unlink once more, we have to define a new embedding~$\epsilon$ that is consistent with the new order of quantum dilogarithms and new number of arrows connecting nodes $i$ and $j$, i.e., $Q_{ij}-1$. 
Note that for $Q_{ij}=1$, all those alterations are trivial, and further unlinking is not possible.\footnote{We can continue unlinking if we introduce negative arrows; see \cite{JKLNS2212}, but even in this generalized case, it is more natural to unlink positive arrows and link negative arrows, ending with fully unlinked nodes.}
\end{rmk}

\subsection{Connectors and their properties}\label{sec:connectors}

Aiming for the proper definition of the symmetrization map $\mathfrak{S}$ for $A_m$ quivers, we need to generalize the relation between wall-crossing in 4d and unlinking in 3d discussed in the previous section. 
In order to achieve this, we have to understand the conditions under which the 3d-4d homomorphism is well-behaved under wall-crossing and unlinking on the respective sides of the embedding. 
In other words, we have to relate the structure of unlinking to pentagon relations in the quantum torus algebra of rank $m$, for the $A_m$ quivers. 
It is a non-trivial correspondence, since any given symmetric quiver admits as many unlinking operations as the number of pairings of its nodes (which can then be iterated any number of times), while pentagon relations which realize mutations of a 4d quiver are very restricted. 
In consequence, we need some tools to help select proper sequences of unlinkings. 

\paragraph{Unlinking equivalences.}
To start with, we recall the relations satisfied by unlinking operators \cite{KLNS} and define an equivalence relation among such operators. 
For any symmetric quiver $Q$ containing distinct nodes $i,j,k,l$ we have the following commutation relation:
\begin{equation}
    U(kl)  U(ij)Q = U(ij)  U(kl)Q.
\end{equation}
At the operator level, we can write this as an equivalence relation, which we call a \emph{square}:
\begin{equation}\label{eq:square}
    U(kl)  U(ij) \sim U(ij)  U(kl).
\end{equation}
If two unlinkings share an index, we can swap their order using either of the following relations:
\begin{equation}\label{eq:unlinking hexagons}
    \begin{aligned}
        U(jk)  U((ij)k)  U(ij) Q &= U(ij)  U(i(jk))  U(jk) Q, \\
        U((ij)k)  U(jk)  U(ij) Q &= U(i(jk))  U(ij)  U(jk) Q. 
    \end{aligned}
\end{equation}
Here, $Q'=U(jk)  U((ij)k)  U(ij) Q$ and $Q''=U((ij)k)  U(jk)  U(ij) Q$ differ by one transposition of a pair of arrows, i.e., they can be unlinked once to the same quiver. 
Likewise, we can write
\begin{subequations}
    \begin{align}
        U(jk)  U((ij)k)  U(ij) &\sim U(ij)  U(i(jk))  U(jk) \label{eq:first hexagon},\\
        U((ij)k)  U(jk)  U(ij) & \sim U(i(jk))  U(ij)  U(jk) \label{eq:second hexagon}.
    \end{align}
\end{subequations}
We call these relations \emph{hexagons }(a) and (b), respectively. 
More generally, we declare
\begin{equation}\label{eq:equivalence of unlinking operators}
    \bfU \sim \bfU' \quad \iff \quad \bfU Q = \bfU' Q
\end{equation}
for any symmetric quiver $Q$ that contains arrows that are deleted by $\bfU$ and $\bfU'$. (We will call such quiver \emph{compatible} with $\bfU$ and $\bfU'$.)
Note that according to the Completeness Theorem \cite{KLNS}, 
$\bfU$ differs from $\bfU'$ by a finite number of transpositions of unlinkings, realized by applications of the square move (\ref{eq:square}), as well as associative replacement $U((ij)k)\leftrightarrow U(i(jk))$ following (\ref{eq:first hexagon}) or (\ref{eq:second hexagon}). 
We will also use the following notation for an equivalence class of operators:
$\llbracket \bfU \rrbracket := \{ \bfU' | \bfU \sim \bfU' \}$. 

\paragraph{Connector.}
One of the main ingredients in our construction is the \emph{connector} of two sequences of unlinking, introduced in \cite{KLNS}:
\begin{dfn}
Given two sequences of unlinking
\begin{equation}
\bfU(\boldsymbol{i}\boldsymbol{j})=  U(i_{n}j_{n})\dots U(i_{2}j_{2})U(i_{1}j_{1}), \quad 
\bfU(\boldsymbol{k}\boldsymbol{l})= U(k_{m}l_{m})\dots U(k_{2}l_{2})U(k_{1}l_{1}),
\end{equation}
a~pair
\begin{equation}
\bfU(\boldsymbol{i'}\boldsymbol{j'})=U(i'_{n'}j'_{n'})\dots U(i'_{2}j'_{2})U(i'_{1}j'_{1}),\quad \bfU(\boldsymbol{k'}\boldsymbol{l'})=U(k'_{m'}l'{}_{m'})\dots U(k'_{2}l'_{2})U(k'_{1}l'_{1})
\end{equation}
is called the connector of $\bfU(\boldsymbol{i}\boldsymbol{j})$ and $\bfU(\boldsymbol{k}\boldsymbol{l})$ if for any compatible quiver $Q$
\begin{equation}\label{eq:connector}
\bfU(\boldsymbol{i'}\boldsymbol{j'})\bfU(\boldsymbol{i}\boldsymbol{j})Q=\bfU(\boldsymbol{k'}\boldsymbol{l'})\bfU(\boldsymbol{k}\boldsymbol{l})Q.
\end{equation}
\end{dfn}

Note that there are generally infinitely many pairs $(\bfU(\boldsymbol{i'}\boldsymbol{j'}),\bfU(\boldsymbol{k'}\boldsymbol{l'}))$ that satisfy (\ref{eq:connector}) for a given $(\bfU(\boldsymbol{i}\boldsymbol{j}),\bfU(\boldsymbol{k}\boldsymbol{l}))$. 
However, the simplest such pairs contain the initial unlinkings, as well as all their associativity unlinkings only -- we call such pairs \emph{optimal}\footnote{To see this, look at the Connector Algorithm in \cite{KLNS}.}. 
Note also that for each repeated index $(^{ij}_{jk})$, we have a choice between hexagons (a) and (b). 
We can store the information about these choices in a set $H$, whose elements are of the form: $h=((^{ij}_{jk}),\rm{a})$ corresponding to (\ref{eq:first hexagon}), and $h=((^{ij}_{jk}),\rm{b})$ corresponding to (\ref{eq:second hexagon}). 
Combining these remarks allows us to define an equivalence class of optimal pairs with respect to (\ref{eq:equivalence of unlinking operators}): 
\begin{dfn}\label{dfn: Lambda class}
For two given sequences of unlinkings $\bfU(\boldsymbol{i}\boldsymbol{j}),\bfU(\boldsymbol{k}\boldsymbol{l})$ and a set $H$ (which specifies the type of the hexagon for each repeated index in $\bfU(\boldsymbol{i}\boldsymbol{j})$ and $\bfU(\boldsymbol{k}\boldsymbol{l})$), 
the \emph{connector class} $\mathbbm{\Lambda}(\bfU(\boldsymbol{i}\boldsymbol{j}),\bfU(\boldsymbol{k}\boldsymbol{l});H)$ is an equivalence class of optimal pairs constructed via the following steps:
\begin{enumerate}
\item Consider an optimal pair $(\bfU(\boldsymbol{i'}\boldsymbol{j'}),\bfU(\boldsymbol{k'}\boldsymbol{l'}))$ satisfying (\ref{eq:connector}), so that $\bfU(\boldsymbol{i'}\boldsymbol{j'})\bfU(\boldsymbol{i}\boldsymbol{j})$ is related to $\bfU(\boldsymbol{k'}\boldsymbol{l'})\bfU(\boldsymbol{k}\boldsymbol{l})$ by a sequence of square moves, as well as hexagon moves specified by $H$. 
\item Compute the equivalence class $\llbracket \bfU(\boldsymbol{i'}\boldsymbol{j'})\bfU(\boldsymbol{i}\boldsymbol{j})\rrbracket = \llbracket \bfU(\boldsymbol{k'}\boldsymbol{l'})\bfU(\boldsymbol{k}\boldsymbol{l})\rrbracket=\mathbbm{\Lambda}(\bfU(\boldsymbol{i}\boldsymbol{j}),\bfU(\boldsymbol{k}\boldsymbol{l});H)$.\footnote{The independence of $\mathbbm{\Lambda}(\bfU(\boldsymbol{i}\boldsymbol{j}),\bfU(\boldsymbol{k}\boldsymbol{l});H)$ from the choice of the optimal pair is guaranteed by the Completeness Theorem \cite{KLNS}.}
\end{enumerate}
\end{dfn}
A concrete realization of \eqref{eq:connector} and step \emph{1.} in the above definition is provided by the Connector Algorithm \cite{KLNS}. 
We stress that $Q' = \mathbbm{\Lambda}(\bfU,\bfU';H)Q$ is defined unambiguously, since it produces the same $Q'$ for any element of $\mathbbm{\Lambda}(\bfU,\bfU';H)$ -- in what follows, we will abuse the notation and denote the operator and its equivalence class by the same symbol. 
For example, 
\begin{equation}
    \mathbbm{\Lambda}(U(12),U(34);\emptyset) := \llbracket U(34)U(12) \rrbracket =  \{U(34)U(12),U(12)U(34)\}.
\end{equation}
Applying this to a particular quiver gives
\begin{equation}
Q' = \mathbbm{\Lambda}(U(12),U(34))Q = U(12)U(34)Q \equiv U(34)U(12)Q. 
\end{equation}
In another example we have
\begin{align}
    \mathbbm{\Lambda}(U(12),U(23);\{((^{12}_{23}),\rm{a})\}) = &\ \llbracket U((12)3)U(23)U(12)\rrbracket \\ =  &\ \{U((12)3)U(23)U(12),U(1(23))U(12)U(23)\}. \nonumber
\end{align}

We will also use the following shorthand notations: $\mathbbm{\Lambda}(\bfU,\bfU')$ for $\mathbbm{\Lambda}(\bfU,\bfU';H)$, where $H$ is considered arbitrary; $\mathbbm{\Lambda}(\bfU,\bfU';\rm{a})$ or $\mathbbm{\Lambda}(\bfU,\bfU';\rm{b})$ for $\mathbbm{\Lambda}(\bfU,\bfU';H)$, where $H$ consists of only type (a) or type (b) hexagons, respectively. 

\begin{prp}\label{prp:Properties of connector}
    For any sequences of unlinking $\bfU,\bfU',\bfU''$, the connector class satisfies the following identities:
    \begin{equation}
        \mathbbm{\Lambda}(\bfU,\bfU') = \mathbbm{\Lambda}(\bfU',\bfU),\qquad
        \mathbbm{\Lambda}(\mathbbm{\Lambda}(\bfU,\bfU'),\bfU'') = \mathbbm{\Lambda}(\bfU,\mathbbm{\Lambda}(\bfU',\bfU'')).
    \end{equation}
\end{prp}

\begin{proof}
    Commutativity follows from ``mirror reflecting'' the corresponding unlinking diagram:
\tikzset{every picture/.style={line width=0.75pt}} 
\begin{equation*}
\begin{tikzpicture}[x=0.75pt,y=0.75pt,yscale=-0.75,xscale=0.75]

\draw    (195.8,1443.8) -- (245.8,1493.8) ;
\draw    (245.8,1493.8) -- (295.8,1443.8) ;
\draw    (395.8,1443.8) -- (445.8,1493.8) ;
\draw    (445.8,1493.8) -- (495.8,1443.8) ;
\draw    (322.8,1468.8) -- (368.8,1468.8) ;
\draw [shift={(370.8,1468.8)}, rotate = 180] [color={rgb, 255:red, 0; green, 0; blue, 0 }  ][line width=0.75]    (10.93,-3.29) .. controls (6.95,-1.4) and (3.31,-0.3) .. (0,0) .. controls (3.31,0.3) and (6.95,1.4) .. (10.93,3.29)   ;
\draw [shift={(320.8,1468.8)}, rotate = 0] [color={rgb, 255:red, 0; green, 0; blue, 0 }  ][line width=0.75]    (10.93,-3.29) .. controls (6.95,-1.4) and (3.31,-0.3) .. (0,0) .. controls (3.31,0.3) and (6.95,1.4) .. (10.93,3.29)   ;

\draw (220.8,1468.8) node [anchor=north east] [inner sep=0.75pt]    {$\bfU$};
\draw (270.8,1468.8) node [anchor=north west][inner sep=0.75pt]    {$\bfU '$};
\draw (420.8,1468.8) node [anchor=north east] [inner sep=0.75pt]    {$\bfU'$};
\draw (470.8,1468.8) node [anchor=north west][inner sep=0.75pt]    {$\bfU$};
\end{tikzpicture}.
\end{equation*}
To confirm associativity, note that
$\mathbbm{\Lambda}(\boldsymbol{U'},\bfU'')\subset \mathbbm{\Lambda}(\mathbbm{\Lambda}(\bfU,\bfU'),\bfU'')$ 
as a set. This allows us to replace the second argument $\bfU''$ and write $\mathbbm{\Lambda}(\mathbbm{\Lambda}(\bfU,\bfU'),\bfU'') = \mathbbm{\Lambda}(\mathbbm{\Lambda}(\bfU,\bfU'),\mathbbm{\Lambda}(\bfU',\bfU''))$. Likewise, $\mathbbm{\Lambda}(\bfU,\bfU')\subset \mathbbm{\Lambda}(\bfU,\mathbbm{\Lambda}(\bfU',\bfU''))$ gives $\mathbbm{\Lambda}(\bfU,\mathbbm{\Lambda}(\bfU',\bfU'')) = \mathbbm{\Lambda}(\mathbbm{\Lambda}(\bfU,\bfU'),\mathbbm{\Lambda}(\bfU',\bfU''))$.
\end{proof}
Associativity can also be confirmed diagrammatically:
\begin{figure}[h!]
    \centering
    \includegraphics[width=0.75\linewidth]{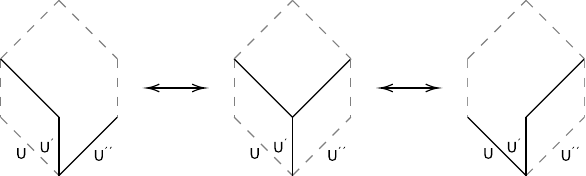}
\end{figure}

\paragraph{Joint connector.}
We now introduce a slightly more general version of a connector for an ordered set of sequences of unlinkings. 
\begin{dfn}
Given an $l$-tuple of sequences of unlinking
\begin{equation}\label{eq:tuple of unlinkseq}
(\bfU(\boldsymbol{i_1}\boldsymbol{j_1}),\dots,
\bfU(\boldsymbol{i_l}\boldsymbol{j_l}))
\end{equation}
and sets $H_k$ encoding the information about the choice of hexagon, 
an $l$-tuple of sequences of unlinking
$
(\bfU(\boldsymbol{i'_1}\boldsymbol{j'_1}),\dots,
\bfU(\boldsymbol{i'_l}\boldsymbol{j'_l}))
$
is called a joint connector for \eqref{eq:tuple of unlinkseq} and $H=\bigcup_{k} H_k$ if
\begin{align}
\bfU(\boldsymbol{i'_1}\boldsymbol{j'_1}) \bfU(\boldsymbol{i_1}\boldsymbol{j_1}) &= \dots = \bfU(\boldsymbol{i'_l}\boldsymbol{j'_l}) \bfU(\boldsymbol{i_l}\boldsymbol{j_l}) = 
 \\
&= \mathbbm{\Lambda}(\dots 
\mathbbm{\Lambda}(\,\mathbbm{\Lambda}
(\bfU(\boldsymbol{i_1}\boldsymbol{j_1}),\bfU(\boldsymbol{i_2}\boldsymbol{j_2}); H_{(12)}),\bfU(\boldsymbol{i_3}\boldsymbol{j_3}); H_{((12)3)}),\dots,\bfU(\boldsymbol{i_l}\boldsymbol{j_l}); H_{((12)\cdots) n)}). \nonumber
\end{align}
We call the equivalence class given by the above equality the joint connector class and denote it by
\begin{equation}\label{eq:joint connector def}
    \mathbbm{\Lambda}(\bfU(\boldsymbol{i_1}\boldsymbol{j_1}),\dots,
\bfU(\boldsymbol{i_l}\boldsymbol{j_l}); H).
\end{equation}
\end{dfn}
Note that associativity implies that $H_{k} = H_{k'}$ if $k'$ differs from $k$ by some change in the bracket order. 
Below, we prove the invariance under permutations. 
\begin{lma} For any permutation $\sigma\in S_l$,
\begin{equation}
    \mathbbm{\Lambda}(\,\bfU(\boldsymbol{i_1}\boldsymbol{j_1}),\dots,
\bfU(\boldsymbol{i_l}\boldsymbol{j_l})\,) = 
\mathbbm{\Lambda}(\sigma[\bfU(\boldsymbol{i_1}\boldsymbol{j_1}),\dots,
\bfU(\boldsymbol{i_l}\boldsymbol{j_l})]).
\end{equation}
\end{lma}
\begin{proof}
    (By induction.) The case $l=2$ is the commutativity relation; see Proposition \ref{prp:Properties of connector}. 
    Assume that the invariance holds for some $n>2$, and append one sequence to the argument:
    \begin{equation}\label{eq:Conn permutation induction step}
        \mathbbm{\Lambda}(\,\bfU(\boldsymbol{i_1}\boldsymbol{j_1}),\dots,
\bfU(\boldsymbol{i_n}\boldsymbol{j_n}),\bfU(\boldsymbol{i_{n+1}}\boldsymbol{j_{n+1}})\,).
    \end{equation}
    Using the invariance for $n$ arguments, we can write: 
\begin{equation}
        \begin{aligned}\label{eq:Conn permutation invariance}
        &\ \mathbbm{\Lambda}(\,\bfU(\boldsymbol{i_1}\boldsymbol{j_1}),\dots,
\bfU(\boldsymbol{i_n}\boldsymbol{j_n}),\bfU(\boldsymbol{i_{n+1}}\boldsymbol{j_{n+1}})\,) \\
    \stackrel{def.}{=} &\ \mathbbm{\Lambda}(\mathbbm{\Lambda}(\,\bfU(\boldsymbol{i_1}\boldsymbol{j_1}),\dots,
\bfU(\boldsymbol{i_n}\boldsymbol{j_n})),\bfU(\boldsymbol{i_{n+1}}\boldsymbol{j_{n+1}})\,) \\
\stackrel{ind.}{=} &\ \mathbbm{\Lambda}(\mathbbm{\Lambda}(\,\bfU(\boldsymbol{i_{\sigma_1}}\boldsymbol{j_{\sigma_1}}),\dots,
\bfU(\boldsymbol{i_{\sigma_n}}\boldsymbol{j_{\sigma_n}})),\bfU(\boldsymbol{i_{n+1}}\boldsymbol{j_{n+1}})\,) \\
\stackrel{a.}{=} &\ \mathbbm{\Lambda}(\,\bfU(\boldsymbol{i_{\sigma_1}}\boldsymbol{j_{\sigma_1}}),\mathbbm{\Lambda}(\bfU(\boldsymbol{i_{\sigma_2}}\boldsymbol{j_{\sigma_2}}),\dots,
\bfU(\boldsymbol{i_{\sigma_n}}\boldsymbol{j_{\sigma_n}}),\bfU(\boldsymbol{i_{n+1}}\boldsymbol{j_{n+1}})\,) \\
\stackrel{c.}{=} &\ \mathbbm{\Lambda}(\,\mathbbm{\Lambda}(\bfU(\boldsymbol{i_{\sigma_2}}\boldsymbol{j_{\sigma_2}}),\dots,
\bfU(\boldsymbol{i_{\sigma_n}}\boldsymbol{j_{\sigma_n}}),\bfU(\boldsymbol{i_{n+1}}\boldsymbol{j_{n+1}})),\bfU(\boldsymbol{i_{\sigma_1}}\boldsymbol{j_{\sigma_1}})\,).
 \end{aligned}
\end{equation}
One can then repeat all steps in (\ref{eq:Conn permutation invariance}) for some other permutation $\sigma'$. 
Iterating over all permutations, we confirm the invariance for $n+1$, completing the proof. 
\end{proof}
\begin{exmp*}
    Permutation invariance for $\mathbbm{\Lambda}(U(12),U(23),U(34))$. 
    Note that it suffices to show that
    \begin{equation}
        \mathbbm{\Lambda}(\mathbbm{\Lambda}(U(12),U(23)),U(34)) = \mathbbm{\Lambda}(\mathbbm{\Lambda}(U(12),U(34)),U(23)) =  \mathbbm{\Lambda}(\mathbbm{\Lambda}(U(23),U(34)),U(12)),
    \end{equation}
    as the remaining three permutations follow from applying commutativity.
    Let us compare the first iteration with $\sigma=(12)$ with the second iteration with $\sigma = \operatorname{Id}$: 
    \begin{equation}
        \begin{aligned}
        &\mathbbm{\Lambda}(U(12),U(23),U(34)) & & \qquad \mathbbm{\Lambda}(U(12),U(23),U(34))\\
    &\stackrel{def.}{=} \mathbbm{\Lambda}(\mathbbm{\Lambda}(U(12),U(23)),U(34)) & & \qquad\stackrel{def.}{=} \mathbbm{\Lambda}(\mathbbm{\Lambda}(U(12),U(23)),U(34))\\
    &\stackrel{ind.}{=}  \mathbbm{\Lambda}(\mathbbm{\Lambda}(U(23),U(12)),U(34)) & & \qquad \stackrel{ind.}{=}  \mathbbm{\Lambda}(\mathbbm{\Lambda}(U(12),U(23)),U(34)) \\
    & \stackrel{a.}{=} \mathbbm{\Lambda}(U(23),\mathbbm{\Lambda}(U(12),U(34))) & & \qquad \stackrel{a.}{=}  \mathbbm{\Lambda}(U(12),\mathbbm{\Lambda}(U(23),U(34))) \\
    & \stackrel{c.}{=} \mathbbm{\Lambda}(\mathbbm{\Lambda}(U(12),U(34)),U(23)). & & \qquad \stackrel{c.}{=}  \mathbbm{\Lambda}(\mathbbm{\Lambda}(U(23),U(34)),U(12)). \\
    \end{aligned}
\end{equation}

When all hexagons are of type (b), the resulting connector is shown in Figure \ref{fig:A4_connector_linear}. 

\begin{figure}[h!]
    \centering
    \includegraphics[width=0.35\linewidth]{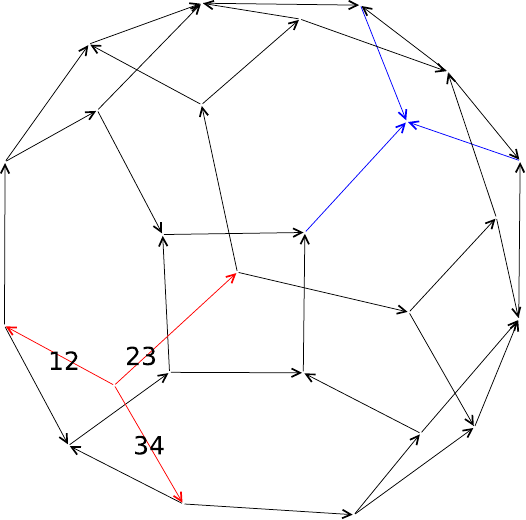}
    \caption{Joint connector $\mathbbm{\Lambda}(U(12),U(23),U(34);\, \rm{b})$ is an equivalence class of unlinking operators (represented by arrows) modulo square and hexagon moves. 
    The composition of two operators is represented by a pair of consecutive arrows. 
    Each such operator starts from a source node, whose outgoing edges ($U(12)$, $U(23)$, $U(34)$) are shown in \textcolor{red}{red}, and follows a path until it reaches the sink, whose outgoing edges are shown in \textcolor{blue}{blue}, here in six steps.}
    \label{fig:A4_connector_linear}
\end{figure}
\end{exmp*}

\subsection{Wall-crossing and connectors}\label{sec:wcf-connector theorem}

After introducing connectors (formed by unlinking operators acting on symmetric quivers corresponding to 3d theories) and showing their crucial properties, we are ready to establish their relation to the structure of wall-crossing for any $A_m$ Argyres-Douglas 4d theories.

\paragraph{Connectors and root systems.}
We start with a key observation which relates a certain type of connectors to root systems. 

\begin{prp}\label{prp:A_m connectors}
Each unlinking in any sequence $\bfU\in\mathbbm{\Lambda}[U(12),\dots,U(m-1,m)]$ (with any choice of hexagons) corresponds bijectively to an element in $\Phi^+(A_m)\setminus\Phi^0(A_m)$, i.e., to some composite root.
\end{prp}
\begin{proof}
We fix $\bfU\in \mathbbm{\Lambda}[U(12),\dots,U(m-1,m)]$ and define the \emph{canonical bijection} $\varphi:\;\bfU\to \Phi^+(A_m)\setminus\Phi^0(A_m)$ in the following steps. 
First, we map the initial unlinkings to their corresponding composite roots:
\begin{equation}
    U(i,i+1) \overset{\varphi}{\longmapsto}\alpha_{i}+\alpha_{i+1},\qquad i=1,\dots, m-1. 
\end{equation}
The other unlinkings may only arise from hexagons. 
Since for every element $U((ij)k)$ in the joint connector there must also be $U(i(jk))$ by associativity, unlinkings of the form $U((ij)k)$ or $U(i(jk))$ occur only for $j=i+1$ and $k=j+1$, the requirement for $U(ij)$ and $U(kl)$ to be initial. 
Moreover, the application of \cite[Lemma 3.10]{KLNS} to the definition of $\mathbbm{\Lambda}[U(12),\dots,U(m-1,m)]$ guarantees that, for any $i\in \{1,\dots, m-2\}$, exactly one bracketing -- either $U((i,i+1),i+2)$ or $U(i,(i+1,i+2))$ -- appears in $\bfU$. 
We define $\varphi$ to map those unlinkings to positive roots $\alpha_{i}+\alpha_{i+1}+\alpha_{i+2}$ for $i=1,\dots,m-2$. 

A simple generalization of this argument to unlinkings of the nodes that are produced from previous unlinkings implies that every unlinking in $\bfU$ is of the form $U[i,i+1,\dots ,i+l]$ -- with appropriate bracketing -- for some positive integers $i,l$ such that $i+l\leq m$. 
Moreover, the application of \cite[Lemma 3.10]{KLNS} guarantees that, for such $i$ and $l$, there is exactly one bracketing of $U[i,i+1,\dots ,i+l]$ that appears in~$\bfU$. 
We define $\varphi$ to map those unlinkings to positive roots  $\alpha_{i}+\alpha_{i+1}+\dots+\alpha_{i+l}$, and the aforementioned argument guarantees its injectivity. 

Since the definition of $\mathbbm{\Lambda}[U(12),\dots,U(m-1,m)]$ is based on optimal pairs, every sequence of unlinkings in this class contains initial unlinkings $U(12),\dots,U(m-1,m)$ and all their associativity unlinkings. 
This means that $\varphi$ is surjective: for every positive root $\alpha_{i}+\alpha_{i+1}+\dots+\alpha_{i+l}$ in $\Phi^+(A_m)\setminus\Phi^0(A_m)$, there exists an unlinking $U[i,i+1,\dots ,i+l]$ (with appropriate bracketing) that appears in $\bfU$. 

The injectivity and surjectivity discussed above imply that $\varphi$ is a bijection.\footnote{The set $H$ was always suppressed in the statement and the proof because the proposition is true for any choice of hexagons.}
\end{proof}

\paragraph{Encoding orientation of arrows in the connector.}
In order to formulate the relation between unlinking and wall-crossing, we first define what it means to follow the orientation of arrows of a given Dynkin quiver from the perspective of hexagons and quantum torus algebra. 

\begin{dfn}\label{dfn:WCF for a symmetric quiver}
For any Dynkin quiver $A_m$
let $\beta_i=\alpha_{i}+\alpha_{i+1}+\dots 
 +\alpha_{i+l_i}$, $\beta_j=\alpha_{j}+\alpha_{j+1}+\dots 
 +\alpha_{j+l_j}$, and $\beta_k=\alpha_{k}+\alpha_{k+1}+\dots 
 +\alpha_{k+l_k}$ be a triple of composite roots such that the corresponding segments (\ref{eq:A_m segments}) in $A_m$ are connected:\footnote{Note that all hexagons in $\mathbbm{\Lambda}$ arise from such a configuration, which follows from Proposition \ref{prp:A_m connectors}.}
\[
\begin{tikzcd}[sep=small]
	{\cdots \alpha_{i}} & \cdots & {\alpha_{i+l_i}} & {\alpha_{j}} & \cdots & {\alpha_{j+l_j}} & {\alpha_{k}} & \cdots & {\alpha_{k+l_k}\cdots}
	\arrow[no head, from=1-1, to=1-2]
	\arrow[no head, from=1-3, to=1-2]
	\arrow[no head, from=1-3, to=1-4]
	\arrow[no head, from=1-4, to=1-5]
	\arrow[no head, from=1-5, to=1-6]
	\arrow[no head, from=1-7, to=1-6]
	\arrow[no head, from=1-7, to=1-8]
	\arrow[no head, from=1-9, to=1-8]
\end{tikzcd}
\]
where connectivity requires $j=i+l_i+1$ and $k=j+l_j+1$, with the appropriate orientation of arrows. 
We say that the set \emph{$H$ follows the orientation of $A_m$} if for the repeated index $(^{ij}_{jk})$ corresponding to $\beta_i,\beta_j,\beta_k$, we assign hexagon \emph{(a)} if the arrow from $\alpha_{i+l_i}$ to $\alpha_{j}$ and the arrow from $\alpha_{j+l_j}$ to $\alpha_{k}$ are co-oriented; otherwise, we assign hexagon \emph{(b)}. 
We denote such a set by $H_{A_m}$. 
\end{dfn}

\begin{rmk}\label{rmk:locality of hexagons}
    Note that by specifying all hexagons in the joint connector in this way, we can say that for any operator $\bfU\in\mathbbm{\Lambda}$, any three unlinkings $U(ij),U(jk),U((ij)k)\in \bfU$ (even if they are separated by some other unlinkings) will follow the pattern of arrows of $A_m$:
    \begin{equation}\label{eq:separated_unlinkings_hexagon_a}
        \bfU = \ldots U(jk) \ldots U((ij)k) \ldots U(ij) \ldots 
    \end{equation}
    whenever the corresponding arrows are co-oriented, and
    \begin{equation}\label{eq:separated_unlinkings_hexagon_b}
        \bfU = \ldots U((ij)k)) \ldots U(jk) \ldots U(ij) \ldots 
    \end{equation}
    otherwise. 
    It follows from the fact that by applying some sequence of square and hexagon moves which do not alter the order of the three unlinkings above, \eqref{eq:separated_unlinkings_hexagon_a} can be brought to a form corresponding to (\ref{eq:first hexagon}):
    \begin{equation}
        \bfU = \ldots U(jk) U((ij)k) U(ij) \ldots \,,
    \end{equation}
    whereas \eqref{eq:separated_unlinkings_hexagon_b} can be transformed into
    \begin{equation}
        \bfU = \ldots U((ij)k)) U(jk) U(ij) \ldots \,,
    \end{equation}
    which corresponds to (\ref{eq:second hexagon}). 
\end{rmk}

\paragraph{Connector computes wall-crossing.}
Our main result is that the joint connector computes the wall-crossing formulae (\ref{eq:Reineke_WCF}) of the quantum torus algebra. 
This can be seen as a direct, yet non-trivial generalization of the relation between the pentagon relation and unlinking from Section \ref{sec:pentagon to unlinking}. 
Namely, by leveraging the 3d-4d homomorphism (\ref{eq:embedding_epsilon}), all combinatorics associated with wall-crossing (i.e., the structure of all possible compositions of pentagon relations) can be expressed at the level of the symmetric quiver using the language of unlinking operators, as stated below. 
\begin{thm}\label{thm:unlinking and A_m wall-crossing}
    For any Dynkin quiver $A_m$, the joint connector class $\mathbbm{\Lambda}(U(12),\dots,U(m-1,m);H_{A_m})$ is in bijection with the set of all maximal sequences of pentagon relations applied to the initial operator $\boldsymbol{\Psi}_{A_m}$ that follow the orientation of arrows in $A_m$. 
\end{thm}

\begin{proof}
    We will construct the bijection in two stages: the first will be a definition for one maximal sequence of pentagon relations, and the second will propagate it across the whole set. 
    After that, we will show the surjectivity and injectivity of the constructed map.
    
    Let us start by fixing the Dynkin quiver $A_m$ and the corresponding initial operator $\boldsymbol{\Psi}_{A_m}$. 
    Then, we perform an iterative construction in which we pick a pair of neighboring\footnote{
    In this reasoning we assume that all necessary transpositions of commuting operators are done without specific mention. 
    In the first iteration $\beta_i$ and $\beta_j$ are simple roots, in the next they may be composite.
    } operators $\Psi(-X_{\beta_i}),\Psi(-X_{\beta_j})$ labeled by roots satisfying $\beta_i\preceq \beta_{j}, \beta_i\nsucceq \beta_{j}$ and apply the pentagon relation, which produces a descendant operator $\Psi(-X_{\beta_i+\beta_j})$ labeled by the root $\beta_i+\beta_j\in\Phi^+(A_m)\setminus\Phi^0(A_m)$. 
    The inverse of the bijection $\varphi$ from Proposition \ref{prp:A_m connectors} maps this root to the unlinking operator $U(ij)$.\footnote{Indices $i$ and $j$ may be composite and in that case the bracketing in $U(ij)$ directly follows from the order of previous unlinkings.} 
    We repeat these steps until we obtain a maximal sequence of pentagon relations labeled by all positive non-simple roots. 
    We combine the assignments following $\varphi^{-1}$ into a map between maximal sequences of pentagon relations and sequences of unlinkings: 
    \begin{align}\label{eq:pentagon relations and unlinking bijection}
        \big( \dots \Psi(-X_{\beta_{i_1}})\Psi(-X_{\beta_{j_1}})\ldots & \rightarrow  \dots \Psi(-X_{\beta_{j_1}})\Psi(-X_{\beta_{i_1}+\beta_{j_1}})\Psi(-X_{\beta_{i_1}})\dots \big) & \overset{\varphi^{-1}}{\longmapsto}\qquad & U(i_1 j_1) \nonumber \\
        \big( \dots \Psi(-X_{\beta_{i_2}})\Psi(-X_{\beta_{j_2}})\ldots & \rightarrow  \dots \Psi(-X_{\beta_{j_2}})\Psi(-X_{\beta_{i_2}+\beta_{j_2}})\Psi(-X_{\beta_{i_2}})\dots \big) & \overset{\varphi^{-1}}{\longmapsto}\qquad & U(i_2 j_2) \nonumber \\
         &\;\; \vdots &\vdots \; \; \; \qquad & \quad\; \vdots \\
         \big( \dots \Psi(-X_{\beta_{i_s}})\Psi(-X_{\beta_{j_s}})\ldots & \rightarrow  \dots \Psi(-X_{\beta_{j_s}})\Psi(-X_{\beta_{i_s}+\beta_{j_s}})\Psi(-X_{\beta_{i_s}})\dots \big) & \overset{\varphi^{-1}}{\longmapsto}\qquad & U(i_s j_s). \nonumber
    \end{align}
    We still need to show that the target sequence of unlinkings $U(i_s j_s)\dots U(i_2 j_2)U(i_1 j_1)\equiv\bfU$ belongs to the joint connector $\mathbbm{\Lambda}[U(12),\dots,U(m-1,m);H_{A_m}]$. 
    Since every maximal sequence of pentagon relations creates operators $\Psi(\beta_k)$ labeled by all $\beta_k\in\Phi^+(A_m)\setminus\Phi^0(A_m)$, Proposition \ref{prp:A_m connectors} guarantees that $U(i_1 j_1), U(i_2 j_2),\dots,U(i_s j_s)$ do appear in any element of $\mathbbm{\Lambda}[U(12),\dots,U(m-1,m)]$ -- we only need to make sure that the arrangement of unlinkings in $\bfU$ is compatible with the structure of hexagons $H_{A_m}$ following the orientation of $A_m$.\footnote{
    We also have to take the bracketing of $U(i_1 j_1), U(i_2 j_2),\dots,U(i_s j_s)$ into account, but it automatically follows from their arrangement in $\bfU$. 
    } 
    In order to show this, let us pick from $\bfU$ any pair of unlinkings that share an index and the corresponding associativity unlinking and denote them $U(ij)$, $ U(jk)$, $U((ij)k)$ -- this choice implicitly assumes that $U(ij)$ appears before $ U(jk)$ in $\bfU$, we can change the labels if this is not the case. 
    Depending on the orientation of arrows in $A_m$, we are in one of the following two situations: 

\begin{enumerate}
    \item If the arrows from $i$ to $j$ and from $j$ to $k$ are co-oriented, then the pentagon relation corresponding to $ U(ij)$ produces an operator containing one of the following ordered products: 
    \begin{equation}\label{eq:pentagon stage after U(ij)}
    \begin{split}
        \ldots \Psi(-X_{\beta_j})&\Psi(-X_{\beta_i+\beta_j})\Psi(-X_{\beta_i})\ldots\Psi(-X_{\beta_k})\ldots\\
        \text{or}\qquad \ldots \Psi(-X_{\beta_k})\ldots\Psi(-X_{\beta_i})&\Psi(-X_{\beta_i+\beta_j})\Psi(-X_{\beta_j})\ldots\;.
    \end{split}
    \end{equation}
    Note that $\Psi(-X_{\beta_i})$ and $\Psi(-X_{\beta_k})$ commute -- the existence of $j$ between $i$ and $k$ means that they cannot be connected by an arrow in $A_m$. 
    Since $U((ij)k)$ is also in $\bfU$, at some stage after applying other pentagon relations to (\ref{eq:pentagon stage after U(ij)}), we must encounter 
    $\Psi(-X_{\beta_i})$ right next to $\Psi(-X_{\beta_k})$:
    \begin{equation}
    \begin{split}
        \ldots \Psi(-X_{\beta_j})&\Psi(-X_{\beta_i+\beta_j})\Psi(-X_{\beta_i})\Psi(-X_{\beta_k})\ldots \\
        \text{or}\qquad \ldots \Psi(-X_{\beta_k})\Psi(-X_{\beta_i})&\Psi(-X_{\beta_i+\beta_j})\Psi(-X_{\beta_j})\ldots\;.
    \end{split}
    \end{equation}
     After permuting $\Psi(-X_{\beta_i})$ with $\Psi(-X_{\beta_k})$, we see that the latter will be adjacent to $\Psi(-X_{\beta_i+\beta_j})$ before $\Psi(-X_{\beta_j})$. 
     It means that
    \begin{equation}
        \bfU = \ldots U(jk)\ldots U((ij)k)\ldots U(ij) \ldots\;,
    \end{equation}
  which is equivalent (modulo hexagon and square moves used to make the above unlinking adjacent -- see also Remark \ref{rmk:locality of hexagons}) to one side of hexagon \eqref{eq:first hexagon} that is assigned to co-oriented arrows in Definition \ref{dfn:WCF for a symmetric quiver}. 
  \item If the arrows from $i$ to $j$ and from $j$ to $k$ have opposite orientations, then the pentagon relation corresponding to $U(ij)$ produces an operator containing one of the following ordered products: 
    \begin{equation}
    \begin{split}
        \ldots \Psi(-X_{\beta_k})\ldots\Psi(-X_{\beta_j})&\Psi(-X_{\beta_i+\beta_j})\Psi(-X_{\beta_i})\ldots\\
        \text{or}\qquad \ldots \Psi(-X_{\beta_i})&\Psi(-X_{\beta_i+\beta_j})\Psi(-X_{\beta_j})\ldots\Psi(-X_{\beta_k})\ldots\;.
    \end{split}    
    \end{equation}
    Repeating the above reasoning, we see that this time $\Psi(-X_{\beta_k})$ will be adjacent to $\Psi(-X_{\beta_j})$ before $\Psi(-X_{\beta_i+\beta_j})$. 
    It means that
    \begin{equation}
        \bfU = \ldots U((ij)k)\ldots U(jk)\ldots U(ij) \ldots
    \end{equation}
    which is equivalent to one side of hexagon \eqref{eq:second hexagon} that is assigned to oppositely oriented arrows in Definition \ref{dfn:WCF for a symmetric quiver}. 
\end{enumerate}
Note that these steps apply as well in the case of composite nodes, where we consider the arrows that connect the outermost initial nodes, according to Definition \ref{dfn:WCF for a symmetric quiver}. 
In each situation, the arrangement of unlinkings in $\bfU$ is compatible with the structure of hexagons $H_{A_m}$, so $\bfU\in \mathbbm{\Lambda}[U(12),\dots,U(m-1,m);H_{A_m}]$ indeed. 

The joint connector class $\mathbbm{\Lambda}[U(12),\dots,U(m-1,m);H_{A_m}]$, by definition, is closed under square and hexagon moves following the orientation of $A_m$ encoded in $H_{A_m}$. 
Therefore, knowing one sequence $\bfU$ as above, we can generate all other elements in $\mathbbm{\Lambda}$ by simply iterating all such moves starting from $\bfU$. 
On the other hand, all maximal sequences of pentagon relations arise from the maximal sequence from \eqref{eq:pentagon relations and unlinking bijection} by changing the arrangement of the pentagon relations. 
Every square move exchanges two unlinkings that do not share an index. 
Flipping their order can also be understood as creating a new maximal sequence of pentagon relations which is mapped to the sequence of unlinkings created by the square move, following $\varphi^{-1}$. 
In other words, the following diagram commutes:
\[
\begin{tikzcd}[row sep=0.1,scale cd=1]
	{\cdots\Psi(-X_{\beta_i})\Psi(-X_{\beta_j})\Psi(-X_{\beta_k})\Psi(-X_{\beta_l})\cdots} \\
	{\cdots\Psi(-X_{\beta_i})\Psi(-X_{\beta_j})\Psi(-X_{\beta_l})\Psi(-X_{\beta_k+\beta_l})\Psi(-X_{\beta_k})\cdots } & \rightarrow & {\cdots U(ij)U(kl)\cdots} \\
	{\cdots\Psi(-X_{\beta_j})\Psi(-X_{\beta_i+\beta_j})\Psi(-X_{\beta_i})\Psi(-X_{\beta_l})\Psi(-X_{\beta_k+\beta_l})\Psi(-X_{\beta_k})\cdots} \\
	\downarrow && \downarrow \\
	{\cdots\Psi(-X_{\beta_i})\Psi(-X_{\beta_j})\Psi(-X_{\beta_k})\Psi(-X_{\beta_l})\cdots} \\
	{\cdots\Psi(-X_{\beta_j})\Psi(-X_{\beta_i+\beta_j})\Psi(-X_{\beta_i})\Psi(-X_{\beta_k})\Psi(-X_{\beta_l})\cdots } & \rightarrow & {\cdots U(kl)U(ij)\cdots} \\
	{\cdots\Psi(-X_{\beta_j})\Psi(-X_{\beta_i+\beta_j})\Psi(-X_{\beta_i})\Psi(-X_{\beta_l})\Psi(-X_{\beta_k+\beta_l})\Psi(-X_{\beta_k})\cdots}
\end{tikzcd}
\]
   Similarly, every hexagon move changes the arrangement of three unlinkings -- two that share an index, and one associativity unlinking -- following the orientation of $A_m$ encoded in $H_{A_m}$. 
   Changing the arrangement of pentagon relations between three operators labeled by roots that are ordered by $\preceq$ follows the orientation of $A_m$ in a completely analogous way (as discussed in the previous part of the proof), creating a new maximal sequence which is mapped to the sequence of unlinkings created by the hexagon move, following $\varphi^{-1}$. 
   
   Since for all square and hexagon moves that create all elements of $\mathbbm{\Lambda}[U(12),\dots,U(m-1,m);H_{A_m}]$ from $\bfU$ there exist corresponding rearrangements of pentagon relations which lead to maximal sequences that are assigned to the respective sequences of unlinkings, we know that such a map is surjective. 
   On the other hand, the arrangement of pentagon relations matches exactly the arrangement of unlinkings, following $\varphi^{-1}$. 
   In other words, any two distinct maximal sequences of pentagons will correspond to distinct sequences of unlinkings in $\mathbbm{\Lambda}$ -- the map is also injective, which completes the proof. 
\end{proof}



\subsection{Connectors as path polytopes}\label{sec:connectors and path polytopes}

In this section, we provide an alternative language to describe combinatorics of unlinkings that encode wall-crossing and make Theorem \ref{thm:unlinking and A_m wall-crossing} more intuitive.
\newline

It is well known that for 4d BPS quivers, the wall-crossing of the first kind is encoded in the combinatorial object called \emph{oriented exchange graph} \cite{keller2013quiver, garver2017maximal, garver2019lattice}, denoted $\overrightarrow{EG}(Q_{4\text{d}})$. 
Its two-dimensional faces are squares and pentagons, representing the commutation and pentagon relations for the quantum dilogarithm, respectively. 
If $Q_{4\text{d}}$ is a finite mutation type, $\overrightarrow{EG}(Q_{4\text{d}})$ is an oriented combinatorial polytope of finite volume.
In particular, $\overrightarrow{EG}(A_m)=K_{m+2}$, the $m$-dimensional associahedron \cite{keller2011cluster,Gaiotto:2009hg,padrol2023associahedra}. 
Note that different orientations of $A_m$ quiver yield $K_{m+2}$ oriented in a different, non-equivalent way. 

We can now use this to reformulate Theorem \ref{thm:unlinking and A_m wall-crossing} in terms of such graphs and make the connector theory more intuitive.
To this end, we need the notion of a path polytope. 
This point of view is especially attractive due to its fully combinatorial nature and broader generality -- in fact, it suggests a far-reaching generalization of Theorem \ref{thm:unlinking and A_m wall-crossing} for any 4d quiver, provided its oriented exchange graph is known. 

Recall that each edge of the diagram in Figure \ref{fig:A4_connector_linear} represents a pentagon relation, i.e., its is dual to a pentagon face in the associahedron $K_6$. 
It is not a coincidence that our diagrams appear like a~dual representation of those, but they apply to 3d quivers instead of 4d ones (with pentagon relations corresponding to unlinkings). 
In the simplest case of $A_2$ and the pentagon relation considered in Section~\ref{sec:pentagon to unlinking}, we can represent the unlinking $U(12)$ as a homotopy between two paths in the associahedron $K_4$: 
\[
\begin{tikzcd}[column sep=0.1]
& {a((bc)d)} && {a(b(cd))} \\
	{(a(bc))d} & {} && {} & {(ab)(cd)} \\
	&& {((ab)c)d}
	\arrow[from=1-2, to=1-4]
	\arrow[from=2-1, to=1-2]
	\arrow["{U(12)}"', Rightarrow, from=2-4, to=2-2]
	\arrow[from=2-5, to=1-4]
	\arrow[from=3-3, to=2-1]
	\arrow[from=3-3, to=2-5]
\end{tikzcd}.
\]
Analogous interpretations also hold for higher-dimensional polytopes arising from the joint connectors. 
We now provide a simpler description of the joint connector using the notion of a path polytope.\footnote{
The notion we define here is a particular instance of a monotonic path polytope, defined in \cite{billera1992fiber}.
}


\begin{dfn}
    Let $P$ be a polytope equipped with an orientation of its edges, such that there is a unique source and a sink vertex. 
    Define the path polytope $\bar{\Lambda}(P)$ of $P$ as a polytope whose vertices correspond to paths in $P$ from the source to the sink, and whose edges correspond to simple homotopies\footnote{Here, by simple homotopies, we mean those which act along one face of $P$.} between these paths. 
\end{dfn}

\begin{exmp*}
    Let $P$ be an $m$-dimensional oriented hypercube, with sink and source vertices separated by the longest diagonal. 
    Then, $\bar{\Lambda}(P) = \Pi_{m-1}$, an $m$-dimensional permutohedron. 
\end{exmp*}

\begin{dfn}\label{dfn:Lambda polytope}
    Let $\overrightarrow{EG}(Q_{4\text{d}})$ be the oriented exchange graph of $Q_{4\text{d}}$. 
    Define $\Lambda(\overrightarrow{EG}(Q_{4\text{d}}))$ as the polytope obtained from $\bar{\Lambda}(\overrightarrow{EG}(Q_{4\text{d}}))$ by identifying those vertices which represent paths related by a square move on the exchange graph. 
\end{dfn}

Now Theorem \ref{thm:unlinking and A_m wall-crossing} implies the following. 
\begin{cor}\label{cor:path polytope}
    For $A_m$ quivers, the joint connector class $\mathbbm{\Lambda}[U(12),\dots,U(m-1,m);H_{A_m}]$ is identified with the path polytope $\Lambda(\overrightarrow{EG}(A_m))$.
\end{cor}

In other words, given a 4d BPS quiver $Q_{4\text{d}}$, we can use its exchange graph to define the path polytope which, according to Theorem \ref{thm:unlinking and A_m wall-crossing}, is equivalent to the joint connector and thus should contain all information about the symmetrization map (which we will define in full generality below). 
For example, the case of $A_3$ is shown in Figure \ref{fig:Associahedra}. 
Here, the exchange graph is the associahedron $K_5$; we can view the connector (hexagon) given by either of the relations \eqref{eq:second hexagon} and \eqref{eq:first hexagon} as the path polytope $\Lambda(K_5)$ of the three-dimensional associahedron $K_5$. 
\begin{figure}[h!]
    \centering
    \includegraphics[width=1\linewidth]{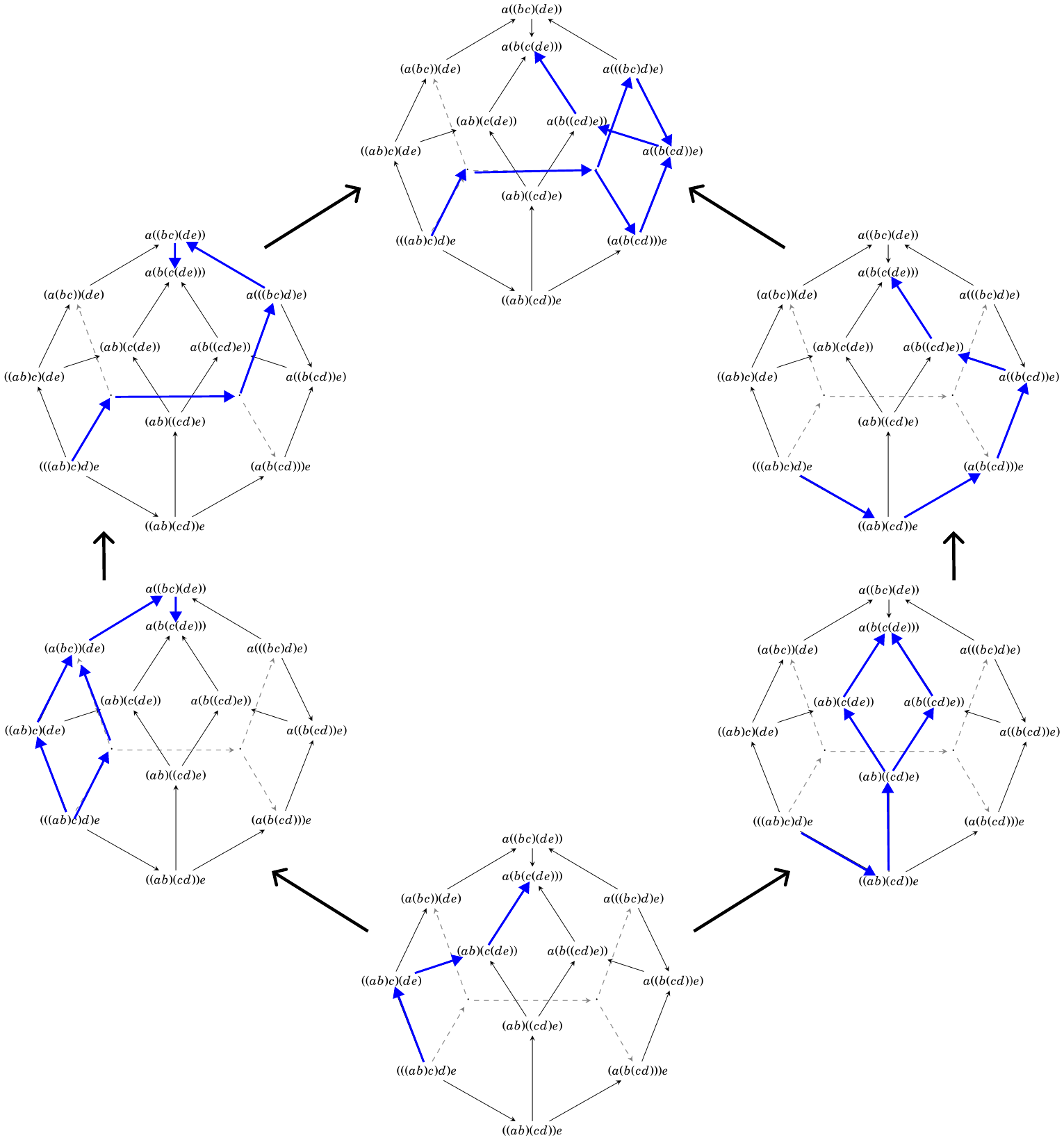}
    \caption{
    The path polytope $\Lambda(K_5)$ is a hexagon. 
    Its vertices correspond to paths (shown in \textcolor{blue}{blue}) on the oriented exchange graph $K_5$ for the $Q_{4\text{d}} = A_3$ quiver. 
    Its edges correspond to simple homotopies of a path in $K_5$ along some pentagonal face, while all homotopies along quadrilateral faces are modded out. 
    As a result, we have a hexagon equivalent to either of the unlinking hexagons $\mathbbm{\Lambda}(U(12),U(23);H)$ \eqref{eq:second hexagon}--\eqref{eq:first hexagon}.
    Oriented arrows of the path polytope connect chambers with different numbers of stable BPS states in 4d: at the bottom is the minimal chamber with 3 BPS states (a unique sequence with 3 blue arrows), while at the top is the maximal chamber with 6 BPS states (sequences with 6 blue arrows). 
    }
    \label{fig:Associahedra}
\end{figure}
It would be interesting to investigate these polytopes further, and, if possible, relate them to higher-categorical analogues of associahedra \cite{backman2024higher}.

\subsection{Definition of symmetrization map}\label{sec:sym-map-def}

Having formulated the bijection between the wall-crossing of the first kind and unlinking, we are now ready to provide a proper definition of the symmetrization map $\mathfrak{S}$. 
\newline

The domain and codomain of the symmetrization map are given by 4d quivers with stability data and 3d symmetric quivers, respectively:
\[
(Q_{4\text{d}},\text{stab.\,data}) \stackrel{\mathfrak{S}}{\longmapsto}Q\,.
\]
The idea for the definition is to start from the minimal chamber discussed in the previous sections and then use Theorem \ref{thm:unlinking and A_m wall-crossing} to define symmetric quivers corresponding to any choice of stability data for a fixed 4d quiver.

\begin{dfn}\label{dfn:symmetrization map}
    The symmetrization map for a Dynkin quiver $A_m$ and stability data $u$ is given by
    \begin{equation}\label{eq:sym-map-eq}
        \Qsym{A_m}{u} := \boldsymbol{U}\vert_{u} Q \,,
    \end{equation}
    where $Q$ is the symmetrization of $A_m$ and $\boldsymbol{U}\vert_u$ is a sequence of unlinkings in joint connector class $\mathbbm{\Lambda}[U(12),\dots,U(m-1,m);H_{A_m}]$ assigned by canonical bijection to the sequence of pentagon relations which transform the minimal chamber of $A_m$ into the chamber described by $u$.
\end{dfn}

For example, if $u$ corresponds to the minimal chamber with $m$ BPS states, then the joint connector consists of the identity operator: $\mathbbm{\Lambda}=\{\operatorname{Id}\}$ and we get:
    \begin{equation}
        \Qsym{A_m}{\text{min}} = Q\,,
    \end{equation}
which agrees with the previous considerations in Section \ref{sec:3d4d}.

\begin{cor}\label{cor:symmetrization map isomorphism}
    Theorem \ref{thm:unlinking and A_m wall-crossing} ensures that the structure of wall-crossing of the first kind for the Dynkin quiver $A_m$ (excluding the chambers described by quivers with superpotential) is isomorphic to the structure of unlinkings encoded in the joint connector class $\mathbbm{\Lambda}[U(12),\dots,U(m-1,m);H_{A_m}]$. It implies that for a~fixed choice of stability data different unlinking operators produce symmetric quivers related by a trivial relabelling of the nodes, which means that the symmetrization map $\mathfrak{S}$ is well defined.
\end{cor}

Generalization of Theorem \ref{thm:unlinking and A_m wall-crossing} (and following isomorphism between wall-crossing and unlinking) to arbitrary 4d quivers $Q_{4\text{d}}$ of finite mutation type, including $ADE$ quivers as well as quivers coming from general triangulated surfaces, is an interesting direction for future research. Without it, one can still use the canonical bijection to assign sequences of unlinkings to sequences of wall-crossing relations, but in case of different paths leading to the same chamber the question whether the resulting symmetric quivers are the same remains open. In case of the negative answer, one can still consider the following extension of the Definition \ref{dfn:symmetrization map}:
    \begin{equation}\label{eq:sym-map-gen}
        \Qsym{Q_{4\text{d}}}{u} := \{ \boldsymbol{U}\vert_{u} Q \}_{\boldsymbol{U}\in \mathbbm{\Lambda}[U(12),\dots,U(m-1,m);H_{Q_{4\text{d}}}]},
    \end{equation}
where we assume that $|Q_{4\text{d}}|=m$.
    
Note that the symmetrization map $\mathfrak{S}$ (as well as its extension) can also be understood in terms of (a~suitable composition of) homotopies between paths in the oriented exchange graph of $Q_{4\text{d}}$ -- see Corollary~\ref{cor:path polytope}. 
In other words, we can view the joint connector as a path polytope $\Lambda(Q_{4\text{d}})$, and the symmetrization map can be defined via paths on $\Lambda(Q_{4\text{d}})$. 
For $A_m$ quivers, where oriented exchange graphs are well-known associahedra, this interpretation is especially effective -- see Figure \ref{fig:Associahedra}. 
\newline

In the remainder of this section, we investigate several examples that illustrate the isomorphism between wall-crossing and unlinking given by the map $\mathfrak{S}$.

\subsection{\texorpdfstring{${A}_2$}{A2} quiver}

Let us start with the simplest non-trivial example which is $Q_{4\text{d}}=A_2=\bullet\rightarrow\bullet$. 
There are two choices of stability data $L_1$ and $L_2$, corresponding to the minimal and maximal BPS chambers of $A_2$ Argyres-Douglas theory, respectively (Table \ref{tab:a2chambers}). 
\begin{table}[h!]
\center
\begin{tabular}{||c|c|c||}
\hline
Chamber & Stability data& \# of BPS States \\
\hline
\hline
$L_{1}$&$\hat{\alpha}_{2} >\hat{\alpha}_{1}$ &  2 \\
\hline
$L_{2}$&$\hat{\alpha}_{1} >\hat{\alpha}_{12} > \hat{\alpha}_{2}$ &  3 \\
\hline
\end{tabular}
\caption{The BPS chambers of the $A_2$ Argyres-Douglas theory, characterized by the partial ordering of the phases of BPS states. 
We use the notation $\hat{\alpha}_i = \arg \alpha_i$ and $\hat{\alpha}_{ij} = \arg(\alpha_i+\alpha_j)$. 
}
\label{tab:a2chambers}
\end{table}
Here, the joint connector $\mathbbm{\Lambda}(U(12);H_{A_2}(u)=\emptyset)$ is a one-element set $\{U(12)\}$, so only one unlinking is required. 
By definition, symmetrization in the minimal chamber agrees with the one considered earlier in Section \ref{sec:3d4d}, and we obtain: 
\begin{equation}
\begin{aligned}
        \Qsym{\bullet \rightarrow \bullet}{L_1} = &\ \operatorname{Id} Q & = \quad &
        \begin{tikzcd}
	\bullet & \bullet
	\arrow[curve={height=-6pt}, from=1-1, to=1-2]
	\arrow[curve={height=-6pt}, from=1-2, to=1-1]
\end{tikzcd}
        \\
        \Qsym{\bullet \rightarrow \bullet}{L_2} = &\ U(12) Q & = \quad & 
\begin{tikzcd}
	\bullet & \bullet & \bullet
	\arrow[from=1-2, to=1-2, loop, in=55, out=125, distance=10mm]
\end{tikzcd}
\end{aligned}.
\end{equation}
As we shall see later, in more complicated cases the structure of symmetric quivers in higher chambers is more involved and is given by a rather intricate pattern of unlinkings. 

\subsection{\texorpdfstring{${A}_3$}{A3} quiver}\label{sec:A3}

Our next example is $Q_{4\text{d}}=A_3$. 
There are three different orientations of the $A_3$ quiver -- linear, inward and outward -- which we denote as $A_3^L$, $A_3^I$, and $A_3^O$, respectively: 
\begin{equation}\label{eq:A_3 orientations}
    \begin{tikzcd}[sep=small,row sep=0.25ex]
	{A_3^L:} & {\bullet} & {\bullet} & {\bullet} \\
	{A_3^I:} & {\bullet} & {\bullet} & {\bullet} \\
	{A_3^O:} & {\bullet} & {\bullet} & {\bullet}
	\arrow[from=1-2, to=1-3]
	\arrow[from=1-3, to=1-4]
	\arrow[from=2-2, to=2-3]
	\arrow[from=2-4, to=2-3]
	\arrow[from=3-3, to=3-2]
	\arrow[from=3-3, to=3-4]
\end{tikzcd}
\end{equation}
These three quivers can be associated with the minimal chambers of the $A_3$ Argyres-Douglas theory, denoted as $L_1$, $I_1$, and $O_1$, respectively. 
The general structure of BPS chambers and their stability data is shown in Table \ref{tab:a3chambers} (see \cite[Section 2.3.2]{Alim:2011kw} for the detailed discussion). 

\begin{table}
\center
\begin{tabular}{||c|c|c|c||}
\hline
Chamber & Stability data & \# of BPS States \\
\hline
\hline
$L_{1}$&$\hat{\alpha}_{3}>\hat{\alpha}_{2} >\hat{\alpha}_{1}$ &  3 \\
\hline
$L_{2}$&$\hat{\alpha}_{2}$ smallest, and $\hat{\alpha}_{1}, \hat{\alpha}_{3} >\hat{\alpha}_{12}$ &  4 \\
\hline
$L_{3}$&$\hat{\alpha}_{2}$ largest, and $ \hat{\alpha}_{23} >\hat{\alpha}_{1}, \hat{\alpha}_{3}$  &  4 \\
\hline
$L_{4}$&$\hat{\alpha}_{1}>\hat{\alpha}_{12} >\hat{\alpha}_{3}>\hat{\alpha}_{2}$ &  5 \\
\hline
$L_{5}$&$\hat{\alpha}_{2}>\hat{\alpha}_{1} >\hat{\alpha}_{23}>\hat{\alpha}_{3}$ &  5 \\
\hline
$L_{6}$&$\hat{\alpha}_{1}>\hat{\alpha}_{2} >\hat{\alpha}_{3}$ &  6 \\
\hline
$O_{1}$&$\hat{\alpha}_{2}$ smallest & 3\\
\hline
$O_{2}$&$\hat{\alpha}_{2}$ intermediate & 4 \\
\hline
$O_{3}$&$\hat{\alpha}_{2}$ largest, and $\hat{\alpha}_{12}<\hat{\alpha}_{3}$ or $\hat{\alpha}_{23}<\hat{\alpha}_{1}$& 5 \\
\hline
$O_{4}$&$\hat{\alpha}_{2}$ largest, and $\hat{\alpha}_{12}>\hat{\alpha}_{3}$ and $\hat{\alpha}_{23}>\hat{\alpha}_{1}$ & 6 \\
\hline
$I_{1}$&$\hat{\alpha}_{2}$ largest & 3\\
\hline
$I_{2}$&$\hat{\alpha}_{2}$ intermediate & 4 \\
\hline
$I_{3}$&$\hat{\alpha}_{2}$ smallest, and $\hat{\alpha}_{3}<\hat{\alpha}_{12}$ or $\hat{\alpha}_{1}<\hat{\alpha}_{23}$& 5 \\
\hline
$I_{4}$&$\hat{\alpha}_{2}$ smallest, and $\hat{\alpha}_{3}>\hat{\alpha}_{12}$ and $\hat{\alpha}_{1}>\hat{\alpha}_{23}$ &  6 \\
\hline
$C_{1}$&not cyclically ordered e.g. $\hat{\alpha}_{2}>\hat{\alpha}_{1} >\hat{\alpha}_{3}$ &  4 \\
\hline
$C_{2}$&cyclically ordered e.g.  $\hat{\alpha}_{1}>\hat{\alpha}_{2} >\hat{\alpha}_{3}$ &  5 \\
\hline
\end{tabular}
\caption{The BPS chambers of the $A_3$ Argyres-Douglas theory, characterised by the partial ordering of phases of BPS states.
}
\label{tab:a3chambers}
\end{table}

Note that the wall-crossing formula (\ref{eq:Reineke_WCF}) relates a minimal 4d chamber with 3 BPS states labeled by simple roots $\alpha_1,\alpha_2,\alpha_3$, to the corresponding maximal chamber with 6 BPS states labeled by all positive roots 
\[
\alpha_1,\,
\alpha_2,\,
\alpha_3,\,
\alpha_{1}+\alpha_2,\,
\alpha_2+\alpha_3,\,
\alpha_1+\alpha_2+\alpha_3\,.
\]
For all three orientations, we get the same symmetric quiver via the symmetrization map (\ref{eq:sym-map-eq}), consistent with Section \ref{sec:3d4d}:
\be
	\Qsym{A_3^L}{L_1} = \Qsym{A_3^I}{I_1} = \Qsym{A_3^O}{O_1} = \begin{tikzcd}
	\bullet & \bullet & \bullet
	\arrow[curve={height=-6pt}, from=1-1, to=1-2]
	\arrow[curve={height=-6pt}, from=1-2, to=1-1]
	\arrow[curve={height=-6pt}, from=1-2, to=1-3]
	\arrow[curve={height=-6pt}, from=1-3, to=1-2]
\end{tikzcd}.
\ee
Let $L_i,I_i,O_i$ be quivers with stability data given in 
Table \ref{tab:a3chambers}, which are obtained by applying certain pentagon relations on the initial operators $L_1, I_1, O_1$. 
Note that in our consideration we exclude chambers~$C_i$, since they correspond to a 4d quiver with a superpotential:
\[
\begin{tikzcd}
	\bullet & \bullet & \bullet
	\arrow[from=1-1, to=1-2]
	\arrow[from=1-2, to=1-3]
	\arrow[curve={height=-12pt}, from=1-3, to=1-1]
\end{tikzcd},
\]
which is not an $A_3$ quiver, but is related to it by a mutation operation.
As a result, the corresponding wall-crossing operators cannot be obtained from the initial operators by applying pentagon relations, and our analysis simply does not apply.\footnote{It would be interesting to understand what symmetric quiver corresponds to a 4d quiver with a superpotential.}

In order to relate $L_i, I_i, O_i$ to 3d symmetric quivers and unlinking, we need to apply Theorem \ref{thm:unlinking and A_m wall-crossing} to every choice of orientation. 
That is, we need to compute the corresponding joint connector which follows the orientation of arrows in the 4d quiver. 
In each case, we need only to connect $U(12)$ with $U(23)$, hence to apply one of the hexagon moves (\ref{eq:unlinking hexagons}) -- most importantly, the type of hexagon will depend on the orientation of the 4d quiver. 

Let us start with $L$, where we can fix the labelling by $\alpha_1\rightarrow\alpha_2\rightarrow\alpha_3$. 
Following Definition \ref{dfn:WCF for a symmetric quiver}, we see that the hexagon type is defined by arrows $\alpha_1\to \alpha_2$ and $\alpha_2\to \alpha_3$. 
Since they are co-oriented, we set $H_{A_3}=\{(^{12}_{23}),\text{a}\}$, i.e., we use hexagon (a). 
In the other cases, the respective arrows have opposite orientation, so for both $I$ and $O$ we assign $H_{A_3} = \{(^{12}_{23}),\text{b}\}$. 
Now the joint connector class in all cases is the simply application of the corresponding hexagon move, following the choice of $H_{A_3}$. 
The combined picture of the joint connector classes together with 4d BPS chambers for all orientations is shown in Figure \ref{fig:A_3 all joint connectors}. 

\begin{figure}[h!]
    \centering
    \[\begin{tikzcd}[column sep={-0.8cm},row sep={1cm}]
	& \begin{array}{c} \substack{\mathlarger{\Qsym{A_3^I}{I_4}} \\ \mathlarger{=\Qsym{A_3^O}{O_4}}} \\
    \includegraphics[scale=0.35]{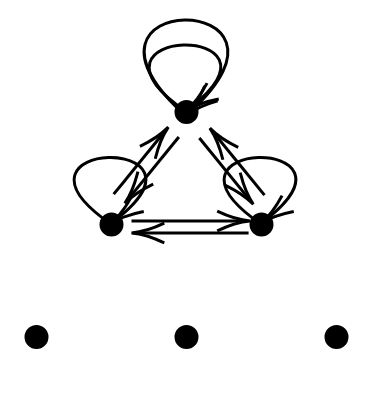}\end{array} &&&& \begin{array}{c}
    {\substack{\mathlarger{\Qsym{A_3^L}{L_6}}}} \\
    \includegraphics[scale=0.35]{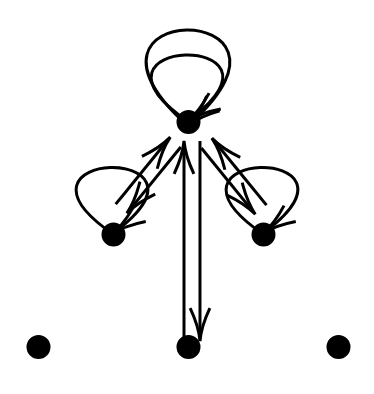}
    \end{array}
    \\
	\begin{array}{c} \substack{\mathlarger{\Qsym{A_3^I}{I_3}}\\ \mathlarger{=\Qsym{A_3^O}{O_3}}} \\
    \includegraphics[scale=0.35]{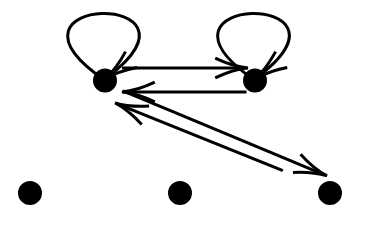}
    \end{array} && \begin{array}{c} \substack{\mathlarger{\Qsym{A_3^I}{I_3}}\\ \mathlarger{=\Qsym{A_3^O}{O_3}}} \\
    \includegraphics[scale=0.35]{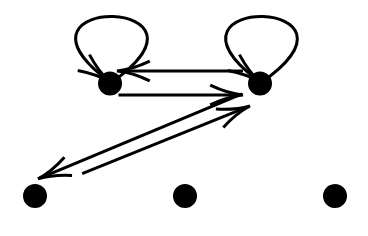}\end{array} && \begin{array}{c}{\substack{\mathlarger{\Qsym{A_3^L}{L_4}}}} \\
    \includegraphics[scale=0.35]{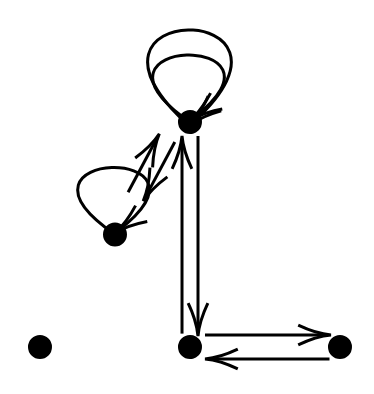}
    \end{array}
    &&
    \begin{array}{c}
    \substack{\mathlarger{\Qsym{A_3^L}{L_5}}} \\
    \includegraphics[scale=0.35]{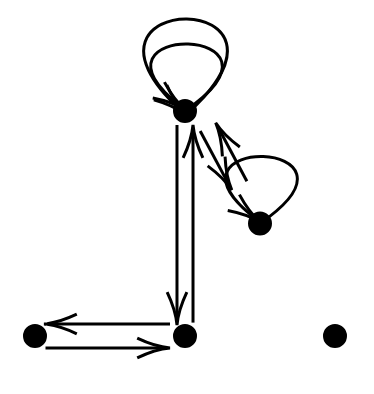}
    \end{array} \\
	& \begin{array}{c} \substack{\mathlarger{\Qsym{A_3^L}{L_2}}\\ \mathlarger{=\Qsym{A_3^I}{I_2}}\\ \mathlarger{=\Qsym{A_3^O}{O_2}}} \\
    \includegraphics[scale=0.35]{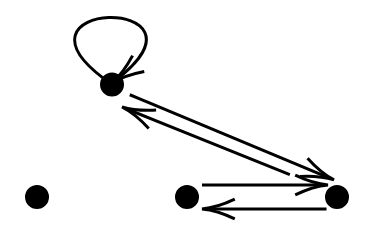}
    \end{array} &&&& \begin{array}{c} \substack{\mathlarger{\Qsym{A_3^L}{L_3}}\\ \mathlarger{=\Qsym{A_3^I}{I_2}}\\ \mathlarger{=\Qsym{A_3^O}{O_2}}} \\
    \includegraphics[scale=0.35]{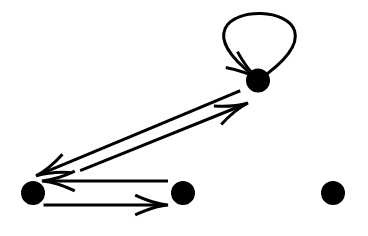}\end{array} \\
	&&& \begin{array}{c} \substack{\mathlarger{\Qsym{A_3^L}{L_1}}\\ \mathlarger{=\Qsym{A_3^I}{I_1}}\\ \mathlarger{=\Qsym{A_3^O}{O_1}}} \\
    \includegraphics[scale=0.35]{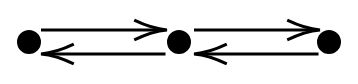}
    \end{array}
    \arrow["{(12)3}"{description}, from=2-1, to=1-2]
	\arrow["{1(23)}"{description}, from=2-3, to=1-2]
	\arrow["23"{description}, from=2-5, to=1-6]
	\arrow["12"{description}, from=2-7, to=1-6]
	\arrow["23"{description}, from=3-2, to=2-1]
	\arrow["{(12)3}"{description}, from=3-2, to=2-5]
	\arrow["12"{description}, from=3-6, to=2-3]
	\arrow["{1(23)}"{description}, from=3-6, to=2-7]
	\arrow["12"{description}, from=4-4, to=3-2]
	\arrow["23"{description}, from=4-4, to=3-6]
    \end{tikzcd}\]
    \caption{The complete structure of the symmetrization map $\mathfrak{S}$ for different orientations of ${A}_3$ quiver and their stability data; note how different stability data are sometimes mapped to the same 3d quiver. 
    Unlinking $U(ij)$ is represented by an oriented edge with subscript $ij$. 
    Each of the two hexagons on this diagram is equivalent to the path polytope $\Lambda(K_5)$ shown in Figure \ref{fig:Associahedra}. 
    That is, a choice of orientation of $A_3$ fixes the orientation (and thus the path structure) of $\overrightarrow{EG}(A_3)=K_5$, whereas unlinkings $U(ij)$ correspond to homotopies between those paths in $K_5$. 
    }
    \label{fig:A_3 all joint connectors}
\end{figure}

Let us illustrate in greater detail one path in Figure \ref{fig:A_3 all joint connectors} which connects the stability chambers corresponding to $A_3^L$: 
\begin{equation}
    \Qsym{A_3^L}{L_1} \xrightarrow{12} \Qsym{A_3^L}{L_2} \xrightarrow{(12)3} \Qsym{A_3^L}{L_4} \xrightarrow{23} \Qsym{A_3^L}{L_6}.
\end{equation}
We will trace the unlinking and pentagon relations simultaneously for the operators corresponding to the 4d quiver. 
We start with the minimal chamber $L_1$, where the corresponding symmetric quiver $Q=\Qsym{A_3^L}{L_1}$ is simply the symmetrization of $A_3$ (note that this also corresponds to the shortest path in $K_5$, Figure \ref{fig:Associahedra}):
\[
\begin{tikzcd}[column sep={0.35cm},row sep={1cm}]
	{\boldsymbol{\Psi}_{L_1}=\Psi(-X_{\alpha_{1}})\Psi(-X_{\alpha_{2}})\Psi(-X_{\alpha_{3}})} & \begin{array}{c} \substack{\bullet \\ 1} \end{array} && \begin{array}{c} \substack{\bullet \\ 2} \end{array} && \begin{array}{c} \substack{\bullet \\ 3} \end{array}
	\arrow[shift left, from=1-2, to=1-4]
	\arrow[from=1-4, to=1-2]
	\arrow[shift left, from=1-4, to=1-6]
	\arrow[from=1-6, to=1-4]
\end{tikzcd}.
\]
Applying unlinking $U(12)$ to $Q$ and the pentagon relation to $\Psi(-X_{\alpha_{1}})\Psi(-X_{\alpha_{2}})$ gives 
\[
\begin{tikzcd}[column sep={0.35cm},row sep={0.5cm}]
	{{\boldsymbol{\Psi}_{L_2}=\Psi(-X_{\alpha_{2}})\Psi(-X_{\alpha_{1}+\alpha_2})\Psi(-X_{\alpha_{1}})\Psi(-X_{\alpha_{3}})}} && \begin{array}{c} \substack{\bullet \\ 12} \end{array} \\
	{{\hphantom{\boldsymbol{\Psi}_{L_2}}\equiv\Psi(-X_{\alpha_{2}})\Psi(-X_{\alpha_{1}+\alpha_2})\Psi(-X_{\alpha_{3}})\Psi(-X_{\alpha_{1}})}} &  \begin{array}{c} \substack{\bullet \\ 1} \end{array} && \begin{array}{c} \substack{\bullet \\ 2} \end{array} && \begin{array}{c} \substack{\bullet \\ 3} \end{array}
	\arrow[from=1-3, to=1-3, loop, in=55, out=125, distance=10mm]
	\arrow[from=1-3, to=2-6]
	\arrow[shift left, from=2-4, to=2-6]
	\arrow[shift right, from=2-6, to=1-3]
	\arrow[from=2-6, to=2-4]
\end{tikzcd}.
\]
(Here we could swap $\Psi(-X_{\alpha_1})$ with $\Psi(-X_{\alpha_3})$ since they commute -- the nodes 1 and 3 are disconnected in the corresponding symmetric quiver.) 
Comparing with Figure \ref{fig:Associahedra}, this corresponds to two paths in $K_5$ of length four which differ along a quadrilateral face.
To reach the chamber with five BPS states, we need to apply $U((12)3)$ to $U(12)Q$ and the pentagon relation to $\Psi(-X_{\alpha_{1}+\alpha_2})\Psi(-X_{\alpha_{3}})$:
\[
\begin{tikzcd}[column sep={0.25cm},row sep={0.25cm}]
	&&& \begin{array}{c} \substack{\bullet \\ (12)3} \end{array} \\
	{{\boldsymbol{\Psi}_{L_4}=\Psi(-X_{\alpha_{2}})\Psi(-X_{\alpha_{3}})\Psi(-X_{\alpha_{1}+\alpha_2+\alpha_3})\Psi(-X_{\alpha_{1}+\alpha_2})\Psi(-X_{\alpha_{1}})}} && \begin{array}{c} \substack{\bullet \\ 12} \end{array} \\
	& \begin{array}{c} \substack{\bullet \\ 1} \end{array} && \begin{array}{c} \substack{\bullet \\ 2} \end{array} && \begin{array}{c} \substack{\bullet \\ 3} \end{array}
	\arrow[from=1-4, to=1-4, loop, in=55, out=125, distance=10mm]
	\arrow[from=1-4, to=1-4, loop, in=50, out=130, distance=15mm]
	\arrow[shift right, from=1-4, to=2-3]
	\arrow[from=1-4, to=3-4]
	\arrow[from=2-3, to=1-4]
	\arrow[from=2-3, to=2-3, loop, in=55, out=125, distance=10mm]
	\arrow[shift right, from=3-4, to=1-4]
	\arrow[shift left, from=3-4, to=3-6]
	\arrow[from=3-6, to=3-4]
\end{tikzcd}.
\]
Lastly, we apply $U(23)$ to $U((12)3)U(12)Q$ and the pentagon relation $\Psi(-X_{\alpha_{2}})\Psi(-X_{\alpha_{3}})$ to reach the maximal chamber with six BPS states, which also corresponds to the longest path in $K_5$:
\[
\begin{tikzcd}[column sep={0.25cm},row sep={0.25cm}]
	&&& \begin{array}{c} \substack{\bullet \\ (12)3} \end{array} \\
	\begin{array}{c} \begin{array}{c} \boldsymbol{\Psi}_{L_6}=\Psi(-X_{\alpha_{3}})\Psi(-X_{\alpha_{1}+\alpha_2})\Psi(-X_{\alpha_{2}})\cdot\\\hphantom{\boldsymbol{\Psi}_{L_6}=\qquad }\Psi(-X_{\alpha_{1}+\alpha_2+\alpha_3})\Psi(-X_{\alpha_{1}+\alpha_2})\Psi(-X_{\alpha_{1}}) \end{array} \end{array} && \begin{array}{c} \substack{\bullet \\ 12} \end{array} && \begin{array}{c} \substack{\bullet \\ 23} \end{array} \\
	& \begin{array}{c} \substack{\bullet \\ 1} \end{array} && \begin{array}{c} \substack{\bullet \\ 2} \end{array} && \begin{array}{c} \substack{\bullet \\ 3} \end{array}
	\arrow[from=1-4, to=1-4, loop, in=55, out=125, distance=10mm]
	\arrow[from=1-4, to=1-4, loop, in=50, out=130, distance=15mm]
	\arrow[shift right, from=1-4, to=2-3]
	\arrow[shift left, from=1-4, to=2-5]
	\arrow[from=1-4, to=3-4]
	\arrow[from=2-3, to=1-4]
	\arrow[from=2-3, to=2-3, loop, in=55, out=125, distance=10mm]
	\arrow[from=2-5, to=1-4]
	\arrow[from=2-5, to=2-5, loop, in=55, out=125, distance=10mm]
	\arrow[shift right, from=3-4, to=1-4]
\end{tikzcd}. 
\]
Note that symmetric quivers corresponding to higher chambers have a non-trivial structure due to unlinking. 
Their quiver generating series are encoded by quantum dilogarithm operators via the 3d-4d homomorphism and the application of the normal ordering. 
The two symmetric quivers corresponding to maximal chambers -- one to $I_4$ and $O_4$, the other to $L_6$ -- can be unlinked once to a single universal quiver, which is guaranteed by the Connector Theorem from \cite{KLNS}.
It would be interesting to see if it has a~4d~interpretation. 
Note also that quiver mutation is understood in 3d as a structural change in the joint connector class $\mathbbm{\Lambda}$. 
For example, mutating the first node for $L$ gives the $O$ orientation, and this alters the choice of hexagon and reorders the unlinking operations correspondingly. 
Some mutations, however, do not change the unlinking operator -- for example, the mutation between the $I$ and $O$ chambers. 


\subsection{\texorpdfstring{${A}_4$}{A4} quiver}

Here we consider one more example of the symmetrization map \eqref{eq:sym-map-eq}, which corresponds to $Q_{4\text{d}}=A_4$. 
We have the following possibilities for the orientation, which we denote as linear, linear-inward, outward-linear, and inward-outward (all other orientations can be obtained by a trivial automorphism, given by relabelling of the nodes): 
\begin{equation}\label{eq:A_4 orientations}
\begin{aligned}
    LL:\qquad &\ \alpha_1 \longrightarrow \alpha_2 \longrightarrow \alpha_3 \longrightarrow \alpha_4 \\
    LI:\qquad &\ \alpha_1 \longrightarrow \alpha_2 \longrightarrow \alpha_3 \longleftarrow \alpha_4 \\
    OL:\qquad &\ \alpha_1 \longleftarrow \alpha_2 \longrightarrow \alpha_3 \longrightarrow \alpha_4 \\
    IO:\qquad &\ \alpha_1 \longrightarrow \alpha_2 \longleftarrow \alpha_3 \longrightarrow \alpha_4
\end{aligned}
\end{equation}
For each of them, the wall-crossing formula (\ref{eq:Reineke_WCF}) relates a minimal 4d chamber with 4 BPS states labeled by simple roots $\alpha_1,\alpha_2,\alpha_3,\alpha_4$ to the corresponding maximal chamber with 10 BPS states labeled by all positive roots. 
Denote the minimal chambers with 4 BPS states as $LL_1,LI_1,OL_1,IO_1$.
The symmetrization map applied to all of them produces the same symmetric quiver, as expected:
\be
	\Qsym{A_4^{LL}}{LL_1} = \Qsym{A_4^{LI}}{LI_1} = \Qsym{A_4^{OL}}{OL_1} = \Qsym{A_4^{IO}}{IO_1} = \begin{tikzcd}
	\bullet & \bullet & \bullet & \bullet
	\arrow[curve={height=-6pt}, from=1-1, to=1-2]
	\arrow[curve={height=-6pt}, from=1-2, to=1-1]
	\arrow[curve={height=-6pt}, from=1-2, to=1-3]
	\arrow[curve={height=-6pt}, from=1-3, to=1-2]
    \arrow[curve={height=-6pt}, from=1-3, to=1-4]
    \arrow[curve={height=-6pt}, from=1-4, to=1-3]
\end{tikzcd}.
\ee

Theorem \ref{thm:unlinking and A_m wall-crossing} assigns the joint connector class to every orientation of arrows in the 4d quiver. 
In each case, we need to connect $U(12)$ with $U(23)$ as well as $U(34)$, so the number of required operations (as well as the number of hexagons) is now larger than in previous examples. 
The type of each hexagon depends on the orientation of the corresponding pair of arrows in the 4d quiver (\ref{eq:A_4 orientations}). 
Let us illustrate this with the $LL$ orientation. 
Following Definition \ref{dfn:WCF for a symmetric quiver}, we see that the hexagons are defined by pairs of arrows: $\alpha_1\to \alpha_2$ and $\alpha_2\to \alpha_3$, $\alpha_2\to \alpha_3$ and $\alpha_3\to \alpha_4$, etc. 
Since all pairs are co-oriented, we use only hexagon (a). 
For other orientations, mixed types of hexagons can appear. 
The resulting joint connectors $\mathbbm{\Lambda}(U(12),U(23),U(34);\, H_{A_4})$ are shown in Figure \ref{fig:A_4 joint connector_linear}.\footnote{Unlike in Figure \ref{fig:A_3 all joint connectors}, we decided that before combining them into a single picture we present them separately for better visual clarity.}

\begin{figure}[h!]
    \centering
    \includegraphics[width=0.7\linewidth]{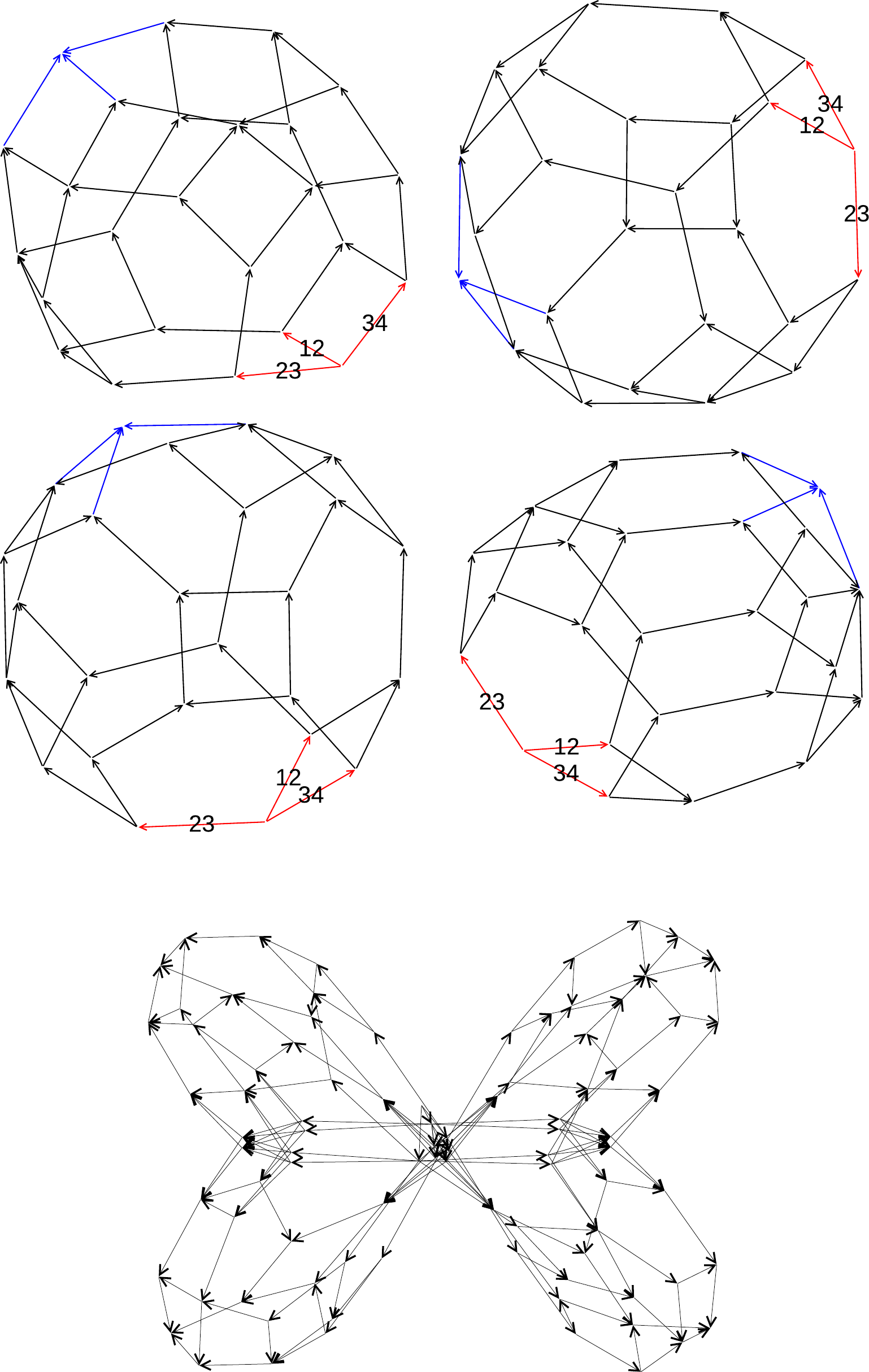}
    \caption{Top: the four joint connectors $\mathbbm{\Lambda}(U(12),U(23),U(34);\, H_{A_4})$ for different orientations of ${A}_4$ quiver:
    $\begin{pmatrix}
    LL & LI \\
    OL & IO
    \end{pmatrix}
    $. 
    Initial symmetric quivers/minimal chambers are represented by a source with three initial unlinkings shown in \textcolor{red}{red}, while the final quivers/maximal chambers are given by a sink with three final unlinkings shown in \textcolor{blue}{blue}. 
    Bottom: the bouquet of the four above polytopes, glued along the common vertex (initial symmetric quiver), as well as along $U(12), U(23), U(34)$, and some other edges.
    This graph gives the complete structure of the symmetrization map for $A_4$.}
    \label{fig:A_4 joint connector_linear}
\end{figure}

Note that in this case, the joint connector forms a 3d polytope, which partially resembles the permutohedron.
We stress again that these polytopes capture the entire structure of stability chambers of $A_4$ Argyres-Douglas theories, excluding the chambers described by a quiver with a superpotential. 
Following Section \ref{sec:connectors and path polytopes}, it is suggestive that such a polytope is a particular example of a three-dimensional 2-associahedron \cite{backman2024higher}. 
The cases $LL, LI, OL$ yield the same polytope, but oriented in different, non-equivalent ways. 
The case $IO$ gives another polytope. 
One can then say that quiver mutations in 4d act on the set of symmetric quivers by changing the orientations of the joint connector in 3d, affecting the shape of the corresponding connector polytope. 
Moreover, one can also combine all the above pictures, similar to the case of a pair of hexagons glued together in the $A_3$ case (Figure \ref{fig:A_3 all joint connectors}). 
The resulting structure is shown at the bottom of Figure \ref{fig:A_4 joint connector_linear}, and it captures all symmetric quivers for all chambers of $A_4$ Argyres-Douglas theories, apart from the chambers with superpotential or, isomorphically, joint connector classes for all choices of hexagons in $H_{A_4}$. Similarly to the $A_3$ case, symmetric quivers corresponding to the maximal chambers in all orientations can be unlinked to the same universal quiver (it is guaranteed by the Connector Theorem from \cite{KLNS}), but this time more than one unlinking is required.


\subsection{\texorpdfstring{$D_4$}{D4} quiver}

We finish our case studies by going beyond the $A_m$ case and considering the most symmetric Dynkin diagram, which is $D_4$. 
Its positive roots are
\begin{gather}
\alpha_1,\ \dots,\ \alpha_4, \alpha_1+\alpha_2,\ \alpha_2+\alpha_3,\ \alpha_2+\alpha_4, \alpha_1+\alpha_2+\alpha_3,\ \alpha_1+\alpha_2+\alpha_4,\ \\ \alpha_2+\alpha_3+\alpha_4, \alpha_1+\alpha_2+\alpha_3+\alpha_4,\ \alpha_1+2\alpha_2+\alpha_3+\alpha_4. \nonumber
\end{gather}
We fix the orientation and labelling as follows:
\begin{equation}\label{eq:D4_quiver}
\begin{tikzcd}
	& {\stackrel{\alpha_4}{\bullet}} \\
	{\stackrel{\alpha_1}{\bullet}} & {\stackrel{\alpha_2}{\bullet}} & {\stackrel{\alpha_3}{\bullet}.}
	\arrow[from=2-2, to=1-2]
	\arrow[from=2-2, to=2-1]
	\arrow[from=2-2, to=2-3]
\end{tikzcd}
\end{equation}
Although Theorem \ref{thm:unlinking and A_m wall-crossing} does not apply in this case, we can still construct the joint connector whose hexagons follow the pattern of arrows in the above 4d quiver. 
Namely, we apply the extension of Definition~\ref{dfn:WCF for a symmetric quiver} to every $A_3$-type subquiver given by triples $(\alpha_1,\alpha_2,\alpha_3)$, $(\alpha_4,\alpha_2,\alpha_3)$ and $(\alpha_1,\alpha_2,\alpha_4)$. 
In particular, one path in this connector is given by the following sequence of unlinkings: 
\begin{equation}\label{eq:D4_partial_unlink}
U(((12)3)4)U((12)4)U((12)3)U((23)4)U(24)U(23)U(12).
\end{equation}
Note, however, that the joint connector in this case computes only the partial wall-crossing formula -- one more step is required to reach the maximal chamber with 12 BPS states. 
This requires appending one more unlinking to (\ref{eq:D4_partial_unlink}), which is $U(12,(23)4)$. 
This unlinking introduces the BPS bound state corresponding to the last remaining positive root $\alpha_1+2\alpha_2+\alpha_3+\alpha_4$. 
This manual insertion was not needed in the case of $A_m$, because their roots do not come with multiplicities. 
(It would be very interesting to find a general description that avoids manual additions to joint connector class.)
After including this extra step, we get the unlinking graph which encodes the complete wall-crossing formula (Figure \ref{fig:D4_tile}). 
We can then apply the extension of Definition~\ref{dfn:symmetrization map} to derive all symmetric quivers corresponding to various chambers of $D_4$ Argyres-Douglas theory, by applying unlinking operators represented by paths on this diagram to $Q$ -- the initial symmetrization of $Q_{4\text{d}}$. 


\begin{figure}[h!]
    \centering
    \includegraphics[scale=0.75]{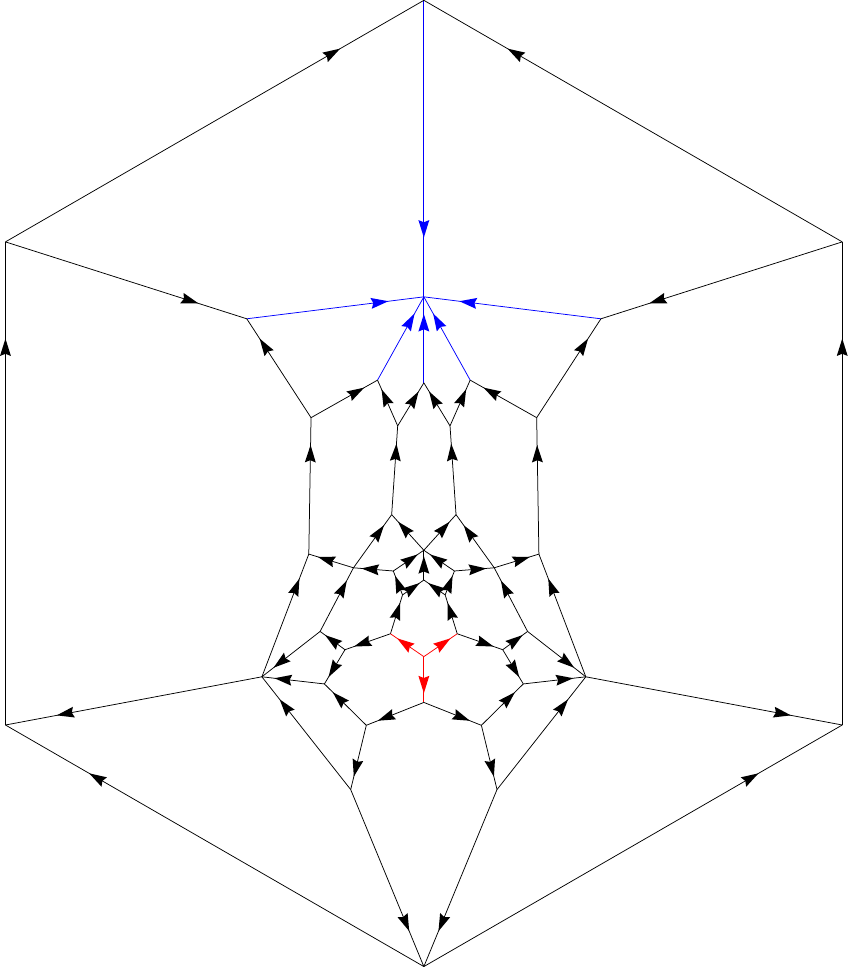}
    \caption{The unlinking graph which encodes the wall-crossing for ${D}_4$ quiver. 
    Initial symmetric quiver/minimal chambers are represented by a source with three initial unlinkings shown in \textcolor{red}{red}, while the final quiver/maximal chambers are given by a sink with six final unlinkings shown in \textcolor{blue}{blue}. 
    It takes 8 operations to reach the sink from the source, which simultaneously corresponds to the wall-crossings in 4d and unlinking in 3d. 
    Interestingly, much like in $A_4$ case, the resulting graph is planar.}
    \label{fig:D4_tile} 
\end{figure}

\section{Schur index and CPT-doubled symmetric quivers} \label{sec-schur}

In previous sections, we considered the relations between Coulomb branches of 4d $\mathcal{N}=2$ theories and 3d $\mathcal{N}=2$ theories, as well as the relations between corresponding quivers. 
We focused on the spectrum of BPS states of the 4d theory in a given chamber and discussed the corresponding geometry and manifestation of wall-crossing in 3d. 
In this section, we focus on a related quantity, which is chamber-independent and thus can also be associated to the superconformal point in the moduli space of 4d theory under consideration. 
This quantity is the trace of the Kontsevich-Soibelman operator $\mathcal{O}(q)$ from \eqref{Schur-intro}, with contributions from all particle and anti-particle states.
\newline

In \cite{Cordova:2015nma}, it was observed in many examples that, up to a simple factor, such a trace is equal to the Schur index \cite{Kinney:2005ej,Gadde:2011ik}, and a proof for class $S$ theories of type $A_1$ was given in \cite{Fluder:2019dpf}. 
The contributions to $\mathcal{O}(q)$ in question are quantum dilogarithms assigned to stable BPS states in a given chamber, analogously to (\ref{KS-intro}). 
However, it is now crucial to include all particle and anti-particle contributions in this trace, in order to avoid the ambiguity in their definition (related to the choice of the half-plane in which their central charges take values). 
More precisely, the claim in \cite{Cordova:2015nma} is that the Schur index is equal to
\begin{equation}
\mathcal{I}=(q^2;q^2)_{\infty}^{2r} \, \textrm{Tr}[\mathcal{O}(q)], \qquad \mathcal{O}(q) =  \prod^\curvearrowleft_{\alpha:\ \textrm{BPS \& anti-BPS states}} \Psi(-X_{\alpha}),  \label{Schur}
\end{equation}
where $(q,q)_{\infty}^2$ is a contribution from a $U(1)$ vector multiplets, and $r$ is the rank of the Coulomb branch. 

In this section, we show that the trace of the Kontsevich-Soibelman operator, and thus also the Schur index (\ref{Schur}), can be written as a motivic generating series (\ref{eq:motivic_generating_series}) for a certain symmetric quiver $Q$. 
More precisely, as a consequence of introducing anti-BPS states into the game, it is a CPT-doubled symmetric quiver arising from a 4d BPS quiver. 
The operation that produces such a quiver is the CPT-doubled symmetrization map (\ref{Sigma-maps})
\begin{equation}
(Q_{4\text{d}},\text{stab.\,data}) \stackrel{\mathfrak{S}^{CPT}}{\longmapsto}Q^{CPT}\, .
\end{equation}
The stability data that we consider in this context are those that produce the minimal number of stable BPS states in the 4d theory, i.e., only those that correspond to the nodes of the 4d BPS quiver. 
In this case, the CPT-doubled quivers that we find have the following properties: 
\begin{itemize}
\item They contain a subquiver which is a symmetrization of the 4d BPS quiver, i.e., for each arrow in the BPS quiver, there is an arrow in the opposite direction, analogously to the quivers we found in other sections. 
\item They contain a second set of nodes, corresponding to anti-BPS states. 
\item They contain extra arrows and loops, which arise from the trace operation in (\ref{Schur}). 
\end{itemize}
%
%
The parameters of the quiver generating series of CPT-doubled quivers are identified with specific powers of $q$ and combinations of fugacities in a given theory. 
We identify such symmetric quivers for a larger class of theories: various ADE-type Argyres-Douglas theories, the $SU(2)$ theory with and without matter, and others. 

In the rest of this section, we show that the Schur indices (\ref{Schur}) can be written in the form of the motivic generating series of symmetric quivers (\ref{eq:motivic_generating_series}). 
While this could be shown by applying in (\ref{Schur}) the homomorphism between the quantum torus algebra of $X_{\alpha}$ and that of $(\hat{x}_i,\hat{y}_j)$, and then computing an appropriately defined trace, we instead take advantage of the results derived already in \cite{Cordova:2015nma}. 
Using then the following identities for the $q$-Pochhammer
\begin{align}
(q^2;q^2)_{\infty} &= \sum_{k=0}^{\infty} \frac{(-1)^k q^{k(k+1)}}{(q^2;q^2)_k}, \label{qP-sum} \\
\frac{1}{(q^2;q^2)_n} & =  \frac{(q^{2(n+1)},q^2)_{\infty}}{(q^2;q^2)_{\infty}} = \frac{1}{(q^2;q^2)_{\infty}} \sum_{k=0}^{\infty} \frac{(-1)^k q^{k(k-1)}}{(q^2;q^2)_k} q^{2(n+1)k}, 
\label{inv-qP}
\end{align}
and some basic manipulations of $q$-series, 
we identify quivers encoding Schur indices of various theories considered in \cite{Cordova:2015nma}, as presented in the rest of this section. 
It is clear from our examples that a~few other indices determined in \cite{Cordova:2015nma} can also be presented in the form of quiver generating series.
Similar relations between the Schur index and the 3d half-index (which is related to symmetric quiver partition functions through the work of \cite{Ekholm:2018eee}) have been independently observed in recent works \cite{Prochazka, Gaiotto:2024ioj}. 

\subsection{Basic examples}

As a warm-up, let us find quivers encoding the indices in the simplest examples of a single vector multiplet, a free hypermultiplet, and QED. 
The quivers we find are shown in Figure \ref{fig-quiver-examples}.

\paragraph{$U(1)$ vector multiplet.}

The contribution to the index of a single $U(1)$ vector multiplet is just $(q^2;q^2)^2_{\infty}$. 
Using (\ref{qP-sum}), we thus get
\begin{equation}
\mathcal{I} = (q^2;q^2)_{\infty}^2 = \sum_{k,l=0}^{\infty} \frac{(-q)^{k^2+l^2} q^{k+l} }{(q^2;q^2)_k(q^2;q^2)_l}.
\end{equation}
This has the form of a quiver generating series for a quiver with two nodes, each with one loop and no other arrows: 
\begin{equation}
Q^{CPT}=\begin{bmatrix}
1 & 0 \\
0 & 1
\end{bmatrix}.
\end{equation}


\paragraph{Free hypermultiplet.}

The index for a free hypermultiplet with fugacity $z$ is \cite{Cordova:2015nma}:
\begin{equation}
\mathcal{I} 
= (q z;q^2)_\infty^{-1}(q z^{-1};q^2)_\infty^{-1}
= \sum_{k,l=0}^{\infty} \frac{ q^{k+l} z^{-k+l}}{(q^2;q^2)_k(q^2;q^2)_l}.
\end{equation}
This has the form of a quiver generating series for a quiver with two nodes and no arrows:
\begin{equation}
Q^{CPT}=\begin{bmatrix}
0 & 0 \\
0 & 0
\end{bmatrix}.
\end{equation}


\paragraph{QED.}

The index for QED reads \cite{Cordova:2015nma}:
\begin{equation}
\mathcal{I} = (q^2;q^2)_{\infty}^2 \sum_{l=0}^{\infty} \frac{q^{2l}}{(q^2;q^2)_l^2} = (q^2;q^2)_{\infty} \sum_{k,l=0}^{\infty} \frac{ (-q)^{k^2+2kl} q^{k+2l} }{(q^2;q^2)_k(q^2;q^2)_l}.
\end{equation}
Apart from $(q^2;q^2)_{\infty}$ prefactor, this has the form of a quiver generating series for a quiver with two nodes, one with a loop, and one pair of arrows between the two nodes:
\begin{equation}
Q^{CPT}=\begin{bmatrix}
0 & 1 \\
1 & 1
\end{bmatrix}.
\end{equation}
In fact, using (\ref{qP-sum}), the prefactor $(q^2;q^2)_{\infty}$ can be reinterpreted as an extra (third) node in the quiver (with one loop and disconnected with other nodes).

\begin{figure}[h!]
\begin{center}
\includegraphics[width=0.8\textwidth]{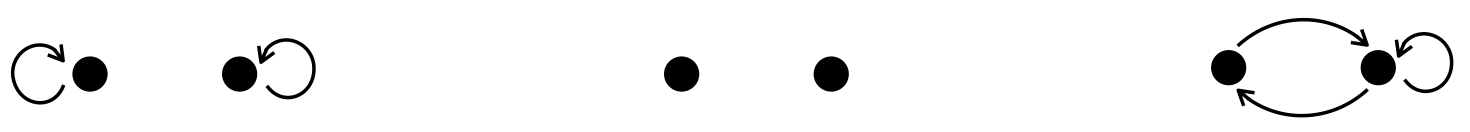}
\caption{Symmetric quivers encoding the Schur indices for a $U(1)$ vector multiplet (left), a free hypermultiplet (middle), and QED (right).}
\label{fig-quiver-examples}
\end{center}
\end{figure}

\subsection{\texorpdfstring{$SU(2)$}{SU2} theory with and without matter}

In turn, let us discuss the $SU(2)$ theory with various numbers of flavours and corresponding fugacities denoted by $z_i$, also taking advantage of the results derived in \cite{Cordova:2015nma}. 

\paragraph{Pure $SU(2)$.}

The Schur index for the pure $SU(2)$ theory takes the form
\begin{align}
\begin{split}
\mathcal{I} &= (q^2;q^2)_{\infty}^2 \sum_{l_1,l_2=0}^{\infty} \frac{q^{2l_1+2l_2+4l_1l_2}}{(q^2;q^2)_{l_1}^2 (q^2;q^2)_{l_2}^2}  = \\
& = \sum_{l_1,l_2,k_1,k_2=0}^{\infty}  \frac{ (-q)^{k_1^2+k_2^2+4l_1l_2+2l_1k_1+2l_2k_2}  }{(q^2;q^2)_{l_1}(q^2;q^2)_{l_2} (q^2;q^2)_{k_1} (q^2;q^2)_{k_2}} q^{2l_1+2l_2+k_1+k_2},
\end{split}
\end{align}
where we used (\ref{inv-qP}) twice and introduced summations over $k_1$ and $k_2$. 
The result has the form of a quiver generating series for a quiver with four nodes, which contains a subquiver (shown in \textcolor{red}{red}, with nodes corresponding to summations over $l_1$ and $l_2$) that is a symmetrized version of 4d quiver. 
Labelling the rows/columns respectively by $l_1,l_2,k_1,k_2$, we get
\begin{equation}
Q^{CPT}=\begin{bmatrix}
\textcolor{red}{0} & \textcolor{red}{2}  & 1 & 0 \\
\textcolor{red}{2}  & \textcolor{red}{0}  & 0 & 1 \\
1 & 0 & 1 & 0 \\
0 & 1 & 0 & 1 
\end{bmatrix}.
\end{equation}
The BPS quiver for this theory and corresponding symmetric quiver are shown in Figure \ref{fig-quiver-SU2}. 

\begin{figure}[h!]
\begin{center}
\includegraphics[width=0.6\textwidth]{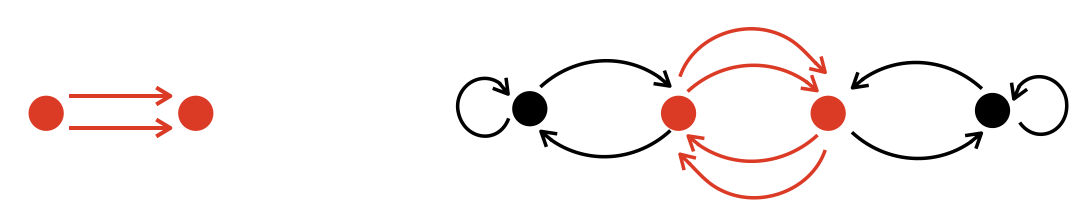}
\caption{BPS quiver for the pure $SU(2)$ theory (left), and a symmetric quiver encoding the Schur index of this theory (right). 
A subquiver of the symmetric quiver that involves a symmetrization of the BPS quiver is shown in red.}
\label{fig-quiver-SU2}
\end{center}
\end{figure}


\paragraph{$SU(2)$ with $N_f=1$.}

For the $SU(2)$ theory with $N_f=1$ flavour, the index reads
\begin{align}
\begin{split}
\mathcal{I} &= (q^2;q^2)_{\infty}^2 \sum_{l_1,l_2,l_3,k_1,k_2,k_3=0}^{\infty} \frac{ (-1)^{\sum_{i=1}^3 (k_i+l_i)}
q^{\sum_{i=1}^3 (k_i+l_i)  -(k_3-l_3)^2 + 2(l_1+l_3)k_2 + 2l_1k_3}}{\prod_{i=1}^3(q^2;q^2)_{k_i} (q^2;q^2)_{l_i}}  \times \\
& \qquad \times z^{2(k_3-l_3)} \delta_{k_1+k_3,l_1+l_3}  \cdot \delta_{k_2+k_3,l_2+l_3}.
\end{split}
\end{align}
From delta function constraints, we fix
\begin{equation}
l_1 = k_1+k_3-l_3,\qquad l_2=k_2+k_3-l_3. 
\end{equation}
We then rewrite terms $(q^2;q^2)_{l_1}^{-1} = (q^2;q^2)_{k_1+k_3-l_3}^{-1}$ and $(q^2;q^2)_{l_2}^{-1} = (q^2;q^2)_{k_2+k_3-l_3}^{-1}$ using (\ref{inv-qP}), introducing new summations respectively over $m$ and $n$, which leads to
\begin{align}
\begin{split}
\mathcal{I} & = \sum_{k_1,k_2,k_3,l_3,m,n=0}^{\infty}  \frac{ (-q)^{f(k_i,l_3,m,n)}  }{(q^2;q^2)_{l_3}(q^2;q^2)_{m} (q^2;q^2)_{n} \prod_{i=1}^3 (q^2;q^2)_{k_i}} q^{4k_1+4k_2+6k_3-2l_3+m+n} z^{2(k_3-l_3)}
\end{split}
\end{align}
where
\begin{equation}
f(k_i,l_3,m,n) = k_3^2 - l_3^2 + m^2 + n^2 + 2k_1 k_2 + 2k_1k_3 + 2 k_2 k_3 + 2m(k_1+k_3-l_3) + 2n(k_2+k_3-l_3).
\end{equation}
Thus, we get a quiver generating series for a quiver with six nodes, which contains a subquiver (shown in \textcolor{red}{red}, with nodes corresponding to summations over $k_i$) that is a symmetrized version of 4d quiver (ignoring an extra loop at node $k_3$). 
Labelling the rows/columns respectively by $k_1,k_2,k_3,l_3,m,n$, we get
\begin{equation}
Q^{CPT}=\begin{bmatrix}
\textcolor{red}{0} & \textcolor{red}{1}  & \textcolor{red}{1} & 0 & 1 & 0 \\
\textcolor{red}{1}  & \textcolor{red}{0}  & \textcolor{red}{1} & 0 & 0 & 1 \\
\textcolor{red}{1}  & \textcolor{red}{1}  & 1 & 0 & 1 & 1 \\
0 & 0 & 0 & -1 & -1 & -1 \\
1 & 0 & 1 & -1 & 1 & 0 \\
0 & 1 & 1 & -1 & 0 & 1 
\end{bmatrix}. \label{C-SU2f1}
\end{equation}
The BPS quiver for this theory and corresponding symmetric quiver are shown in Figure \ref{fig-quiver-SU2f1}.

\begin{figure}[h!]
\begin{center}
\includegraphics[width=0.5\textwidth]{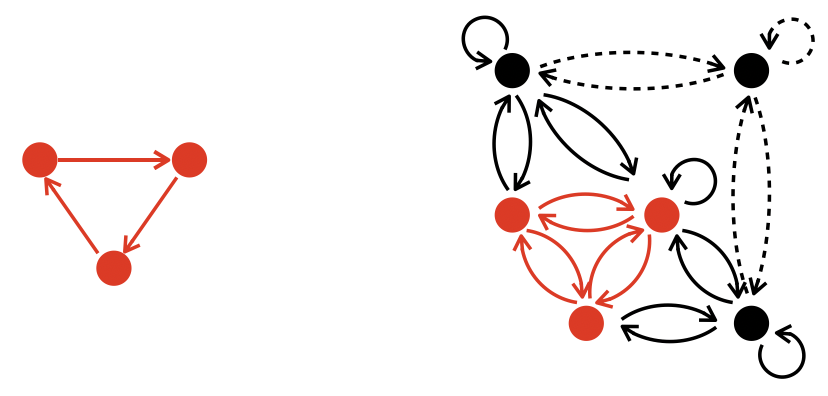}
\caption{BPS quiver for $SU(2)$ theory with $N_f=1$ flavour (left), and symmetric quiver encoding the Schur index of this theory (right). 
A subquiver of the symmetric quiver that involves a symmetrization of the BPS quiver is shown in red. 
Dashed arrows represent entries in the quiver matrix (\ref{C-SU2f1}) with negative values $-1$.}
\label{fig-quiver-SU2f1}
\end{center}
\end{figure}


\paragraph{$SU(2)$ with $N_f=2$.}

For the $SU(2)$ theory with $N_f=2$ flavours, the index reads
\begin{align}
\begin{split}
\mathcal{I} &= (q^2;q^2)_{\infty}^2 \sum_{l_1,\ldots,l_4,k_1,\ldots,k_4=0}^{\infty} \frac{
q^{2(l_1+\ldots+l_4) + 2(l_1+l_2)(l_3+l_4)}}{\prod_{i=1}^4(q^2;q^2)_{k_i} (q^2;q^2)_{l_i}}  \times \\
& \qquad \times z_1^{2(l_1-k_1)} z_2^{2(l_3-k_3)} \delta_{k_1+k_2,l_1+l_2} \cdot \delta_{k_3+k_4,l_3+l_4}.
\end{split}
\end{align}
From delta function constraints, we fix
\begin{equation}
l_1 = k_1+k_2-l_2,\qquad l_3=k_3+k_4-l_4.
\end{equation}
We then rewrite terms $(q^2;q^2)_{l_1}^{-1} = (q^2;q^2)_{k_1+k_2-l_2}^{-1}$ and $(q^2;q^2)_{l_3}^{-1} = (q^2;q^2)_{k_3+k_4-l_4}^{-1}$ using (\ref{inv-qP}), introducing new summations respectively over $m$ and $n$, which leads to
\begin{align}
\begin{split}
\mathcal{I} & = \sum_{k_1,\ldots,k_4,l_2,l_4,m,n=0}^{\infty}  \frac{ (-q)^{f(k_i,l_2,l_4,m,n)}  q^{2(k_1+\ldots +k_4)+m+n} }{(q^2;q^2)_{l_2} (q^2;q^2)_{l_4}(q^2;q^2)_{m} (q^2;q^2)_{n} \prod_{i=1}^4 (q^2;q^2)_{k_i}}  z_1^{2(k_2-l_2)} z_2^{2(k_4-l_4)} 
\end{split}
\end{align}
where
\begin{equation}
f(k_i,l_2,l_4,m,n)= m^2 + n^2 + 2(k_1 k_3 + k_1k_4 + k_2 k_3 + l_2k_4) + 2m(k_1+k_2-l_2) + 2n(k_3+k_4-l_4)
\end{equation}
Thus we get a quiver generating series for a quiver with eight nodes, which contains a subquiver (shown in \textcolor{red}{red}, with nodes corresponding to summations over $k_i$) that is a symmetrized version of 4d quiver. 
Labelling the rows/columns respectively by $k_1,k_2,k_3,k_4,l_2,l_4,m,n$ we get
\begin{equation}
Q^{CPT}=\begin{bmatrix}
\textcolor{red}{0} & \textcolor{red}{0}  & \textcolor{red}{1} & \textcolor{red}{1} & 0 & 0 & 1 & 0 \\
\textcolor{red}{0}  & \textcolor{red}{0}  & \textcolor{red}{1} & \textcolor{red}{1} & 0 & 0 & 1 & 0 \\
\textcolor{red}{1}  & \textcolor{red}{1}  & \textcolor{red}{0} & \textcolor{red}{0}  & 0 & 0 & 0 & 1 \\
\textcolor{red}{1}  & \textcolor{red}{1}  & \textcolor{red}{0} & \textcolor{red}{0}  & 0 & 0 & 0 & 1 \\
0 & 0 & 0 & 0 & 0 & 0 & -1 & 0 \\
0 & 0 & 0 & 0 & 0 & 0 & 0 & -1  \\
1  & 1 & 0 & 0 & -1 & 0 & 1 & 0 \\
0 & 0 & 1 & 1 & 0 & -1 & 0 & 1
\end{bmatrix}
\end{equation}


\paragraph{$SU(2)$ with $N_f=3$.}

For the $SU(2)$ theory with $N_f=3$ flavours, the index takes form
\begin{align}
\begin{split}
\mathcal{I} &= (q^2;q^2)_{\infty}^2 \sum_{l_1,\ldots,l_5,k_1,\ldots,k_5=0}^{\infty} \frac{
q^{2(l_1+\ldots+l_5) + 2l_1(l_2+\ldots+l_5)}}{\prod_{i=1}^5(q^2;q^2)_{k_i} (q^2;q^2)_{l_i}}  \times \\
& \qquad \times z_1^{(k_2-l_2)+(l_5-k_5)} z_2^{(l_2-k_2)+(k_3-l_3)} z_3^{(l_3-k_3)+(k_4-l_4)}  \delta_{k_1,l_1} \cdot \delta_{k_2+\ldots+k_5,l_2+\ldots+l_5}.
\end{split}
\end{align}
From delta function constraints, we fix
\begin{equation}
k_1 = l_1,\qquad k_2=(l_2+\ldots+l_5)-(k_3+k_4+k_5).
\end{equation}
We then rewrite terms $(q^2;q^2)_{k_1}^{-1} = (q^2;q^2)_{l_1}^{-1}$ and $(q^2;q^2)_{k_2}^{-1} = (q^2;q^2)_{(l_2+\ldots+l_5)-(k_3+k_4+k_5)}^{-1}$ using (\ref{inv-qP}), introducing new summations respectively over $m$ and $n$, which leads to
\begin{align}
\begin{split}
\mathcal{I} & = \sum_{l_1,\ldots,l_5,k_3,k_4,k_5,m,n=0}^{\infty}  \frac{ (-q)^{f(l_i,k_3,k_4,k_5,m,n)}  q^{2(l_1+\ldots +l_5)+m+n} }{(q^2;q^2)_{k_3} (q^2;q^2)_{k_4} (q^2;q^2)_{k_5} (q^2;q^2)_{m} (q^2;q^2)_{n} \prod_{i=1}^5 (q^2;q^2)_{l_i}} \times \\ 
&\qquad \times z_1^{l_3+l_4+2l_5-k_3-k_4-2k_5} z_2^{-2l_3-l_4-l_5+2k_3+k_4+k_5} z_3^{(l_3-k_3)+(k_4-l_4)}
\end{split}
\end{align}
where
\begin{equation}
f(l_i,k_3,k_4,k_5,m,n)= m^2 + n^2 + 2l_1(l_2 +\ldots + l_5) + 2ml_1 + 2n(l_2+\ldots+l_5-k_3-k_4-k_5)
\end{equation}
Thus we get a quiver generating series for a quiver with ten nodes, which contains a subquiver (shown in \textcolor{red}{red}, with nodes corresponding to summations over $l_i$) that is a symmetrized version of 4d quiver. 
Labelling the rows/columns respectively by $l_1,\ldots,l_5,k_3,k_4,k_5,m,n$, we get
\begin{equation}
Q^{CPT}=\begin{bmatrix}
\textcolor{red}{0} & \textcolor{red}{1}  & \textcolor{red}{1} & \textcolor{red}{1} & \textcolor{red}{1} & 0 & 0 & 0 & 1 & 0 \\
\textcolor{red}{1} & \textcolor{red}{0} & \textcolor{red}{0} & \textcolor{red}{0} & \textcolor{red}{0} & 0 & 0 & 0 & 0 & 1 \\
\textcolor{red}{1} & \textcolor{red}{0} & \textcolor{red}{0} & \textcolor{red}{0} & \textcolor{red}{0} & 0 & 0 & 0 & 0 & 1 \\
\textcolor{red}{1} & \textcolor{red}{0} & \textcolor{red}{0} & \textcolor{red}{0} & \textcolor{red}{0} & 0 & 0 & 0 & 0 & 1 \\
\textcolor{red}{1} & \textcolor{red}{0} & \textcolor{red}{0} & \textcolor{red}{0} & \textcolor{red}{0} & 0 & 0 & 0 & 0 & 1 \\
0 & 0 & 0 & 0 & 0 & 0 & 0 & 0 & 0 & -1  \\
0 & 0 & 0 & 0 & 0 & 0 & 0 & 0 & 0 & -1  \\
0 & 0 & 0 & 0 & 0 & 0 & 0 & 0 & 0 & -1  \\
1 & 0 & 0 & 0 & 0 & 0 & 0 & 0 & 1 & 0 \\  
0 & 1 & 1 & 1 & 1 & -1 & -1 & -1 & 0 & 1  
\end{bmatrix}
\end{equation}

\subsection{Type \texorpdfstring{$A$}{A} theories}

As another class of examples, we consider the $A_m$ Argyres-Douglas theories with even $m$. 
Examples with odd $m$ can be worked out analogously.

\paragraph{$A_2$ theory.}

Following analogous manipulations as above, the index for $A_2$ theory can be rewritten as
\begin{align}
\begin{split}
\mathcal{I} &= (q^2;q^2)_{\infty}^2 \sum_{l_1,l_2=0}^{\infty} \frac{q^{2l_1+2l_2+2l_1l_2}}{(q^2;q^2)_{l_1}^2 (q^2;q^2)_{l_2}^2}  = \\
& = \sum_{l_1,l_2,m,n=0}^{\infty}  \frac{ (-q)^{m^2+n^2+2ml_1 + 2nl_2 + 2l_1l_2}  }{(q^2;q^2)_{l_1}(q^2;q^2)_{l_2} (q^2;q^2)_{m} (q^2;q^2)_{n}} q^{2l_1+2l_2+m+n}
\end{split}
\end{align}
where we used (\ref{inv-qP}) twice and introduced summations over $m$ and $n$. 
The result has the form of a quiver generating series for a quiver with four nodes, which contains a subquiver (shown in \textcolor{red}{red}, with nodes corresponding to summations over $l_1$ and $l_2$) that is a symmetrized version of 4d quiver. 
Labelling the rows/columns respectively by $l_1,l_2,m,n$, we get
\begin{equation}
Q^{CPT}=\begin{bmatrix}
\textcolor{red}{0} & \textcolor{red}{1}  & 1 & 0 \\
\textcolor{red}{1}  & \textcolor{red}{0}  & 0 & 1 \\
1 & 0 & 1 & 0 \\
0 & 1 & 0 & 1 
\end{bmatrix}
\end{equation}
The BPS quiver for this theory and corresponding symmetric quiver are shown in Figure \ref{fig-quiver-A2}.

\begin{figure}[h!]
\begin{center}
\includegraphics[width=0.6\textwidth]{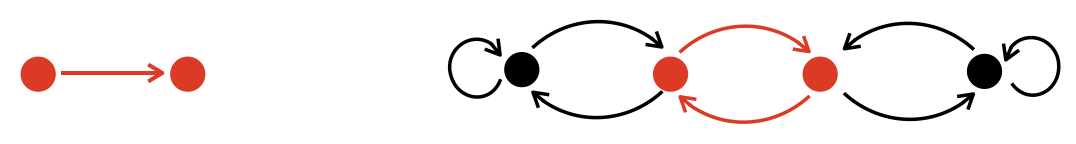}
\caption{BPS quiver for $A_2$ Argyres-Douglas theory (left), and symmetric quiver encoding the Schur index of this theory (right). 
A subquiver of the symmetric quiver that involves a symmetrization of the BPS quiver is shown in red.}
\label{fig-quiver-A2}
\end{center}
\end{figure}


\paragraph{$A_4$ theory.}

For $A_4$ Argyres-Douglas theory, we obtain
\begin{align}
\begin{split}
\mathcal{I} &= (q^2;q^2)_{\infty}^4 \sum_{l_1,l_2,l_3,l_4=0}^{\infty} \frac{q^{2(l_1+l_2+l_3+l_4)+2l_1l_2+2l_2l_3+2l_3l_4}}{(q^2;q^2)_{l_1}^2 (q^2;q^2)_{l_2}^2 (q^2;q^2)_{l_3}^2 (q^2;q^2)_{l_4}^2}  = \\
& = \sum_{l_1,l_2,l_3,l_4m,n,r,s=0}^{\infty}  \frac{ (-q)^{f(l_i,m,n,r,s)}  }{(q^2;q^2)_{m} (q^2;q^2)_{n} (q^2;q^2)_{r} (q^2;q^2)_{s} \prod_{i=1}^4 (q^2;q^2)_{l_i}} q^{2(l_1+l_2+l_3+l_4)+m+n+r+s}
\end{split}
\end{align}
where
\begin{equation}
f(l_i,m,n,r,s) = m^2+n^2+r^2+s^2+2l_1l_2 + 2l_2l_3 + 2l_3l_4 + 2ml_1 + 2nl_2 + 2rl_3+2sl_4.
\end{equation}
The result has the form of a quiver generating series for a quiver with eight nodes, which contains a subquiver (shown in \textcolor{red}{red}, with nodes corresponding to summations over $l_1,\ldots,l_4$) that is a symmetrized version of 4d quiver. 
Labelling the rows/columns respectively by $l_1,l_2,l_3,l_4,m,n,r,s$ we get
\begin{equation}
Q^{CPT}=\begin{bmatrix}
\textcolor{red}{0} & \textcolor{red}{1}  & \textcolor{red}{0}  & \textcolor{red}{0}  & 1 & 0 & 0 & 0 \\
\textcolor{red}{1}  & \textcolor{red}{0}  & \textcolor{red}{1}  & \textcolor{red}{0} & 0 & 1 & 0 & 0 \\
\textcolor{red}{0} & \textcolor{red}{1}  & \textcolor{red}{0}  & \textcolor{red}{1}  & 0 & 0 & 1 & 0 \\
\textcolor{red}{0} & \textcolor{red}{0}  & \textcolor{red}{1}  & \textcolor{red}{0}  & 0 & 0 & 0 & 1 \\
1 & 0 & 0 & 0 & 1 & 0 & 0 & 0 \\ 
0 & 1 & 0 & 0 & 0 & 1 & 0 & 0 \\ 
0 & 0 & 1 & 0 & 0 & 0 & 1 & 0 \\ 
0 & 0 & 0 & 1 & 0 & 0 & 0 & 1  
\end{bmatrix}   \label{C4-quiver}
\end{equation}
The BPS quiver for this theory and corresponding symmetric quiver are shown in Figure \ref{fig-quiver-A4}.

\begin{figure}[h!]
\begin{center}
\includegraphics[width=0.8\textwidth]{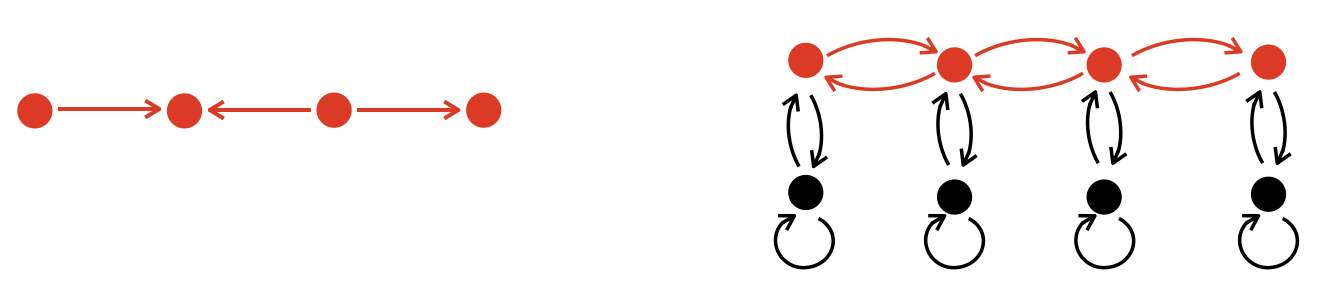}
\caption{BPS quiver for $A_4$ Argyres-Douglas theory (left), and symmetric quiver encoding the Schur index of this theory (right). 
A subquiver of the symmetric quiver that involves a symmetrization of the BPS quiver is shown in red.}
\label{fig-quiver-A4}
\end{center}
\end{figure}


\paragraph{$A_{2n}$ theory.}

The index for $A_{2n}$ with arbitrary $n$ takes the form
\begin{align}
\begin{split}
\mathcal{I} = (q^2;q^2)_{\infty}^{2n} \sum_{l_1,\ldots,l_{2n}=0}^{\infty} \frac{q^{2\sum_{i=1}^{2n-1}l_i l_{i+1} + 2\sum_{i=1}^{2n} l_i}}{\prod_{i=1}^{2n}(q^2;q^2)_{l_i}^2 }  .    \label{I-A2n}
\end{split}
\end{align}
Applying (\ref{inv-qP}) for each $l_i$, i.e., 
\begin{equation}
\frac{1}{(q^2;q^2)_{l_i}}  = \frac{1}{(q^2;q^2)_{\infty}} \sum_{k_i=0}^{\infty} \frac{(-1)^{k_i} q^{k_i(k_i-1)}}{(q^2;q^2)_{k_i}} q^{2(l_i+1)k_i},  \label{inv-qP-i}
\end{equation}
we get
\begin{equation}
\mathcal{I} = \sum_{l_1,\ldots,l_{2n},k_1,\ldots,k_{2n}=0}^{\infty}  \frac{ (-q)^{\sum_{i=1}^{2n} (k_i^2+2l_ik_i) + \sum_{i=1}^{2n-1}l_il_{i+1}}  }{ \prod_{i=1}^{2n} (q^2;q^2)_{l_i} (q^2;q^2)_{k_i}} q^{\sum_{i=1}^{2n}(2l_i + k_i)}.
\end{equation}
The result has the form of a quiver generating series for a quiver with $4n$ nodes, whose matrix takes form of a straightforward generalization of (\ref{C4-quiver}), with rows/columns labeled respectively by $l_1,\ldots,l_{2n},k_1,\ldots,k_{2n}$. 
It contains a subquiver (with nodes corresponding to summations over $l_1,\ldots,l_{2n}$) that is a symmetrized version of 4d quiver.


\subsection{Type \texorpdfstring{$E$}{E} theories}

Finally, we consider two examples of Argyres-Douglas theories of type $E$.

\paragraph{$E_6$ theory.}

The index of $E_6$ theory takes the form
\begin{align}
\begin{split}
\mathcal{I} = (q^2;q^2)_{\infty}^{6} \sum_{l_1,\ldots,l_{6}=0}^{\infty} \frac{q^{\sum_{i,j=1}^{6 }b_{ij} l_i l_{j} + 2\sum_{i=1}^{6} l_i}}{\prod_{i=1}^{6}(q^2;q^2)_{l_i}^2 }  ,
\end{split}
\end{align}
where $b_{ij}=-C^{E_6}_{ij}+2\delta_{ij}$ and $C^{E_6}$ is $E_6$ Cartan matrix. 
This is analogous to $A_{2n}$ case in (\ref{I-A2n}), so that applying (\ref{inv-qP-i}) to each $l_i$, we get
\begin{align}
\begin{split}
\mathcal{I} = \sum_{l_1,\ldots,l_{6},k_1,\ldots,k_{6}=0}^{\infty}  \frac{ (-q)^{\sum_{i=1}^{6} (k_i^2+2l_ik_i) + \sum_{i,j=1}^{6}b_{ij}l_il_{j}}  }{ \prod_{i=1}^{6} (q^2;q^2)_{l_i} (q^2;q^2)_{k_i}} q^{\sum_{i=1}^{6}(2l_i + k_i)}.
\end{split}
\end{align}
The result has the form of a quiver generating series for a quiver with $12$ nodes, which contains a subquiver (with nodes corresponding to summations over $l_1,\ldots,l_{6}$) that is a symmetrized version of $E_6$ BPS quiver.


\paragraph{$E_8$ theory.}

The index of $E_8$ theory takes the form
\begin{align}
\begin{split}
\mathcal{I} = (q^2;q^2)_{\infty}^{8} \sum_{l_1,\ldots,l_{8}=0}^{\infty} \frac{q^{\sum_{i,j=1}^{8}b_{ij} l_i l_{j} + 2\sum_{i=1}^{8} l_i}}{\prod_{i=1}^{8}(q^2;q^2)_{l_i}^2 }  ,
\end{split}
\end{align}
where $b_{ij}=-C^{E_8}_{ij}+2\delta_{ij}$ and $C^{E_8}$ is $E_8$ Cartan matrix. 
This is completely analogous to $E_6$ case and applying (\ref{inv-qP-i}) to each $l_i$, we get
\begin{align}
\begin{split}
\mathcal{I} = \sum_{l_1,\ldots,l_{8},k_1,\ldots,k_{8}=0}^{\infty}  \frac{ (-q)^{\sum_{i=1}^{8} (k_i^2+2l_ik_i) + \sum_{i,j=1}^{8}b_{ij}l_il_{j}}  }{ \prod_{i=1}^{8} (q^2;q^2)_{l_i} (q^2;q^2)_{k_i}} q^{\sum_{i=1}^{8}(2l_i + k_i)}.
\end{split}
\end{align}
The result has the form of a quiver generating series for a quiver with $16$ nodes, which contains a subquiver (with nodes corresponding to summations over $l_1,\ldots,l_{8}$) that is a symmetrized version of $E_8$ BPS quiver.

\section{Directions for future work}\label{sec:future-work}

Our work reveals an intriguing unifying role of symmetric quivers and the 3d $\mathcal{N}=2$ theories $T[Q]$ they encode. 
We showed that symmetric quivers capture wall-crossing phenomena and various observables of 4d $\mathcal{N}=2$ theories. 
While most of our analysis focused on (some classes of) Argyres-Douglas theories, we believe that this relation is more general and opens an interesting direction of research. 

First, it is important to generalize the topological and physical interpretation of the symmetrization map between 3d and 4d quivers, that we provide for $A_m$ Argyres-Douglas theories, to other classes of such theories, for which we have stated these relations only at the algebraic or combinatorial level. 

In particular, we expect our construction of 3d $\mathcal{N}=2$ QFTs from 4d BPS states to directly generalize to all class $S$ theories with a finite BPS chamber. That is, we expect that the resulting 3d QFT admits a symmetric quiver description $T[Q]$, although in general the relation between the 3d quiver $Q$ and the originating 4d BPS $Q_{4d}$ may be more involved.

Another interesting generalization concerns higher rank class $S$ theories, inparticular the $K$-lifts of $A_1$ theories defined in \cite{Gaiotto:2012db}. For these the corresponding 3d QFTs and 3-manifolds should be provided by the work of \cite{Dimofte:2013iv}, see also \cite{Gabella:2017hpz, Gang:2017ojg} for the related BPS quivers and mutations.
Similarly, the isomorphism between the structure of wall-crossing and the structure of unlinkings could be generalized to other classes of theories -- in particular to 4d quivers with superpotentials and to a study of mutations in this context. Of particular interest are also wild wall-crossing phenomena. Specifically, their description for Kronecker quivers and corresponding symmetric quivers has been recently found in \cite{Bryan:2025mwi}; an interesting challenge is to understand other classes of wild wall-crossing phenomena and to generalize geometric constructions presented in this work also in this context.

On the other hand, it is known that the Schur indices can be also identified as characters of 2d conformal field theories. 
This implies a connection between such 2d models and 3d $T[Q]$ theories, which also deserves better understanding, potential generalization to a~broader family of conformal theories, and possibly making contact with other recently found connections between 3d and 4d theories and vertex operator algebras \cite{Cheng:2022rqr,Dedushenko:2023cvd}. 

Our results may have impact on other lines of research too, both in mathematis and physics. Among others, they add up to the developments in the knots-quivers correspondence, which relates symmetric quivers to observables in Chern-Simons theory via the geometric transition \cite{KRSS1707short,KRSS1707long,Ekholm:2018eee}, and to analogous connections of symmetric quivers to topological strings \cite{Panfil:2018faz,Kimura:2020qns,Cheng:2021nex}, and  $\widehat{Z}$ and $F_K$ invariants of knot complements via 3d-3d correspondence \cite{Kucharski:2020rsp,Ekholm:2021irc,Cheng:2022rqr,Chung:2023qth}. 

Another relation that may emerge from our results connects the above mentioned systems and a vast reasearch area of topological recursion. It was shown in \cite{Larraguivel:2020sxk} that asymptotic expansion of a quiver generating series can be reproduced by the topological recursion. Our results should then imply that topological recursion captures analogous expansions of observables in 4d $\mathcal{N}=2$ theories, characters of 2d conformal field theories, and other systems that we relate to 3d theories encoded in symmetric quivers.

From yet another perspective, our work may lead to new results or viewpoints on holography. Typically, holographic considerations involve theories of higher rank, or with some other dependence on an appropriate (large) parameter $N$. In the context of Argyres-Douglas theories, the quiver form of Schur indices that we analyzed in section \ref{sec-schur} can be found also for related theories considered \cite{Cordova:2015nma}, in particular for Argyres-Douglas theories of type $(A_N,A_k)$. On the other hand, holographic duals of such theories have been found e.g. in \cite{Bah:2021mzw}. The structure of quivers that arise for such Argyres-Douglas theories and properties of associated integral Donaldson-Thomas invariants thus might encode certain features of their holographic dual geometries. Similarly, it turns out \cite{PurSul} that a quiver structure arises also for Schur indices of gauge theories whose holographic duals admit giant graviton expansion \cite{McGreevy:2000cw,Murthy:2022ien}. Therefore symmetric quivers and associated 3d $\mathcal{N}=2$ theories may play an interesting role in the realm of holography and AdS/CFT correspondence too.

We believe there should be a deeper reason for the role of symmetric quivers in all the contexts mentioned above, which deserves further investigation.


\section*{Acknowledgements}

We thank H\'{e}lder Larragu\'{i}vel for collaboration in initial stages of this project.
The work of P.K. is supported by the Polish National Science Centre through SONATA grant (2022/47/D/ST2/02058). 
The work of P.L. is supported by the Knut and Alice Wallenberg Foundation grant KAW 2020.0307 and by the Swedish Research Council, VR 2022-06593, Centre of Excellence in Geometry and Physics at Uppsala University.
S.P. gratefully acknowledges support from Simons Foundation through Simons Collaboration on Global Categorical Symmetries. 
The work of P.S. has been supported by the OPUS grant no. 2022/47/B/ST2/03313 ``Quantum geometry and BPS states'' funded by the National Science Centre, Poland. 

\appendix
\section{Notations and conventions}\label{sec:notations}

\begin{small} 
    \begin{list}{}{ \itemsep -1pt \labelwidth 23ex \leftmargin 13ex } 

      \item[$C$\;\;--] a Riemann surface

      \item[$\Sigma_\tau$\;\;--] the branched double cover of $C$ associated with ideal triangulation $\tau$

      \item[$I$\;\;--] the unit interval $[0,1]$

      \item[$\Sk^{\mathfrak{gl}_1}_q(Y)$\;\;--] the \emph{$\mathfrak{gl}_1$-skein module} of a 3-manifold $Y$ decorated with a branch locus (Definition \ref{defn:gl1_skein})

      \item[$\widehat{\Sk}^{\mathfrak{gl}_1}_q(Y)$\;\;--] action completion of the $\mathfrak{gl}_1$-skein module

      \item[$\SkAlg^{\mathfrak{gl}_1}_q(\Sigma)$\;\;--] the skein module $\Sk^{\mathfrak{gl}_1}_q(\Sigma \times I)$, with the natural algebra structure

      \item[$\hat{x},\hat{y}$\;\;--] the elements of the $\mathfrak{gl}_1$-skein algebra of the torus with four branch points \eqref{eq:xhat-yhat-A1}. They generate quantum torus and satisfy $\hat{y}\hat{x} = q^2\hat{x}\hat{y}$

      \item[$(a;q)_n$\;\;:=] $\prod_{r=0}^{n-1}(1-aq^r)$, $(a;q)_\infty:=\prod_{r\ge0}(1-aq^r)$, $q$-Pochhammer symbols

      \item[$\Psi(x)$\;\;:=] $(qx;q^2)^{-1}$, the quantum dilogarithm

      \item[$\widehat{\mathbb{C}(q)}$\;\;=] $ \left\{ \sum_{n\geq N}a_nq^n\ :\ a_n\in\mathbb{C},N\in\mathbb{Z} \right\}$, the field of formal Laurent series with coefficients in $\mathbb{C}$

      \item[$Z$\;\;--] the $\mathfrak{gl}_1$-skein-valued partition function constructed via sequence of flips; the distinguished element of $\Sk^{\mathfrak{gl}_1}_q(Y)$.

      \item[$Q$\;\;--] \emph{symmetric quiver} encoded by a symmetric integer matrix $Q=(Q_{ij})$; $Q_{ij}$ counts arrows between $i\neq j$ and $Q_{ii}$ counts loops

      \item[$Y_{in},Y_{out}$\;\;--] the cup/cap bordisms. $Y_{in}:\emptyset\rightarrow \Sigma_r$, $Y_{out}:\Sigma_{r'}\rightarrow\emptyset$. Here $Y:\Sigma_r\rightarrow \Sigma_{r'}$ is viewed as a bordism as well

      \item[$\ket{d_1, \cdots, d_g}$\;\;:=] $ L_{\nu_1}^{d_1} \cdots L_{\nu_g}^{d_g}[\emptyset]\in \Sk_{q}^{\mathfrak{gl}_1}(Y_{in}),\quad d_1, \cdots, d_g \in \mathbb{Z}$ the basis of states for $\Sk_{q}^{\mathfrak{gl}_1}(Y_{in})$ viewed as a vector space (similarly, $\bra{d_1, \cdots, d_g}$ for $\Sk_{q}^{\mathfrak{gl}_1}(Y_{out})$)

      \item[{$T[M]$}\;\;--] 3d $\mathcal{N}=2$ theory associated to a 3-manifold $M$ via the 3d-3d correspondence. In our work we consider splitting $M=M_+ \cup M_0 \cup M_-$

      \item[{$T[C]$}\;\;--] 4d $\mathcal{N}=2$ theory of class $S$ associated to a Riemann surface $C$. We take $\partial M_-=C$, $\partial M_+ = \overline{C}$, $\partial M_0=C\cup \overline{C}$ where $\overline{C}$ is orientation reversal of $C$.

      \item[{$T[Q]$}\;\;--] 3d $\mathcal{N}=2$ theory associated to a symmetric quiver $Q$

      \item[$\mathcal{P}_S$\;\;:=] $\mathcal{M}_{\rm flat}(S,SL(2,\mathbb{C}))$ the character variety of flat $SL(2,\mathbb{C})$ connections on $S=\partial M$

      \item[$\Xi_E$\;\;--] $\in \IC / 2\pi i \IZ$ sheer coordinate corresponding to an edge of triangulation $E$ in $\tau_S$, triangulated boundary of $M$

      \item[$\{X_i,Y_i\}$\;\;--] the Darboux basis corresponding to a choice of polarisation of $\mathcal{P}_S$, i.e. a set of coordinates with canonical Poisson brackets \eqref{eq:Darboux}

      \item[$\mathcal{T}_\Delta$\;\;--] the tetrahedron theory \eqref{eq:tetrahedron-theory}. For a triangulated 3-manifold, we write the product theory as $\otimes_i \mathcal{T}_{\Delta_i}$  

      \item[$\mathcal{W}$\;\;--] superpotential of $T[M]$

      \item[$\alpha_1,\dots,\alpha_m$\;\;--] IR electromagnetic charges of $T[C]=A_m$ theory with Dirac pairing $\langle \alpha_{2i},\alpha_{2i+1}\rangle = \langle \alpha_{2i},\alpha_{2i-1}\rangle = 1$

      \item[$Q_{\rm 4d}$\;\;--] 4d BPS quiver corresponding to the minimal chamber of $T[C]$

      \item[$\operatorname{lk}(\cdot,\cdot)$\;\;--] Gauss linking number. In case of the double covering manifold $L_0$ with $\partial L_0=\Sigma_+\cup \Sigma_-$, this linking number between (the boundaries of) holomorphic discs $\partial D_i$ coincides with $Q_{ij}$ 

      \item[$\hat{x}_i,\hat{y}_i$\;\;--] generators of quantum tori associated to a symmetric quiver $Q$, see the discussion around \eqref{eq:xi-hat-yi-hat-def} 

      \item[$P_Q(\boldsymbol{x},q)$\;\;--] motivic generating series of a symmetric quiver $Q$ \eqref{eq:Efimov-PQ}. Here we denote $\boldsymbol{x}=(x_1,\dots,x_{|Q_0|})$ where $Q_0$ is the set of nodes of $Q$.

      \item[$\epsilon$\;\;--] the \emph{3d-4d homomorphism} which embeds the quantum torus of rank $m$ generated by $\{X_{\alpha_i}\}$ into the quantum torus algebra of rank $2m$ generated by $\{\hat{x}_i^{\pm},\hat{y}_i^{\pm}\}$ \eqref{eq:embedding_epsilon}

      \item[$\Phi^0,\Phi^+$\;\;--] the sets of simple and positive roots of $A_m$ root system. Here we use the same notation for roots $\alpha_i$ as for BPS charges 

      \item[$\mu_{\alpha_i}$\;\;--] mutation operator at node $\alpha_i$ of $Q_{4\text{d}}$ 

      \item[$U(ij)$\;\;--] unlinking operator between nodes $i$ and $j$ ($i\neq j$). It acts on symmetric quivers: $U(ij)Q$, \eqref{eq:unlinking_definition}. We also denote composition of unlinkings $\boldsymbol{U}=U(i_nj_n)\dots U(i_1j_1)$ 

      \item[$\mathbbm{\Lambda}(\boldsymbol{U},\boldsymbol{U}';H)$\;\;--] connector class of $\boldsymbol{U},\boldsymbol{U}'$ with hexagon specification $H$, Definition \ref{dfn: Lambda class}

      \item[$\Qsym{Q_{4\text{d}}}{\Theta}$\;\;=] $Q$ the symmetrization map applied to $Q_{4\text{d}}$ with $T[C]$ BPS chamber $\Theta$, Definition \ref{dfn:symmetrization map}. By definition, different choice of a chamber produces different symmetric quivers for the same $Q_{4\text{d}}$

      \item[$\overrightarrow{EG}(Q_{4\text{d}})$\;\;--] the oriented exchange graph of $Q_{4\text{d}}$

      \item[$\Lambda(\overrightarrow{G})$\;\;--] the path polytope of oriented graph $\overrightarrow{G}$. In particular, $\Lambda(\overrightarrow{EG}(Q_{4\text{d}}))$ gives a combinatorial description of the symmetrization map (Definition \ref{dfn:Lambda polytope})

      \item[$\mathcal{I}$\;\;--] Schur index of a 4d $\mathcal{N}=2$ theory
      
    \end{list}
\end{small}

\bibliography{ref}
\bibliographystyle{hep}

\end{document}